\pdfoutput=1
\documentclass{ucetd}
\usepackage{epsfig,amsfonts}
\usepackage[numbers]{natbib}
\usepackage{amsmath}
\usepackage{amssymb}
\usepackage{amsthm}
\usepackage{comment}
\usepackage{subcaption}

\def\showauthornotes{1}
\ifnum\showauthornotes=1
\newcommand{\gnote}[1]{\textcolor{magenta}{ {\textbf{(Goutham: #1)}}}}
\else
\newcommand{\gnote}[1]{}
\fi

\title{Nonlinear random matrices and applications to the Sum of Squares hierarchy}
\author{Goutham Rajendran}
\department{Computer Science}
\division{Physical Sciences}
\degree{Doctor of Philosophy}
\date{August 2022}

\dedication{\textit{Dedicated to the memory of my dad, Dr. A. S. Rajendran (1958--2017),\\
who instilled in me the value of education and the pursuit of knowledge\\
and made me into the person that I am now}}

\usepackage[pdfusetitle]{hyperref}
\hypersetup{unicode=true,
    linktoc=all,
    pdfsubject=subject here,
    pdfkeywords=keyword1 keyword2 keyword3,
    pdfborder={0 0 0},
    breaklinks=true}
\makeatletter
\let\ORG@hyper@linkstart\hyper@linkstart
\protected\def\hyper@linkstart#1#2{%
    \lowercase{\ORG@hyper@linkstart{#1}{#2}}}
\makeatother

\usepackage{xspace,xcolor,enumerate}
\usepackage{thmtools}
\usepackage{thm-restate}
\usepackage[capitalise,nameinlink]{cleveref}
\usepackage{mathtools}
\newtheorem{theorem}{Theorem}[section]
\newtheorem{conjecture}[theorem]{Conjecture}

\newtheorem{definition}[theorem]{Definition}

\newtheorem{lemma}[theorem]{Lemma}
\newtheorem{remark}[theorem]{Remark}
\newtheorem{proposition}[theorem]{Proposition}
\newtheorem{corollary}[theorem]{Corollary}
\newtheorem{claim}[theorem]{Claim}
\newtheorem{fact}[theorem]{Fact}

\newtheorem{remk}[theorem]{Remark}
\newtheorem{example}[theorem]{Example}

\def\FullBox{\hbox{\vrule width 6pt height 6pt depth 0pt}}

\def\qed{\ifmmode\qquad\FullBox\else{\unskip\nobreak\hfil
		\penalty50\hskip1em\null\nobreak\hfil\FullBox
		\parfillskip=0pt\finalhyphendemerits=0\endgraf}\fi}

\def\qedsketch{\ifmmode\Box\else{\unskip\nobreak\hfil
		\penalty50\hskip1em\null\nobreak\hfil$\Box$
		\parfillskip=0pt\finalhyphendemerits=0\endgraf}\fi}

\newenvironment{proofof}[1]{\begin{trivlist} \item {\bf Proof of
			#1:~~}}
	{\qed\end{trivlist}}

\def\to{\rightarrow}
\def\epsilon{\varepsilon}
\def\eps{\epsilon}
\def\phi{\varphi}
\def\cal{\mathcal}

\def\psdgeq{\succeq}
\newcommand{\defeq}{:=}
\newcommand{\ie}{i.e.,\xspace}

\newcommand{\etal}{et al.\xspace}

\newcommand{\mper}{\,.}
\newcommand{\mcom}{\,,}

\newcommand{\R}{{\mathbb R}}
\newcommand{\E}{{\mathbb E}}
\newcommand{\N}{{\mathbb{N}}}

\newcommand{\F}{{\mathbb F}}

\newcommand{\B}{\{0,1\}\xspace}
\newcommand{\pmone}{\{-1,1\}\xspace}

\newcommand{\gauss}[2]{{\calN(#1, #2)}}

\newcommand{\abs}[1]{\ensuremath{\left\lvert #1 \right\rvert}}

\newcommand{\norm}[1]{\ensuremath{\left\lVert #1 \right\rVert}}
\newcommand{\ip}[2] {\ensuremath{\langle #1 , #2 \rangle}}

\newcommand{\calA}{{\cal A}}
\newcommand{\calB}{{\cal B}}
\newcommand{\calC}{{\cal C}}
\newcommand{\calD}{{\cal D}}
\newcommand{\calE}{{\cal E}}
\newcommand{\calF}{{\cal F}}

\newcommand{\calI}{{\cal I}}

\newcommand{\calL}{{\cal L}}
\newcommand{\calM}{{\cal M}}
\newcommand{\calN}{{\cal N}}

\newcommand{\calP}{{\cal P}}

\newcommand{\calS}{{\cal S}}

\newcommand{\poly}{{\mathrm{poly}}}
\newcommand{\polylog}{{\mathrm{polylog}}}
\DeclareMathOperator{\aut}{\operatorname{Aut}}
\newcommand{\ceil}[1]{\ensuremath{\left\lceil #1 \right\rceil}}

\newcommand{\bigoh}{\operatorname{O}}

\newcommand{\bigomega}{\mathop{\Omega}}
\newcommand{\littleoh}{\operatorname{o}}
\newcommand{\pE}{\widetilde{\E}}
\def\RR{\mathbb R}
\DeclareMathOperator{\slice}{\mathcal{S}}
\newcommand\set[1] {\lbrace #1 \rbrace}

\def\one{\mathbf 1}
\def\multiset#1#2{\ensuremath{\left(\kern-.3em\left(\genfrac{}{}{0pt}{}{#1}{#2}\right)\kern-.3em\right)}}

\usepackage{nicefrac}

\DeclareMathOperator\supp{Supp}

\newcommand{\Erdos}{Erd\H{o}s\xspace}
\newcommand{\Renyi}{R\'enyi\xspace}
\newcommand{\Lovasz}{Lov\'asz\xspace}
\newcommand{\Godel}{G\"{o}del\xspace}
\newcommand{\Poincare}{Poincar\'e\xspace}
\newcommand{\inparen}[1]{\left(#1\right)}             
\newcommand{\insquare}[1]{\left[#1\right]}             

\DeclareMathOperator{\tr}{\operatorname {tr}}
\newcommand{\Esymb}{\mathbb{E}}
\newcommand{\Psymb}{\mathbb{P}}

\newcommand{\Varsymb}{\mathrm{Var}}
\DeclareMathOperator*{\ExpOp}{\Esymb}

\makeatletter
\def\Pr#1{%
    \ProbabilityRender{\Psymb}{#1}%
}

\def\Ex#1{%
    \ProbabilityRender{\Esymb}{#1}%
}

\def\Var#1{%
    \ProbabilityRender{\Varsymb}{#1}%
}

\def\ProbabilityRender#1#2{
    \@ifnextchar\bgroup%
    {\renderwithdist{#1}{#2}}
    {\singlervrender{#1}{#2}}
}
\def\singlervrender#1#2{%
    \ensuremath{\mathchoice
        {{#1}\left[ #2 \right]}
        {{#1}[ #2 ]}
        {{#1}[ #2 ]}
        {{#1}[ #2 ]}
    }
}
\def\renderwithdist#1#2#3{%
    \@ifnextchar\bgroup
    {\superfancyrender{#1}{#2}{#3}}
    {\ensuremath{\mathchoice
            {\underset{#2}{#1}\left[ #3 \right]}
            {{#1}_{#2}[ #3 ]}
            {{#1}_{#2}[ #3 ]}
            {{#1}_{#2}[ #3 ]}
        }
    }
}
\def\superfancyrender#1#2#3#4#5{
    \ensuremath{\mathchoice
        {\underset{#1}{{#1}}\left#4 #3 \right#5}
        {{#1}_{#2}#4 #3 #5}
        {{#1}_{#2}#4 #3 #5}
        {{#1}_{#2}#4 #3 #5}
    }
}
\makeatother

\usepackage{tikz}
\usetikzlibrary{shapes.geometric}
\usepackage{float}
\newcommand{\unif}{\in_{\text{R}}}
\newcommand{\T}{\intercal}
\DeclareMathOperator{\GOE}{GOE}
\DeclareMathOperator{\OPT}{OPT}
\DeclareMathOperator{\nullspace}{Null}
\DeclareMathOperator{\parents}{parents}
\DeclareMathOperator{\body}{body}
\renewcommand*{\circle}[1]{\scalebox{0.85}{\footnotesize
		\tikz[baseline=(char.base)]{
			\node[shape=circle,draw,inner sep=1.5pt](char) {   \ifx&#1&
				\color{white} $i$
				\else
				$#1$
				\fi};
}}}
\renewcommand*{\square}[1]{\scalebox{0.85}{\footnotesize
		\tikz[baseline=(char.base), square/.style={regular polygon,regular polygon sides=4}]{
			\node[draw,square, inner sep=0.5pt](char) {
				\ifx&#1&
				\color{white} $t$
				\else
				$#1$
				\fi};
}}}
\newcommand{\NN}{\mathbb{N}}

\newcommand{\EE}{\mathop{\mathbb{E}}}
\newcommand{\GN}{\mathcal{N}}

\newcommand{\sig}{\sigma}
\newcommand{\Sig}{\Sigma}
\newcommand{\al}{\alpha}
\newcommand{\lda}{\lambda}
\newcommand{\Lda}{\Lambda}
\newcommand{\Gam}{\Gamma}
\newcommand{\gam}{\gamma}
\newcommand{\Del}{\Delta}

\newcommand{\tens}[2] {\ensuremath{#1 ^ {\otimes #2}}}

\newtheorem{propn}[theorem]{Proposition}
\usepackage{tikz,pgf}

\usepackage{caption}
\newcommand{\mat}[1]{\mathbf{#1}}
\newcommand{\HH}{\mathbb{H}}

\newcommand{\resamp}[1]{\widetilde{#1}}
\newcommand{\coef}[2]{\widehat{#1}(#2)}
\newcommand{\dpoly}{d_p}
\newcommand{\icL}{\cL^{-1}}

\newcommand{\grad}{\nabla}

\newcommand{\graphmat}[1]{\mat{M}_{#1}}
\newcommand{\sch}[2]{\norm{#1}_{#2}^{#2}}
\newcommand{\Etr}[1]{\EE \tr\left[#1\right]}
\newcommand{\Esch}[2]{\EE \norm{#1}_{#2}^{#2}}
\newcommand{\var}[1]{\Varsymb[#1]}
\newcommand{\herm}[1]{\overline{#1}}

\newcommand{\mA}{{\mat{A}}}
\newcommand{\mB}{{\mat{B}}}
\newcommand{\mC}{{\mat{C}}}
\newcommand{\mD}{{\mat{D}}}

\newcommand{\mF}{{\mat{F}}}
\newcommand{\mG}{{\mat{G}}}
\newcommand{\mH}{{\mat{H}}}
\newcommand{\mI}{{\mat{I}}}

\newcommand{\mK}{{\mat{K}}}

\newcommand{\mM}{{\mat{M}}}
\newcommand{\mN}{{\mat{N}}}

\newcommand{\mR}{{\mat{R}}}

\newcommand{\mU}{{\mat{U}}}
\newcommand{\mV}{{\mat{V}}}

\newcommand{\mX}{{\mat{X}}}

\newcommand{\mDel}{{\mat{\Del}}}
\newcommand{\mPi}{{\mat{\Pi}}}
\newcommand{\mSig}{{\mat{\Sig}}}

\newcommand{\cG}{{\cal{G}}}

\newcommand{\cI}{{\cal{I}}}
\newcommand{\cJ}{{\cal{J}}}
\newcommand{\cK}{{\cal{K}}}
\newcommand{\cL}{{\cal{L}}}

\newcommand{\cP}{{\cal{P}}}

\newcommand{\cR}{{\cal{R}}}
\newcommand{\cS}{{\cal{S}}}

\newcommand{\iid}{i.i.d.\xspace}

\usepackage{blkarray}
\usepackage{bigstrut}

\newcommand{\dsos}{D_{\text{SoS}}}

\newcommand{\psdmass}{\normalfont{(PSD mass) }}
\newcommand{\middleshapebounds}{\normalfont{(Middle shape bounds) }}
\newcommand{\intersectionbounds}{\normalfont{(Intersection term bounds) }}
\newcommand{\truncationbounds}{\normalfont{(Truncation error bounds) }}

\newcommand{\middleshapeboundstwo}{Middle shape bounds}
\newcommand{\intersectionboundstwo}{Intersection term bounds}
\newcommand{\truncationboundstwo}{Truncation error bounds}

\begin{document}
\maketitle

\makecopyright
\makededication

\tableofcontents
\listoffigures

\acknowledgments

Some of the best years of my life were at UChicago, academically, professionally and personally. Incontrovertibly, this is entirely due to the nourishing and encouraging environment I was in, rather than because of the research I did, the progress I made, or how productive I was.
A few pages do not do justice in acknowledging the various people who made this happen. However, I am going to make an attempt.

First and foremost, I am immensely grateful to my advisors Madhur Tulsiani and Aaron Potechin. Their extreme patience, generosity, kindness, and guidance are what kept me going over the years. It's hard to overstate the profound and positive impact they had on my life. I could not have asked for better advisors. I extend my heartfelt gratitude to them.

I recall being fascinated by Madhur's Mathematical Toolkit course in 2016 and approaching him later to work on research problems. Since then, Madhur has been so generous with his time and we have constantly battled many hard problems together. I have learnt so much from him, both related and unrelated to research.
His taste for research is fantastic and learning his intuitive approach to research has been incredible for my own development.
It still amazes me to see his command of so many different fields of Computer Science; moreover, his work ethic is an absolute inspiration.
Apart from research, I recall fondly all our conversations on all kinds of things under the sun. I thank him profusely for his wisdom, constant source of encouragement and indefatigable support especially when I needed it.

Aaron joined UChicago around the time I was starting to get seriously interested in the Sum of Squares method, which is now one of the main topics of my dissertation. In the following years, I made contributions to this field that I'm very proud of. This was made possible by his guidance and by adapting his approach to problem solving, where we keep chipping away at a problem from first principles until either the problem has fallen or something remarkable has been discovered regardless.
Some of the most memorable moments in my research have been in my work with Aaron. I also thank him for his careful reading of and many comments on this dissertation.

I am extremely grateful to Bryon Aragam from the UChicago Booth School of Business. When my colleague Bohdan and I approached him asking for problems to work on, he took the time to educate us on his research field, and gave us many intriguing problems and ideas to think about. We ended up having a fantastic collaboration that resulted in many works that I'm incredibly proud of. Strikingly, they led to topics I'm now passionate about and continue to think deeply about.
Bryon's thinking is lucid, his communication is excellent, and he has the remarkable talent of taking any problem or result and distilling its impact to computer science. I am happy to have had the excellent opportunity to learn from him.


I am very fortunate to have worked with Aravindan Vijayaraghavan. Whenever we were stuck, Aravindan's optimism helped me never to lose motivation and as a consequence, we constantly kept attacking our problems with renewed energy. This experience has sculpted me into a better researcher. Moreover, he is extraordinarily kind and has given me a lot of sagacious advice, for which I'm very grateful.

I want to thank Prof. Janos Simon for agreeing to be on my committee and for supporting me throughout my PhD. He was always willing to hear what I had been working on and offer thoughtful insights.
I am honored to have learnt from Prof. Laci Babai. Every conversation with him is stimulating, and his infinite exuberance for research and teaching never fails to inspire and reinvigorate me.
I thank Borja Sotomayor for involving me in the UChicago community of ICPC programming competitions, which kept my coding skills honed.


In the initial stages of my PhD, I didn't have many collaborators, and being stuck for so long on various concepts or problems was a frustrating experience. However, halfway into my PhD, I realized that the frustration diminished significantly when working with others.
In a sense, one of the primary outcomes of a PhD is to learn how to learn and to learn how to make the intense struggle not suck. For me, collaborative research was the answer. Moreover, collaboration has led me to enjoy not just the end goals but also  the process of conducting research.
I am fortunate to have worked with some incredible people during my PhD, which has helped me grow both as a researcher and as a person.
These people include Bryon Aragam, Ainesh Bakshi, Xue Chen, Ming Gao, Mrinalkanti Ghosh, Fernando Granha Jeronimo, Chris Jones, Bohdan Kivva, Sidhanth Mohanty, Aaron Potechin, Pradeep Ravikumar, Madhur Tulsiani, Aravindan Vijayaraghavan, Jeff Xu and everyone else I've discussed math with.

I want to highlight two of them --- Chris and Bohdan. I am thankful to the cheerful Chris for the hundreds of hours we've worked together on research problems, reading papers, learning new topics and solving problem sets. He has also been a close friend and confidante, and our many non-research conversations and activities are memorable.
I am also fortunate to have worked with the intelligent Bohdan, whose incredible intuition and clarity of thought continue to be an inspiration. I am thankful for his kindness and camaraderie.

Finally on the academic side, l would like to extend my sincere gratitude to my professors and teachers from before UChicago. I thank Sourav Chakraborty, also a UChicago PhD graduate advised by Laci Babai. Sourav gave me professional advice, helped me with my graduate school application, and encouraged me to apply (and join) UChicago for my PhD.
I thank Geevarghese Philip for introducing me to research, while hosting me at the serene Max-Planck-Institut f\"{u}r Informatik in the summer of 2015 where we worked on parameterized algorithms.
I am grateful to Narayan Kumar, Madhavan Mukund, K. V. Subrahmanyam and B. Srivathsan for helping me at various stages of college at Chennai Mathematical Institute and for making my life there memorable.
Thanks to Mr. Sadagopan Rajesh for introducing me to olympiad math, which perhaps changed my entire life.
Thanks to my good friend Amir Goharshady for teaching me {C\nolinebreak[4]\hspace{-.05em}\raisebox{.4ex}{\tiny\bf ++}} and introducing me to competitive programming, which I continue to enjoy to this day and which was instrumental in pivoting my research from mathematics to computer science.

Thanks to UChicago and TTIC for providing me with compute clusters to conduct some of my research and for the incredible staff and facillities that made my life here pleasant.

Although graduate school is primarily an academic place, I had way more fun than I anticipated. I will remember fondly the ungodly amount of football (soccer) I played, especially in the Stagg and South fields, but sometimes in the Jackman field, the Henry Crown field, the Midway Plaisance, the Stony Island field, the Fire Pitch and random fields all around Chicago.
Perhaps the best outcome of it is that it led me to meet some of the most wonderful people in my life.
Apart from football, I have also met many incredible people in this city, thanks in part to the International House at UChicago for organizing social events on campus.
I am privileged to have had a tender and caring social circle, which has enabled me to grow into a better human being.
These are some people close to me, some even closer than family, sorted in lexicographic order of last name (I apologize in advance if I may have missed your name):
Sudarshan Babu, Alan Chang, Alexis Comte, Leonardo Coregliano, Adam Dziedzic, Animesh Fatehpuria, Nicolas Fierro, Mrinalkanti Ghosh, Mina Gian, Amir Goharshady, Neng Huang, Jafar Jafarov, Fernando Jeronimo, Reza Jokar, Chris Jones, Marie Kim, Akash Kumar, Hyunku Kwon, Kunal Marwaha, Matt McPartlon, Omar Montasser, Tushant Mittal, Sidhanth Mohanty, Sayan Mukherjee, Nathan Mull, Kshitij Patel, Wojciech Nadara, Rachit Nimavat, Jair Pinedo, Aravind Reddy, Suhail Rehman, Harry Ren, Amar Risbud, Diego Rojas, Aritra Sen, Hy Truong Son, Harsha Srinivas, Zihan Tan, Akilesh Tangella, Shubham Toshniwal, Shubhendu Trivedi, Xiaoyan Wang, Jeff Xu, Xifan Yu, Ivan Zelich, Meiqing Zhang and Wei Zou.

Lastly, I owe it to my family for my happy childhood, for cultivating my carefree and joyful personality, and for providing constant support throughout my life. In this regard, I am grateful to my brother Aravind Raj, my sister R. Ishwarya, my mom R. Anandhi and my dad Dr. A. S. Rajendran. Most of my personality, my character, my passion, my thoughts, and my views on life are heavily inspired by my late dad. It is to his loving memory that I dedicate this dissertation.

\abstract

We develop new tools in the theory of nonlinear random matrices and apply them to study the performance of the Sum of Squares (SoS) hierarchy on average-case problems.

The SoS hierarchy is a powerful optimization technique that has achieved tremendous success for various problems in combinatorial optimization, robust statistics and machine learning. It's a family of convex relaxations that lets us smoothly trade off running time for approximation guarantees. In recent works, it's been shown to be extremely useful for recovering structure in high dimensional noisy data. It also remains our best approach towards refuting the notorious Unique Games Conjecture.

In this work, we analyze the performance of the SoS hierarchy on fundamental problems stemming from statistics, theoretical computer science and statistical physics. In particular, we show subexponential-time SoS lower bounds for the problems of the Sherrington-Kirkpatrick Hamiltonian, Planted Slightly Denser Subgraph, Tensor Principal Components Analysis and Sparse Principal Components Analysis. These SoS lower bounds involve analyzing large random matrices, wherein lie our main contributions. These results offer strong evidence for the truth of and insight into the low-degree likelihood ratio hypothesis, an important conjecture that predicts the power of bounded-time algorithms for hypothesis testing.

We also develop general-purpose tools for analyzing the behavior of random matrices which are functions of independent random variables. Towards this, we build on and generalize the matrix variant of the Efron-Stein inequalities. In particular, our general theorem on matrix concentration recovers various results that have appeared in the literature. We expect these random matrix theory ideas to have other significant applications.

\mainmatter

\chapter{Introduction}\label{chap: intro}

Algorithm design, mathematical optimization and computational complexity are close-knit fields of computer science that have largely developed in parallel in the beginning. In recent decades, there has been an explosion of research in these fields that often borrowed ideas from the other ones, and there is no longer a discernible wall separating them. Indeed, these fields of computer science can now be construed as trying to achieve the same goal --- Which problems are easy and which are hard?

Early researchers have mainly focused on search problems. Given an input, the objective is to search for a desired hidden structure. Often, this can be equivalently restated as the problem of optimizing an appropriate objective function under various constraints.

\begin{figure}[!h]
    \centering
    \includegraphics[trim={2cm 20cm 2cm 2cm}, clip, scale=0.9]{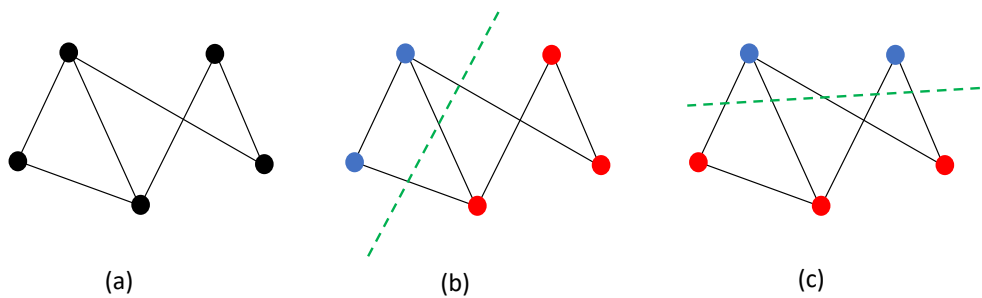}
    \caption{An example graph on $5$ vertices and two possible cuts.}
    \label{fig: maxcut}
\end{figure}

For example, consider the \textit{Maximum Cut} problem, where the input is a graph and the goal is to partition the set of vertices into two subsets that maximizes the number of edges with endpoints in different parts. If we take for instance the graph $(a)$ in \cref{fig: maxcut}, two possible partitions are shown in $(b)$ and $(c)$ where blue colored vertices form a part and red colored vertices form a part. Then, the partition in $(b)$ cuts $3$ edges and the partition in $(c)$ cuts $5$ edges, namely the edges intersecting the green line. It's easy to see via a simple parity argument that we cannot do better than $5$ edges.

In the search problem formulation, we would like our search algorithm to output a partition that cuts the maximum possible number of edges. In the optimization problem formulation, we would like our optimization algorithm to output the maximum value correctly.

Another formulation of computational problems are decision problems. Given an input, the objective is to decide whether there exists a hidden structure or if the objective value satisfies some properties, with the restriction that the algorithm can only return a boolean output --- for example, true or false; or yes or no. In the above example of maximum cut, the decision problem perspective could be to ask if the maximum cut in the given graph contains at least $0.6$ (say) fraction of the total number of edges.

These types of problems are all intimately related and in many cases, essentially boil down to the search for algorithms. For practicality, we require various properties like efficiency, accuracy, etc. This has led to the development of a rich theory of computability, complexity theory and optimization. In this dissertation, we will also consider the viewpoints of related types of problems, namely certification problems and hypothesis testing. As we will see, these other formulations are related to the former and to each other but it's not clear how deep the connections go, and trying to understand this is an important pursuit in theoretical computer science.
That said, underlying all these formulations is the goal of searching for efficient algorithms to detect and extract structure from data, or arguing that no such algorithms exist unless we're willing to compromise on other things like efficiency or accuracy.

\section{Certification problems}

As opposed to search or decision problems, certification problems, given an input, ask for a bound on the objective value that holds with probability $1$, along with a certificate of the output bound. The quality of the algorithm is usually measured in terms of how close the bound gets to the true optimum.

In the running example of maximum cut, given a graph, the task could be to output a value that's always an upper bound on the size of the maximum cut. A simple algorithm could be to simply return the total number of edges in the graph. Indeed, this is a valid certification algorithm but we could ask if one could do better.

This is fundamentally a different approach to algorithm design. Consider the scenario when we are maximizing some objective function and so we desire an upper bound on the optimal value. Then, designing a certification algorithm can be construed as attacking a problem from \textit{above} as opposed to from \textit{below}, the latter of which is the more standard notion of algorithm design.

The notion of linear programming relaxations already provide such certification algorithms. Given a problem that can be formulated as an integer program (as many can be), a natural way to obtain a certification solution is to widen the search space from integral variables to real variables, adding other appropriate constraints as necessary. This is known as relaxing the program. This enables a faster algorithm to attempt to compute the solution, but comes at a loss of only obtaining an approximate solution. More importantly, the objective value obtained by the return solution is a definite bound bound on the optimal solution, no matter the input. This is what a certification algorithm desires. Measuring the quality of the returned output often depends on the type of relaxation considered and problem specific structure.

In many cases, it's possible to obtain an approximation algorithm to a problem by looking at a relaxation of the program, obtaining a non-integral solution and rounding it to a valid solution. For the maximum cut problem, this was done by Goemans and Williamson in their seminal work \cite{GW94} where they used a semidefinite programming relaxation, which is more powerful than linear programming relaxations.

In this dissertation, we will focus on a specific class of such certification algorithms, namely the Sum of Squares (SoS) hierarchy, sometimes referred to as the Lasserre hierarchy. The SoS hierarchy is a series of convex relaxations to a given program. By virtue of being a relaxation, they can be used for certification.
Due to it's tremendous success for various fundamental optimization problems such as maximum cut, constraint satisfaction, etc., the SoS hierarchy has become a powerful optimization technique. This is further amplified by results that say that the SoS hierarchy is the optimal relaxation among a broad class of semidefinite programming relaxations \cite{lrs15}, and assuming the famous unique games conjecture, it's the best approximation algorithm for every constraint satisfaction problem \cite{Raghavendra08}.
A chief goal of this dissertation is to understand the limits of this powerful technique. We especially focus on the so-called average-case setting, that we will define now.

\section{Average-case analysis}

An important theme in this work is the study of random instances of problems, which is termed average-case analysis. As opposed to traditional worst-case algorithm design, where we wish to design an algorithm that performs well on the worst possible input, there has been an exciting development of research on problems where the input is randomly sampled from a distribution. For instance, in the maximum cut problem, we could assume that the input comes from the \Erdos-\Renyi family of random graphs, where the number of vertices in the graph is chosen beforehand and each edge is present independently with probability $0.5$.

In average-case algorithm design, we wish to design algorithms that perform well on average-case inputs with high probability, as opposed to all inputs.
This is important because studying the worst case complexity of a problem may not shed light on the intrinsic hardness of the problem. This happens because the worst-case instance input for an algorithm could be highly artificial and contrived. Put another way, in real world scenarios, the inputs for various optimization or search problems we encounter are unlikely to be such instances. This is seen in practice as well. For example, the simplex method for linear programming \cite{dantzig2016linear} is exponentially slow in the worst-case, as was shown by Klee and Minty \cite{klee1972good}, but performs extremely well practically. Various works have tried to explain this behavior, e.g. \cite{borgwardt1982average, smale1983average, borgwardt1988probabilistic, spielman2004smoothed}, a highlight is the work of Spielman and Teng for which they were awarded the \Godel prize in 2008.

Tremendous effort has been invested to understand the average-case complexity for a wide variety of problems. Research towards designing average-case algorithms brings about a deeper understanding of the core of the problem, enabling the design of worst-case algorithms as well. This can be seen for example for the famous Densest $k$-subgraph problem \cite{bhaskara2010detecting}. In this work, we will focus on average-case analysis.

In our pursuit, fundamental mathematical objects that occur repeatedly are large random matrices. We often desire to understand their behavior.

\section{Underlying theme of this work: Random matrices}

Random matrices are abundant in computer science, especially in the fields of optimization and statistics. Often, the analysis of an algorithm requires analyzing the behavior of certain random matrices that can be constructed from the input. Even outside computer science, random matrix theory is a fundamental field of it's own right, having been studied since the early $1900$s, with applications also extending to many branches of mathematics and physics. For a short survey, see \cite{forrester2003developments}.

There has been tremendous effort over the last few decades to develop the theory of random matrices, see the book by Tropp \cite{tropp2015:book}. For example, the matrix-Bernstein inequality studies the behavior of a random weighted sum of matrices; the Wigner semicircle law studies the distribution of the eigenvalues of a random matrix sampled from the Gaussian Orthogonal ensemble. On the other hand, fewer tools are available to understand the behavior of nonlinear random matrices, where each matrix entry is a nonlinear function of the input, say for instance low-degree polynomials.

In our setting, this occurs frequently when trying to analyze the SoS hierarchy for various problems. This is true both when trying to design algorithms via SoS as well as when trying to study the limitations of SoS algorithms, for example, \cite{barak2012hypercontractivity, hopkins2015tensor, schramm2017fast, moitra2019spectral, jones2022sum}. Therefore, we begin with this important endeavor of understanding the behavior of nonlinear random matrices. In the first part of this thesis, we are interested specifically in concentration behavior. We emphasize that this is an important research direction in it's own right.

To bound the fluctuations of a random matrix from its mean, measured in terms of spectral or Schatten $t$-norm of the difference, a simple but powerful technique that has been widely used (including in many of the works cited above) is the so-called trace method. In this method, the (centered) random matrix is raised to a large power and the expected trace of the resulting matrix is bounded. While this method gives satisfactory results, it often requires ingenious observations and highly nontrivial combinatorics.

Another approach is as follows. Consider a random matrix that is a function of several independent input variables. We can study it's behavior by studying how much it deviates when a single uniformly chosen input entry is resampled. By bounding these local fluctuations, we can bound the global fluctuation of the random matrix.
This technique gives rise to the Efron-Stein inequalities. Originally, they were developed for scalar random variables (which can be thought of as $1 \times 1$ matrices). In this special case, they turned out to be extremely powerful since they have been shown to recover many standard concentration inequalities. Recently, the work \cite{paulin2016} showed a matrix version of the Efron-Stein inequalities.
In this work, we build on this to obtain a general framework for proving concentration of large random matrices.

In the second part of this thesis, in the analysis of SoS algorithms, the fundamental difficulty that appears is to analyze the behavior of a large nonlinear random matrix. In particular, we want to argue that this random matrix is positive semidefinite with high probability over the choice of the input. For this, we exhibit an approximate Cholesky decomposition of the matrix and the proof extensively builds on the concentration results we develop above.

In conclusion, the motif in this work is the study of nonlinear random matrices, where we both build a general framework for analyzing concentration and apply them to study algorithms on fundamental problems.

\section{The Sum of Squares Hierarchy}

Given an optimization problem in the form of a program with polynomial inequality constraints, there have been many works proposing generic approaches to relax the program, in order to obtain good solutions efficiently. Some of the more dominant approaches have been the \Lovasz-Schrijver hierarchy \cite{LoS91} and the Sherali-Adams hierarchy \cite{SA90}. Informally speaking, these hierarchies of algorithms lift the program to a larger set of variables, tied together via various constraints, relax and solve the larger program, and finally project the solution down to the original variable space. They are parameterized by an integer known as the degree, where larger degrees offer tighter relaxations at the cost of larger running times.

The Sum of Squares (SoS) hierarchy is a similar optimization technique that harnesses the power of semidefinite programming. For polynomial optimization problems, the SoS hierarchy, first independently investigated by Shor \cite{shor1987approach}, Nesterov \cite{nesterov2000squared}, Parillo \cite{parrilo2000structured}, Lasserre \cite{lasserre2001global} and Grigoriev \cite{grigoriev2001complexity, Grigoriev01}, offers a sequence of convex relaxations parameterized by an integer called the degree of the SoS hierarchy.
As we increase the degree $d$ of the hierarchy, we get progressively stronger convex relaxations which are solvable in $n^{O(d)}$ time.
This has paved the way for the SoS hierarchy to be almost a blackbox tool for algorithm design. As has been shown in multiple works, it serves as a strong algorithm for various problems, both in the worst case and the average case settings.

Consider our running example of the Maximum Cut problem. The seminal Goemans-Williamson algorithm \cite{GW94:stoc} achieves an approximation factor of $\approx 0.878$ for this problem via a semidefinite programming relaxation. As it turns out, this algorithm is just the degree $2$ SoS hierarchy. This approximation factor is conjectured to be optimal and there has been increasing evidence that this is indeed the case. This highlights an example of why the SoS hierarchy is powerful.

Indeed, there has been tremendous success in using the SoS hierarchy to obtain efficient algorithms for combinatorial optimization problems (e.g., \cite{GW94, AroraRV04, GuruswamiS11, raghavendra2017strongly}) as well as problems stemming from Statistics and Machine Learning (e.g., \cite{barak2012hypercontractivity, bks15, HopSS15, pot17, kothari2017outlier}). In fact, SoS achieves the state-of-the-art approximation guarantees for many fundamental problems such as Sparsest Cut \cite{AroraRV04}, Maximum Cut \cite{GW94}, Tensor PCA \cite{HopSS15} and all Max-$k$-CSPs \cite{Raghavendra08}. As mentioned earlier, for a large class of problems, it's been shown that SoS relaxations are the most efficient among all semidefinite programming relaxations \cite{lrs15}.

The term ``Sum of Squares'' comes from a dual view in proof complexity.
Besides being an algorithmic technique, SoS can be equivalently viewed as giving a proof or certificate of a bound on the optimal value of a polynomial optimization problem.
This work can be traced back to Hilbert's seventeeth problem which has led to work on a proof complexity result known as the Positivstellensatz, which gives conditions under which polynomial systems can be shown to have no solutions, see e.g. \cite{stengle1974nullstellensatz, putinar1993positive, reznick2000some}. The algorithmic implications were originally observed by Lasserre \cite{lasserre2001global} and Parillo \cite{parrilo2000structured, parrilo2003semidefinite} leading to the interpretation of SoS as an optimization technique as we study in this work.
This duality can be completely formalized and has led to the so-called framework of ``proofs to algorithms'' that has achieved tremendous success, especially recently in robust statistics, see e.g., \cite{kothari2017outlier, karmalkar2019list, hopkins2020mean, bakshi2021robust}. The adage is that if we can find an ``easy'' proof of an identifiability result for a search problem, then it can be automatized to give an algorithm.
We will not explore this in detail here, and we refer the reader to the monograph \cite{FKP19}.

Next, we move onto SoS lower bounds but before that, we highlight some related techniques that has gained traction in the community recently.

\subsection{Related Algorithmic Techniques}\label{subsec: related_techniques}

Apart from search, decision and certification, researchers have also considered other related types of problems. Consider a problem where the input is sampled from one of two known distributions and we would like to identify which distribution it was sampled from. This is known generally as hypothesis testing. For example, one distribution could be the distribution of \Erdos-\Renyi random graphs while the other could be the distribution of \Erdos-\Renyi random graphs but with a large cut planted in them. It's clear that this problem is a different flavor of the maximum cut problem on random graphs. Beyond being interesting in their own right, studying these related formulations offer alternate perspectives and interesting insights into the search or certification variants as well.
Another type of problem, known as recovery problems, is to recover the planted structure when the input is sampled from the latter distribution.

For all the type of problems considered so far, apart from SoS, there have also been several other framework of algorithms that have been considered and in some cases, extensively studied. Examples include
\begin{itemize}
    \item \Lovasz-Schrijver and Sherali-Adams hierarchies --- As discussed earlier, these hierarchies lift a program to a larger set of variables and then relax any integrality constraints. The resulting solution is then projected back to the original variables which may then be rounded to an integral solution. These hierarchies are captured by the SoS hierarchy, or in other words, the SoS hierarchy is at least as powerful as these hierarchies \cite{FKP19}.
    \item Low degree polynomials --- For hypothesis testing, low degree polynomials can be used to try and distinguish the two distributions. More precisely, if there is a low degree polynomial such that its expected value on the two distributions behave differently and the variance isn't too large, this can be used to distinguish the two distributions. This is related to the SoS hierarchy and we will revisit this point in more detail later.
    \item Statistical query algorithms --- For hypothesis testing, the statistical query model (SQ) is another popular restricted class of algorithms introduced by \cite{kearns1998efficient}. In this model, for an underlying distribution, we can access it indirectly by querying expected values of functions, up to some error.
    Given access to this oracle, we would like to hypothesis test. SQ algorithms capture a broad class of algorithmic techniques in statistics and machine learning including spectral methods, moment and tensor methods (see e.g. \cite{feldman2017statistical, feldman2021statistical}). SQ algorithms has also been used to study information-computation tradeoffs and more broadly has been studied in other contexts \cite{Feldman2016}. There has also been significant work trying to understand the limits of SQ algorithms (e.g. \cite{feldman2017statistical, feldman2018complexity, diakonikolas2017statistical}). Recent work \cite{brennan2020statistical} has shown that low degree polynomials and statistical query algorithms have equivalent power under mild conditions.
    \item Approximate message passing and other statistical physics techniques such as belief propagation, see e.g. the review \cite{zdeborova2016statistical}.
    \item Local algorithms, see e.g. \cite{elek2010borel, fan2017well, hoppen2018local}.
    \item Circuit models of computation of bounded size, see e.g. \cite{rossman2010average, rossman2014monotone}.
\end{itemize}

\section{Lower bounds against the Sum of Squares Hierarchy}

Because of the incredible success of the SoS hierarchy for a variety of problems, it's an important research direction to study the limits of the SoS hierarchy, which we endeavour in this dissertation. In particular, we will focus on average-case problems and as we will see, most of the technical difficulty boils down to the analysis of nonlinear random matrices, to handle which we develop various techniques.

There are many reasons for why studying lower bounds against the SoS hierarchy is important. The SoS hierarchy is general enough to capture a broad class of algorithmic reasoning \cite{FKP19}. In particular, SoS captures the \Lovasz-Schrijver and Sherali-Adams hierarchies and under mild restrictions, also statistical query algorithms and algorithms based on low degree polynomials. Therefore, SoS lower bounds indicate to the algorithm designer the intrinsic hardness of the problem and suggest that if they want to break the algorithmic barrier, they need to search for algorithms that are not captured by SoS. Secondly, in average case problem settings, standard complexity theoretic assumptions such as P $\neq$ NP have not been shown to give insight into the limits of efficient algorithms. Instead, lower bounds against powerful techniques such as SoS have served as strong evidence of computational hardness \cite{hop17, hop18}. Thus, understanding the power of the SoS hierarchy on these problems is an important step towards understanding the approximability of these problems. See also the surveys \cite{BS14:ICM, moitra2020sum} for more on this.

There have been relatively fewer works on SoS lower bounds, as opposed to some other classes of algorithms we have discussed, which can be attributed to the sheer technical difficulty of proving such lower bounds. For example, the works \cite{Grigoriev01, Schoenebeck08, KothariMOW17} studied SoS lower bounds for random constraint satisfaction problems. A series of works \cite{feige2000finding, meka2015sum, deshpande2015improved, BHKKMP16, Pang21} studied SoS lower bounds for maximum clique on random graphs. Some other SoS lower bounds, not including the ones in this thesis, are the works \cite{ma_wigderson_15, kothari2018sum, mohanty2020lifting, kunisky2020,  kothari2021stress}.

\section{A summary of our main results}

In the first part of this work, we study concentration behavior of nonlinear random matrices. In the second part, we study lower bounds against the SoS hierarchy for several fundamental problems.

\subsection{Nonlinear matrix concentration via Matrix Efron-Stein}

We start by giving a general theorem on concentration of random matrices whose entries are polynomials of independent random variables. The famous matrix-Bernstein inequality answers this question when we only have linear polynomials. However, understanding the setting of non-linear polynomials is just as important yet it poses significant challenges. When they arise in various applications in the literature, the usual way to handle such random matrices has been the so-called trace method. While this method gives the desired results, sometimes to great effect, applying it usually turns out to be highly nontrivial. In this work, we propose an alternate way to prove matrix concentration via the Matrix Efron-Stein inequalities. We propose a general matrix concentration inequality, the proof of which relies on the powerful method of exchangeable pairs. We show some applications of this inequality and expect it to have significant applications outside what we have explored here.

\subsection{Sum of Squares lower bounds}

We obtain strong sub-exponential time lower bounds against the SoS hierarchy for a variety of fundamental problems in computer science. All our applications start with the so-called pseudocalibration heuristic, reducing the problem to analyzing the behavior of a large random matrix, known as the \textit{moment matrix}. Our conceptual and technical innovations happen at this step. The results we present are as follows.

\subsubsection{Sherrington-Kirkpatrick Hamiltonian}

An important problem in statistical physics, the Sherrington-Kirkpatrick problem is to optimize the quadratic form of a random matrix sampled from the Gaussian Orthogonal Ensemble, over boolean vectors. It's been known for a long time that the true optimal value concentrates at a particular constant, up to scaling. Recently, an efficient algorithm was proposed for this optimization problem. Certification on the other hand was widely believed to be hard beyond the simple spectral algorithm. We provide strong evidence for this by exhibiting lower bounds against SoS for this problem. This work requires us to understand the nullspace of the moment matrix and \textit{nullify it} before applying our matrix concentration tools. Conceptually, this work provides a lot of insight into the behavior of SoS on other fundamental problems such as maximum cut and learning mixtures of Gaussians.

\subsubsection{Sparse PCA}

Sparse PCA is a variant of principal components analysis (PCA), a fundamental routine in statistics and machine learning. We work with the spiked Wishart model, which is the most natural version of this problem, but which has proved quite hard to analyze in SoS.
Prior works have predicted the computational barrier of the recovery of the sparse component, as a tradeoff between the dimension, sparsity and number of samples. We confirm this barrier by proving lower bounds, matching known algorithms, against sub-exponential time SoS. This work involves splitting the random moment matrix into different matrices and using innovative combinatorial charging arguments to study how these matrices interact with each other. Conceptually, this work confirms the computational barrier diagram for this problem, that has been predicted and believed to be true for a long time.

\subsubsection{Planted Slightly Denser subgraph}

Finding a dense subgraph in a given graph is an important problem that has received much scrutiny over the years, both algorithmically as well as from the algorithmic hardness angle. For random instances of the problem under certain parameter regimes, the difficulty of this problem has been conjectured, usually referred to as the PDS conjecture, and this problem has been used as a canonical hard problem to reduce to various other problems and study their computational barriers. Moreover, these hard instances have also been used as a basis for cryptographic schemes.
Therefore, SoS lower bounds against this problem go a long way towards confirming this conjecture. In this work, we exhibit such sub-exponential time lower bounds for certain parameter regimes, where it has been widely believed to require sub-exponential time.

\subsubsection{Tensor PCA}

Tensor PCA is the average-case version of the problem of optimizing homogeneous polynomials over the sphere, which is a fundamental and important problem in optimization due to it's connections to a variety of fields. In this work, we prove SoS lower bounds matching known algorithms for this problem, settling the computational barrier for SoS for this problem. It also offers insight on  the approximability-inapproximability threshold for general homogeneous polynomial optimization and suggests that random instances may not be the hardest for this problem.

\section{Excluded work}

This dissertation contains the main body of my research conducted during my PhD but there have also been other research directions that have been left out, regrettably. This includes the following works.

\subsection{SoS Lower bounds for Sparse Independent Set}

In our work \cite{jones2022sum}, we show SoS lower bounds for the maximum independent set problem on sparse \Erdos-\Renyi random graphs, matching the \Lovasz theta function up to low order terms. To do this, we build on the tools developed in this dissertation as well as develop a variety of new techniques. In particular, this work is the first venture in the important research direction of understanding the limitations of SoS on sparse random graphs. We highlight that for this work, our nonlinear matrix concentration tools from \cref{chap: efron_stein} are very useful. We will elaborate on this result in \cref{chap: future_work} since it builds on much of the work we will develop in this dissertation.

\subsection{Causal Inference }

Causal inference is the study of discovering and understanding causal relationships in observed data, which has diverse applications in medicine, genetics, economics, epidemics, artificial intelligence, etc. In our work \cite{rajendran2021structure}, we focus on the problem of learning a class of causal models known as Bayesian Networks (BN), from data. This is a classical and fundamental problem since BNs are compact, modular and offer intuitive causal interpretation, which has made them very useful in various fields. We propose and study a new practical algorithm for this problem. It is efficient, provably differs from the widely used Greedy-Equivalence-Search algorithm, and since the algorithm is a general-purpose score-based learning algorithm, it is widely applicable. Also, under some statistical assumptions that are inspired from and which generalize recent works, our algorithm provably recovers the true Bayesian Network, even for non-parametric models, while making no assumptions on linearity, additivity, independent noise or faithfulness. It also suggests interesting potential connections to other machine learning fields such as clustering, forward-backward greedy methods, and kernel methods.

\subsection{Latent Variable modeling}

In our work \cite{kivva2021learning}, we study a relatively understudied but important problem of latent variable modeling of observed data. Building from the previous section, we now have unobserved (sometimes even unmeasurable!) latent causes or confounders for the observed variables. We focus on the setting of probabilistic mixture models, which naturally comes up in machine learning, economics, finance, biology, etc. Under some natural assumptions on the model, we develop an algorithm that takes the observed data and uncovers the hidden variables and the underlying causal relationships. Prior works related to this problem have usually focused on special settings such as linear models. We instead propose an algorithm to this problem in the highly nonlinear mixture models setting which works atop existing algorithms for mixture model order estimation (which is easier than density estimation).

\subsection{Causal representation learning}

An exciting new branch of machine learning, known as causal representation learning, takes as input raw, unstructured data, and aims to learn the underlying generative model that generated it. On top of this, it also aims to learn the causal relationships among the learnt latent variables, hence the name causal representation learning. In particular, this field brings together ideas from two fields which have largely developed separately, namely causal inference and latent variable modeling, the two topics described above. In our work \cite{kivva2022identifiability}, we prove an interesting and surprising result in this direction. We show that a broad class of generative models with a mixture of Gaussians prior is identifiable (which means it can be recovered from raw data). In particular, our models have universal approximation capabilities and have been used extensively (without theoretical validation) in many practical works on deep representation learning \cite{dilokthanakul2016deep, jiang2016variational, willetts2021don}.

In deep learning, there has been tremendous effort to identify the latent features and the mechanisms that generate observed data. Instead of handcrafting low level features of data, this process is largely automated via algorithms that learn low level representations. The models thus learnt are quite useful for a variety of downstream tasks such as sampling, prediction, classification, clustering, interventions, etc. A prominent player here is variational autoencoders \cite{kingma2013auto, rezende2014stochastic}. Various improvements to variational autoencoders have been made over the last decade, with a wide variety of applications. A much-desired property of the training process is stability, i.e. whether repeated trainings will lead to the same latent variable generative model. This can be captured by the mathematical notion of identifiability, which is a crucial primitive which guarantees that there is a unique parameter and generation mechanism that could have generated the data. Putting computational feasibility aside, identifiability is a necessary condition for stable and repeatable training. Apart from stability of training, this also paves the way for other important considerations in machine learning, such as the increasing need to learn representations of data that are robust, interpretable, explainable and fair.

In our work \cite{kivvaidentifiability}, we show that for commonly used variational autoencoders with a mixture of Gaussians prior, identifiability holds under the assumption that the warping mechanism is affine (in particular, deep neural networks with ReLU activations satisfy this property) and importantly, without assuming that auxiliary information is available. This significantly improves upon a flurry of recent works (initiated by \cite{khemakhem2020variational}) that have shown identifiability in the presence of auxiliary variables or side information. Also, several prior works have made empirical observations that a mixture of Gaussians prior often leads to stable and repeatable training for variational autoencoders, thereby suggesting identifiability.  Our work theoretically grounds these observations.

\section{Organization of the thesis}

In \cref{chap: efron_stein}, we develop our nonlinear matrix concentration results and show it's applications towards various nonlinear random matrices that have arisen in the literature. We then introduce the Sum of Squares hierarchy in \cref{chap: sos}, introduce the technique of pseudocalibration used for showing SoS lower bounds and show it's connections to low-degree algorithms. In \cref{chap: main_results}, we formally state the main SoS lower bounds we show in this thesis and put them in context with known prior works. In \cref{chap: sk}, we prove the SoS lower bound for the Sherrington-Kirkpatrick problem. In the next two chapters, \cref{chap: qual} and \cref{chap: quant}, we prove the SoS lower bounds for Planted Slightly Denser Subgraph, Tensor PCA and Sparse PCA. We conclude with follow-up and potential future works in \cref{chap: future_work}.

\chapter{Nonlinear matrix concentration}\label{chap: efron_stein}

In this chapter, we will describe our techniques for nonlinear matrix concentration via Efron-Stein inequalities. The material in this chapter is adapted from \cite{rajendran2023concentration}. While we develop general techniques that can be applied to study nonlinear concentration and this chapter is completely self-contained, our application to graph matrices will serve as a good warmup to segue into the technical sections of the Sum of Squares lower bounds that'll appear in later chapters of this dissertation.

\section{Introduction}\label{sec: intro}
In  optimization, statistics, and spectral algorithms, we often want to understand the concentration of various random matrices. To do this, we can appeal to the powerful theory of matrix-deviation inequalities~\cite{tropp2015:book}.
For example, the matrix-Bernstein inequality addresses random matrices of the form
\[
\mat{M} ~=~ x_1 \cdot \mat{C_1} + \cdots + x_n \cdot \mat{C_n}
\]
where $x_1, \ldots, x_n$ are independent scalar random variables, and $\mat{C_1}, \ldots, \mat{C_n}$ are fixed matrices.
A large selection of such inequalities are available when the random matrix (say) $\mat{M}$ is a \emph{linear} function of independent random variables. However, several recent works require us to understand random matrices which are \textit{non-linear} functions, and in particular low-degree polynomial functions, of scalar random variables. This forms the focus of our work.

As a motivating example, consider the random matrix $\mat{M} \in \R^{[n]^2 \times [n]^2}$ obtained as
\[
\mat{M} ~=~ \mat{A_1} \otimes \mat{A_1} + \cdots + \mat{A_m} \otimes \mat{A_m} \mcom
\]
where $\mat{A_1}, \ldots, \mat{A_m} \in \R^{[n]\times[n]}$ are independent random matrices, with \iid entries uniformly distributed in $\pmone$.
It is easy to see that the entries of the matrix $\mat{M}$ are degree-2 polynomial functions of the independent random variables describing the entries of $\mat{A_1}, \ldots, \mat{A_m}$. The concentration of such a matrix was analyzed by Hopkins \etal \cite{hopkins2015tensor, hopkins2018statistical}, who use it to design spectral algorithms for a variant of the principal components analysis (PCA). This matrix is a special case of a more general setting that we study in this work.

\paragraph{Matrix-valued polynomial functions.}
In the example above, the entries of the matrices are low-degree polynomials in independent (Rademacher) random variables.
In this work, we consider a general setting where we take an $n$-tuple $Z = (Z_1, \ldots, Z_n)$ of independent and identically distributed random variables\footnote{Our framework also applies when the variables are not necessarily identically distributed, as long as they are independent.} distributed in $\Omega$.
We consider random matrices given by a matrix-valued function $\mat{F}(Z)$ taking values in $\R^{\cI \times \cJ}$ for arbitrary index sets $\cI, \cJ$, where each entry $\mat{F}[I,J](Z)$ is a polynomial in $Z_1, \ldots, Z_n$.
We develop a general framework to analyze concentration of such matrices.
Our matrix concentration results are simpler to state in the case when $Z_1, \ldots, Z_n$ are independent Rademacher variables uniformly distributed in $\pmone$, but apply for the general case as well.

Special cases of such non-linear random matrices have been used in several applications in spectral algorithms and lower bounds. We now briefly discuss a few examples below.

\begin{enumerate}
\item \textbf{Tensor networks.}
Random matrices such as the above were viewed as a special case of ``flattened tensor networks'' by Moitra and Wein~\cite{moitra2019spectral}, who also considered spectral algorithms obtained via somewhat larger tensor networks.
A tensor network is a graph with nodes corresponding to tensors (see the figure below for an example). An edge between two nodes corresponds to shared indices for one of the dimensions and the degree of each node is equal to the order of the corresponding tensor (the number of dimensions).
%
Such networks indicate how tensors of different orders can be multiplied to obtain larger ones.
%
For example, the first network in the figure below illustrates the network corresponding to simple multiplication $\mA \cdot \mB$ of two matrices $\mat{A} \in \R^{m \times n}$ and $\mat{B} \in \R^{n \times m}$, where the red and blue edges indicate the row and column indices respectively.
Similarly, the second network in the figure below illustrates the network corresponding to the application by Hopkins \etal \cite{hopkins2016fast}, where $\mat{T} \in \R^{n  \times n \times m}$ is a random tensor with \iid entries in $\pmone$.
%
While the latter network yields an order-4 tensor, they obtain a matrix in $\RR^{n^2 \times n ^2}$ by ``flattening'' it, where the row is indicated by the indices in the red edges and the column is indicated by the indices in the blue edges.
%
In the figure, we also indicate the index sets corresponding to each of the edges (though these are often supressed in the diagrams).
Moitra and Wein~\cite{moitra2019spectral} analyzed a larger tensor network, with a graph consisting of 10 nodes, in their algorithm for the continuous multi-reference alignment problem.
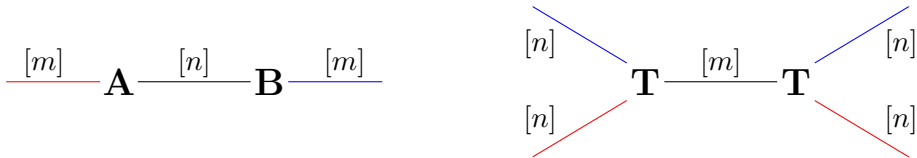
\begin{figure}[!ht]
\label{fig:tensor-network-example}
\begin{center}
\begin{tikzpicture}
\draw (1,2) node {{\large $\mat{A}$}};
\draw (3,2) node {{\large $\mat{B}$}};
\draw (2,2.25) node {\small {$[n]$}};
\draw (0,2.25) node {\small {$[m]$}};
\draw (4,2.25) node {\small {$[m]$}};
\draw (1.25, 2) -- (2.75,2);
\draw[red]  (0.75, 2) -- (-0.5,2);
\draw[blue] (3.25, 2) -- (4.5,2);
\draw (8,2) node {{\large $\mat{T}$}};
\draw (10,2) node {{\large $\mat{T}$}};
\draw (8.25, 2) -- (9.75,2);
\draw[red] (7.75, 1.75) -- (6.5,1);
\draw[blue] (7.75, 2.25) -- (6.5,3);
\draw[red] (10.25, 1.75) -- (11.5,1);
\draw[blue] (10.25, 2.25) -- (11.5,3);
\draw (9,2.25) node {\small {$[m]$}};
\draw (6.6,1.5) node {\small {$[n]$}};
\draw (6.6, 2.5) node {\small {$[n]$}};
\draw (11.4,1.5) node {\small {$[n]$}};
\draw (11.4, 2.5) node {\small {$[n]$}};
\end{tikzpicture}
\end{center}
\caption{Tensor networks for matrix multiplication and the algorithm in \cite{hopkins2016fast}}
\end{figure}

\item \textbf{Graph matrices.}
Another setting of nonlinear concentration arises from the analysis of the so-called ``graph matrices'' ~\cite{medarametla2016bounds, ahn2016graph}. Graph matrices play an important role in lower bounds for average-case problems, against algorithms based on the powerful Sum-of-Squares (SoS) SDP hierarchy running in polynomial time and even sub-exponential time~\cite{meka2015sum, deshpande2015improved, hopkins2015sos, raghavendra2015tight, BHKKMP16, mohanty2020lifting, ghosh2020sum, potechin2020machinery, jones2022sum}.

Let $\mat{X}$ be the $\{\pm1\}$-adjacency matrix of a random graph in $\cG_{n,1/2}$ \ie $\mat{X}[i,j]$ is uniform $\pmone$ when $i \neq j$ and 0 when $i=j$.
Graph matrices are random matrices corresponding to the occurences of a small graph pattern called a ``shape''.
A shape $\tau$ is a small, fixed graph with two ordered subsets $U_{\tau}, V_{\tau}$ of vertices. For simplicity, let $\tau$ be a shape of a fixed size, where the vertex set $V(\tau)$ is partitioned into two ordered sets $V(\tau) = U_{\tau} \sqcup V_{\tau}$.
%
For such a shape $\tau$, the corresponding \emph{graph matrix} $\mat{M}_{\tau}$ has rows and columns indexed by $[n]^{|U_{\tau}|}$ and $[n]^{|V_{\tau}|}$ respectively, and we view the row and column indices $I$ and $J$ as defining a (unique in this case) map $\phi: U_{\tau} \sqcup V_{\tau} \to [n]$. The corresponding entry is given by
\[
\mat{M}_{\tau}[I,J]
~=~ \mat{M}_{\tau}[\phi(U_{\tau}),\phi(V_{\tau})]
~=~
\begin{cases}
\prod_{(u,v)\in E(\tau)}\mat{X}[\phi(u), \phi[v]] & \text{if}~\phi~\text{is injective} \\[5 pt]
0 & \text{otherwise}
\end{cases}
\]
In the case of general graph matrices (defined formally in \cref{sec: dense_graph_matrices}), $U_{\tau}, V_{\tau}$ are arbitrary ordered subsets of the vertex set of $\tau$, and we sum over all feasible injective maps $\phi$ \footnote{In later chapters, for technical reasons, we move to an alternate definition where we sum over distinct Fourier characters as opposed to distinct injective maps}.
As an example, consider the case shown in \cref{fig: tau}, where $\tau$ is a triangle on three vertices $\{u_1, v_1, v_2\}$ with $U_{\tau} = (u_1)$ and $V_{\tau} = (v_1, v_2)$. Then, the corresponding matrix is given by
\[
\mat{M}_{\tau}[i_1,(i_2,i_3)] ~=~ \mat{X}[i_1,i_2] \cdot \mat{X}[i_2,i_3] \cdot \mat{X}[i_3,i_1] \mcom
\]
where $\mat{X}$ automatically enforces injectivity.
%

Graph matrices are closely related to tensor networks (ignoring the injectivity constraint on $\phi$). For instance, the above matrix can be viewed as the flattened tensor network below, where the tensor $\mat{I}$ denotes the ``diagonal'' tensor of order 3 with entries being 1 if all indices are equal and 0 otherwise.
%
\begin{figure}[!ht]
\begin{center}
\begin{tikzpicture}
\draw (1,4) node {$u_1$};
\draw (1,4) circle (0.5 cm);
\draw[dashed, blue] (0.35,3.35) rectangle (1.65, 4.65);
\draw (3,5) node {$v_1$};
\draw (3,5) circle (0.5 cm);
\draw (3,3) node {$v_2$};
\draw (3,3) circle (0.5 cm);
\draw[dashed, red] (2.35,2.35) rectangle (3.65, 5.65);
\draw (1.5,4) -- (2.5,5);
\draw (1.5,4) -- (2.5,3);
\draw (3,4.5) -- (3,3.5);
\draw (8,4) node {\large $\mat{I}$};
\draw (11,5.5) node {\large $\mat{I}$};
\draw (11,2.5) node {\large $\mat{I}$};
\draw (9.5,4.75) node {\large $\mat{X}$};
\draw (11,4) node {\large $\mat{X}$};
\draw (9.5,3.25) node {\large $\mat{X}$};
\draw (8.25,4.15) -- (9.25,4.65);
\draw (10.75,5.35) -- (9.75,4.85);
\draw (8.25,3.85) -- (9.25,3.4);
\draw (10.75,2.65) -- (9.75,3.15);
\draw (11,2.75) -- (11,3.75);
\draw (11,5.25) -- (11,4.25);
\draw[blue] (7.75,4) -- (7,4);
\draw[red] (11.15,5.65) -- (11.7,6.2);
\draw[red] (11.15,2.35) -- (11.7,1.8);
\end{tikzpicture}
\end{center}
\caption{The graph $\tau$ and corresponding flattened tensor network}
\label{fig: tau}
\end{figure}
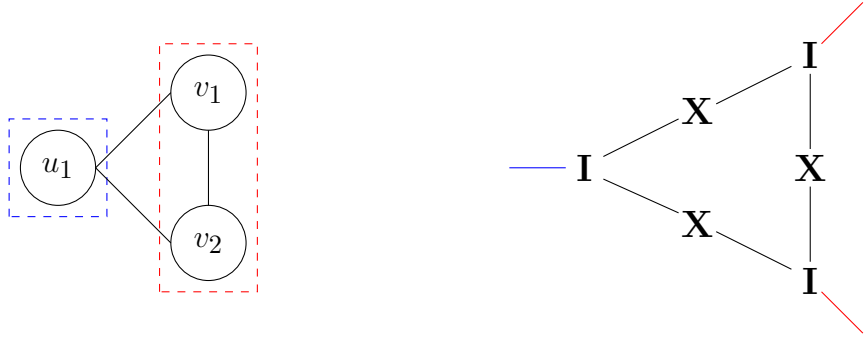
%
%
%

\end{enumerate}

\paragraph{Analyzing concentration}
Recall that our objective is to analyze the concentration of polynomial random matrices.
To motivate our approach, consider first the problem of obtaining concentration bounds on a \emph{scalar} polynomial $f(Z)$ with mean zero. To obtain such bounds, because of Markov's inequality, it suffices to compute moment estimates
\[
\Pr{\abs{f(Z)} \geq \lambda}
~=~ \Pr{\inparen{f(Z)}^{2t} \geq \lambda^{2t}}
~\leq~ \lambda^{-2t} \cdot {\Ex{\inparen{f(Z)}^{2t}}}
\]
While in some cases $\Ex{\inparen{f(Z)}^{2t}}$ can be computed by direct expansion, it often involves an intricate analysis of the structure of terms with degrees growing with $t$, and therefore indirect methods may be more convenient.
One such method is based on hypercontractive inequalities.
In particular for Rademacher variables, the hypercontractive inequality~\cite{ODonnell08} gives that for a polynomial $f$ of degree $d_p$, we have
\[
\Ex{\inparen{f(Z)}^{2t}} ~\leq~ (2t-1)^{d_p \cdot t} \cdot \inparen{\Ex{\inparen{f(Z)}^2}}^t \mper
\]
Thus, for (scalar) polynomial functions, the hypercontractive inequality gives moment estimates using $\inparen{f(Z)}^2$, which is convenient because $\inparen{f(Z)}^2$ is a polynomial of \emph{fixed} degree and therefore is much easier to understand. In fact, it can often be conveniently analyzed using the Fourier coefficients of $f$.

The matrix analog of the above argument involves the Schatten-$2t$ norm $\norm{.}_{2t}$, which is defined for a matrix $\mat{M} $ with non-zero singular values $\sigma_1, \ldots, \sigma_r$ as $\norm{\mat{M}}_{2t}^{2t} ~:=~ \sum_{j \in [r]} \sigma_j^{2t}$.
For a function $\mat{F}$ with $\Ex{\mat{F}(Z)} = 0$, we have the following bound using Schatten norms.
\[
\Pr{\sigma_1(\mat{F}) \geq \lambda}
~\leq~ \lambda^{-2t} \cdot \Esch{\mF}{2t} 
~=~ \lambda^{-2t} \cdot \Etr{\inparen{\mF(Z)\mF(Z)^\T}^{t}}
\]
Known norm bounds for tensor networks~\cite{moitra2019spectral} (which involves Gaussian variables) and graph matrices~\cite{ahn2016graph, jones2022sum} rely on direct expansion of the trace above. They analyze terms in the expansion as being formed by $2t$ copies of the network/shape, which leads them to consider graphs formed by $2t$ copies of the network/shape, with possibly overlapping vertex sets. To analyze such graphs, they both rely on intricate combinatorics.

Hypercontractive inequalities are also known for matrix-valued functions of Rademacher variables \cite{BARDW08}. However, their form involves Schatten-$p$ norms for $p \in [1,2]$ and (to the best of our knowledge) are not known to imply matrix concentration.
%
To get around this, we consider another indirect method based on Efron-Stein inequalities. In the scalar case, Efron-Stein inequalities gives us a slight weakening of the above scalar bound. Interestingly, it turns out that this can indeed be generalized to the matrix case.

\paragraph{Efron-Stein inequalities.}
Efron-Stein inequalities bound the global variance of a function of independent random variables, in terms of local variance estimates obtained by changing one variable at a time.
For $i \in [n]$ and tuple $Z = (Z_1, \ldots, Z_n)$, let $Z^{(i)}$ denote the tuple $(Z_1, \ldots, Z_{i-1}, \resamp{Z_i}, Z_{i+1}, \ldots, Z_n)$, where $\resamp{Z_i}$ is an independent copy of $Z_i$.
For a scalar function $f(Z)$, the Efron-Stein inequality states that
\[
\Var{f(Z)} ~=~ \Ex{\inparen{f(Z) - \EE f}^2} ~\leq~ \frac12 \cdot \sum_{i \in [n]} \Ex{\inparen{f(Z) - f\inparen{Z^{(i)}}}^2} ~=~ \Ex{V(Z)}\mcom
\]
where $V(Z) ~:=~ \sum_{i \in [n]} \EE \insquare{\inparen{f(Z) - f\inparen{Z^{(i)}}}^2 | Z}$.
For Rademacher variables, $\Ex{V(Z)}$ is equal to the total influence from boolean Fourier analysis and indeed, the above inequality can also be observed via Fourier analysis. In fact, when $f$ is a polynomial of degree $d_p$, the two sides are within a factor $d_p$.

A moment version of the Efron-Stein inequality was developed by Boucheron \etal~\cite{BBLM05}, who obtain bounds in terms of $V(Z)$ (in fact, in terms of more refined quantities $V_+(Z)$ and $V_{-}(Z)$) which serves as a proxy for the variance. Their results imply that for a function $f$,
\[
\Ex{\inparen{f(Z) - \EE f}^{2t}} ~\le~ (C_0 \cdot t)^t \cdot \Ex{\inparen{V(Z)}^{t}} \mper
\]
A beautiful matrix generalization of the above inequality (\cref{thm: main_efron_stein} below) was obtained by Paulin, Mackey and Tropp~\cite{paulin2016}, via the method of exchangeable pairs (see also~\cite{HT21:poincare} for a different proof).
Their inequality is stated for Hermitian matrix valued functions $\mH$. But we can also use it for non-Hermitian functions $\mat{F}$, where we simply apply it to the Hermitian dilation $\mH = \begin{bmatrix}
	0 & \mF\\
	\mF^\T & 0
\end{bmatrix}$ instead.
\begin{theorem}[\cite{paulin2016}]\label{thm: main_efron_stein}
Let $\mat{H}(Z)$  be a Hermitian matrix valued function of independent random variables $Z = (Z_1, \ldots, Z_n)$ with $\ExpOp\norm{\mat{H}} < \infty$.
Then, for each natural number $t \ge 1$,
\[
\Etr{(\mat{H} - \EE\mat{H})^{2t}}  ~\le~  (4t-2)^t \cdot \Etr{\mV^t} \mcom
\]
where $\mV(Z)$ is the variance proxy defined as
\[
\mV(Z) := \frac{1}{2} \cdot \sum_{i = 1}^n \Ex{\inparen{\mat{H}(Z) - \mat{H}\inparen{Z^{(i)}}}^2 \mid Z} \mper
\]
\end{theorem}

\paragraph{A simple bound for Rademacher variables.}
The form of the variance proxy suggests a recursive approach for polynomial functions (say of degree $d_p$) of Rademacher variables. Consider the scalar case again in particular the Efron-Stein inequality by Boucheron \etal~\cite{BBLM05}, where the variance proxy can be written as
\begin{align*}
V(Z)
~=~ \frac12 \cdot \sum_{i \in [n]} \Ex{\inparen{f(Z) - f\inparen{Z^{(i)}}}^2 \mid Z}
&~=~ \frac12 \cdot \sum_{i \in [n]} \Ex{(Z_i - \resamp{Z_i})^2 \cdot \inparen{\frac{\partial f(Z)}{\partial Z_i}}^2 \mid Z} \\
&~=~ \sum_{i \in [n]} \inparen{\frac{\partial f(Z)}{\partial Z_i}}^2
~=~ \norm{\mat{f}_1(Z)}_2^2 \mcom
\end{align*}
where $\mat{f}_1(Z)$ is a vector-valued function given by $\mat{f}_1[i](Z) = \frac{\partial f(Z)}{\partial Z_i}$. Thus, to estimate $\ExpOp \inparen{f(Z)}^{2t}$, we just need to estimate $\ExpOp \norm{\mat{f}_1(Z)}_2^{2t}$, where $\mat{f}_1(Z)$ is now a vector valued function. The key observation is that $\mat{f}_1(Z)$ has entries of degree at most $d_p - 1$. This suggests that we can apply this inequality recursively until we end up with constant polynomials, which we fully understand.
We can do a similar computation for matrix-valued functions $\mat{F}(Z)$ using \cref{thm: main_efron_stein}. This yields two matrices $\mat{F}_{0,1}$ and $\mat{F}_{1,0}$ of partial derivatives, where an extra index $i$ is added either to the row or column indices. Iterating this yields the following result, which we state in terms of the partial derivative operators
$\grad_{\alpha}(f) = \inparen{\prod_{i: \alpha_i = 1} \frac{\partial}{\partial Z_i}}(f)$ for $\alpha \in \B^n$ (extended entry-wise to matrices).
\begin{restatable}[Rademacher recursion]{theorem}{basicframework}\label{thm: main_rademacher}
Let $\mat{F}: \pmone^n \to \R^{\cI \times \cJ}$ be a matrix valued polynomial function of degree at most $d_p$. Then, for each natural number $t \ge 1$,
\[
\Esch{\mF - \EE \mF}{2t} ~\le~ \sum_{1 \le a + b \leq \dpoly} (16t\dpoly)^{(a + b) \cdot t} \cdot \sch{\EE\mF_{a, b}}{2t} \mcom
\]
where $\mF_{a,b}$ is a matrix of partial derivatives indexed by the sets $\cI \times \B^n$ and $\cJ \times \B^n$ with
\[
\mF_{a, b}[(\cdot, \al), (\cdot, \beta)] = \begin{dcases}
	\grad_{\alpha + \beta}(\mF) & \text{ if $|\al| = a, |\beta| = b, \al\cdot \beta = 0$}\\
	0 & \text{otherwise}
      \end{dcases}
    \]
\end{restatable}
%
Similar to the hypercontractive bound for the scalar case, the bound above is in terms of a small number ($O(d_p^2)$) of matrices that arise from polynomials of fixed degree (not growing with $t$), but importantly, they are \emph{deterministic} matrices. Because they are deterministic, analyzing them is considerably easier.
When we apply this theorem to the case $\mF = \mM_{\tau}$, the graph matrix of a shape $\tau$, we obtain bounds in terms of combinatorial objects known as ``vertex separators'' of the shape $\tau$. This recovers the bounds by Ahn \etal~\cite{ahn2016graph} and perhaps surprisingly (to the authors), this gives an alternative and direct derivation of these combinatorial structures such as vertex separators, compared to the ingenious observations made in Ahn \etal~\cite{ahn2016graph}. We cover this and other applications of the Rademacher framework in \cref{sec:rademacher-applications}.
\paragraph{Extending the framework to general product distributions.}
A key contribution of our work is to show how the above framework can be extended to arbitrary product distributions (with bounded moments).
A motivating example of this is norm bounds for the so-called ``sparse graph matrices''. In sparse graph matrices, the variables $Z_i$ can be thought of as (normalized) edges of a $\cG_{n', p}$ graph, that is, $Z_i = -\sqrt{\frac{1-p}{p}}$ with probability $p$ and $Z_i = \sqrt{\frac{p}{1-p}}$ with probability $1-p$. These variables are standard in $p$-biased Fourier analysis~\cite{o2014analysis} and are chosen to satisfy $\ExpOp Z_i = 0$ and $\ExpOp Z_i^2 = 1$. Sparse graph matrices naturally arise when analyzing average case problems on $\cG_{n, p}$ graphs for $p = o(1)$, as opposed to $\cG_{n, 1/2}$ graphs.

Until recently, little was known about norm bounds for sparse graph matrices. The difficulty stems partly from the fact that when $p = o(1)$, it is important that sparse graph matrix norm bounds have the right dependence on $p$ and not just on $n$. Such norm bounds were obtained recently by Jones \etal~\cite{jones2022sum}, via the trace power method which involved a delicate combinatorial counting argument.
On the other hand, we obtain similar norm bounds using our framework but in a more mechanical fashion.
%
%
We can also readily apply our framework in the even more general case of sub-Gaussian random variables and our bounds will depend on the sub-Gaussian norm of the distributions.

To extend our framework to general product distributions, we could take inspiration from the Rademacher case and could attempt to simply recursively apply the Efron-Stein inequality. Unfortunately, this idea will fail. The issue can be observed by again considering the scalar case.
%
Assume that $Z_1, \ldots, Z_n$ are \iid with $\ExpOp Z_i = 0$ and $\ExpOp Z_i^2 = 1$ for all $i \in [n]$.
Also assume for simplicity that $f(Z)$ is a multi-linear polynomial of degree $\dpoly$. Analyzing the variance proxy as before, we get
\[
V(Z)
~=~ \frac12 \cdot \sum_{i \in [n]} \Ex{(Z_i - \resamp{Z_i})^2 \cdot \inparen{\frac{\partial f(Z)}{\partial Z_i}}^2 \mid Z}
~=~ \frac12 \sum_{i \in [n]} \Ex{(Z_i - \resamp{Z_i})^2 | Z} \cdot \inparen{\frac{\partial f(Z)}{\partial Z_i}}^2 \mper
\]
In the Rademacher case, we had $\Ex{(Z_i - \resamp{Z_i})^2 | Z} = 2$. This left us with the polynomials corresponding to partial derivatives but which importantly had a strictly lower degree. However, for a general product distribution, we instead have $\Ex{(Z_i - \resamp{Z_i})^2 | Z} = 1+Z_i^2$. This gives back a term $\inparen{Z_i \cdot \frac{\partial f}{\partial Z_i}}^2$ where the polynomial inside the square could have degree possibly still equal to $d_p$.
%
This means that in the next step of the recursion, we may again have to consider a derivative with respect to $Z_i$ and may again end up with the same polynomial $f$. Therefore, the recursion is stalled! A similar issue occurs for matrices, which is elaborated in \cref{sec: failure_of_basic}. To get around this, we generalize the work of \cite{paulin2016}.

\paragraph{Generalizing \cite{paulin2016} via explicit inner kernels.}
To resolve the above issue, we modify the proof of \cite{paulin2016} and our proof techniques may be of independent interest.

We first recall how the matrix Efron-Stein inequality, \cref{thm: main_efron_stein}, was proved in \cite{paulin2016}. Their basic strategy is to utilize the theory of \textit{exchangeable pairs} \cite{stein1972bound, stein1986approximate, chatterjee2005concentration, chatterjee2006stein}, in particular \textit{kernel Stein pairs}.
%
A kernel Stein pair is an exchangeable pair of random matrices that has a ``kernel'', a bivariate function that ``reproduces'' the matrices in the pair.
More concretely, consider an exchangeable pair of random variables $(Z,Z')$ (which means $(Z',Z)$ has the same distribution). For this exchangeable pair, a bivariate matrix-valued function $\mK(z,z')$ is said to be a kernel for a matrix-valued function $\mF$ if it satisfies
\begin{itemize}
    \item Anti-symmetry: $\mK(z',z) = -\mK(z,z')$ for all inputs $(z,z')$.
    \item Reproducing property: $\Ex{\mK(Z,Z') \mid Z} = \mF(Z)$.
\end{itemize}
If such a kernel $\mK$ exists, then the pair of random variables $(\mF(Z), \mF(Z'))$ is said to be a kernel Stein pair.

Building on ideas from \cite{stein1986approximate, chatterjee2005concentration}, Paulin, Mackey and Tropp~\cite{paulin2016} first show the existence of a kernel, by exhibiting it as a limit of coupled Markov Chains. By studying the evolution of this kernel coupling, they prove analytic properties of the kernel.
Then, using this kernel, they employ the powerful method of exchangeable pairs to evaluate moments of the random matrix, which in turn will imply concentration.

For a Hermitian random matrix $\mX$, they introduce two matrices - the \textit{conditional variance} $\mV_{\mX}$ which measures the squared fluctuations of $\mX$ when resampling a coordinate of $Z$; and the \textit{kernel conditional variance} $\mV^{\mK}$ which measures the squared fluctation of the kernel when resampling a coordinate of $Z$. With these matrices in hand, they bound the Schatten $2t$-norm of $\mX$ by the Schatten $t$-norm of $s\mV_{\mX} + s^{-1}\mV^{\mK}$ for any parameter $s > 0$. Finally, they choose $s$ appropriately to make these two quantities approximately equal, in which case it simplifies to the variance proxy $\mV$, proving \cref{thm: main_efron_stein}.

In our setting, no such choice of $s$ is feasible because for any choice of $s$, either the conditional variance term $s\mV_{\mX}$ will dominate $\mX^2$ or the kernel conditional variance term $s^{-1}\mV^{\mK}$ will dominate $\mX^2$. This will make the main inequality \cref{thm: main_efron_stein} trivial.

To get around this, we will exploit the structure of the matrix we have, i.e. $\mF = \mD\mG\mD$ where $\mD$ is a diagonal matrix that encodes all variables that have already been differentiated on and $\mG$ is a polynomial matrix of the remaining variables. Since $\mD$ is a simple diagonal matrix with low degrees, most of the deviations exhibited by $\mF$ are in fact likely to be exhibited by $\mG$. To capture this intuition, we consider a kernel for only the inner matrix $\mG$ instead of $\mF$ as a whole. We call this an \textit{inner kernel}.

This helps us avoid the root cause of the issue, i.e. differentiating on variables we have already encountered (which correspond to entries in $\mD$).
Therefore, the recursion will not stall!

However, in general, this is not realizable since $\mD$ and the kernel of $\mG$ can interact in unexpected ways. To study this interaction, we construct explicit polynomial kernels (\cref{thm: explicit_kernel_for_poly}) (compared to \cite{paulin2016} who show the existence of the kernel but for all functions).

We study how this explicit inner kernel interacts with $\mD$ (see \cref{lem: props_of_exp_kernel_mat}) and use it to obtain a generalization of the inequalities by \cite{paulin2016} (generalized because setting $\mD = \mI$ will give back their result) stated in \cref{lem: main_pmt_bound}.

A subtle issue is that the conditional variance of $\mX$ may still have additional deviations due to the diagonal matrices $\mD$ (which still involve random variables). We control the additional deviations using Jensen's operator trace inequality (for non-commuting averages)~\cite{hansen2003jensen} (stated in \cref{lem: jensen_trace}).
Putting these ideas together lets us obtain a version of the Efron-Stein inequality where the variance proxy only corresponds to the conditional variance of the inner kernel. In the setting of polynomial functions, this inequality generalizes the work of \cite{paulin2016}.

With the modified Efron-Stein inequality from above, we cannot guarantee that the matrices $\mF$ at intermediate steps are of lower degree, but on the other hand, the degree of the inner matrix $\mG$ reduces at each step. Therefore, we can recursively apply this inequality to obtain our final bounds. The final bounds are then stated in terms of norm bounds for the simplified matrices of the form $\mD\mG\mD$ where $\mG$ are deterministic matrices and $\mD$ are diagonal matrices which are still functions of $Z$.
%
While random, these matrices can be easily analyzed via simple scalar concentration tools.
%

The main theorem is stated in \cref{sec: general_recursion}, in particular \cref{thm: main_general}, with the proof following in \cref{sec: proof_of_general}. While our proof builds on the work by \cite{paulin2016}, the argument here is self-contained.

\paragraph{Applications.}
Our framework is suitable for many nonlinear concentration results obtained in the literature \cite{barak2012hypercontractivity, ge2015decomposing, hopkins2015tensor, medarametla2016bounds, ahn2016graph, hopkins2016fast, schramm2017fast, hopkins2018statistical, hopkins2019robust, moitra2019spectral, jones2022sum}.
%
We show a few of these applications in \cref{sec:rademacher-applications} and \cref{sec: sparse_graph_matrices}.
We expect similar future applications to benefit from our framework because the task is mechanically reduced to analyzing considerably simpler matrices.

In \cref{sec: dense_graph_matrices}, we derive norm bounds on dense graph matrices. In earlier works, dense graph matrices have been used extensively in analysis of semidefinite programming hierarchies, especially the Sum of Squares (SoS) hierarchy \cite{meka2015sum, deshpande2015improved, hopkins2015sos, raghavendra2015tight, BHKKMP16, mohanty2020lifting, ghosh2020sum, potechin2020machinery}. For more applications and a detailed treatment of graph matrices, see \cite{ahn2016graph}.

In \cref{sec: sparse_graph_matrices}, we derive norm bounds for sparse graph matrices. Sparse graph matrices have been relatively less understood until recently, when \cite{jones2022sum} obtained norm bounds for such matrices via the trace power method. They use these bounds to prove SoS lower bounds for the maximum independent set problem on sparse graphs.

\paragraph{Potential extensions}\label{par: extension}

In this work, we assumed that the input forms a product distribution. In other words, the variables $Z_1, \ldots, Z_n$ are independent. A natural extension is the case when they are not independent. This has important applications for many problems such as when the input is a uniform $d$-regular graph, or when the input is sampled from a distribution with a global constraint, etc. In such cases, the input variables are not independent but it may be possible to use similar ideas to analyze concentration.

More concretely, to study concentration in the non-independent setting, one can use the recent work of Huang and Tropp~\cite{HT21:poincare} on matrix concentration from \Poincare inequalities, together with our framework. For this, we just need to exhibit a Markov process that converges to our desired distribution.

\paragraph{Organization of the chapter}

We start with preliminaries in \cref{sec: prelims}. In \cref{sec: basic_recursion}, we state and prove the Rademacher recursion. We illustrate some applications of this framework in \cref{sec:rademacher-applications}. In \cref{sec: failure_of_basic}, we explain why similar ideas may not be enough in the general case. We then propose our general framework in \cref{sec: general_recursion} and prove it in \cref{sec: proof_of_general}. We end with an application of the general framework to sparse graph matrices in \cref{sec: sparse_graph_matrices}.

\section{Preliminaries}\label{sec: prelims}
\paragraph{Notation}

We use boldface letters such as $\mI, \mM, \mX\ldots, $ to denote matrices.
Entries of a matrix $\mX \in \R^{\cI \times \cJ}$ will be denoted by $\mX[I,J]$ for $I \in \cI, J \in \cJ$. Let $\HH^n$ denote the set of $n \times n$ real symmetric matrices. The trace of a matrix $\mX \in \HH^n$ equals $\sum_{i \in [n]} \mX[i,i]$ and is denoted by $\tr \mX$.

\subsubsection*{Multi-index notation}

For any pair of vectors $\al, \beta \in \NN^n$ and scalar $c \in \NN$, we define $\al + \beta, \al \cdot \beta, c\al$ entrywise. We also define the orderings $\al \le \beta$ and $\al \unlhd \beta$ where we say $\al \le \beta$ if for each $i$, $\al_i \le \beta_i$, and $\al \unlhd \beta$ if for each $i$, $\al_i$ is either $0$ or $\beta_i$. We denote by $|\al|_0$ the number of nonzero entries of $\al$ and by $|\al|_1$, the sum of entries of $\al$. For a boolean vector $\gam \in \{0, 1\}^n$, we define $1 - \gam$ the vector with all its bits flipped.




\subsubsection*{Derivatives}

For variables $Z_1, \ldots, Z_n$ and $\al \in \NN^n$, define the monomial $Z^{\al} := \prod_{i = 1}^n Z_i^{\al_i}$. This forms a standard basis for polynomials.

For $\al \in \NN^n$, we define the linear operator $\grad_{\al}$ that acts on polynomials by defining its action on the elements $Z^{\beta}$ as follows and then extend linearly to all polynomials.
\[\grad_{\al}(Z^{\beta}) = \begin{dcases}
	Z^{\beta - \al} & \text{ if $\al \unlhd \beta$}\\
	0 & \text{ o.w.}
\end{dcases}\]

Informally, for a polynomial $f$ written as a linear combination of the standard basis polynomials $Z^{\beta}$, $\grad_{\al}(f)$ isolates the terms that precisely contain the powers $Z_i^{\al_i}$ for all $i$ such that $\al_i \neq 0$ and then truncates these powers. In other words, it's the coefficient of $Z^{\alpha}$ in $f$. In particular, observe that $\grad_{\al}(f)$ does not depend on $Z_i$ for any $i$ such that $\al_i \neq 0$.

Supose $f$ is multilinear, as we can assume in the Rademacher case when we are working with $Z_i \in \pmone$. For $\al \in \{0, 1\}^n$ with nonzero indices $i_1, \ldots, i_k \in [n]$, we have $\grad_{\al}(f) = \frac{\partial}{\partial Z_{i_1}}\ldots \frac{\partial}{\partial Z_{i_k}}f$. So this linear operator generalizes the partial derivative operator. But note that in general, $\grad$ is not simply the standard partial derivative operator.


\subsubsection*{Matrix Analysis}

Linear operators that act on polynomials can also be naturally defined to act on matrices by acting on each entry.

We define $\mI_m$ to be the $m \times m$ identity matrix. We drop the subscript when it's clear.
For matrices $\mF, \mG$, define $\mF \oplus \mG$ to be the matrix $\begin{bmatrix}
	0 & \mF\\
	\mG & 0
\end{bmatrix}$. For a matrix $\mF$, define its Hermitian dilation $\herm{\mF}$ as $\mF \oplus \mF^T$. Denote by $\preceq$ the Loewner order, that is, $\mA \preceq \mB$ for $\mA, \mB \in \HH^n$ if and only if $\mB - \mA$ is positive semi-definite.

\begin{definition}
	For a matrix $\mF$ and an integer $t \ge 0$, define the Schatten $2t$-norm as
	\[\norm{\mF}_{2t}^{2t} = \tr[{(\mF\mF^T)^t}]\]
\end{definition}

\begin{fact}\label{fact: cs}
	For real symmetric matrices $\mX_1, \ldots, \mX_n$, we have
	\begin{align*}
		(\mX_1 + \ldots + \mX_n)^2 \preceq n(\mX_1^2 + \ldots + \mX_n^2)
	\end{align*}
\end{fact}

\begin{fact}\label{fact: holder}
	For positive semidefinite matrices $\mX, \mX_1, \ldots, \mX_n$ such that $\mX \preceq \mX_1 + \ldots + \mX_n$ and for any integer $t \ge 1$,
	\begin{align*}
		\tr [\mX^t] \le n^{t - 1}(\tr[\mX_1^t] + \ldots + \tr[\mX_n^t])
	\end{align*}
\end{fact}

\begin{proof}
    By H\"{o}lder's inequality, $n^{t - 1}(\tr[\mX_1^t] + \ldots + \tr[\mX_n^t]) \ge (\norm{\mX_1}_t + \ldots + \norm{\mX_n}_t)^t$. By triangle inequality of Schatten norms, this is at least $\norm{\mX_1 + \ldots + \mX_n}_t^t$. Finally, because $\mX_1 + \ldots + \mX_n \succeq \mX\succeq 0$, we can use the monotonicity of trace functions (see \cite[Proposition 1]{petz1994survey}) where we use the increasing function $f(x) = x^t$ on $x \in [0, \infty)$. This proves the result.
\end{proof}

\begin{lemma}[Jensen's operator trace inequality]\cite[Corollary 2.5]{hansen2003jensen}\label{lem: jensen_trace}
	Let $f$ be a convex, continuous function defined on an interval $I$ and suppose that $0 \in I$ and $f(0) \le 0$. Then, for all integers $m, n \ge 1$, for every tuple $\mB_1, \ldots, \mB_n$ of real symmetric $m \times m$ matrices with spectra contained in $I$ and every tuple $\mA_1, \ldots, \mA_n$ of $m \times m$ matrices with $\sum_{i = 1}^n \mA_i^T\mA_i \preceq \mI$, we have
	\[\tr[f(\sum_{i = 1}^n \mA_i^T \mB_i \mA_i)] \le \tr[\sum_{i = 1}^n \mA_i^T f(\mB_i) \mA_i]\]
\end{lemma}

\section{The basic framework for Rademacher random variables} \label{sec: basic_recursion}


%
%
Let $Z = (Z_1, \ldots, Z_n)$ be sampled uniformly from $\{-1, 1\}^n$.
We will consider matrix-valued functions $\mF: \pmone^n \to \RR^{\cI \times \cJ}$, with rows and columns indexed by arbitrary sets $\cI, \cJ$ respectively such that for all $I \in \cI, J \in \cJ$,
\[
\mF[I, J] ~=~ f_{I, J}(Z)
\]
where $f_{I, J}$ are polynomials of $Z_1, \ldots, Z_n$.
Since $Z_i \in \{-1, 1\}$, we can assume without loss of generality that $f_{I, J}$ are multilinear.
Let $\dpoly$ be the maximum degree of any $f_{I, J}$ in $\mF$.
In this section, we will give a general framework using which we can obtain bounds on
$\Esch{\mF - \EE\mF}{2t}$ for any integer $t \ge 1$.

\basicframework*

\begin{remark}
Note that while the matrices $\mF_{a,b}$ are stated above as having rows and colmns indexed by $\cI \times \B^n$ and $\cJ \times \B^n$ for convenience, we only need to consider the submatrices with $\abs{\cI} \cdot \binom{n}{a}$ rows and $\abs{\cJ} \cdot \binom{n}{b}$ columns, since all other entries will be zero (when $\abs{\alpha} \neq a$ or $\abs{\beta} \neq b$).
\end{remark}

\begin{remark}
	To obtain high probability norm bounds from moment estimates, we can set $t = \polylog(n)$ and invoke Markov's inequality. Since we do not attempt to optimize the dependence on the logarithmic factors, we do not attempt to optimize the exponent of $t$ in the main theorem.
\end{remark}

To prove this, we will prove \cref{lem: main_rademacher} and then recursively apply it.

For each $i \le n$, define the random vector
\[Z^{(i)} ~:=~ (Z_1, \ldots, Z_{i - 1}, \resamp{Z_i}, Z_{i + 1}, \ldots, Z_n)\]
where $\resamp{Z_i}$ is an independent copy of $Z_i$, that is,
is independently resampled from $\{-1, 1\}$.

Let $\mX := \mF- \EE\mF$. When the input is $Z$, we denote the matrices as $\mF, \mX$, etc and when the input is $Z^{(i)}$, denote the corresponding matrices as $\mF^{(i)}, \mX^{(i)}$, etc. That is, for $I \in \cI, J \in \cJ$, we have $\mF^{(i)}[I, J] = f_{I, J}(Z^{(i)})$. Define $\mX_{a, b} = \mF_{a, b} - \EE\mF_{a, b}$.

\begin{lemma}\label{lem: main_rademacher}
	For integers $a, b \ge 0$, we have
	\[\Esch{\mX_{a, b}}{2t} \le (16t\dpoly)^t(\Esch{\mX_{a, b + 1}}{2t} + \Esch{\mX_{a + 1, b}}{2t} + \sch{\EE\mF_{a, b + 1}}{2t} +\sch{\EE\mF_{a + 1, b}}{2t})\]
\end{lemma}

Using this lemma, we can complete the proof of the main theorem.

\begin{proof}[Proof of \cref{thm: main_rademacher}]
	Observing that $\mX$ is a principal submatrix of $\mX_{0, 0}$ with all other entries being $0$, we can apply \cref{lem: main_rademacher} repeatedly until $\mX_{a, b} = 0$, which will be the case if $a + b > \dpoly$.
\end{proof}

In the rest of this section, we will prove \cref{lem: main_rademacher}. We start with a basic fact. Let $\mat{e}_i \in \{0, 1\}^n$ be the vector with a unique nonzero entry $(\mat{e}_i)_i = 1$.

\begin{propn}\label{propn: basic}
	For a multilinear polynomial $f(Z) = f(Z_1, \ldots, Z_n)$, we have
	\[f(Z) - f(Z^{(i)}) ~=~ (Z_i - \resamp{Z_i})\cdot \grad_{\mat{e}_i}f(Z)\]
\end{propn}

\begin{proof}[Proof of \cref{lem: main_rademacher}]
	Consider the Hermitian dilation $\herm{\mF}_{a, b} = \mF_{a, b} \oplus \mF_{a, b}^T$. Define $\herm{\mX}_{a, b} = \herm{\mF}_{a, b} - \EE \herm{\mF}_{a, b} = \mX_{a, b} \oplus \mX_{a, b}^T$. By \cref{thm: main_efron_stein} applied to $\herm{\mX}_{a, b}$,
	\[\Etr{\herm{\mX}_{a, b}^{2t}} \le (2(2t - 1))^t \Etr{\mV_{a, b}^t}\]
	where $\mV_{a, b}$ is the variance proxy
	\[\mV_{a, b} = \frac{1}{2} \sum_{i = 1}^n\EE[(\herm{\mX}_{a, b} - \herm{\mX}^{(i)}_{a, b})^2 | Z]\]

	Firstly, by a simple computation,
	\[\Etr{\herm{\mX}_{a, b}^{2t}} = \Etr{(\mX_{a, b}\mX_{a, b}^\T)^t} + \Etr{(\mX_{a, b}^\T\mX_{a, b})^t} = 2 \Esch{\mX_{a, b}}{2t}\]
	and
	\begin{align*}
		\mV_{a, b} &= \frac{1}{2} \sum_{i = 1}^n \EE[(\herm{\mX}_{a, b} - \herm{\mX}_{a, b}^{(i)})^2|Z]\\
		&= \frac{1}{2}\sum_{i = 1}^n \EE\bigg[\begin{bmatrix}
			(\mX_{a, b} - \mX_{a, b}^{(i)})(\mX_{a, b} - \mX_{a, b}^{(i)})^\T & 0\\
			0 & (\mX_{a, b} - \mX_{a, b}^{(i)})^\T(\mX_{a, b} - \mX_{a, b}^{(i)})
		\end{bmatrix}|Z\bigg]\\
		&= \frac{1}{2} \begin{bmatrix}
			\sum_{i = 1}^n\EE[(\mF_{a, b} - \mF_{a, b}^{(i)})(\mF_{a, b} - \mF_{a, b}^{(i)})^\T|Z] & 0\\
			0 & \sum_{i = 1}^n\EE[(\mF_{a, b} - \mF_{a, b}^{(i)})^\T(\mF_{a, b} - \mF_{a, b}^{(i)})|Z]
		\end{bmatrix}
	\end{align*}

	We will use the following claim that we will prove later.
	\begin{claim}\label{claim: reduction}
		We have the following relations.
		\[\sum_{i = 1}^n\EE[(\mF_{a, b} - \mF_{a, b}^{(i)})(\mF_{a, b} - \mF_{a, b}^{(i)})^\T|Z] = 2(b + 1)\mF_{a, b + 1}\mF_{a, b + 1}^\T\]
		\[\sum_{i = 1}^n\EE[(\mF_{a, b} - \mF_{a, b}^{(i)})^\T(\mF_{a, b} - \mF_{a, b}^{(i)})|Z] = 2(a + 1)\mF_{a + 1, b}^\T\mF_{a + 1, b}\]
	\end{claim}
	This gives $\Etr{\mV_{a, b}^t} = (b + 1)^t\Esch{\mF_{a, b + 1}}{2t} + (a + 1)^t\Esch{\mF_{a + 1, b}}{2t}$. Therefore, we get
    {\footnotesize
	\begin{align*}
		2 \Esch{\mX_{a, b}}{2t} &= \Etr{\herm{\mX}_{a, b}^{2t}}\\
		&\le (2(2t - 1))^t \Etr{\mV_{a, b}^t}\\
		&\le (2(2t - 1))^t((b + 1)^t\Esch{\mF_{a, b + 1}}{2t} + (a + 1)^t\Esch{\mF_{a + 1, b}}{2t})\\
		&\le (2(2t - 1))^t((b + 1)^t\Esch{\mX_{a, b + 1} + \EE\mF_{a, b + 1}}{2t} + (a + 1)^t\Esch{\mX_{a + 1, b} + \EE\mF_{a + 1, b}}{2t})\\
		&\le (16t)^t((b + 1)^t(\Esch{\mX_{a, b + 1}}{2t} + \sch{\EE\mF_{a, b + 1}}{2t}) + (a + 1)^t(\Esch{\mX_{a + 1, b}}{2t} + \sch{\EE\mF_{a + 1, b}}{2t})\\
		&\le (16t\dpoly)^t(\Esch{\mX_{a, b + 1}}{2t} + \sch{\EE\mF_{a, b + 1}}{2t} + \Esch{\mX_{a + 1, b}}{2t} + \sch{\EE\mF_{a + 1, b}}{2t})
	\end{align*}
}
\end{proof}

It remains to prove the claim.
\begin{proof}[Proof of~\cref{claim: reduction}]
	We will prove the first equality. The second one is analogous.
	For $I \in \cI, J \in \cJ, \al, \beta \in \{0, 1\}^n$, we have
	\[(\mF_{a, b} - \mF^{(i)}_{a, b})[(I, \al), (J, \beta)] = \begin{dcases}
		\grad_{\al + \beta} (f_{I, J}(Z) - f_{I, J}(Z^{(i)})) & \text{ if $|\al|_0 = a, |\beta|_0 = b, \al\cdot \beta = 0$}\\
		0 & \text{o.w.}
	\end{dcases}\]
	By \cref{propn: basic}, the first expression simplifies to $(Z_i - \resamp{Z_i})\grad_{\mat{e}_i}\grad_{\al + \beta} f_{I, J}(Z)$. Define the matrix $\mF_{a, b, i}$ to be the matrix with the same set of rows and columns as $\mF_{a, b}$ and whose only nonzero entries are given by
	\[\mF_{a, b, i}[(I, \al), (J, \beta + \mat{e}_i)] = \grad_{\mat{e}_i}\grad_{\al + \beta} f_{I, J}(Z) \text{ if $|\al|_0 = a, |\beta|_0 = b, \beta \cdot \mat{e}_i = 0, \al\cdot (\beta + \mat{e}_i) = 0$}\]

	Then, it's easy to see that $\sum_{i = 1}^n \mF_{a, b, i}\mF_{a, b, i}^\T = (b + 1)\mF_{a, b + 1}\mF_{a, b + 1}^T$ and $(\mF_{a, b} - \mF_{a, b}^{(i)})(\mF_{a, b} - \mF_{a, b}^{(i)})^\T = (Z - \resamp{Z_i})^2\mF_{a, b, i}\mF_{a, b, i}^\T$. The latter equality implies
	\[\EE[(\mF_{a, b} - \mF_{a, b}^{(i)})(\mF_{a, b} - \mF_{a, b}^{(i)})^\T|Z] = \EE[(Z_i - \resamp{Z_i})^2\mF_{a, b, i}\mF_{a, b, i}^\T|Z] = 2\mF_{a, b, i}\mF_{a, b, i}^\T\]

	Therefore,
	\[\sum_{i = 1}^n\EE[(\mF_{a, b} - \mF_{a, b}^{(i)})(\mF_{a, b} - \mF_{a, b}^{(i)})^\T|Z] = 2\sum_{i = 1}^n\mF_{a, b, i}\mF_{a, b, i}^\T = 2(b + 1)\mF_{a, b + 1}\mF_{a, b + 1}^\T\]
\end{proof}

\section{Applications}\label{sec:rademacher-applications}

To illustrate our framework, we apply it to obtain concentration bounds for nonlinear random matrices that have been considered in the literature before. The first one is a simple tensor network that arose in the analysis of spectral algorithms for a variant of principal components analysis (PCA) \cite{hopkins2015tensor, hopkins2018statistical}.
The second application is to obtain norm bounds on dense graph matrices \cite{medarametla2016bounds, ahn2016graph}. In the second application, the norm bounds are governed by a combinatorial structure called \textit{the minimum vertex separator of a shape}. We will see how this notion arises naturally under our framework, while prior works that derived such bounds used the trace power method and required nontrivial combinatorial insights.

\subsection{A simple tensor network}

We consider the following result from \cite{hopkins2015tensor, hopkins2018statistical}.

\begin{lemma}[\cite{hopkins2018statistical}, Theorem 6.7.1]
	Let $c \in \{1, 2\}$ and let $d \ge 1$ be an integer. Let $\mA_1, \ldots, \mA_{n^c}$ be i.i.d. random matrices uniformly sampled from $\pmone^{n^d \times n^d}$. Then, with probability $1 - O(n^{-100})$,
	\[\norm{\sum_{k \le n^c} \mA_k \otimes \mA_k- \EE \sum_{k \le n^c} \mA_k \otimes \mA_k} \le C\sqrt{d}n^{(2d + c) / 2} (\log n)^{1/2}\]
	for an absolute constant $C > 0$.
\end{lemma}

Using our framework, we will prove a slightly relaxed version of the inequality where $\sqrt{d} (\log n)^{1/2}$ is replaced by $\log n$.
We remark that we have not attempted to optimize these extra factors in front of the dominating term $n^{(2d + c)/2}$, so it's plausible that a more careful analysis can obtain a slightly better bound.

\begin{proof}[Proof of the relaxed bound]
	Let the $i, j$-th entry of $\mA_k$ be $a_{k, i, j}$.
	Let $\mF = \sum_{i \le n^c} \mA_k \otimes \mA_k - \EE \sum_{i \le n^c} \mA_k \otimes \mA_k$ be a random matrix on the variables $a_{k, i, j}$ for $k \le n^c, i, j \le n^d$. So $\EE \mF = 0$ and we are looking for bounds on $\norm{\mF}$. The entries are given by
	\[\mF[(i_1, i_2), (j_1, j_2)] = \begin{dcases}
		\sum_{k \le n^c} a_{k, i_1, j_1}a_{k, i_2, j_2} & \text{ if $(i_1, j_1) \neq (i_2, j_2)$}\\
		0 & \text{ if $(i_1, j_1) = (i_2, j_2)$}
	\end{dcases}\]
	The nonzero entries are homogeneous polynomials of degree $2$. Using \cref{thm: main_rademacher},
	\[\Esch{\mF}{2t} \le (32t)^{2t}(\sch{\EE\mF_{2, 0}}{2t} + \sch{\EE\mF_{1, 1}}{2t} + \sch{\EE\mF_{0, 2}}{2t})\]




	We will consider each of these terms.
	In the following arguments, we restrict attention to indices $i_1, i_2, j_1, j_2$ such that $(i_1, j_1) \neq (i_2, j_2)$.

	\begin{enumerate}
		\item $\EE\mF_{2, 0}$ has nonzero entries in row $((i_1, i_2), \{(k, i_1, j_1), (k, i_2, j_2)\})$ and column $(j_1, j_2)$ and all these entries are $1$.  The Schatten norm does not change when we permute the rows and columns. So, we can group the rows on $k, i_1, i_2$ and within each group, we can sort $j_1, j_2$ in both rows and columns. We get a matrix having $n^{2d + c}$ identity matrices, each of dimensions $n^{2d} \times n^{2d}$, stacked on top of each other. Using the definition, the Schatten-$2t$ norm of this matrix is easily computed to be $\sch{\EE\mF_{2, 0}}{2t}  = n^{c + 4d} n^{t(2d + c)}$.

	\item $\EE\mF_{1, 1}$ has nonzero entries in either row $((i_1, i_2), \{(k, i_1, j_1)\})$ and column $((j_1, j_2)$, $\{(k, i_2, j_2)\})$; or row $((i_1, i_2), \{(k, i_2, j_2)\})$ and column $((j_1, j_2), \{(k, i_1, j_1)\})$ and all these entries are $1$. So we can write $\EE\mF_{1, 1} = \mA + \mB$ corresponding to the 2 sets of entries. Arguing just as in the previous case,  we can obtain $\sch{\mA}{2t} = n^{c + 4d} n^{t(2d + c)}$ where we group the rows on $k, i_2, j_1$ and $\sch{\mB}{2t} = n^{c + 4d} n^{t(2d + c)}$ where we group the rows on $k, i_1, j_2$.
    Therefore,	$\sch{\EE\mF_{1, 1}}{2t} \le 2^{2t} (\sch{\mA}{2t} + \sch{\mB}{2t}) = 2^{2t + 1} n^{c + 4d} n^{t(2d + c)}$.


	\item The case $\EE\mF_{0, 2}$ is identical to $\EE\mF_{2, 0}$.
	\end{enumerate}

	Putting them together, $\Esch{\mF}{2t} \le (C't)^{2t} n^{c + 4d}n^{t(2d + c)}$
	for an absolute constant $C' > 0$. Now, we apply Markov's inequality to get
	\begin{align*}
		Pr[\norm{\mF - \EE \mF} \ge \theta] ~\le~ Pr[\sch{\mF - \EE \mF}{2t} \ge \theta^{2t}] &~\le~ \theta^{-2t} \EE\sch{\mF - \EE \mF}{2t}\\
        & ~\le~ \theta^{-2t}(C't)^{2t} n^{c + 4d}n^{t(2d + c)}
	\end{align*}
	We now set $\theta = \eps^{-1/(2t)} (C't)n^{(c+4d)/t} n^{(2d + c)/2}$ to make this expression at most $\eps$. Plug in $\eps = n^{-100}$ and set $t = \log n$ to obtain that $\norm{\mF - \EE \mF} \le Cn^{(2d + c) / 2} \log n$ holds with probability $1 - n^{-100}$, where $C > 0$ is an absolute constant.
\end{proof}

\subsection{Graph matrices}\label{sec: dense_graph_matrices}

In this section, we first define graph matrices and then show how to obtain norm bounds for \textit{dense graph matrices}, i.e. the case when $G \sim \cG_{n, 1/2}$, using our framework. Handling \textit{sparse graph matrices}, i.e. the case when $G \sim \cG_{n, p}$ for $p = o(1)$, may not work well with our basic framework as we will explain in \cref{sec: failure_of_basic}. Instead, our general framework in \cref{sec: general_recursion} will handle this case well and we obtain sparse graph matrix norm bounds in \cref{sec: sparse_graph_matrices}.


\subsubsection{Definitions}

Define by $\cG_{n, p}$ the \Erdos-\Renyi random graph on the vertex set $[n]$ with $n$ vertices, where each edge is present independently with probability $p$. Let the graph be encoded by variables $G_{i, j} \in \Omega = \{-\sqrt{\frac{1 - p}{p}}, \sqrt{\frac{p}{1 - p}}\}$ where $-\sqrt{\frac{1 - p}{p}}$ indicates the presence of the edge $\{i, j\}$ and $\sqrt{\frac{p}{1 - p}}$ indicates absence, for all $1 \le i, j \le n$.

So, each $G_{i, j}$ for $i < j$ is sampled from $\Omega$ where $G_{i, j}$ takes the value $-\sqrt{\frac{1 - p}{p}}$ with probability $p$ and takes the value $\sqrt{\frac{p}{1 - p}}$ otherwise. Here, $\Omega$ has been normalized so that $\EE_{x \sim \Omega}[x] = 0, \EE_{x \sim \Omega}[x^2] = 1$. as is standard in $p$-biased Fourier analysis.

When $p = \nicefrac{1}{2}$, we are in the setting of \textit{dense graph matrices}. Then, $\cG_{n, 1/2}$ can be thought of as a sampling of the $G_{i, j}, i < j$ independently and uniformly from $\Omega = \{-1, 1\}$.

For a set of edges $E \subseteq \binom{[n]}{2}$, define $G_E := \prod_{e \in E} G_e$. When $p = \nicefrac{1}{2}$, the $G_E$ correspond to the Fourier basis for functions of the graph.

Define $\cI$ to be the set of sub-tuples of $[n]$, including the empty tuple. Graph matrices will have rows and columns indexed by $\cI$. Each graph matrix has a succinct representation as a graph with some extra information, that is called a \textit{shape}.

\begin{definition}[Shape]
	A shape is a tuple $\tau = (V(\tau), E(\tau), U_{\tau}, V_{\tau})$ where $(V(\tau), E(\tau))$ is a graph and $U_{\tau}, V_{\tau}$ are ordered subsets of the vertices.
\end{definition}

\begin{definition}[Realization]
	Given a shape $\tau$, a realization of $\tau$ is an injective map $\varphi: V(\tau) \to [n].$
\end{definition}

\begin{definition}[Graph matrices]
	Let $\tau$ be a shape.
	Corresponding to $\tau$, the graph matrix $\graphmat{\tau}  : \{ \pm 1\}^{n \choose 2} \rightarrow \R^{\cI\times \cI}$ is defined to be the matrix-valued function with $I, J$-th entry defined as follows.
	\[
	\mM_{\tau}[I, J] := \sum_{\substack{\text{Realization }\phi\\ \phi(U_{\tau}) = I, \phi(V_{\tau}) = J}}{G_{\phi(E(\tau))}} = \sum_{\substack{\text{Realization }\phi\\ \phi(U_{\tau}) = I, \phi(V_{\tau}) = J}}\prod_{(u, v) \in E(\tau)} G_{\phi(u), \phi(v)}
	\]
	In other words, we sum over all realizations of $\tau$ that map $U_{\tau}, V_{\tau}$ to $I, J$ respectively and for each such realization, we have a term corresponding to the Fourier character that the realization gives.
\end{definition}

\begin{figure}[!h]
	\centering
	\includegraphics[trim={3cm 21cm 3cm 1cm}, clip, scale=1]{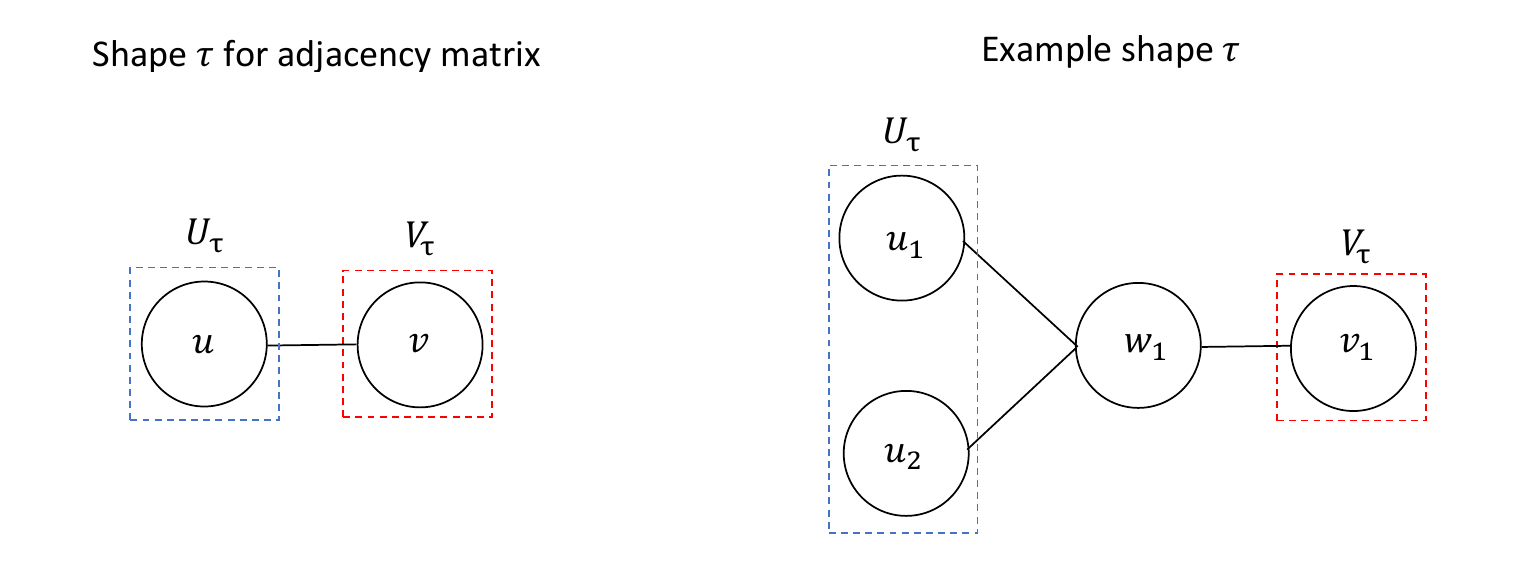}
	\caption{Left: Shape corresponding to adjacency matrix, Right: Example of a more complicated shape}
	\label{fig: shape}
\end{figure}

The following examples illustrate some simple graph matrices.

\begin{example}[Adjacency matrix]
	Let $\tau$ be the shape on the left in \cref{fig: shape}, with two vertices $V(\tau) = \{u,v\}$ and a single edge $E(\tau) = \{\{ u,v\}\}$. $U_\tau, V_\tau$ are $(u), (v)$ respectively where we use tuples to indicate ordering.
	Then $\mM_\tau$ has nonzero entries $\mM_\tau[(i), (j)](G) = G_{i, j}$ for all $i \neq j$.
	If $G \in \{ \pm 1\}^{n \choose 2}$ is thought of as a graph, then $\mM_\tau$ has as principal submatrix the $\pm 1$ adjacency matrix of $G$ with zeros on the diagonal, and the other entries are $0$.
\end{example}

\begin{example}
	In \cref{fig: shape}, consider the shape $\tau$ on the right. We have $U_{\tau} = (u_1, u_2), V_{\tau} = (v_1), V(\tau) = \{u_1, u_2, v_1, w_1\}$ and $E(\tau) = \{\{u_1, w_1\}, \{u_2, w_1\}, \{w_1, v_1\}\}$. $\mM_{\tau}$ is a matrix with rows and columns indexed by sub-tuples of $[n]$. Its nonzero entries are in rows $I$ and columns $J$ with $|I| = |U_{\tau}| = 2$ and $|J| = |V_{\tau}| = 1$ respectively. More specifically, for all distinct $a_1, a_2, b_1$, the entry corresponding to row $(a_1, a_2)$ and column $(b_1)$ is $\sum_{c_1 \in [n] \setminus \{a_1, a_2, b_1\}} G_{a_1, c_1}G_{a_2, c_1}G_{c_1, b_1}$.
	Here, each term is obtained via the realization $\phi$ that maps $u_1, u_2, w_1, v_1$ to $a_1, a_2, c_1, b_1$ respectively. Succinctly, \[\mM_{\tau} =
	\begin{blockarray}{rl@{}c@{}r}
		& & \makebox[0pt]{column $(b_1)$} \\[-0.5ex]
		& & \,\downarrow \\[-0.5ex]
		\begin{block}{r(l@{}c@{}r)}
			&  & \vdots & \\[-0.2ex]
			\text{row }(a_1, a_2) \to \mkern-9mu & \raisebox{0.5ex}{\makebox[3.2em][l]{\dotfill}} & \sum_{c_1 \in [n] \setminus \{a_1, a_2, b_1\}} G_{a_1, c_1}G_{a_2, c_1}G_{c_1, b_1} & \raisebox{0.5ex}{\makebox[4.2em][r]{\dotfill}} \\[+.5ex]
			&  & \vdots & \\
		\end{block}
	\end{blockarray}\]
\end{example}


Intuitively, graph matrices are symmetrizations of the Fourier basis, where the symmetry is incorporated by summing over all realizations of ``free'' vertices $V(\tau) \setminus U_{\tau} \setminus V_{\tau}$ of the shape $\tau$.
For more examples of graph matrices and why they can be a useful tool to work with, see \cite{ahn2016graph}.

\subsubsection{Norm bounds for dense graph matrices}\label{sec: norm_bounds_for_dense_graph_matrices}

In this section, we study the concentration of the so-called ``dense graph matrices'' which is a term that refers to graph matrices $M_{\tau}$ in the setting $p = \nicefrac{1}{2}$.
Since the edges of a random graph sampled from $\cG_{n,1/2}$ can be viewed as independent Rademacher random variables, we can apply our framework in this setting.
%


In particular, we will obtain bounds on $\Esch{\mM_{\tau} - \EE\mM_{\tau}}{2t}$.
The $G_{i, j} \in \{-1, 1\}$ correspond to the $Z_i$s in \cref{sec: basic_recursion} and for a fixed shape $\tau$, $\mM_{\tau}$ will be the matrix $\mF$ we are interested in analyzing. For $I, J \in \cI$, $\mM_{\tau}[I, J]$ is a nonzero polynomial only when there exists at least one realization of $\tau$ that maps $U_{\tau}, V_{\tau}$ to $I, J$ respectively. In particular, we must have $|I| = |U_{\tau}|$ and $|J| = |V_{\tau}|$. In this case, $\mM_{\tau}[I, J]$ is a homogenous polynomial of degree $|E(\tau)|$.

By \cref{thm: main_rademacher}, we have
\[\Esch{\mM_{\tau} - \EE\mM_{\tau}}{2t} ~\le~ \sum_{a + b \ge 1\atop a, b \ge 0}(16t|E(\tau)|)^{(a + b)t}\sch{\EE\mM_{\tau, a, b}}{2t}\]
where for integers $a, b \ge 0$, $\mM_{\tau, a, b}$ is defined to be the matrix with rows and columns each indexed by $\cI \times \{0, 1\}^{\binom{n}{2}}$ such that for all $I, J \in \cI$, we have
\[\mM_{\tau, a, b}[(I, \al), (J, \beta)] ~=~ \begin{dcases}
	\grad_{\al + \beta} \mM_{\tau}[I, J] & \text{ if $|\al|_0 = a, |\beta|_0 = b, \al \cdot \beta = 0$}\\
	0 & \text{o.w.}
\end{dcases}
\]

For any multilinear homogenous polynomial $f$ of degree $d$, since $\EE[G_{i, j}] = 0$ for all $i, j$, we have $\grad_{\al}f = 0$ whenever $|\al|_0 < d$. Therefore, $\EE\mM_{\tau, a, b} = 0$ for all $a + b < |E(G)|$. Moreover, $\EE\mM_{\tau, a, b} = 0$ whenever $a + b \neq |E(G)|$ otherwise $\EE\mM_{\tau, a, b} = \mM_{\tau, a, b}$. So, we can further simplify the above expression to
\[\Esch{\mM_{\tau} - \EE\mM_{\tau}}{2t} ~\le~ \sum_{a + b = |E(\tau)|\atop a, b \ge 0}(16t|E(\tau)|)^{|E(\tau)|t}\sch{\mM_{\tau, a, b}}{2t}\]

It remains to analyze $\sch{\mM_{\tau, a, b}}{2t}$ for $a + b = |E(G)|$. We will see that analyzing these matrices is much simpler since they are deterministic matrices and simple computations using the Frobenius norm bound will work well. To state our final bounds, we need to define the notion of vertex separators of shapes.

\begin{remark}
	As we will see, when analyzing the Frobenius norms for these deterministic matrices, the notion of the minimum vertex separator arises naturally. In prior trace method calculations (e.g. \cite{medarametla2016bounds}, \cite{ahn2016graph}), this required ingenious combinatorial observations.
\end{remark}

\begin{restatable}[Vertex separator]{definition}{vertexseparator}
	For a shape $\tau$, define a vertex separator to be a subset of vertices $S \subseteq V(\tau)$ such that there is no path from $U_{\tau}$ to $V_{\tau}$ in $\tau \setminus S$, which is the shape obtained by deleting all the vertices of $S$ (including all edges they're incident on).
\end{restatable}

For a shape $\tau$, denote by $S_{\tau}$ a vertex separator of the smallest size. Also, let $I_{\tau}$ be the set of isolated vertices (vertices with degree $0$) in $V(\tau) \setminus U_{\tau} \setminus V_{\tau}$, so the presence of these vertices essentially scale the matrix by a scalar factor.

\begin{theorem}\label{thm: dense_graph_matrix_norm_bounds}
	For a shape $\tau$ and any integer $t \ge 1$,
	\[\EE\sch{\mM_{\tau} - \EE \mM_{\tau}}{2t} \le \bigg(C^{t|E(\tau)|}n^{|V(\tau)|} t^{t|E(\tau)|}|E(\tau)|^{2t|E(\tau)|}\bigg)n^{t(|V(\tau)| - |S_{\tau}| + |I_{\tau}|)}\]
	for an absolute constant $C > 0$.
\end{theorem}

Up to lower order terms, the same result has been shown before in \cite{medarametla2016bounds, ahn2016graph}. To interpret this bound, assume that $\tau$ has a constant number of vertices. By setting $t \approx \polylog(n)$, we get \[\norm{\mM_{\tau}} = \widetilde{\bigoh}\left(\sqrt{n}^{|V(\tau)| - |S_{\tau}| + |I_{\tau}|}\right)\] with high probability, where $\widetilde{\bigoh}$ hides logarithmic factors.
This is obtained by applying Markov's inequality on the bound on $\Esch{\mM_{\tau}}{2t}$. If $\tau$ has at least one edge, then $\EE \mM_{\tau} = 0$ and \cref{thm: dense_graph_matrix_norm_bounds} yields such bounds. If $\tau$ has no edges, then it's quite simple to obtain such a bound and we include it in \cref{lem: empty_shape} for the sake of completeness. \cref{cor: dense_graph_matrix_norm_bounds} makes precise the high probability bound above. Therefore, this power of $n$ is essentially what controls the norm bound and this is utilized heavily in applications (e.g. \cite{BHKKMP16, ghosh2020sum, potechin2020machinery}).

\begin{proof}[Proof of \cref{thm: dense_graph_matrix_norm_bounds}]
    We first argue that we can assume $I_{\tau} = \emptyset$. This is because of the following reason. Each distinct vertex in $\tau$ of degree $0$ essentially scales the matrix by a factor of at most $n$. And in the right hand side of the inequality, each vertex in $I_{\tau}$ contributes a factor of $n^{2t}$ accordingly, from $n^{t|V(\tau)|}$ and from $n^{t|I_{\tau}|}$, and the other changes only weaken the inequality.

	Now, fix $a, b \ge 0$ such that $a + b = |E(\tau)|$ and consider $\mM_{\tau, a, b}$. For $I, J \in \cI, \al, \beta \in \{0, 1\}^{\binom{n}{2}}$ such that $|\al|_0 = a, |\beta|_0 = b, \al \cdot \beta = 0$, by definition,
    \begin{align*}
        \mM_{\tau, a, b}[(I, \al), (J, \beta)] &~=~ \grad_{\al + \beta} \left(\sum_{\phi: \phi(U_{\tau}) = I, \phi(V_{\tau})= J} \prod_{u, v \in E(\tau)} G_{\phi(u), \phi(v)}\right)\\
        &~=~ |\{\phi ~|~ \phi(U_{\tau}) = I, \phi(V_{\tau})= J, \phi(E(\tau)) = \supp(\al + \beta)\}|
    \end{align*}
    where $\supp(.)$ denotes the support. We will now obtain norm bounds on these deterministic matrices by reinterpreting them as graph matrices for different shapes.



    Let $P = (E_1, E_2)$ denote the partition of $E(\tau) = E_1 \sqcup E_2$ into two ordered sets $E_1, E_2$, where $\sqcup$ denotes disjoint union. Then, we can write $\mM_{\tau, a, b} = \sum_{P \in \cP} \mM_{\tau, a, b, P}$ where
    \[\mM_{\tau, a, b, P}[(I, \al), (J, \beta)] ~=~ |\{\phi ~|~ \phi(U_{\tau}) = I, \phi(V_{\tau})= J, \phi(E_1) = \supp(\al), \phi(E_2) = \supp(\beta)\}|\]

    Let the set of ordered partitions $P$ be $\cP$. Then, $|\cP| \le (4|E(\tau)|)^{|E(\tau)|}$ and so, by \cref{fact: holder},
    \[\sch{\mM_{\tau, a, b}}{2t} \le (4|E(\tau)|)^{t|E(\tau)|} \sum_{P \in \cP}\sch{\mM_{\tau, a, b, P}}{2t}\]


    Each $\mM_{\tau, a, b, P}$ can be interpreted as a graph matrix for a different shape $\tau_P$, with the same vertex set and no edges. Let $V(\tau_P) = V(\tau), E(\tau_P) = \emptyset$ and set $U_{\tau_P} = U_{\tau} \cup V(E_1), V(\tau_P) = V_{\tau} \cup V(E_2)$ using a canonical ordering. Then, $\mM_{\tau, a, b}$ is equal to $\mM_{\tau_P}$ up to renaming of the rows and columns. For an illustration, see \cref{fig: evolution_new}.


	\begin{figure}[!h]
		\centering
		\includegraphics[trim={2cm 20cm 2cm 2cm}, clip, scale=0.9]{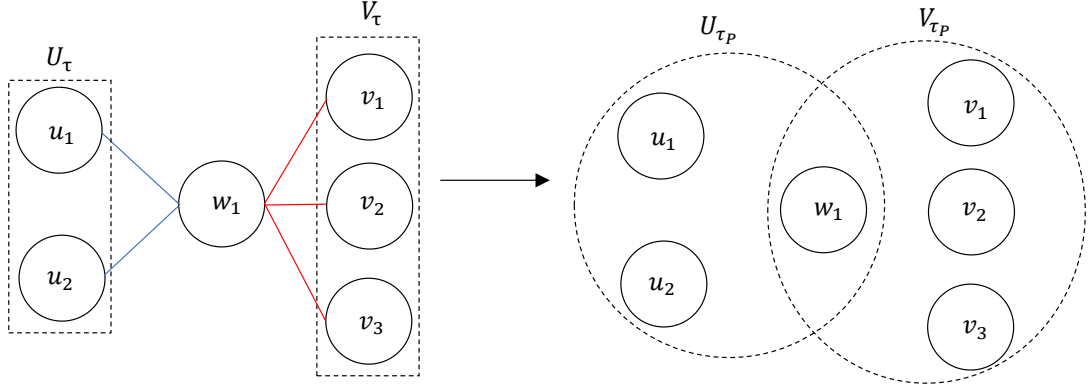}
		\caption{An example illustrating how $\tau_P$ is defined. In this example, $P$ constraints the blue and red edges to go to $\al$ and $\beta$ respectively. $U_{\tau_P}, V_{\tau_P}$ have an ordering on the vertices (not shown here).}
		\label{fig: evolution_new}
	\end{figure}

	This graph matrix has a block diagonal structure indexed by the realizations of the set of common vertices $S = U_{\tau_P} \cap V_{\tau_P}$. Indeed, for $K \in [n]^S$, let $\mM_{\tau_P, K}$ be the block of $\mM_{\tau_P}$ with $\phi(S) = K$. Then, $\mM_{\tau_P, K}\mM_{\tau_P, K'}^\T = \mM_{\tau_P, K}^\T\mM_{\tau_P, K'} = 0$ for $K \neq K'$ and so,
	\begin{align*}
		\Esch{\mM_{\tau, a, b}}{2t} &\le (4|E(\tau)|)^{t|E(\tau)|} \sum_{P \in \cP}\sch{\mM_{\tau_P}}{2t}\\
		&= (4|E(\tau)|)^{t|E(\tau)|} \sum_{P \in \cP} \sum_{T \in [n]^S}\sch{\mM_{\tau_P, T}}{2t}\\
		&\le (4|E(\tau)|)^{t|E(\tau)|} \sum_{P \in \cP} \sum_{T \in [n]^S}\left(\sch{\mM_{\tau_P, T}}{2}\right)^t
	\end{align*}
	where we bounded the Schatten norm by the appropriate power of the Frobenius norm.

	For any fixed $K \in [n]^S$, the entries of $\mM_{\tau_P, K}$ take values in $\{0, 1\}$ and the number of nonzero entries is at most $n^{|V(\tau)| - |S|}$ because the realizations of vertices in $S$ are fixed and the other vertices have at most $n$ choices each. Therefore, $\sch{\mM_{\tau_P, K}}{2} \le n^{|V(\tau)| - |S|}$.

	Finally, we bound $|S|$ to estimate how large this term can be over all possibilities of $P$.
    We argue that $S$ blocks all paths from $U_{\tau}$ to $V_{\tau}$. To see this, consider any path from $U_{\tau}$ to $V_{\tau}$, it must contain an edge $(u, v) \in E(\tau)$ such that $u \in U_{\tau_P}, v \in V_{\tau_P}$. We must either have $(u, v) \in E_1$, in which case $u,  v \in U_{\tau_P}$ and $v \in S$, or $(u, v) \in E_2$, in which case $u, v \in V_{\tau_P}$ and $u \in S$. In either case, $S$ must contain either $u$ or $v$. This argument implies $S$ must be a vertex separator of $\tau$, giving $|S| \ge |S_{\tau}|$.
    For a proof by picture, see \cref{fig: proof_by_picture}.

	\begin{figure}[!h]
		\centering
		\includegraphics[trim={2cm 20cm 2cm 2cm}, clip, scale=0.9]{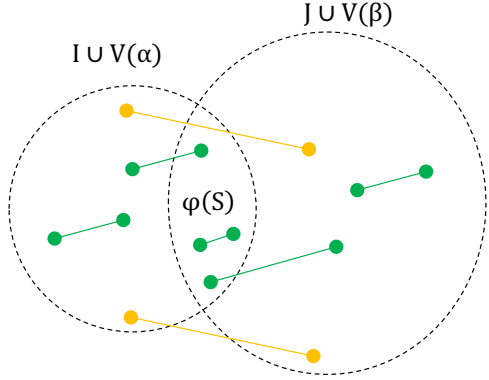}
		\caption{Proof by picture that $|S| \ge |S_{\tau}|$. Green edges can occur in $\tau$, orange edges cannot, so $S$ blocks all paths from $U_{\tau}$ to $V_{\tau}$.}
		\label{fig: proof_by_picture}
	\end{figure}

	We also have the trivial upper bound $|S| \le |V(\tau)|$. Ultimately, this gives
	\begin{align*}
		\sch{\mM_{\tau, a, b}}{2t} &\le (4|E(\tau)|)^{t|E(\tau)|} \sum_{P \in \cP} \sum_{T \in [n]^S}n^{t(|V(\tau)| - |S_{\tau}|)}\\
		&\le (4|E(\tau)|)^{t|E(\tau)|}(4|E(\tau)|)^{|E(\tau)|}n^{|V(\tau)|}n^{t(|V(\tau)| - |S_{\tau}|)}
	\end{align*}
	Along with our prior discussion, we get
    {\footnotesize
	\begin{align*}
		\Esch{\mM_{\tau} - \EE\mM_{\tau}}{2t} &\le 	\sum_{a + b = |E(\tau)|}(16t|E(\tau)|)^{|E(\tau)|t}\sch{\mM_{\tau, a, b}}{2t}\\
		&\le \sum_{a + b = 	|E(\tau)|}(16t|E(\tau)|)^{|E(\tau)|t}(4|E(\tau)|)^{t|E(\tau)|}(4|E(\tau)|)^{|E(\tau)|}n^{|V(\tau)|}n^{t(|V(\tau)| - |S_{\tau}|)}\\
		&\le \bigg(C^{t|E(\tau)|}n^{|V(\tau)|} t^{t|E(\tau)|}|E(\tau)|^{2t|E(\tau)|}\bigg)n^{t(|V(\tau)| - |S_{\tau}|)}
	\end{align*}}
	for an absolute constant $C > 0$.
\end{proof}

In the proof above, our analysis of the shape $\tau_P$ which has no edges, applies in general to any shape $\tau$ with no edges. For the sake of completeness, we state it explicity in the following lemma.

\begin{lemma}\label{lem: empty_shape}
	For a shape $\tau$ with no edges and any integer $t \ge 1$,
	\[\Esch{\mM_{\tau}}{2t} \le n^{|U_{\tau} \cap V_{\tau}|}n^{t(V(\tau) - |U_{\tau} \cap V_{\tau}| + |I_{\tau}|)}\]
\end{lemma}

Note that this has the same form as \cref{thm: dense_graph_matrix_norm_bounds} because for a shape $\tau$ with no edges, the minimum vertex separator $S_{\tau}$ is just $U_{\tau} \cap V_{\tau}$.

The following corollary obtains high probability norm bounds for norms of graph matrices via Markov's inequality.

\begin{corollary}\label{cor: dense_graph_matrix_norm_bounds}
	For a shape $\tau$, for any constant $\eps > 0$, with probability $1 - \eps$,
	\[\norm{\mM_{\tau}} \le (C|E(\tau)| \log(n^{|V(\tau)|}/\eps))^{|E(\tau)|}\cdot\sqrt{n}^{|V(\tau)| - |S_{\tau}| + |I_{\tau}|}\]
	for an absolute constant $C > 0$.
\end{corollary}

\begin{proof}
	If $E(\tau) = \emptyset$, we invoke \cref{lem: empty_shape}. Otherwise, $\EE\mM_{\tau} = 0$ and we invoke \cref{thm: dense_graph_matrix_norm_bounds}. By an application of Markov's inequality,
	\begin{align*}
		Pr[\norm{\mM_{\tau}} \ge \theta] &\le Pr[\sch{\mM_{\tau}}{2t} \ge \theta^{2t}]\\
		&\le \theta^{-2t} \EE\sch{\mM_{\tau}}{2t}\\
		&\le \theta^{-2t}\bigg((C')^{t|E(\tau)|}n^{|V(\tau)|} t^{t|E(\tau)|}|E(\tau)|^{2t|E(\tau)|}\bigg)n^{t(|V(\tau)| - |S_{\tau}| + |I_{\tau}|)}
	\end{align*}
	for an absolute constant $C' > 0$. We now set
	\[\theta = \bigg(\eps^{-1/(2t)} (C'')^{|E(\tau)|} n^{|V(\tau)|/(2t)} t^{|E(\tau)|/2}|E(\tau)|^{|E(\tau)|}\bigg)\sqrt{n}^{|V(\tau)| - |S_{\tau}| + |I_{\tau}|}\] for an absolute constant $C'' > 0$, to make this expression at most $\eps$. Set $t = \frac{1}{2} \log(n^{|V(\tau)|}/\eps)$ to complete the proof.
\end{proof}

\section{Why a na\"ive application of \cite{paulin2016} may fail for general product distributions} \label{sec: failure_of_basic}

In this section, we elaborate on the difficulties that arise when working with random variables that are not necessarily Rademacher. In this case, note that we cannot assume that the polynomial entries are multilinear as well.

To recall the setting, we are given a random matrix $\mF$ whose entries are low degree polynomials in random variables $Z_1, \ldots, Z_n$ which are independently sampled from arbitrary distributions. And we wish to obtain concentration bounds on how much $\mF$ can deviate from its mean, by way of controlling $\Esch{\mF- \EE\mF}{2t}$.

Building on the ideas from \cref{sec: basic_recursion}, we could attempt to use matrix Efron-Stein, \cref{thm: main_efron_stein} and hope to obtain a similar recursion framework. We now discuss what happens if we do this. Assume $\EE[Z_i] = 0, \EE[Z_i^2] = 1$. We can proceed similar to the proof of \cref{thm: main_rademacher}. So, we consider $\mX$ as a principal submatrix of $\mX_{0, 0}$ and follow through \cref{lem: main_rademacher}. The main change will happen in \cref{claim: reduction}. In particular, the equation $\EE[(Z_i - \resamp{Z_i})^2|Z] = 2$ is no longer true. Instead, we will have $\EE[(Z_i - \resamp{Z_i})^2|Z] = 1 + Z_i^2$. So, we get the expression
\[\sum_{i = 1}^n (1 + Z_i^2)\mF_{a, b, i}\mF_{a, b, i}^\T = \sum_{i = 1}^n \mF_{a, b, i}\mF_{a, b, i}^\T + \sum_{i = 1}^n Z_i^2\mF_{a, b, i}\mF_{a, b, i}^\T\]

The first term can been handled just as in the basic framework. Unfortunately, the second term will be a source of difficulty. To get around this difficulty, we could attempt to apply the matrix Efron-Stein inequality again on an appropriately constructed matrix.
To do this, we can interpret the second term as having been obtained after differentiating with respect to the variable $Z_i$ and then \textit{putting the variable back}. In contrast, we didn't need to put it back when working with Rademacher random variables.
But after we do this, when we recurse on these extra matrices, the new second term will contain the left hand side as a sub-term, thereby giving a trivial inequality and stalling the recursion.

To see this more clearly, consider the simplest case $a = b = 0$. Then, the first term $\sum_{i = 1}^n \mF_{a, b, i}\mF_{a, b, i}^\T$ will be equal to $\mF_{0, 1}\mF_{0, 1}^\T$ as we saw earlier. To evaluate the second term $\sum_{i = 1}^n Z_i^2\mF_{a, b, i}\mF_{a, b, i}^\T$ in a similar manner, we define the matrix $\mH$ to be the same as $\mF_{0, 1}$ except that each entry is now multiplied by $Z_i$ where $i$ is the differentiated variable in the column. That is, $\mH[I, (J, \mat{e}_i)] = Z_i\mF_{0, 1}[I, (J, \mat{e}_i)]$. Observe that in the definition of $\mH$, $Z_i$ has been put back after differentiating with respect to it. Then, the second term will be $\mH\mH^\T$ and we can hope to use Efron-Stein again on this matrix $\mH$ recursively.

We could do that and proceed similarly to the proof of \cref{lem: main_rademacher} with appropriate modifications as above. But since $\beta_i = 1$ already, differentiating with respect to $Z_i$ and putting it back, will return the same matrix $\mH$! So, we end up with an inequality of the form
\[\Esch{\mH}{2t} \le \bigoh(t)^{t}(\Esch{\mH}{2t} + \text{ other nonnegative terms})\]
Indeed, this is a tautology and will not be useful to us.

For a quick and dirty bound, suppose we had a parameter $L$ such that $1 + Z_i^2 \le L$ for our distributions, then we will be able to obtain a similar framework while incurring a loss of $\sqrt{L}$ at each step of the recursion. But unfortunately, this bound will be lossy. For example, if we do this computation for the centered normalized adjacency matrix of $G \sim \cG_{n, p}$, we will obtain a norm bound of $\widetilde{\bigoh}(\frac{\sqrt{n(1 - p)}}{\sqrt{p}})$ where $\widetilde{\bigoh}$ hides logarithmic factors.. This bound is tight for constant or even inverse polylogarithmic $p$. But for $p = n^{-\theta}$ for some constant $0 < \theta < 1$, this is not tight because in this regime, the true norm bound is known to be $\widetilde{\bigoh}(\sqrt{n})$ (see the early works of \cite{furedi1981eigenvalues, vu2005spectral} and for tighter bounds, see \cite{benaych2020spectral} and references therein).

If we dig into the details of what happened, this example illustrates that the matrix Efron-Stein inequality \cref{thm: main_efron_stein} becomes a tautology for certain kinds of matrices, that yield $\mV= \bigoh(1) \mX\mX^\T + \text{ other positive semidefinite matrices}$.

But in our framework in general, the aforementioned bad matrices occur when we differentiate with respect to variables that have already been differentiated on. In other words, the current definition of the variance proxy $\mV$ doesn't take into account whether we have already differentiated with respect to some variable $Z_i$. So, for the general recursion, we dive into the proof due to \cite{paulin2016} and modify it using structural properties of the intermediate matrices we obtain in our framework.

\section{The general recursion framework}\label{sec: general_recursion}

We now assume $Z_1, \ldots, Z_n$ are i.i.d. random variables sampled from a distribution $\Omega$ with finite moments.
We assume that they are identically distributed for simplicity but our technique easily extends even when they are not identically distributed, as long as they are independent.
For each $i \le n$, define $\resamp{Z_i}$ to be an independent copy of $Z_i$ and define the vector $Z^{(i)} := (Z_1, \ldots, Z_{i - 1}, \resamp{Z_i}, Z_{i + 1}, \ldots, Z_n)$. Define $Z'$ to be the random vector defined by sampling $i$ from $[n]$ uniformly at random and then setting $Z' = Z^{(i)}$.

Let $\mF \in \RR[Z]^{\cI \times \cJ}$ be a matrix with rows and columns indexed by arbitrary sets $\cI, \cJ$ respectively such that for all $I \in \cI, J \in \cJ$, $\mF[I, J]$ are polynomials of $Z_1, \ldots, Z_n$. Let $\dpoly$ the maximum degree of $\mF[I, J]$ over all entries $I, J$ and let $d$ be the maximum degree of $Z_i$ over all entries $\mF[I, J]$ and $i \le n$.


Similar to the Rademacher case, let $\mX := \mF- \EE\mF$. When the input is $Z$, we denote the matrices as $\mF, \mX$, etc and when the input is $Z^{(i)}$, denote the corresponding matrices as $\mF^{(i)}, \mX^{(i)}$, etc. In this section, we will give a general framework using which we can obtain bounds on $\Esch{\mF - \EE\mF}{2t}$ for any integer $t \ge 1$.

We set up a few preliminaries in order to state the main theorem.

\begin{definition}[Space $\cS$]
	Let $\cS$ be the space of mean-zero polynomials in $Z_1, \ldots, Z_n$ of degree at most $\dpoly$.
\end{definition}

For $\al \neq 0$, we also define the centered monomials
\[\chi_{\al}(Z) = \prod_{\al_i > 0} (Z_i^{\al_i} - \EE[Z_i^{\al_i}])\]

By definition, $\chi_{\al} \in \cS$ for all $\al \neq 0, |\al|_1 \le \dpoly$. The following proposition is straightforward.

\begin{propn}\label{propn: basis}
	The set $\{\chi_{\al}(Z) | 1 \le |\al|_1 \le \dpoly\}$ forms a basis for $\cS$.
\end{propn}

For the general framework, we work over this basis because as we will see in \cref{sec: proof_of_general}, the ``inner kernel matrix'' is convenient to state in this basis.
The $\grad$ operator also works nicely with our polynomials $\chi_{\beta}$. Indeed, observe that $\grad_{\al}(\chi_{\beta}) = \begin{dcases}
	\chi_{\beta - \al} & \text{ if $\al \unlhd \beta$}\\
	0 & \text{ o.w.}
\end{dcases}$.

For a polynomial $f(Z)$ in $\cS$, denote by $\coef{f}{\al}$ the coefficient of $\chi_{\al}(Z)$ in the expansion of $f$, that is, \[f(Z) = \sum_{0 \neq \al \in \NN^n} \coef{f}{\al}\chi_{\al}(Z)\]

We can naturally extend this notation to matrices that have mean $0$. So, we can write $\mX = \sum_{\al \neq 0} \coef{\mX}{\al} \chi_{\al}(Z)$ where $\coef{\mX}{\al}$ are deterministic matrices. In order to apply our recursion framework, we group this sum into terms based on $|\al|_0$. For $k \ge 1$, define $\mX_k = \sum_{|\al|_0 = k} \coef{\mX}{\al} \chi_{\al}(Z)$. Then,
\[\mX = \sum_{k \ge 1} \mX_k\]

Note that when $k > \dpoly$, $\mX_k = 0$.

\begin{definition}[Indexing set $\cK$]
	We define $\cK \subseteq \NN^n \times \{0, 1\}^n$ to be the set of pairs $(\al, \gam)$ such that $|\al|_1 \le \dpoly, \al \in \NN^n$ and $\gam \le \al, \gam \in \{0, 1\}^n$.
\end{definition}

Define the diagonal matrix $\mD_1 \in \RR[Z]^{\cI \times \cK} \times \RR[Z]^{\cI \times \cK}$ with nonzero entries
\[\mD_1[(I, \al, \gam), (I, \al, \gam)] = \sqrt{\EE[Z^{2\al\cdot (1 - \gam)}]}Z^{\al\cdot \gam}\]

Similarly, define the diagonal matrix $\mD_2 \in \RR[Z]^{\cJ \times \cK} \times \RR[Z]^{\cJ \times \cK}$ with nonzero entries
\[\mD_2[(J, \al, \gam), (J, \al, \gam)] = \sqrt{\EE[Z^{2\al\cdot (1 - \gam)}]}Z^{\al\cdot \gam}\]

\begin{definition}[Matrices $\mG_{k, a, b}, \mF_{k, a, b}$]
	For integers $k, a, b$ such that $k \ge 1, a, b \ge 0$, define the matrix $\mG_{k, a, b}$ to have rows and columns indexed by $\cI \times\cK$ and $\cJ \times \cK$ respectively such that for all $(I, \al_1, \gam_1) \in \cI \times \cK$, $(J, \al_2, \gam_2) \in \cJ \times \cK$,
	\[\mG_{k, a, b}[(I, \al_1, \gam_1), (J, \al_2, \gam_2)] = \begin{dcases}
		\grad_{\al_1 + \al_2} \mX_k[I, J] & \text{ if $|\al_1|_0 = a, |\al_2|_0 = b, \al_1 \cdot \al_2 = 0$}\\
		0 & \text{o.w.}
	\end{dcases}
	\]
	Also, define $\mF_{k, a, b} := \mD_1 \mG_{k, a, b}\mD_2$.
\end{definition}

Note that when $k > \dpoly$, $\mF_{k, a, b} = 0$.

\begin{propn}
	For integers $k, a, b$ such that $k \ge 1, a, b \ge 0$, suppose $a + b < k$. Then each nonzero entry $f$ of $\mG_{k, a, b}$ has the property that $\coef{f}{\al}$ is nonzero only when $|\al|_0 = k - a - b$
\end{propn}

\begin{proof}
	The nonzero entries of $\mX_k$ only has terms containing exactly $k$ variables and $\grad_{\al_1 + \al_2}$ either zeroes out the term, or it truncates exactly $|\al_1 + \al_2|_0 = |\al_1|_0 + |\al_2|_0 = a + b$ variables.
\end{proof}

This also immediately implies that $\EE[\mG_{k, a, b}] = 0$ whenever $a + b < k$. Finally, when $k = a + b$, we have that $\mG_{k, a, b}$ is a deterministic matrix independent of the $Z_i$. These give rise to the matrices $\mF_{a + b, a, b}$ that appears in our main theorem.

We are now ready to state the main theorem.

\begin{theorem}[General recursion]
\label{thm: main_general}
Let the tuple of random variables $Z$ and the function $\mF$ be as above. Then, for all integers $t \ge 1$,
\[\Esch{\mF - \EE \mF}{2t} \le \sum_{a, b \ge 0, a + b \ge 1}(Ct^2d\dpoly^4)^{(a + b)t}\Esch{\mF_{a + b, a, b}}{2t}\]
	for an absolute constant $C > 0$.
\end{theorem}

Note that $\mF_{a + b, a, b} = \mD_1 \mG_{a + b, a, b} \mD_2$ where $\mD_1, \mD_2$ are diagonal matrices and $\mG_{a + b, a, b}$ is a deterministic matrix that's independent of $Z$. To analyze the expected Schatten norm of such matrices, we can resort to far simpler techniques. For instance, we can obtain a simple bound using an appropriate power of the Frobenius norm, and apply standard scalar concentration tools. We will see an example of this in \cref{sec: sparse_graph_matrices}.

\begin{remk}
	We have made no attempts to optimize the factors in front of the expectation in \cref{thm: main_general}, which we suspect can be improved. 
\end{remk}

We prove the main theorem by repeatedly applying the following technical lemma, the proof of which we defer to the next section.

\begin{restatable}{lemma}{maingeneral}\label{lem: main_general}
	For all integers $t \ge 1$, integers $k \ge 1, a, b \ge 0$ such that $a + b < k$,
	\[\Esch{\herm{\mF}_{k, a, b}}{2t} \le (Ct^2d\dpoly^2)^t (\Esch{\herm{\mF}_{k, a, b + 1}}{2t} + \Esch{\herm{\mF}_{k, a + 1, b}}{2t})\]
	for an absolute constant $C > 0$.
\end{restatable}

Using this lemma, we can complete the proof of the main theorem.

\begin{proof}[Proof of \cref{thm: main_general}]
	Using \cref{fact: holder}, we have $\Esch{\mX}{2t} \le \dpoly^{2t} \sum_{k = 1}^{\dpoly}\Esch{\mX_k}{2t}$. Note that for any $k \ge 1$, the matrix $\mX_k$ is a principal submatrix of $\mF_{k, 0, 0}$ with all other entries being $0$, so $\Esch{\mX_k}{2t} = \Esch{\mF_{k, 0, 0}}{2t} = \frac{1}{2} \Esch{\herm{\mF}_{k, 0, 0}}{2t}$. Therefore,
	\[\Esch{\mX}{2t} \le \frac{1}{2}\dpoly^{2t} \sum_{k = 1}^{\dpoly}\Esch{\herm{\mF}_{k, 0, 0}}{2t}\]
	We now apply \cref{lem: main_general} repeatedly to all our terms until $k = a + b$, ultimately giving
	\[\Esch{\mX}{2t} \le \frac{1}{2}\dpoly^{2t}(Ct^2d\dpoly^2)^{(a + b)t} \sum_{a, b \ge 0, a + b \ge 1}\Esch{\herm{\mF}_{a + b, a, b}}{2t}\]
	Observing that $\Esch{\herm{\mF}_{a + b, a, b}}{2t} = 2\Esch{\mF_{a + b, a, b}}{2t}$ completes the proof.
\end{proof}

\section{A generalization of \cite{paulin2016} and proof of \cref{lem: main_general}}\label{sec: proof_of_general}

In this section, we will prove \cref{lem: main_general} using the high level strategy described in \cref{sec: intro}. This requires generalizing the results in \cite{paulin2016}, and the proof techniques may be of independent interest.

\subsection{Generalizing \cite{paulin2016} via explicit inner kernels}\label{sec: explicit_inner_kernels}

In our setting, observe that $(Z, Z')$ has the same distribution as $(Z', Z)$. This is what is known as an \textit{exchangeable pair} of variables, that will be extremely useful for our analysis. In particular, $Z, Z'$ have the same distribution and $\EE f(Z, Z') = \EE f(Z', Z)$ for every integrable function $f$.

\begin{definition}[Laplacian operator $\cL$]
	Define the operator $\cL$ on the space $\cS$ as
	\[\cL(f)(Z) = \EE[f(Z) - f(Z') | Z]\]
	for all polynomials $f \in \cS$.
\end{definition}

Note that this operator is well-defined since for any $f \in \cS$, $\EE[L(f)] = \EE[\EE[f(Z) - f(Z') | Z]] = \EE[f(Z) - f(Z')] = 0$ and hence, $L(f) \in \cS$.

\begin{lemma}\label{lem: eigenvector}
	For all $\al \in \NN^n$, $\chi_{\al}$ is an eigenvector of $\cL$ with eigenvalue $\frac{|\al|_0}{n}$.
\end{lemma}

\begin{proof}
	Recall that $Z'$ is obtained by choosing $i \in [n]$ uniformly at random and then setting $Z' = Z^{(i)}$. Therefore,
	\begin{align*}
		\cL(\chi_{\al})(Z) &= \EE[\chi_{\al}(Z) - \chi_{\al}(Z') | Z]\\
		&= \frac{1}{n}\sum_{i \le n} \EE[\chi_{\al}(Z) - \chi_{\al}(Z^{(i)}) | Z]
	\end{align*}
	When $\al_i = 0$, $\chi_{\al}(Z) - \chi_{\al}(Z^{(i)}) = 0$. Otherwise, $\EE[\chi_{\al}(Z) - \chi_{\al}(Z^{(i)})|Z] = \chi_{\al}(Z)$. Therefore, the above expression simplifies to $\frac{|\al|_0}{n} \chi_{\al}(Z)$.
\end{proof}

\begin{theorem}[Explicit Kernel]\label{thm: explicit_kernel_for_poly}
	For any mean-centered polynomial $f \in \cS$, there exists a polynomial $K_f$ on $2n$ variables $z_1, \ldots, z_n, z_1', \ldots, z_n'$, denoted collectively as $(z, z')$, with the following properties
	\begin{enumerate}
		\item $K_f(z', z) = -K_f(z, z')$
		\item $\EE[K_f(Z, Z') | Z] = f(Z)$ where $(Z, Z')$ is the exchangeable pair we consider above.
	\end{enumerate}
\end{theorem}

\begin{proof}
	Using \cref{propn: basis} and \cref{lem: eigenvector}, under the basis of polynomials $\chi_{\al}$, the operator $\cL$ is a diagonal matrix with nonzero diagonal entries and therefore, $\cL^{-1}$ exists and is explicitly given by
	\[\icL(f)(Z) = \sum_{\al} \frac{n}{|\al|_0}\coef{f}{\al} \chi_{\al}(Z)\]
	We then take $K_f(z, z') = \icL(f)(z) - \icL(f)(z')$. The first condition is obvious and for the second condition, we have
	\[\EE[K_f(Z, Z')|Z] = \EE[\icL(f)(Z) - \icL(f)(Z') | Z] = \cL(\icL(f)) = f\]
\end{proof}

As seen in the proof of \cref{thm: explicit_kernel_for_poly}, $\cL$ has a well-defined inverse $\icL$. We now define the matrix $\mK_{k, a, b}$ that we call the \textit{inner kernel}.

\begin{definition}[The inner kernel matrix $\mK_{k, a, b}$]
	For integers $k \ge 1, a, b \ge 0$ such that $a + b < k$, define the matrix $\mK_{k, a, b} \in \RR[Z]^{\cI \times \cK} \times \RR[Z]^{\cJ \times \cK}$ taking $2n$ variables $(z, z') = (z_1, \ldots, z_n, z_1', \ldots, z_n')$ as input as follows
	\[\mK_{k, a, b}(z, z') = \icL(\mG_{k, a, b})(z) - \icL(\mG_{k, a, b})(z')\]
\end{definition}

In the rest of this section except where explicitly stated, fix integers $k \ge 1, a, b \ge 0$ such that $a + b < k$. Then, the inner kernel $\mK_{k, a, b}$ is well-defined.

\begin{lemma}\label{lem: explicit_kernel_for_matrices}
	$\mK_{k, a, b}(Z, Z') = \frac{n}{k - a - b}(\mG_{k, a, b}(Z) - \mG_{k, a, b}(Z'))$
\end{lemma}

\begin{proof}
	\begin{align*}
		\mK_{k, a, b}(Z, Z') &= \cL^{-1}(\mG_{k, a, b})(Z) - \cL^{-1}(\mG_{k, a, b})(Z')\\
		&= \sum_{|\al|_0 = k - a - b} \coef{\mG_{k, a, b}}{\al} (\cL^{-1}(\chi_{\al})(Z) - \cL^{-1}(\chi_{\al})(Z'))\\
		&= \frac{n}{k - a - b}\sum_{|\al|_0 = k - a - b} \coef{\mG_{k, a, b}}{\al} (\chi_{\al}(Z) - \chi_{\al}(Z'))\\
		&= \frac{n}{k - a - b}(\mG_{k, a, b}(Z) - \mG_{k, a, b}(Z'))
	\end{align*}
\end{proof}

The following lemma postulates important properties of the the inner kernel, including how it interacts with $\mD_1$ and $\mD_2$.

\begin{lemma}\label{lem: props_of_exp_kernel_mat}
	$\mK_{k, a, b}$ satisfies the following properties
	\begin{enumerate}
		\item $\mK_{k, a, b}(z', z) = -\mK_{k, a, b}(z, z')$
		\item $\EE[\mK_{k, a, b}(Z, Z') | Z] = \mG_{k, a, b}(Z)$
		\item $(\mD_1(Z) - \mD_1(Z'))\mK_{k, a, b}(Z, Z') = \mK_{k, a, b}(Z, Z')(\mD_2(Z) - \mD_2(Z')) = 0$.
	\end{enumerate}
\end{lemma}

\begin{proof}
	The first equality is obvious from the definition. For the second equality, note that $\EE[\mG_{k, a, b}] = 0$ and $\mK_{k, a, b}$ is defined by replacing each entry $f$ of $\mG_{k, a, b}$ by the kernel polynomial $K_f$ as exhibited in \cref{thm: explicit_kernel_for_poly}. Now, we prove the third equality.

	Consider the matrix $(\mD_1(Z) - \mD_1(Z'))\mK_{k, a, b}(Z, Z')$ whose $[(I, \al_1, \gam_1), (J, \al_2, \gam_2)]$ entry is given by
	\[\frac{n}{k - a - b}\sqrt{\EE[Z^{2\al_1\cdot (1 - \gam_1)}]}(Z^{\al_1 \cdot \gam_1} - (Z')^{\al_1\cdot \gam_1})(\grad_{\al_1 + \al_2}\mX_k[I, J](Z) - \grad_{\al_1 + \al_2}\mX_k[I, J](Z'))\]
	where we have used \cref{lem: explicit_kernel_for_matrices}.
	We will argue that this term is identically $0$.
	We must have $Z' = Z^{(i)}$ for some $i \le n$. If $(\al_1 \cdot \gam_1)_i = 0$, then $Z^{\al_1 \cdot \gam_1} = (Z')^{\al_1\cdot \gam_1}$ and the above term is $0$.
	Otherwise, $(\al_1 + \al_2)_i \neq 0$ and so $\grad_{\al_1 +\al_2}$ on any polynomial $f$ will only contain the terms independent of $Z_i$, in which case $\grad_{\al_1 + \al_2}\mX_k[I, J](Z) = \grad_{\al_1 + \al_2}\mX_k[I, J](Z')$. In this case was well, the above term is $0$. The proof of the other equality is analogous.
\end{proof}

The reason we call $\mK_{k, a, b}$ the inner kernel is because, as seen above, it serves as a kernel for the inner matrix $\mG$ in the decomposition $\mF = \mD\mG\mD$.

Since we will need to work with Hermitian dilations, we define \[\mD = \begin{bmatrix}
	\mD_1 & 0\\
	0 & \mD_2
\end{bmatrix}\]

We will use the following basic fact extensively in our manipulations.

\begin{fact}
	For any matrix $\mA \in \RR[Z]^{\cI\times \cK} \times \RR[Z]^{\cJ\times \cK}$, $\mD \herm{\mA}\mD = \herm{\mD_1\mA\mD_2}$.
\end{fact}

\begin{proof}
	We have
	\begin{align*}
		\mD \herm{\mA}\mD =
		\begin{bmatrix}
			\mD_1 & 0\\
			0 & \mD_2
		\end{bmatrix}
		\begin{bmatrix}
			0 & \mA \\
			\mA^\T & 0
		\end{bmatrix}
		\begin{bmatrix}
			\mD_1 & 0\\
			0 & \mD_2
		\end{bmatrix}
		&=
		\begin{bmatrix}
			0 & \mD_1\mA\\
			\mD_2\mA^\T  & 0
		\end{bmatrix}
		\begin{bmatrix}
			\mD_1 & 0\\
			0 & \mD_2
		\end{bmatrix}\\
		&=
		\begin{bmatrix}
			0 & \mD_1\mA\mD_2\\
			\mD_2\mA^\T\mD_1  & 0
		\end{bmatrix}\\
		&= \herm{\mD_1\mA\mD_2}
	\end{align*}
\end{proof}

We start with a generalized version of a result from \cite{paulin2016}.

\begin{lemma}\label{lem: deviation_bound}
	Let $\mK = \herm{\mK}_{k, a, b}$. For any symmetric matrix valued function $\mR$ on the variables $Z$ of the same dimensions as $\mK$, such that $\EE\norm{\mK(Z, Z')\mR(Z)}  < \infty$, we have
	\[\EE[\herm{\mF}_{k, a, b}(Z)\mR(Z)] = \frac{1}{2}\EE[\mD(Z)\mK(Z, Z')\mD(Z) (\mR(Z) - \mR(Z'))]\]
\end{lemma}

\begin{proof}
	By \cref{lem: props_of_exp_kernel_mat}, we have
	\begin{align*}
		\EE[\herm{\mF}_{k, a, b}(Z)\mR(Z)] &= \EE[\mD(Z)\herm{\mG}_{k, a, b}(Z)\mD(Z) \mR(Z)]\\
		&= \EE[\mD(Z)\EE[\mK(Z, Z') | Z]\mD(Z) \mR(Z)]\\
		&= \EE[\mD(Z)\mK(Z, Z')\mD(Z) \mR(Z)]
	\end{align*}
	where the first equality follow from condition $2$ of \cref{lem: props_of_exp_kernel_mat} and the second follows from the pull-through property of expectations. Continuing,
	\begin{align*}
		\EE[\herm{\mF}_{k, a, b}(Z)\mR(Z)] &= \EE[\mD(Z)\mK(Z, Z')\mD(Z)\mR(Z)]\\
		&= \EE[\mD(Z')\mK(Z', Z)\mD(Z') \mR(Z')]\\
		&= -\EE[\mD(Z')\mK(Z, Z')\mD(Z') \mR(Z')]\\
		&= -\EE[\mD(Z)\mK(Z, Z')\mD(Z') \mR(Z')]\\
		&= -\EE[\mD(Z)\mK(Z, Z')\mD(Z) \mR(Z')]
	\end{align*}
	Here, the second equality follows from the fact that $(Z, Z')$ has the same distribution as $(Z', Z)$, so we can exchange them. The third, fourth and fifth equalities follow from conditions $1, 3, 3$ of \cref{lem: props_of_exp_kernel_mat} respectively. Adding the two displays, we get the result.
\end{proof}

\begin{definition}[Matrices $\mU_{k, a, b}, \mV_{k, a, b}$]
	We define the following matrices
	\[\mU_{k, a, b} = \EE[(\herm{\mF}_{k, a, b}(Z) - \herm{\mF}_{k, a, b}(Z'))^2|Z]\]
	\[\mV_{k, a, b} = \EE[(\mD(Z)\herm{\mK}_{k, a, b}(Z, Z')\mD(Z))^2|Z]\]
\end{definition}

The definition of $\mU_{k, a, b}$ is essentially unchanged from \cite{paulin2016}, where it is called the \textit{conditional variance}. The definition of $\mV_{k, a, b}$ is slightly different in our setting. This lets us exploit the specific product structure exhibited by $\herm{\mF}_{k, a, b}$ and the special properties of the inner kernel from \cref{lem: props_of_exp_kernel_mat}.

We will now prove a lemma which is similar to a lemma shown in \cite{paulin2016}.

\begin{lemma}\label{lem: main_pmt_bound}
	For any $s > 0$ and for any integer $t \ge 1$,
	\begin{align*}
		\Esch{\herm{\mF}_{k, a, b}}{2t}
		&\le \left(\frac{2t - 1}{4}\right)^t\Esch{s\mU_{k, a, b} + s^{-1}\mV_{k, a, b}}{t}
	\end{align*}
\end{lemma}

To prove this, we will need the following inequality.

\begin{lemma}[Polynomial mean value trace inequality, \cite{paulin2016}]\label{lem: mean_value_trace_inequality}
	For all matrices $\mA, \mB, \mC \in \HH^d$, all integers $q \ge 1$ and all $s > 0$,
	\begin{align*}
		\tr [\mC(\mA^q - \mB^q)]| \le \frac{q}{4} \tr[(s(\mA - \mB)^2 + s^{-1}\mC^2)(\mA^{q - 1} + \mB^{q - 1})]
	\end{align*}
\end{lemma}

\begin{proof}[Proof of \cref{lem: main_pmt_bound}]
	We start by invoking \cref{lem: deviation_bound} by setting $\mR(Z) = \herm{\mF}_{k, a, b}^{2t - 1}(Z)$.
	\begin{align*}
		\Esch{\herm{\mF}_{k, a, b}}{2t} &= \EE\tr [\herm{\mF}_{k, a, b}\cdot\herm{\mF}_{k, a, b}^{2t - 1}]\\
		&= \frac{1}{2}\EE[\mD(Z)\herm{\mK}_{k, a, b}(Z, Z')\mD(Z) (\herm{\mF}_{k, a, b}^{2t - 1}(Z) - \herm{\mF}_{k, a, b}^{2t - 1}(Z'))]
	\end{align*}

	Applying \cref{lem: mean_value_trace_inequality},
    {\footnotesize
	\begin{align*}
		&\Esch{\herm{\mF}_{k, a, b}}{2t}\\
		&\le (\frac{2t - 1}{8})\EE\tr[(s(\herm{\mF}_{k, a, b}(Z) - \herm{\mF}_{k, a, b}(Z'))^2 + s^{-1}(\mD(Z)\herm{\mK}_{k, a, b}(Z, Z')\mD(Z))^2)(\herm{\mF}_{k, a, b}^{2t - 2}(Z) + \herm{\mF}_{k, a, b}^{2t - 2}(Z'))]\\
		&= (\frac{2t - 1}{4})\EE\tr[(s(\herm{\mF}_{k, a, b}(Z) - \herm{\mF}_{k, a, b}(Z'))^2 + s^{-1}(\mD(Z)\herm{\mK}_{k, a, b}(Z, Z')\mD(Z))^2)\herm{\mF}_{k, a, b}^{2t - 2}(Z)]
	\end{align*}
}
	where the last line used the fact that $(Z, Z')$ has the same distribution as $(Z', Z)$ and applied condition $3$ of \cref{lem: props_of_exp_kernel_mat}. Using the definitions of $\mU_{k, a, b}$ and $\mV_{k, a, b}$, we get
	\begin{align*}
		\Esch{\herm{\mF}_{k, a, b}}{2t} &\le \frac{2t - 1}{4}\EE\tr[(s\mU_{k, a, b} + s^{-1}\mV_{k, a, b})\herm{\mF}_{k, a, b}^{2t - 2}]\\
		&\le \frac{2t - 1}{4}\left(\Esch{s\mU_{k, a, b} + s^{-1}\mV_{k, a, b}}{t}\right)^{1/t}(\Esch{\herm{\mF}_{k, a, b}}{2t})^{(t - 1)/t}
	\end{align*}
	where we used H\"{o}lder's inequality for the trace and H\"{o}lder's inequality for the expectation. Rearranging gives the result.
\end{proof}

\subsection{Proof of \cref{lem: main_general}}

\cref{lem: main_pmt_bound} suggests that in order to bound $\Esch{\herm{\mF}_{k, a, b}}{2t}$, it suffices to bound $\Esch{\mU_{k, a, b}}{t}$ and $\Esch{\mV_{k, a, b}}{t}$. Indeed, this will be our strategy. To bound $\Esch{\mU_{k, a, b}}{t}$, we will bound it via the matrices that we define below.

\begin{definition}[Matrices $\mDel_1^{k, a, b}, \mDel_2^{k, a, b}, \mDel_3^{k, a, b}$]
	Define the matrices
	\[\mDel_1^{k, a, b} = \EE[((\mD(Z) - \mD(Z'))\herm{\mG}_{k, a, b}(Z)\mD(Z))^2|Z]\]
	\[\mDel_2^{k, a, b} = \EE[(\mD(Z)(\herm{\mG}_{k, a, b}(Z) - \herm{\mG}_{k, a, b}(Z'))\mD(Z))^2|Z]\]
	\[\mDel_3^{k, a, b} = \EE[(\mD(Z)\herm{\mG}_{k, a, b}(Z)(\mD(Z) - \mD(Z')))^2|Z]\]
\end{definition}

\begin{lemma}\label{lem: bound_U_by_Deltas}
	$\mU_{k, a, b} \preceq 3(\mDel_1^{k, a, b} + \mDel_2^{k, a, b} + \mDel_3^{k, a, b})$.
\end{lemma}

To prove this lemma, we will use the following lemma.

\begin{lemma}\label{lem: orthogonality}
	We have the relations
	\[(\mD(Z) - \mD(Z'))(\herm{\mG}_{k, a, b}(Z)\mD(Z) - \herm{\mG}_{k, a, b}(Z')\mD(Z')) = 0\]
	\[(\herm{\mG}_{k, a, b}(Z) - \herm{\mG}_{k, a, b}(Z'))(\mD(Z) - \mD(Z')) = 0\]
\end{lemma}

\begin{proof}[Proof sketch]
	The proof is similar to the proof of third equality in \cref{lem: props_of_exp_kernel_mat}. When $Z'$ is set to $Z^{(i)}$ for some $i \le n$, when a diagonal entry of $\mD(Z) - \mD(Z')$ is nonzero, then the corresponding row of $\herm{\mG}_{k, a, b}(Z)\mD(Z) - \herm{\mG}_{k, a, b}(Z')\mD(Z')$ will be $0$. The second equality is analogous.
\end{proof}

\begin{proof}[Proof of \cref{lem: bound_U_by_Deltas}]
	We have
	\begin{align*}
		&(\herm{\mF}_{k, a, b}(Z) - \herm{\mF}_{k, a, b}(Z'))^2\\
		&= (\mD(Z)\herm{\mG}_{k, a, b}(Z)\mD(Z) - \mD(Z')\herm{\mG}_{k, a, b}(Z')\mD(Z'))^2\\
		&= \bigg(\mD(Z)\herm{\mG}_{k, a, b}(Z)(\mD(Z) - \mD(Z')) + \mD(Z)(\herm{\mG}_{k, a, b}(Z) - \herm{\mG}_{k, a, b}(Z'))\mD(Z')\\
        &\qquad + (\mD(Z) - \mD(Z'))\herm{\mG}_{k, a, b}(Z')\mD(Z')\bigg)^2\\
		&= \bigg(\mD(Z)\herm{\mG}_{k, a, b}(Z)(\mD(Z) - \mD(Z')) + \mD(Z)(\herm{\mG}_{k, a, b}(Z) - \herm{\mG}_{k, a, b}(Z'))\mD(Z)\\
        & \qquad + (\mD(Z) - \mD(Z'))\herm{\mG}_{k, a, b}(Z)\mD(Z)\bigg)^2
	\end{align*}
	where the last equality follows from \cref{lem: orthogonality}. Taking expectations conditioned on $Z$ and applying \cref{fact: cs}, we immediately get $\mU_{k, a, b} \preceq 3(\mDel_1^{k, a, b} + \mDel_2^{k, a, b} + \mDel_3^{k, a, b})$.
\end{proof}

In subsequent sections, we will prove the following technical bounds on the matrices we have considered so far.

\begin{restatable}{lemma}{boundDelTwo}\label{lem: bound_on_Del2}
	For all integers $t \ge 1$, \[\Esch{\mDel_2^{k, a, b}}{t} \le \frac{(2\dpoly)^t}{n^t} (\Esch{\herm{\mF}_{k, a, b + 1}}{2t} + \Esch{\herm{\mF}_{k, a + 1, b}}{2t})\]
\end{restatable}

\begin{restatable}{lemma}{boundV}\label{lem: bound_on_V}
	For all integers $t \ge 1$, \[\Esch{\mV_{k, a, b}}{t} \le (2\dpoly)^tn^t (\Esch{\herm{\mF}_{k, a, b + 1}}{2t} + \Esch{\herm{\mF}_{k, a + 1, b}}{2t})\]
\end{restatable}

\begin{restatable}{lemma}{boundDelOne}\label{lem: bound_on_Del1}
	For all integers $t \ge 1$, \[\Esch{\mDel_1^{k, a, b}}{t} \le \frac{(8d\dpoly)^t}{n^t}\Esch{\herm{\mF}_{k, a, b}}{2t}\]
\end{restatable}

\begin{restatable}{lemma}{boundDelThree}\label{lem: bound_on_Del3}
	For all integers $t \ge 1$, \[\Esch{\mDel_3^{k, a, b}}{t} \le \frac{(4\dpoly)^t}{n^t}\Esch{\herm{\mF}_{k, a, b}}{2t}\]
\end{restatable}

Assuming the above lemmas, we can complete the proof of \cref{lem: main_general}, which we restate for convenience.

\maingeneral*

\begin{proof}[Proof of \cref{lem: main_general}]
	Using \cref{lem: main_pmt_bound}, \cref{lem: bound_U_by_Deltas}, we get that for any $s > 0$,
	\begin{align*}
		\Esch{\herm{\mF}_{k, a, b}}{2t}
		&\le (\frac{2t - 1}{4})^t\Esch{s\mU_{k, a, b} + s^{-1}\mV_{k, a, b}}{t}\\
		&\le t^t(s^t\Esch{\mU_{k, a, b}}{t} + s^{-t}\Esch{\mV_{k, a, b}}{t})\\
		&\le (9st)^t(\Esch{\mDel_1^{k, a, b}}{t} + \Esch{\mDel_2^{k, a, b}}{t} + \Esch{\mDel_3^{k, a, b}}{t}) + t^ts^{-t}\Esch{\mV_{k, a, b}}{t}
	\end{align*}
	Let $\rho = s / n$. Since the inequality is true for any choice of $s > 0$, it is true for any choice of $\rho > 0$.
	Now, using \cref{lem: bound_on_Del1}, \cref{lem: bound_on_Del3},
	\begin{align*}
		(9st)^t(\Esch{\mDel_1^{k, a, b}}{t} + \Esch{\mDel_3^{k, a, b}}{t}) &\le (9st)^t\bigg(\frac{(8d\dpoly)^t}{n^t} + \frac{(4\dpoly)^t}{n^t}\bigg)\Esch{\herm{\mF}_{k, a, b}}{2t}\\
		&= \rho^t (C_1td\dpoly)^t\Esch{\herm{\mF}_{k, a, b}}{2t}
	\end{align*}
	for an absolute constant $C_1 > 0$. Using \cref{lem: bound_on_Del2}, \cref{lem: bound_on_V},
    {\footnotesize
	\begin{align*}
		(9st)^t\Esch{\mDel_2^{k, a, b}}{t} + t^ts^{-t}\Esch{\mV_{k, a, b}}{t} & \le\bigg((9st)^t\frac{(2\dpoly)^t}{n^t} + t^ts^{-t}(2\dpoly)^tn^t\bigg)(\Esch{\herm{\mF}_{k, a, b + 1}}{2t} + \Esch{\herm{\mF}_{k, a + 1, b}}{2t})\\
		&\le (\rho^tC_2^t + \rho^{-t}C_3^t) (t\dpoly)^t (\Esch{\herm{\mF}_{k, a, b + 1}}{2t} + \Esch{\herm{\mF}_{k, a + 1, b}}{2t})
	\end{align*}
}
	for absolute constants $C_2, C_3 > 0$.
	Therefore,
	\begin{align*}
		\Esch{\herm{\mF}_{k, a, b}}{2t} &\le \rho^t (C_1td\dpoly)^t\Esch{\herm{\mF}_{k, a, b}}{2t} \\&\qquad+ (\rho^tC_2^t + \rho^{-t}C_3^t) (t\dpoly)^t(\Esch{\herm{\mF}_{k, a, b + 1}}{2t} + \Esch{\herm{\mF}_{k, a + 1, b}}{2t})
	\end{align*}
	We choose $\rho > 0$ so that $\rho^t (C_1td\dpoly)^t = \frac{1}{2}$ to get
	\begin{align*}
		\Esch{\herm{\mF}_{k, a, b}}{2t} &\le \frac{1}{2}\Esch{\herm{\mF}_{k, a, b}}{2t} + \frac{1}{2}(Ct^2d\dpoly^2)^t (\Esch{\herm{\mF}_{k, a, b + 1}}{2t} + \Esch{\herm{\mF}_{k, a + 1, b}}{2t})
	\end{align*}
	for an absolute constant $C > 0$.
	Rearranging yields the result.
\end{proof}

\subsection{Bounding $\mDel_2^{k, a, b}$ and $\mV_{k, a, b}$}

The next lemma relates $\mV_{k, a, b}$ to $\mDel_2^{k, a, b}$ upto a factor of $n^2$ which will be enough for us. We can then focus on bounding $\mDel_2^{k, a, b}$.

\begin{lemma}\label{lem: bounding_V_loewner}
	$\mV_{k, a, b} \preceq n^2 \mDel_2^{k, a, b}$
\end{lemma}

\begin{proof}
	Using \cref{lem: explicit_kernel_for_matrices},
	\begin{align*}
		\mV_{k, a, b} &= \EE[(\mD(Z)\herm{\mK}_{k, a, b}(Z, Z')\mD(Z))^2|Z]\\
		&= \EE[(\mD(Z)\bigg(\frac{n}{k - a - b}(\herm{\mG}_{k, a, b}(Z) - \herm{\mG}_{k, a, b}(Z'))\bigg)\mD(Z))^2|Z]\\
		&\preceq n^2\EE[(\mD(Z)(\herm{\mG}_{k, a, b}(Z) - \herm{\mG}_{k, a, b}(Z'))\mD(Z))^2|Z]\\
		&= n^2 \mDel_2^{k, a, b}
	\end{align*}
\end{proof}

For $1 \le i \le n$ and $1 \le l \le d$, let $\mat{e}_{i, l} \in \NN^n$ denote the vector $\al$ with $\al_i = l$ and $\al_j = 0$ for $j \neq i$.
We note the following simple proposition.

\begin{propn}\label{propn: difference_equality}
	For any polynomial $f$ such that the degree of $Z_i$ is at most $d$, \[f(Z) - f(Z^{(i)}) = \sum_{1 \le l \le d} (Z_i^l - \resamp{Z_i}^l)\grad_{\mat{e}_{i, l}}(f)\]
\end{propn}

We now restate and prove \cref{lem: bound_on_Del2}.

\boundDelTwo*

\begin{proof}
	Consider
	\begin{align*}
		\mDel_2^{k, a, b} &= \EE[(\mD(Z)(\herm{\mG}_{k, a, b}(Z) - \herm{\mG}_{k, a, b}(Z'))\mD(Z))^2|Z]\\
		&= \EE\bigg[ \begin{bmatrix}
			\mM\mM^\T & 0\\
			0 & \mM^\T\mM
		\end{bmatrix} | Z\bigg]\\
		&= \begin{bmatrix}
			\EE[\mM\mM^\T|Z] & 0\\
			0 & \EE[\mM^\T\mM|Z]
		\end{bmatrix}
	\end{align*}
	where $\mM = \mD_1(Z)(\mG_{k, a, b}(Z) - \mG_{k, a, b}(Z'))\mD_2(Z)$. Using \cref{propn: difference_equality},
    {\footnotesize
    \begin{align*}
		\EE[\mM\mM^T | Z] &= \EE[\mD_1(Z)(\mG_{k, a, b}(Z) - \mG_{k, a, b}(Z'))\mD_2(Z)\cdot \mD_2(Z) (\mG_{k, a, b}(Z) - \mG_{k, a, b}(Z'))^\T\mD_1(Z)|Z]\\
		&= \frac{1}{n} \sum_{i = 1}^n\EE[\mD_1(Z)(\mG_{k, a, b}(Z) - \mG_{k, a, b}(Z^{(i)}))\mD_2(Z)\cdot \mD_2(Z) (\mG_{k, a, b}(Z) - \mG_{k, a, b}(Z^{(i)}))^\T\mD_1(Z)|Z]\\
		&= \frac{1}{n}\sum_{i = 1}^n \sum_{l = 1}^d\EE[(Z_i^l - \resamp{Z_i}^l)^2|Z]\cdot \mD_1(Z)(\grad_{\mat{e}_{i, l}} \mG_{k, a, b})(Z)\mD_2(Z)\cdot \mD_2(Z) (\grad_{\mat{e}_{i, l}} \mG_{k, a, b})(Z)^\T\mD_1(Z)
	\end{align*}
}
	Define $\mN_{i, l}(Z) := \mD_1(Z)(\grad_{\mat{e}_{i, l}} \mG_{k, a, b})(Z)\mD_2(Z)$. Then,
	\begin{align*}
		\EE[\mM\mM^T | Z] &= \frac{1}{n}\sum_{i = 1}^n \sum_{l = 1}^d\EE[(Z_i^l - \resamp{Z_i}^l)^2|Z]\cdot \mN_{i, l}(Z)\mN_{i, l}(Z)^\T\\
		&\preceq \frac{2}{n}\sum_{i = 1}^n \sum_{l = 1}^d(Z_i^{2l} + \EE[Z_i^{2l}])\cdot \mN_{i, l}(Z)\mN_{i, l}(Z)^\T
	\end{align*}
	Similarly,
	\begin{align*}
		\EE[\mM^\T\mM | Z] &\preceq \frac{2}{n}\sum_{i = 1}^n \sum_{l = 1}^d(Z_i^{2l} + \EE[Z_i^{2l}])\cdot \mN_{i, l}(Z)^\T\mN_{i, l}(Z)
	\end{align*}

	\begin{claim}\label{claim: reduction_general}
		We have the relations
		\[\sum_{i = 1}^n\sum_{l = 1}^d (Z_i^{2l} + \EE[Z_i^{2l}]) \cdot \mN_{i, l}(Z)\mN_{i, l}(Z)^\T = (b + 1)\mF_{k, a, b + 1}\mF_{k, a, b + 1}^\T\]
		\[\sum_{i = 1}^n \sum_{l = 1}^d(Z_i^{2l} + \EE[Z_i^{2l}])\cdot \mN_{i, l}(Z)^\T\mN_{i, l}(Z) = (a + 1)\mF_{k, a + 1, b}^\T\mF_{k, a + 1, b}\]
	\end{claim}

	Using this claim, we have
	\[\EE[\mM\mM^T | Z] \preceq \frac{2(b + 1)}{n}\mF_{k, a, b + 1}\mF_{k, a, b + 1}^\T \preceq \frac{2\dpoly}{n}\mF_{k, a, b + 1}\mF_{k, a, b + 1}^\T\]
	\[\EE[\mM^\T\mM | Z] \preceq \frac{2(a + 1)}{n}\mF_{k, a + 1, b}^\T\mF_{k, a + 1, b} \preceq \frac{2\dpoly}{n}\mF_{k, a + 1, b}^\T\mF_{k, a + 1, b}\]
	Therefore,
	\begin{align*}
		\Esch{\mDel_2^{k, a, b}}{t} &= \Esch{\EE[\mM\mM^\T|Z]}{t} + \Esch{\EE[\mM^\T\mM|Z]}{t}\\
		&\le \frac{(2\dpoly)^t}{n^t} (\Esch{\mF_{k, a, b + 1}}{2t} + \Esch{\mF_{k, a + 1, b}}{2t})\\
		&\le \frac{(2\dpoly)^t}{n^t} (\Esch{\herm{\mF}_{k, a, b + 1}}{2t} + \Esch{\herm{\mF}_{k, a + 1, b}}{2t})
	\end{align*}
\end{proof}

It remains to prove the claim.
\begin{proof}[Proof of~\cref{claim: reduction_general}]
	We will prove the first relation, the second is analogous.
	For a fixed $i \le n, l \le d$, consider any nonzero entry $[(I_1, \al_1, \gam_1), (I_2, \al_2, \gam_2)]$ of $\sum_{i = 1}^n \sum_{l = 1}^d (Z_i^{2l} + \EE[Z_i^{2l}]) \mN_{i, l}(Z) \mN_{i, l}(Z)^\T$, where $I_1, I_2 \in \cI, (\al_1, \gam_1), (\al_2, \gam_2) \in \cK$. We must have $|\al_1|_0 = |\al_2|_0 = a$, in which case the entry  is equal to
    {\footnotesize
	\begin{align*}
		\sum_{\substack{(J, \al_3, \gam_3) \in \cJ\times \cK\\ |\al_3| = b \\ \al_1\al_3 = \al_2\al_3 = 0}} &(Z_i^{2l} + \EE[Z_i^{2l}]) \cdot (\sqrt{\EE[Z^{2\al_1\cdot (1 - \gam_1) + 2\al_3\cdot (1 - \gam_3)}]}Z^{\al_1\cdot \gam_1 + \al_3\cdot\gam_3}\grad_{\mat{e}_{i, l}} \grad_{\al_1 + \al_3} \mX_k[I_1, J])\\
		&\cdot (\sqrt{\EE[Z^{2\al_2\cdot (1 - \gam_2) + 2\al_3\cdot (1 - \gam_3)}]}Z^{\al_2\cdot \gam_2 + \al_3\cdot\gam_3}\grad_{\mat{e}_{i, l}} \grad_{\al_2 + \al_3} \mX_k[I_2, J])
	\end{align*}}
	Note that the term inside the summation is nonzero only when $\mat{e}_{i, l}\cdot (\al_1 + \al_3) = \mat{e}_{i, l} \cdot (\al_2 + \al_3) = 0$. Hence, this sum can be written as
	\begin{align*}
		\sum_{\substack{(J, \al_3, \gam_3) \in \cJ\times \cK\\ |\al_3| = b + 1 \\ \mat{e}_{i, l} \unlhd \al_3, \al_1\al_3 = \al_2\al_3 = 0}} &(\sqrt{\EE[Z^{2\al_1\cdot (1 - \gam_1) + 2\al_3\cdot (1 - \gam_3)}]}Z^{\al_1\cdot \gam_1 + \al_3\cdot\gam_3}\grad_{\al_1 + \al_3} \mX_k[I_1, J])\\
		&\cdot (\sqrt{\EE[Z^{2\al_2\cdot (1 - \gam_2) + 2\al_3\cdot (1 - \gam_3)}]}Z^{\al_2\cdot \gam_2 + \al_3\cdot\gam_3}\grad_{\al_2 + \al_3} \mX_k[I_2, J])
	\end{align*}
	When we add this entry over all $i \le n, l \le d$, this simplifies to
	\begin{align*}
		(b + 1) \cdot \sum_{\substack{(J, \al_3, \gam_3) \in \cJ\times \cK\\ |\al_3| = b + 1 \\ \al_1\al_3 = \al_2\al_3 = 0}} &(\sqrt{\EE[Z^{2\al_1\cdot (1 - \gam_1) + 2\al_3\cdot (1 - \gam_3)}]}Z^{\al_1\cdot \gam_1 + \al_3\cdot\gam_3}\grad_{\al_1 + \al_3} \mX_k[I_1, J])\\
		&\cdot (\sqrt{\EE[Z^{2\al_2\cdot (1 - \gam_2) + 2\al_3\cdot (1 - \gam_3)}]}Z^{\al_2\cdot \gam_2 + \al_3\cdot\gam_3}\grad_{\al_2 + \al_3} \mX_k[I_2, J])
	\end{align*}
	The factor of $(b + 1)$ came because the index $i$ could have been chosen from among all the active indices in $\al_3$. But this is precisely the $[(I_1, \al_1, \gam_1), (I_2, \al_2, \gam_2)]$ entry of $(b + 1)\mF_{k, a, b + 1}\mF_{k, a, b + 1}^\T$, proving the claim.
\end{proof}

We restate and prove \cref{lem: bound_on_V}.

\boundV*

\begin{proof}
	Using \cref{lem: bounding_V_loewner} and \cref{lem: bound_on_Del2}, we get
	\begin{align*}
		\Esch{\mV_{k, a, b}}{t} &\le n^{2t}\Esch{\mDel_2^{k, a, b}}{t}\\
		&\le (2\dpoly)^tn^t (\Esch{\herm{\mF}_{k, a, b + 1}}{2t} + \Esch{\herm{\mF}_{k, a + 1, b}}{2t})
	\end{align*}
\end{proof}

\subsection{Bounding $\mDel_1^{k, a, b}$ and $\mDel_3^{k, a, b}$}

Define $\sqcup$ to be the disjoint union of sets. For $1 \le i \le n$ and $1 \le l \le d$, define the diagonal matrices $\mPi_{i, l}, \mPi_{i, l}', \mPi_i, \mPi_i' \in \RR^{(\cI \times \cK) \sqcup (\cJ \times \cK)} \times \RR^{(\cI \times \cK) \sqcup (\cJ \times \cK)}$ (the same dimensions as $\mD$) as
{\footnotesize
\[\mPi_{i, l}[(I, \al, \beta), (I, \al, \beta)] = \begin{dcases}
	1 & \text{ if $(\al \cdot \gam)_i \neq 0$ and $\al_i = l$}\\
	0 & \text{ o.w.}
\end{dcases}\qquad \mPi_i[(I, \al, \beta), (I, \al, \beta)] = \begin{dcases}
	1 & \text{ if $(\al \cdot \gam)_i \neq 0$}\\
	0 & \text{ o.w.}
\end{dcases}\]
\[\mPi'_{i, l}[(I, \al, \beta), (I, \al, \beta)] = \begin{dcases}
	1 & \text{ if $\al_i \neq 0$ and $\al_i = l$}\\
	0 & \text{ o.w.}
\end{dcases}\qquad \mPi_i'[(I, \al, \beta), (I, \al, \beta)] = \begin{dcases}
	1 & \text{ if $\al_i \neq 0$}\\
	0 & \text{ o.w.}
\end{dcases}\]
}
for all $I \in \cI \sqcup \cJ$.
Note that for all $i \le n$, $\mPi_i = \sum_{l = 1}^d \mPi_{i, l}$.

Also, for all $1 \le i \le n$, we define the permutation matrices $\mSig_i \in\RR^{(\cI \times \cK) \sqcup (\cJ \times \cK)} \times \RR^{(\cI \times \cK) \sqcup (\cJ \times \cK)}$ as follows. Consider the permutation $\sig_1$ on $\cI\times \cK$ that transposes $(I, \al, \gam)$ and $(I, \al, \gam + \mat{e}_i)$ for all $(I, \al, \gam) \in \cI\times \cK$ such that $\al_i \neq 0$. Here, $\mat{e}_i \in \{0, 1\}^n$ has exactly one nonzero entry, which is in the $i$th position, and $\gam + \mat{e}_i$ is the usual addition over $\mathbb{F}_2$. $\sig_1$ leaves other positions fixed. Let $\mSig^{(1)}_i$ be the permutation matrix for $\sig$. Similarly, let $\mSig^{(2)}_i$ be the permutation matrix of the permutation $\sig_2$ on $\cJ \times \cK$ that transposes $(J, \al, \gam)$ and $(J, \al, \gam + \mat{e}_i)$ for all $(J, \al, \gam) \in \cJ\times \cK$ such that $\al_i \neq 0$, and leaves all other positions fixed. Then, we define $\mSig_i = \begin{bmatrix}
	\mSig^{(1)}_i & 0\\
	0 & \mSig^{(2)}_i
\end{bmatrix}$. The following fact is easy to verify.

\begin{fact}\label{fact: commutativity}
	$\mPi'_{i, l}\mSig_i = \mSig_i\mPi'_{i, l}$ and $\mPi_i' \mSig_i = \mSig_i \mPi_i'$.
\end{fact}

We are now ready to prove \cref{lem: bound_on_Del1} which we restate for convenience.

\boundDelOne*

\begin{proof}
	Firstly,
	\begin{align*}
		\mDel_1^{k, a, b} &= \EE[((\mD(Z) - \mD(Z'))\herm{\mG}_{k, a, b}(Z)\mD(Z))^2|Z]\\
		&= \EE[(\mD(Z) - \mD(Z'))\herm{\mG}_{k, a, b}(Z)\mD(Z)\cdot \mD(Z)\herm{\mG}_{k, a, b}(Z)(\mD(Z) - \mD(Z'))|Z]\\
		&= \EE[(\mD(Z) - \mD(Z')) \mM(Z) (\mD(Z) - \mD(Z'))|Z]
	\end{align*}
	where we define $\mM(Z) = \herm{\mG}_{k, a, b}(Z)\mD(Z)\cdot \mD(Z)\herm{\mG}_{k, a, b}(Z)$.
	Recall that $Z' = Z^{(i)}$ for some $i$ randomly chosen from $[n]$ uniformly. Observing that $\mD(Z) - \mD(Z^{(i)}) = \mPi_i(\mD(Z) - \mD(Z^{(i)}))$ for all $i$, we get
	\begin{align*}
		\mDel_1^{k, a, b}
		&= \EE[ \EE_{i \in [n]} [(\mD(Z) - \mD(Z^{(i)})) \mM(Z) (\mD(Z) - \mD(Z^{(i)}))]|Z]\\
		&= \EE[\EE_{i \in [n]} [\mPi_i(\mD(Z) - \mD(Z^{(i)})) \mM(Z) (\mD(Z) - \mD(Z^{(i)}))\mPi_i]|Z]\\
		&\preceq 2\bigg(\EE[\EE_{i \in [n]} [\mPi_i\mD(Z)\mM(Z)\mD(Z)\mPi_i]|Z] + \EE[\EE_{i \in [n]} [\mPi_i\mD(Z^{(i)})\mM(Z)\mD(Z^{(i)})\mPi_i]|Z]\bigg)\\
		&\preceq 2\bigg(\EE_{i \in [n]} [\mPi_i\herm{\mF}_{k, a, b}^2\mPi_i] + \EE[\EE_{i \in [n]} [\mPi_i\mD(Z^{(i)})\mM(Z)\mD(Z^{(i)})\mPi_i]|Z]\bigg)\\
		& \preceq 2(\mDel_{10} + \mDel_{11})
	\end{align*}
	where we define
	\[\mDel_{10} = \EE_{i \in [n]} [\mPi_i\herm{\mF}_{k, a, b}^2\mPi_i], \qquad \mDel_{11} = \EE[\EE_{i \in [n]} [\mPi_i\mD(Z^{(i)})\mM(Z)\mD(Z^{(i)})\mPi_i]|Z]\]
	Invoking \cref{lem: jensen_trace} over the interval $[0, \infty)$ with the convex continuous function $f(x) = x^t$, $\mB_i = \herm{\mF}_{k, a, b}^2, \mA_i = \frac{1}{\sqrt{\dpoly}}\mPi_i$ where we observe that $\sum_{i = 1}^n \mA_i \mA_i^T = \frac{1}{\dpoly}\sum_{i = 1}^n \mPi_i^2\preceq \mI$, we get

	\begin{align*}
		\Esch{\mDel_{10}}{t} = \EE\tr[\mDel_{10}^t] = \EE\tr[\bigg(\EE_{i \in [n]} [\mPi_i\herm{\mF}_{k, a, b}^2\mPi_i]\bigg)^t] &= \frac{1}{n^t}\EE\tr[\bigg(\sum_{i = 1}^n\mPi_i\herm{\mF}_{k, a, b}^2\mPi_i\bigg)^t]\\
		&\le \frac{\dpoly^{t - 1}}{n^t}\EE\tr[\bigg(\sum_{i = 1}^n\mPi_i\herm{\mF}_{k, a, b}^{2t}\mPi_i\bigg)]\\
		&\le \frac{\dpoly^{t - 1}}{n^t}\EE\tr[\bigg(\sum_{i = 1}^n\mPi_i^2\bigg)\herm{\mF}_{k, a, b}^{2t}]\\
		&\le \frac{\dpoly^t}{n^t}\EE\tr[\herm{\mF}_{k, a, b}^{2t}]\\
		&= \frac{\dpoly^t}{n^t}\Esch{\herm{\mF}_{k, a, b}}{2t}
	\end{align*}

	Now, consider
	\begin{align*}
		\mDel_{11} &= \EE[\EE_{i \in [n]} [\mPi_i\mD(Z^{(i)})\mM(Z)\mD(Z^{(i)})\mPi_i]|Z]\\
		&= \EE[\EE_{i \in [n]} 	[(\sum_{l = 1}^d\mPi_{i, l})\mD(Z^{(i)})\mM(Z)\mD(Z^{(i)})(\sum_{l = 1}^d\mPi_{i, l})]|Z]\\
		&\preceq d\cdot \EE[\EE_{i \in [n]} [\sum_{l = 1}^d\mPi_{i, l}\mD(Z^{(i)})\mM(Z)\mD(Z^{(i)})\mPi_{i, l}]|Z]\\
		&= d\cdot \EE_{i \in [n]} [\sum_{l = 1}^d \frac{\EE[Z_i^{2l}]}{Z_i^{2l}}\mPi_{i, l}\mD(Z)\mM(Z)\mD(Z)\mPi_{i, l}]\\
		&= \frac{d}{n} \sum_{i = 1}^n\sum_{l = 1}^d \frac{\EE[Z_i^{2l}]}{Z_i^{2l}}\mPi_{i, l}\mD(Z)\mM(Z)\mD(Z)\mPi_{i, l}\\
		&= \frac{d}{n} \sum_{i = 1}^n\sum_{l = 1}^d \mPi_{i, l}\mSig_i\mD(Z)\mM(Z)\mD(Z)\mSig_i^\T\mPi_{i, l}\\
		&= \frac{d}{n} \sum_{i = 1}^n\sum_{l = 1}^d \mPi_{i, l}\mSig_i\herm{\mF}_{k, a, b}^2\mSig_i^\T\mPi_{i, l}
	\end{align*}
	We now invoke \cref{lem: jensen_trace} on $d\dpoly$ terms with $\mB_{i, l} = \herm{\mF}_{k, a, b}^2$ and $\mA_{i, l} = \frac{1}{\sqrt{\dpoly}} \mPi_{i, l}\mSig_i$ where we observe that
	\[\sum_{i = 1}^n\sum_{l = 1}^d \mA_{i, l} \mA_{i, l}^T = \frac{1}{\dpoly}\sum_{i = 1}^n\sum_{l = 1}^d \mPi_{i, l}\mSig_i\mSig_i^\T\mPi_{i, l}^\T = \frac{1}{\dpoly}\sum_{i = 1}^n\sum_{l = 1}^d \mPi_{i, l}^2 \preceq \mI\]
	to get
	\begin{align*}
		\Esch{\mDel_{11}}{t} = \EE \tr[\mDel_{11}^t] &\le \frac{d^t}{n^t} \EE\tr[(\sum_{i = 1}^n\sum_{l = 1}^d \mPi_{i, l}\mSig_i\herm{\mF}_{k, a, b}^2\mSig_i^\T\mPi_{i, l})^t]\\
		&\le \frac{(d\dpoly)^t}{n^t}\EE\tr[\bigg(\frac{1}{\dpoly}\sum_{i = 1}^n\sum_{l = 1}^d \mPi_{i, l}\mSig_i\herm{\mF}_{k, a, b}^{2t}\mSig_i^\T\mPi_{i, l}\bigg)]\\
		&= \frac{(d\dpoly)^t}{n^t}\EE\tr[\bigg(\frac{1}{\dpoly}\sum_{i = 1}^n\sum_{l = 1}^d \mSig_i^\T\mPi_{i, l}\mPi_{i, l}\mSig_i\herm{\mF}_{k, a, b}^{2t}\bigg)]
	\end{align*}
	To simplify this, we use \cref{fact: commutativity} to get
    \begin{align*}
    	\sum_{i = 1}^n\sum_{l = 1}^d \mSig_i^\T(\mPi_{i, l})^2\mSig_i \preceq \sum_{i = 1}^n\sum_{l = 1}^d \mSig_i^\T(\mPi'_{i, l})^2\mSig_i = \sum_{i = 1}^n\sum_{l = 1}^d \mPi'_{i, l}\mSig_i^\T\mSig_i\mPi'_{i, l} &= \sum_{i = 1}^n\sum_{l = 1}^d \mPi'_{i, l}\mPi'_{i, l}\\ &\preceq \dpoly \mI
    \end{align*}
	Therefore,
	\[\Esch{\mDel_{11}}{t} \le  \frac{(d\dpoly)^t}{n^t}\EE\tr[\herm{\mF}_{k, a, b}^{2t}] = \frac{(d\dpoly)^t}{n^t}\Esch{\herm{\mF}_{k, a, b}}{2t}\]
	Putting them together, using \cref{fact: holder},
	\begin{align*}
		\Esch{\mDel_1^{k, a, b}}{t} &\le 4^t(\Esch{\mDel_{10}}{t} + \Esch{\mDel_{11}}{t})\\
		&\le \frac{(8d\dpoly)^t}{n^t}\Esch{\herm{\mF}_{k, a, b}}{2t}
	\end{align*}
\end{proof}

We now restate and prove \cref{lem: bound_on_Del3}.

\boundDelThree*

\begin{proof}
	Recall that $Z' = Z^{(i)}$ for $i$ sampled uniformly from $[n]$. Then,
	\begin{align*}
		\mDel_3^{k, a, b} &= \EE[(\mD(Z)\herm{\mG}_{k, a, b}(Z)(\mD(Z) - \mD(Z')))^2|Z]\\
		&= \EE[\EE_{i \in [n]} [(\mD(Z)\herm{\mG}_{k, a, b}(Z)(\mD(Z) - \mD(Z^{(i)})))^2] | Z]\\
		&= \EE[\EE_{i \in [n]} [(\mD(Z)\herm{\mG}_{k, a, b}(Z)\mPi_i(\mD(Z) - \mD(Z^{(i)})))^2] | Z]
	\end{align*}
	where we use the fact that $\mD(Z) - \mD(Z^{(i)}) = \mPi_i(\mD(Z) - \mD(Z^{(i)}))$ for all $i$. Define $\mM(Z) = \mD(Z)\herm{\mG}_{k, a, b}$ to get
	\begin{align*}
		\mDel_3^{k, a, b} &= \EE[\EE_{i \in [n]} [\mM(Z) \mPi_i(\mD(Z) - \mD(Z^{(i)}))^2\mPi_i\mM(Z)^\T] | Z]\\
		&\preceq 2(\EE[\EE_{i \in [n]} [\mM(Z) \mPi_i\mD(Z)^2\mPi_i\mM(Z)^\T] | Z] + \EE[\EE_{i \in [n]} [\mM(Z) \mPi_i\mD(Z^{(i)})^2\mPi_i\mM(Z)^\T] | Z])\\
		&= 2(\EE_{i \in [n]} [\mM(Z) \mPi_i\mD(Z)^2\mPi_i\mM(Z)^\T] + \EE[\EE_{i \in [n]} [\mM(Z) \mPi_i\mD(Z^{(i)})^2\mPi_i\mM(Z)^\T] | Z])\\
		&= 2(\mDel_{30} + \mDel_{31})
	\end{align*}
	where we define
	\[\mDel_{30} = \EE_{i \in [n]} [\mM(Z) \mPi_i\mD(Z)^2\mPi_i\mM(Z)^\T], \qquad \mDel_{31} = \EE[\EE_{i \in [n]} [\mM(Z) \mPi_i\mD(Z^{(i)})^2\mPi_i\mM(Z)^\T] | Z]\]
	We have
	\begin{align*}
		\mDel_{30} = \EE_{i \in [n]} [\mM(Z) \mPi_i\mD(Z)^2\mPi_i\mM(Z)^\T] &= \EE_{i \in [n]} [\mM(Z) \mD(Z)\mPi_i\mPi_i\mD(Z)\mM(Z)^\T]\\
		&= \mM(Z) \mD(Z)(\frac{1}{n}\sum_{i = 1}^n\mPi_i^2)\mD(Z)\mM(Z)^\T\\
		&\preceq \frac{\dpoly}{n} \mM(Z) \mD(Z)\mD(Z)\mM(Z)^\T\\
		&= \frac{\dpoly}{n} \herm{\mF}_{k, a, b}^2
	\end{align*}
	For the other term, using \cref{fact: commutativity},
	\begin{align*}
		\mDel_{31} &= \EE[\EE_{i \in [n]} [\mM(Z) \mPi_i\mD(Z^{(i)})^2\mPi_i\mM(Z)^\T] | Z]\\
		&= \EE_{i \in [n]} [\mM(Z) \mPi_i\mSig_i\mD(Z)^2\mSig_i\mPi_i\mM(Z)^\T]\\
		&\preceq \EE_{i \in [n]} [\mM(Z) \mPi'_i\mSig_i\mD(Z)^2\mSig_i\mPi'_i\mM(Z)^\T]\\
		&= \EE_{i \in [n]} [\mM(Z) \mSig_i\mPi'_i\mD(Z)^2\mPi'_i\mSig_i\mM(Z)^\T]\\&= \EE_{i \in [n]} [\mD(Z)\herm{\mG}_{k, a, b} \mSig_i\mPi'_i\mD(Z)^2\mPi'_i\mSig_i\herm{\mG}_{k, a, b}\mD(Z)]
	\end{align*}
	Observe that $\herm{\mG}_{k, a, b} \mSig_i = \herm{\mG}_{k, a, b}$ because the entries of $\herm{\mG}$ only depend on $\al$ and not on $\gam$, so permuting the $\gamma$s will not have any effect on the matrix. Therefore,

	\begin{align*}
		\mDel_{31} &\preceq \EE_{i \in [n]} [\mD(Z)\herm{\mG}_{k, a, b}\mPi'_i\mD(Z)^2\mPi'_i\herm{\mG}_{k, a, b}\mD(Z)]\\
		&\preceq \EE_{i \in [n]} [\mD(Z)\herm{\mG}_{k, a, b}\mD(Z)\mPi'_i\mPi'_i\mD(Z)\herm{\mG}_{k, a, b}\mD(Z)]\\
		&= \EE_{i \in [n]} \herm{\mF}_{k, a, b}\mPi'_i\mPi'_i\herm{\mF}_{k, a, b}\\
		&= \frac{1}{n} \sum_{i = 1}^n \herm{\mF}_{k, a, b}\mPi'_i\mPi'_i\herm{\mF}_{k, a, b}\\
		&\preceq \frac{\dpoly}{n}\herm{\mF}_{k, a, b}^2
	\end{align*}
	where we used the fact that $\sum_{i = 1}^n\mPi'_i\mPi'_i \preceq \dpoly \mI$. Putting them together,
	\begin{align*}
		\Esch{\mDel_3^{k, a, b}}{t} \le 2^t(\Esch{\mDel_{30}}{t} + \Esch{\mDel_{31}}{t}) \le 2^t \cdot 2 \frac{\dpoly^t}{n^t}\Esch{\herm{\mF}_{k, a, b}}{2t} \le \frac{(4\dpoly)^t}{n^t}\Esch{\herm{\mF}_{k, a, b}}{2t}
	\end{align*}
\end{proof}

\section{Application: Sparse graph matrices} \label{sec: sparse_graph_matrices}

We now consider sparse graph matrices, i.e., the setting $G \sim \cG_{n, p}$ for $p \le \frac{1}{2}$.
The main difference from dense graph matrices is the contribution of the edge factors.  Na\"ively bounding the contribution of each edge by it's absolute value, as explained in \cref{sec: failure_of_basic}, each edge in the shape contributes a factor of $\sqrt{\frac{1 - p}{p}}$. But in many cases, these bounds are not tight. In fact, they are not tight even in the basic case of the adjacency matrix. In this section, we obtain tighter bounds using our general recursion. As we will see, the improved bound will contain the edge factors only for edges within the vertex separator.

Let $\graphmat{\tau}$ be the graph matrix corresponding to shape $\tau$ where we use $p$-biased Fourier characters $G_{i, j}$. In this section, we obtain bounds on $\Esch{\mM_{\tau} - \EE \mM_{\tau}}{2t}$ and use it to obtain high probability bounds on $\norm{\mM_{\tau}}$.
Since many of the details are similar to \cref{sec: norm_bounds_for_dense_graph_matrices} and the proof of \cref{thm: dense_graph_matrix_norm_bounds}, we will pass lightly over some details. We recommend the reader to read that section first.

The $G_{i, j}$ correspond to the $Z_i$s in \cref{sec: general_recursion} and $\mF$ corresponds to $\mM_{\tau}$. Let $\cI$ denote the set of sub-tuples of $[n]$. Each nonzero entry of $\mM_{\tau}$ is a homogenous polynomial of degree $|E(\tau)|$. If $E(\tau) = \emptyset$, then, $\mM_{\tau} - \EE \mM_{\tau} = 0$ so we can focus on the case when $\tau$ has at least one edge. Moreover, since degree-$0$ vertices in $V(\tau)\setminus U_{\tau} \setminus V_{\tau}$ simply scale the matrix by a factor of at most $n$, we can handle them separately and for our main analysis, we assume there are no such vertices in $\tau$.

We will use \cref{thm: main_general} but the matrices and the statement can be drastically simplified in our application. Instate the notation of \cref{sec: general_recursion}. Since we are dealing with multilinear polynomials, in the definition of $\cK$, we can restrict our attention to $\al \in \{0, 1\}^{\binom{n}{2}}$ because for any other $\al \in \NN^n$, the corresponding row or column of $\mG_{a + b, a, b}$ and hence $\mF_{a + b, a, b}$, will be $0$. So, we can accordingly redefine $\cK$ to only contain these $(\al, \gam)$, hence $\cK \subseteq \{0, 1\}^n \times \{0, 1\}^n$.

Next, the diagonal matrices $\mD_1, \mD_2$ will both be equal to the diagonal matrix $\mD \in \RR[Z]^{\cI \times \cK} \times \RR[Z]^{\cI \times \cK}$ with nonzero entries
\[\mD[(I, \al, \gam), (I, \al, \gam)] = \sqrt{\EE[ \prod_{i, j} G_{ij}^{2\al_{ij}(1 - \gam)_{ij}}]} \prod_{i, j} G_{i}^{\al_{ij}\gam_{ij}} =  \prod_{i, j} G_{i}^{\al_{ij}\gam_{ij}}\]
where we used the fact that for any $i, j$, $\EE[G_{ij}^2] = 1$.

For integers $a, b \ge 0$ such that $a + b = |E(\tau)|$, define the matrix $\mM_{\tau, a, b}$ to be the matrix $\mG_{a + b, a, b}$. We use this notation in order to be streamlined with \cref{sec: norm_bounds_for_dense_graph_matrices}. That is, $\mM_{\tau, a, b}$ has rows and columns indexed by $\cI\times \cK$ such that for all $(I, \al_1, \gam_1), (J, \al_2, \gam_2) \in \cI \times \cK$,
\[\mM_{\tau, a, b}[(I, \al_1, \gam_1), (J, \al_2, \gam_2)] = \begin{dcases}
	\grad_{\al_1 + \al_2} \mM_{\tau}[I, J] & \text{ if $|\al_1|_0 = a, |\al_2|_0 = b, \al_1 \cdot \al_2 = 0$}\\
	0 & \text{o.w.}
\end{dcases}
\]

This is almost identical to the $\mM_{\tau, a, b}$ matrix defined in \cref{sec: norm_bounds_for_dense_graph_matrices}, with the difference being that the row and column indices now have $\gam$ in them. Therefore, for $I, J \in \cI, (\al_1, \gam_1), (\al_2, \gam_2) \in \cK$ such that $|\al_1|_0 = a, |\al_2|_0 = b, \al_1 \cdot \al_2 = 0$, the entry in row $(I, \al_1, \gam_1)$ and column $(J, \al_2, \gam_2)$ is the number of realizations $\phi$ of $\tau$ such that
\begin{itemize}
	\item $U_{\tau}, V_{\tau}$ map to $I, J$ respectively under $\phi$, and
	\item Under $\phi$, the edges of $\tau$ map to the edges in $\al_1$ and $\al_2$ viewed as a set.
\end{itemize}

By \cref{thm: main_general}, for integers $t \ge 1$,
\begin{align*}
	\Esch{\mM_{\tau} - \EE \mM_{\tau}}{2t} &\le \sum_{a, b \ge 0, a + b \ge 1}(Ct^2d\dpoly^4)^{(a + b)t}\Esch{\mF_{a + b, a, b}}{2t}\\
	&= \sum_{a, b \ge 0, a + b  = |E(\tau)|}(Ct^2|E(\tau)|^4)^{t|E(\tau)|}\Esch{\mD\mM_{\tau, a, b}\mD}{2t}
\end{align*}
for an absolute constant $C > 0$.

Now, we would like to analyze $\Esch{\mD\mM_{\tau, a, b}\mD}{2t}$. Just as in the proof of \cref{thm: dense_graph_matrix_norm_bounds}, let $P$ specify which edges of $E(\tau)$ go to $\al_1, \al_2$ respectively and in what order. Moreover, we now store extra information in $P$ that indicates which entries of $\gam_1, \gam_2$ (relative to $\al_1, \al_2$) are set to $1$. Let the set of such information $P$ be denoted $\cP$, then $|\cP| \le (4|E(\tau)|)^{t|E(\tau)|} 2^{|E(\tau)|}$. Thus,
\begin{align*}
	\Esch{\mD\mM_{\tau, a, b}\mD}{2t} \le (8|E(\tau)|)^{t|E(\tau)|} \sum_{P \in \cP} \Esch{\mD\mM_{\tau, a, b, P}\mD}{2t}
\end{align*}
where we define $\mM_{\tau, a, b, P}$ similar to $\mM_{\tau, a, b}$ with the extra condition that $\phi, \al_1, \al_2, \gam_1, \gam_2$ must respect $P$.

At this point, in contrast to the proof of \cref{thm: dense_graph_matrix_norm_bounds}, note that the matrices $\mM_{\tau, a, b, P}$ here have rows and columns indexed by $\cI\times \cK$. We will again define the shape $\tau_P$ that is equal to the nonzero block of the matrix $\mD\mM_{\tau, a, b, P}\mD$, up to renaming of the rows and columns. $V(\tau_P), U_{\tau_P}, V_{\tau_P}$ are defined the same way as in \cref{sec: norm_bounds_for_dense_graph_matrices} but to incorporate the action of $\mD$ on these entries, we simply keep the edges that are active in $\gam_1$ or $\gam_2$, as prescribed by $P$. For an illustration, see \cref{fig: evolution_sparse}.

\begin{figure}[!h]
	\centering
	\includegraphics[trim={2cm 20cm 2cm 2cm}, clip, scale=0.9]{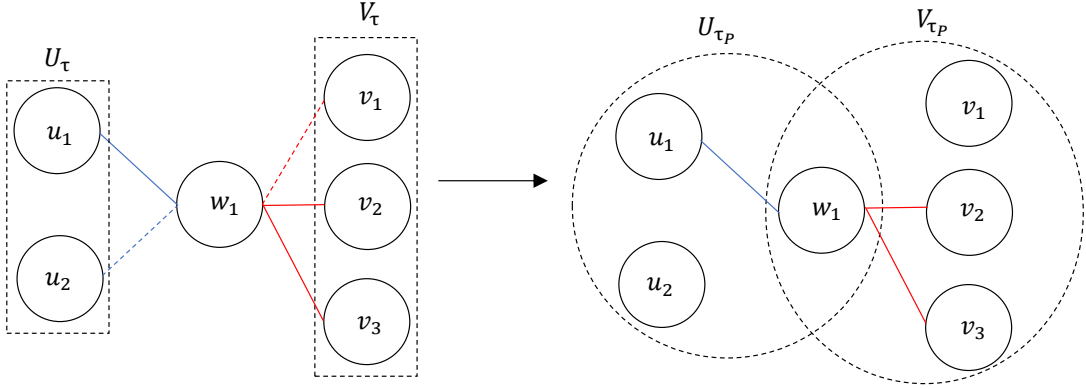}
	\caption{An example illustrating how $\tau_P$ is defined. In this example, $P$ constraints the blue and red edges to go to $\al_1$ and $\al_2$ respectively. Moreover, $P$ indicates that some edges are active in $\gam_1, \gam_2$ (indicated by a solid edge) and some are not active (indicated by a dashed edge) in $\gam_1, \gam_2$. We keep the solid edges in $\tau_P$. $U_{\tau_P}, V_{\tau_P}$ also have an ordering on the vertices (not shown here).}
	\label{fig: evolution_sparse}
\end{figure}

Then, by similar renaming of the rows and columns of $\mD\mM_{\tau, a, b, P}\mD$ and dropping the $\gam$s, we obtain $\mM_{\tau_P}$. We therefore obtain the bound
\begin{align*}
	\Esch{\mD\mM_{\tau, a, b}\mD}{2t} &\le (8|E(\tau)|)^{t|E(\tau)|} \sum_{P \in \cP}\Esch{\mM_{\tau_P}}{2t}
\end{align*}

We would like to analyze norm bounds on the matrices $\mM_{\tau_P}$. Observe that $\tau_P$ are shapes with the properties
\begin{itemize}
	\item there are no vertices in $V(\tau_P) \setminus U_{\tau_P} \setminus V_{\tau_P}$
	\item each edge is either entirely contained in $U_{\tau_P}$ or entirely contained in $V_{\tau_P}$
\end{itemize}
Call such shapes \textit{simple}.

In the following lemma, whose proof is deferred to the next section, we prove norm bounds on simple shapes. Recall that in \cref{lem: empty_shape}, we analyzed the norm bounds of simple shapes with no edges (because in this case, the graph distribution doesn't matter). The analysis for simple shapes is very similar but this time, we use scalar concentration tools to bound the Frobenius norm.

For a set $S$ of vertices, denote by $E(S)$ the set of edges with both endpoints in $S$.

\begin{restatable}{lemma}{simplegraphmatrixnormbounds}\label{lem: graphmatrixnormbound_nomiddlevertices}
	For all even integers $t \ge 2$, if $\tau$ is a simple shape,
	\[\Esch{\graphmat{\tau}}{2t} \le \bigg(n^{|V(\tau)|} (Ct)^{t|E(\tau)|} |V(\tau)|^{t|V(\tau)|}\bigg)\max_{U_{\tau} \cap V_{\tau} \subseteq S \subseteq V(\tau)}\left(\frac{1 - p}{p}\right)^{t|E(S)|}n^{t(|V(\tau)| - |S|)}\]
	for an absolute constant $C > 0$.
\end{restatable}

For simple shapes, the main difference from norm bounds on corresponding dense graph matrices is that each edge within $S$ contributes a factor of $\sqrt{\frac{1 - p}{p}}$. Edge contributions are unavoidable when handling sparse graph matrices, but we have identified that we need not consider all edges in the shape but only a subset of it.

Using this lemma, we can obtain norm bounds on general graph matrices. We recall the definition of a vertex separator.

\vertexseparator*

Let $I_{\tau}$ be the set of isolated vertices (vertices of degree $0$) in $V(\tau) \setminus U_{\tau} \setminus V_{\tau}$, so they essentially scale the matrix by a scalar factor. We now state the main theorem of this section.

\begin{theorem}\label{thm: sparse_graph_matrix_norm_bounds}
	For all even integers $t \ge 2$, for any shape $\tau$,
    {\footnotesize
    \begin{align*}
	&\Esch{\mM_{\tau} - \EE\mM_{\tau}}{2t}\\
    &\le \bigg(n^{|V(\tau)|} |V(\tau)|^{t|V(\tau)|} (Ct^3|E(\tau)|^5)^{t|E(\tau)|}\bigg) \max_{\text{vertex separator }S}\left(\frac{1 - p}{p}\right)^{t|E(S)|}n^{t(|V(\tau)| - |S| + |I_{\tau}|)}
    \end{align*}
}
	where the maximum is over all vertex separators $S$.
\end{theorem}

To interpret this bound, if we assume that there are a constant number of vertices in $\tau$, then by choosing $t \approx \polylog(n)$, we get
\[\norm{\mM_{\tau}} = \widetilde{\bigoh}\bigg(\max_{\text{vertex separator }S}\left(\sqrt{\frac{1 - p}{p}}\right)^{|E(S)|}\sqrt{n}^{|V(\tau) - |S| + |I_{\tau}|}\bigg)\] with high probability, where $\widetilde{\bigoh}$ hides logarithmic factors.  This result follows from \cref{thm: sparse_graph_matrix_norm_bounds} if $\tau$ has at least one edge, but also applies if $\tau$ has no edges, in which case we can directly use the far simpler \cref{lem: empty_shape}. A precise form of the above characterization is given in \cref{cor: sparse_graph_matrix_norm_bounds}.

\cref{thm: sparse_graph_matrix_norm_bounds} gives us the right dependence on $p, n$ for norm bounds in the case of sparse graph matrices. The same bound, up to lower order terms, was also obtained in \cite{jones2022sum} via the trace power method, where they use these bounds to prove semidefinite-programming lower bounds for the maximum independent set problem on sparse graphs.

\begin{proof}[Proof of \cref{thm: sparse_graph_matrix_norm_bounds}]
	If $E(\tau) = \emptyset$, then $\mM_{\tau} = \EE\mM_{\tau}$ and we are done. So, assume $E(\tau) \neq \emptyset$. Since vertices in $I_{\tau}$ only scale the matrix by a factor of at most $n$, we can handle them separately and our bound has the appropriate power of $n$ coming from these. Therefore, we can assume $I_{\tau} = \emptyset$. Continuing our prior discussions, for an absolute constant $C_1 > 0$,
	\begin{align*}
		\Esch{\mM_{\tau} - \EE\mM_{\tau}}{2t} &\le \sum_{a, b \ge 0, a + b  = |E(\tau)|}(C_1t^2|E(\tau)|^4)^{t|E(\tau)|}\Esch{\mD\mM_{\tau, a, b}\mD}{2t}\\
		&\le \sum_{a, b \ge 0, a + b  = |E(\tau)|}(C_1t^2|E(\tau)|^4)^{t|E(\tau)|}(8|E(\tau)|)^{t|E(\tau)|} \sum_{\psi \in \Gam_{a, b}}\Esch{\mM_{\psi}}{2t}
	\end{align*}
	where $\Gam_{a, b}$ are the set of simple shapes we obtain for $\mD\mM_{\tau, a, b} \mD$, as per our discussion above. Using \cref{lem: graphmatrixnormbound_nomiddlevertices}, for an absolute constant $C_2 > 0$, we have
    {\footnotesize
	\begin{align*}
		\Esch{\mM_{\tau} - \EE\mM_{\tau}}{2t} &\le \bigg(n^{|V(\tau)|} |V(\tau)|^{t|V(\tau)|} (C_2t^3|E(\tau)|^5)^{t|E(\tau)|}\bigg)\\
        &\qquad \cdot \sum_{a, b \ge 0, a + b  = |E(\tau)|}\sum_{\psi \in \Gam_{a, b}}\max_{U_{\psi} \cap V_{\psi} \subseteq S \subseteq V(\psi)}\left(\frac{1 - p}{p}\right)^{t|E(S)|}n^{t(|V(\psi)| - |S|)}
	\end{align*}
}
	For any $a, b$, consider any simple shape $\psi \in \Gam_{a, b}$ that can be obtained. As observed in the proof of \cref{thm: dense_graph_matrix_norm_bounds} (see in particular \cref{fig: proof_by_picture}), $U_{\psi} \cap V_{\psi}$ must be a vertex separator of $\tau$. Therefore, any $S \supseteq U_{\psi} \cap V_{\psi}$ must be a vertex separator of $\tau$. It's easy to see that as $S$ ranges over all sets such that $U_{\psi} \cap V_{\psi} \subseteq S \subseteq V(\psi)$, it ranges over all vertex separators of $\tau$.

	Also, the number of different $\psi$ is at most $4^{|E(\tau)|}$ since each edge can go either to $U_{\psi}$ or $V_{\psi}$ and for each such choice, it can either be active in $\gam$ or not. Therefore,
    {\footnotesize
	\begin{align*}
		&\Esch{\mM_{\tau} - \EE\mM_{\tau}}{2t}\\
		&\le \bigg(n^{|V(\tau)|} |V(\tau)|^{t|V(\tau)|} (C_2t^3|E(\tau)|^5)^{t|E(\tau)|}\bigg) 4^{|E(\tau)|}\max_{\text{vertex separator }S}\left(\frac{1 - p}{p}\right)^{t|E(S)|}n^{t(|V(\tau)| - |S|)}\\
		&\le \bigg(n^{|V(\tau)|} |V(\tau)|^{t|V(\tau)|} (Ct^3|E(\tau)|^5)^{t|E(\tau)|}\bigg) \max_{\text{vertex separator }S}\left(\frac{1 - p}{p}\right)^{t|E(S)|}n^{t(|V(\tau)| - |S|)}
	\end{align*}}
	for an absolute constant $C > 0$.
\end{proof}

The following corollary obtains high probability norm bounds for norms of graph matrices via Markov's inequality. We assume the graph has at least one edge, otherwise it is deterministic and its norm bound was already analyzed in \cref{lem: empty_shape}, \cref{cor: dense_graph_matrix_norm_bounds}, where we observe that the distinction between sparse and dense graph matrices does not matter if the random matrix is deterministic.

\begin{corollary}\label{cor: sparse_graph_matrix_norm_bounds}
	For a shape $\tau$ with at least one edge, for any constant $\eps > 0$, with probability $1 - \eps$,
    {\footnotesize
	\[\norm{\mM_{\tau}} \le \bigg(|V(\tau)|^{|V(\tau)|/2} (C|E(\tau)|^5\log^3(n^{|V(\tau)|}/\eps))^{|E(\tau)|/2}\bigg)\cdot\max_{\text{vertex separator }S}\left(\sqrt{\frac{1 - p}{p}}\right)^{|E(S)|}\sqrt{n}^{|V(\tau) - |S| + |I_{\tau}|}\]
}
	for an absolute constant $C > 0$.
\end{corollary}

\begin{proof}
	Since $|E(\tau)| \ge 1$, $\EE\mM_{\tau} = 0$. By an application of Markov's inequality,
    {\footnotesize
	\begin{align*}
		Pr&[\norm{\mM_{\tau}} \ge \theta] \\
        &\le Pr[\sch{\mM_{\tau}}{2t} \ge \theta^{2t}]\\
		&\le \theta^{-2t} \EE\sch{\mM_{\tau}}{2t}\\
		&\le \theta^{-2t}\bigg(n^{|V(\tau)|} |V(\tau)|^{t|V(\tau)|} (C't^3|E(\tau)|^5)^{t|E(\tau)|}\bigg) \max_{\text{vertex separator }S}\left(\frac{1 - p}{p}\right)^{t|E(S)|}n^{t(|V(\tau)| - |S| + |I_{\tau}|)}
	\end{align*}
}
	for an absolute constant $C' > 0$. We now set
	\begin{align*}
		\theta = &\bigg(\eps^{-1/(2t)} (C'')^{|E(\tau)|} n^{|V(\tau)|/(2t)}|V(\tau)|^{|V(\tau)|/2} t^{3|E(\tau)|/2}|E(\tau)|^{5|E(\tau)|/2}\bigg)\\
		&\qquad\cdot \max_{\text{vertex separator }S}\left(\sqrt{\frac{1 - p}{p}}\right)^{|E(S)|}\sqrt{n}^{|V(\tau) - |S| + |I_{\tau}|}
	\end{align*}
	for an absolute constant $C'' > 0$, to make this expression at most $\eps$. Set $t = \frac{1}{2}\log(n^{|V(\tau)|}/\eps)$ to complete the proof.
\end{proof}

\subsection{Norm bounds on simple graph matrices}

In this section, we will prove \cref{lem: graphmatrixnormbound_nomiddlevertices}. First, we recall the following scalar concentration result from \cite{schudy2011bernstein}.

\subsubsection{Schudy-Sviridenko moment bound}

The definitions and main bound in this section are from \cite{schudy2011bernstein}.

\begin{definition}
	A random variable $Z$ is central moment bounded with real parameter $L > 0$ if for any integer $i \ge 1$,
	\[\EE[|Z - \EE[Z]|^i] \le i\cdot L\cdot \EE[|Z - \EE[Z]|^{i - 1}]\]
\end{definition}

\begin{propn}
	The $p$-biased Bernoulli random variable $Z$ is central moment bounded with real parameter $L =\sqrt{\frac{1 - p}{p}}$.
\end{propn}

\begin{proof}
	We have $\EE[Z] = 0$ and for $p \le \frac{1}{2}$, $|Z| \le \sqrt{\frac{1 - p}{p}}$, therefore,
	\begin{align*}
		\EE[|Z - \EE[Z]|^i] &= p\sqrt{\frac{p}{1 - p}}^i + (1 - p)\sqrt{\frac{1 - p}{p}}^i\\
		&\le \sqrt{\frac{1 - p}{p}}\bigg(p\sqrt{\frac{p}{1 - p}}^{i - 1} + (1 - p)\sqrt{\frac{1 - p}{p}}^{i- 1}\bigg)\\
		&= \sqrt{\frac{1 - p}{p}}\EE[|Z - \EE[Z]|^{i - 1}]
	\end{align*}
	therefore, we can take $L = \sqrt{\frac{1 - p}{p}}$.
\end{proof}

For a given multilinear polynomial $f(x)$ on variables $x_1, \ldots, x_n$, we can naturally associate with it a hypergraph $H$ on vertices $[n]$ and weighted hyperedges $E(H)$ where each $h \in E(H)$ corresponds to a distinct term of $f(x)$. Each hyperedge $h$ is a subset $V(h)$ of vertices and has a real valued weight $w_h$ which is the coefficient of that monomial in $f$. Therefore,
\[f(x) = \sum_{h \in E(H)} w_h \prod_{v \in V(h)} x_v\]

Assume $f$ has degree $\dpoly$, then each hyperedge of $H$ has at most $\dpoly$ vertices.

Now, for a given collection of independent random variables $Y_1, \ldots, Y_n$, a multilinear poynomial $f$ with associated hypergraph $H$ and weights $w$, and an integer $r \ge 0$, define
\[\mu_r(f, Y) = \max_{S \subseteq [n], |S| = r} \bigg(\sum_{h \in E(H), S \subseteq V(h)} |w_h|\prod_{v \in V(h) \setminus S} \EE[|Y_v|]\bigg)\]

\begin{lemma}[\cite{schudy2011bernstein}, Lemma 5.1]\label{lem: schudy_sviridenko}
	Given $n$ independent central moment bounded random variables $Y_1, \ldots, Y_n$ with the same parameter $L > 0$ and a degree $\dpoly$ multilinear polynomial $f(x)$. Let $t \ge 2$ be an even integer, then
	\[\EE[|f(Y) - \EE [f(Y)]|^t] \le \max\bigg\{\bigg(\sqrt{tR_4^{\dpoly}\var{f(Y)}}\bigg)^t, \max_{r \in [\dpoly]}(t^rR_4^{\dpoly}L^r\mu_r(f, Y))^t\bigg\}\]
	where $R_4 \ge 1$ is some absolute constant.
\end{lemma}

In our setting, we can also bound the variance in terms of the $\mu_r$ as was shown in \cite{schudy2011bernstein}, which will simplify our calculations.

\begin{lemma}[\cite{schudy2011bernstein}, Lemma 1.5]\label{lem: var_bound}
	For the same setting as in \cref{lem: schudy_sviridenko},
	\[\var{f(Y)} \le 2\dpoly 4^{\dpoly}\max_{r \in [\dpoly]} (\mu_0(f, Y) \mu_r(f, Y)4^rL^r)\]
\end{lemma}

\subsubsection{Proof of \cref{lem: graphmatrixnormbound_nomiddlevertices}}

We are ready to prove \cref{lem: graphmatrixnormbound_nomiddlevertices} which we restate for convenience.

\simplegraphmatrixnormbounds*

We will prove it the same way as \cref{lem: empty_shape}, by bounding the schatten norm of each diagonal block by an appropriate power of its Frobenius norm. In this case, to bound the expected power of the Frobenius norm, we use the scalar concentration inequality from the previous section.

\begin{proof}[Proof of \cref{lem: graphmatrixnormbound_nomiddlevertices}]
	First, we note that $\mM_{\tau}$ has a block diagonal structure indexed by the realizations of the set of common vertices $S_0 = U_{\tau_P} \cap V_{\tau_P}$. For $T \in [n]^{S_0}$, let $\mM_{\tau, T}$ be the block of $\mM_{\tau}$ with $\phi(S_0) = T$. Then, $\mM_{\tau, T}\mM_{\tau, T'}^\T = \mM_{\tau, T}^\T\mM_{\tau, T'} = 0$ for $T \neq T'$ and so,
	\begin{align*}
		\Esch{\mM_{\tau}}{2t} = \sum_{T \in [n]^{S_0}} \Esch{\mM_{\tau, T}}{2t} \le \sum_{T \in [n]^{S_0}}\EE(\sch{\mM_{\tau, T}}{2})^t
	\end{align*}
	where we bounded the Schatten norm by a power of the Frobenius norm.

	Fix $T \in [n]^{S_0}$ and consider $\Esch{\mM_{\tau, T}}{2}$. Let $\cR$ be the set of realizations $\phi$ of $\tau$ such that $\phi(S_0) = T$. Then, for $\phi \in \cR$ and $e \in E(S_0)$, the value of $\phi(e)$ is fixed. Using this,
	\begin{align*}
		\norm{\mM_{\tau, T}}_2^2 &= \sum_{\phi \in \cR} \prod_{e \in E(\tau)} G_{\phi(e)}^2\\
		&= \prod_{e \in E(S_0)} G_{\phi(e)}^2 \sum_{\phi \in \cR} \prod_{e \in E(\tau) \setminus E(S_0)} G_{\phi(e)}^2\\
		&\le L^{|E(S_0)|} \sum_{\phi \in \cR} \prod_{e \in E(\tau) \setminus E(S_0)} G_{\phi(e)}^2
	\end{align*}
	where $L = \frac{1 - p}{p}$ is an upper bound on $G_{ij}^2$ for $p \le \frac{1}{2}$. Define the quantity
	\[A = \max_{S_0 \subseteq S \subseteq V(\tau)}L^{|E(S)|}n^{|V(\tau)| - |S|}\]

	\begin{claim}\label{claim: lil_claim}
		$\EE(\norm{\mM_{\tau, T}}_2)^t \le (Ct)^{t|E(\tau)|} |V(\tau)|^{t|V(\tau)|}A^t$ for an absolute constant $C > 0$.
	\end{claim}

	Using this claim, we have
	\begin{align*}
		\Esch{\mM_{\tau}}{2t} &\le \sum_{T \in [n]^{S_0}} \EE(\norm{\mM_{\tau, T}}_2)^t\\
		&\le n^{|S_0|}(Ct)^{t|E(\tau)|} |V(\tau)|^{t|V(\tau)|}A^t\\
		&= n^{|V(\tau)|} (Ct)^{t|E(\tau)|} |V(\tau)|^{t|V(\tau)|}\max_{U_{\tau} \cap V_{\tau} \subseteq S \subseteq V(\tau)}\left(\frac{1 - p}{p}\right)^{t|E(S)|}n^{t(|V(\tau)| - |S|)}
	\end{align*}
	as required.
\end{proof}

It remains to prove the claim.

\begin{proof}[Proof of \cref{claim: lil_claim}]
	For $1\le i, j \le n$, define the variables $Y_{ij} = G_{ij}^2$ with $\EE[|Y_{ij}|] = 1$. Let $f(Y)$ be the polynomial $L^{|E(S_0)|} \sum_{\phi \in \cR} \prod_{e \in E(\tau) \setminus E(S_0)} Y_{\phi(e)}$. It suffices to prove that $\EE[f(Y)^t] \le (Ct)^{t|E_1|}A^t$.

	We will first prove that $\EE[(f(Y) - \EE[f(Y)])^t] \le (C't)^{t|E(\tau)|}|V(\tau)|^{t|V(\tau)|}A^t$ for a sufficiently large constant $C' > 0$.

	$f$ is a homogeneous multilinear polynomial of degree $|E(\tau) \setminus E(S_0)|$. If we had $E(\tau) \setminus E(S_0) = \emptyset$, then $f$ is a constant and so, the inequality is obvious because $f(Y) = \EE[f(Y)]$. Now, assume $E(\tau) \setminus E(S_0) \neq \emptyset$. We invoke \cref{lem: schudy_sviridenko}. Let $f$ have associated hypergraph $H$ and weights $w$. Then,
    {\footnotesize
	\[\EE[|f(Y) - \EE [f(Y)]|^t] \le \max\bigg\{\bigg(\sqrt{tR_4^{|E(\tau) \setminus E(S_0)|}\var{f(Y)}}\bigg)^t, \max_{r \in [|E(\tau) \setminus E(S_0)|]}(t^rR_4^{|E(\tau) \setminus E(S_0)|}L^r\mu_r(f, Y))^t\bigg\}\]
}
	For all $r \ge 0$, we will prove that $L^r\mu_r(f, Y) \le |V(\tau)|^{|V(\tau)|}A$. By definition,
	\begin{align*}
		\mu_r(f, Y) &= \max_{F \subseteq \binom{[n]}{2}, |F| = r} \sum_{h \in E(H), F \subseteq V(h)} |w_h|
	\end{align*}
	Consider any set of edge labels $F \subseteq \binom{[n]}{2}, |F| = r$.
	Then, $\sum_{h \in E(H), F \subseteq V(h)} |w_h|$ is at most $L^{|E(S_0)|}c$ where $c$ is the number of realizations $\phi \in \cR$ such that $\phi(E(\tau))$ contains $F$.
	Suppose $F$ contains $v$ new labels apart from $\phi(S_0) = T$.
	Then $c \le |V(\tau)|^v n^{|V(\tau)| - |S_0| - v}$ because we can first choose and label the set of vertices that get these $v$ labels and then label the remaining vertices freely, each of which has at most $n$ choices.

	Observe that $L^{|E(S_0)|} L^r n^{|V(\tau)| - |S_0| - v} \le A$ because in the definition of $S$, we can set $S$ to be the union of $S$ and any valid choice of these $v$ vertices. Putting this together, we get
	\begin{align*}
		L^r\mu_r(f, Y) &\le L^r\max_{F \subseteq \binom{[n]}{2}, |F| = r} \sum_{h \in E(H), F \subseteq V(h)} |w_h|\\
		&\le |V(\tau)|^{|V(\tau)|} A
	\end{align*}
	which implies
	\[\max_{r \in [|E(\tau) \setminus E(S_0)|]}(t^rR_4^{|E(\tau) \setminus E(S_0)|}L^r\mu_r(f, Y))^t \le |V(\tau)|^{t|V(\tau)|}(R_4t)^{t|E(\tau)|}A^t\]
	and using \cref{lem: var_bound},
	\begin{align*}
		\var{f(Y)} &\le 2|E(\tau)|4^{|E(\tau)|}\max_{r \in [|E(\tau) \setminus E(S_0)|]} (\mu_0(f, Y) \mu_r(f, Y)4^rL^r)\\
		&\le 2|E(\tau)|16^{|E(\tau)|} |V(\tau)|^{2|V(\tau)|}A^2
	\end{align*}
	Putting them together, we get
    {\footnotesize
	\begin{align*}
		\EE[(f(Y) - \EE[f(Y)])^t] &\le \max\bigg\{\bigg(\sqrt{2tR_4^{|E(\tau)|}|E(\tau)|16^{|E(\tau)|} |V(\tau)|^{2|V(\tau)|}A^2}\bigg)^t, |V(\tau)|^{t|V(\tau)|}(R_4t)^{t|E(\tau)|}A^t\bigg\}\\
		&\le (C't)^{t|E(\tau)|}|V(\tau)|^{t|V(\tau)|}A^t
	\end{align*}
}
	for an absolute constant $C' > 0$.

	Finally, $\EE[f(Y)] \le L^{|E(S_0)|} |\cR| \le L^{|E(S_0)|}n^{|V(\tau) \setminus S_0|} \le A$ which gives
	\begin{align*}
		\EE[f(Y)^t] &\le 2^t(\EE[(f(Y) - \EE[f(Y)])^t] + \EE[f(Y)]^t)\\
		&\le 2^t((C't)^{t|E(\tau)|}|V(\tau)|^{t|V(\tau)|}A^t + A^t)\\
		&\le (Ct)^{t|E(\tau)|}|V(\tau)|^{t|V(\tau)|}A^t
	\end{align*}
	for an absolute constant $C > 0$.
\end{proof}

\chapter{The Sum of Squares Hierarchy}\label{chap: sos}

In this chapter, we formally introduce the Sum of Squares (SoS) hierarchy.
Then, we take a minor detour and define low-degree distinguishers and related concepts for hypothesis testing, which will set the stage for us to discuss SoS lower bounds.
We then go back to SoS and discuss the heuristic known as pseudo-calibration, that will be a basic ingredient we use in our SoS lower bounds.
We finally show a formal connection between pseudo-calibration and low-degree distinguishers and conclude with a note on recent successes of SoS.

\section{The Sum of Squares hierarchy}

We start by defining convex relaxations for polynomial optimization problems. The SoS hierarchy will then be a special family of convex relaxations. For a more detailed treatment, see e.g. \cite{sos_course, BS14:ICM, FKP19}.

\subsection{Polynomial optimization and convex relaxations}

In polynomial optimization, we are given multivariate polynomials $p, g_1, \ldots, g_m$ on $n$ variables $x_1, \ldots, x_n$ taking real values, denoted collectively by $x$, and the task is to:
\[\text{maximize } p(x)\text{ such that }g_1(x) = 0, \ldots, g_m(x) = 0\]

In general, we could also allow inequality constraints, e.g., $g_i(x) \ge 0$. For technical convenience in our setup, we work only with equality constraints but much of the theory generalizes, with some modifications, when we have inequality constraints instead. An alternate approach is to replace each inequality $g_i(x) \ge 0$ by $g_i(x) = y^2$ where $y$ is a new variable that we can introduce.

In this formulation, many optimization problems can be formulated as polynomial optimization problems.

\begin{example}[Maximum Cut]
    Given a graph $G = (V, E)$, we would like to partition the set of vertices into two subsets such that the number of edges with endpoints in different subsets is maximized. To formulate this as a polynomial optimization problem, let the graph have $n$ vertices and let $x_1, \ldots, x_n$ be variables, one for each vertex. We wish to enforce $x_i \in \{-1, 1\}$ where all vertices $i$ with $x_i = -1$ form one subset and the rest form the other subset. We can enforce this set containment constraint via the polynomial constraint $x_i^2 = 1$. For any edge $(i, j) \in E$, it is cut if and only if $x_ix_j = -1$. Therefore, the total number of edges cut is $\sum_{(i, j) \in E} \frac{1}{2}(1 - x_ix_j)$. The polynomial formulation therefore becomes
    \begin{align*}
        \max_{x \in \RR^n} \sum_{(i, j) \in E} \frac{1}{2}&(1 - x_ix_j) \text{ such that }\\
        x_i^2 &= 1 \text{ for all }i \le n
    \end{align*}
\end{example}

\begin{example}[Maximum Clique]\label{ex: max_clique}
    Given a graph $G = (V, E)$, we would like to find the maximize size subset of vertices that form a clique. Again, let $x_1, \ldots, x_n$ be variables, one for each vertex. This time, we wish to enforce $x_i \in \{0, 1\}$, which we can easily do so using the polynomial constraint $x_i^2 = x_i$, with the intent being that all vertices $i$ with $x_i = 1$ form a clique. To enforce this clique constraint, we can add the polynomial constraint $x_ix_j = 0$ for all non-edges $(i, j) \not\in E$. Finally, to maximize the size of the subset, we simply maximize $\sum_{i \le n} x_i$. Therefore, the polynomial optimization is
    \begin{align*}
        \max_{x \in \RR^n} \sum_{i \le n} &x_i \text{ such that}\\
        x_ix_j &= 0 \text{ for all }(i, j) \not\in E\\
        x_i^2 &= x_i \text{ for all } i \le n
    \end{align*}
\end{example}

There can be other equivalent formulations for these problems. In general, many optimization problems can be stated in this manner, therefore generic polynomial optimization contains a large class of fundamental problems that appear in computer science.

Since exactly solving maximum cut or maximum clique is NP-hard \cite{karp1972reducibility}, exactly solving these polynomial optimization problems is also NP-hard. Therefore, we turn to convex relaxations.

A convex relaxation of a polynomial optimization problem widens the search space of solution vectors $x$ into a larger space that one can efficiently optimize over. We will describe one way to do this. We identify a convex space $\calC$ that contains the space $\calS= \{g_1(x) = 0, \ldots, g_m(x) = 0\}$ up to a map, that is, for each $x \in \calS$, there exists a corresponding $y \in \calC$ such that $y$ is a representative of $x$. We also identify a convex function $\tilde{p}(y)$ such that if $y$ is a representative of $x$, then $\tilde{p}(y) = p(x)$. Then, we simply optimize $\tilde{p}(y)$ over $\calC$. There has been significant work on efficiently optimizing a convex function over a convex body, which is possible under reasonable assumptions (see e.g. \cite{PS82}). It's clear that from the above properties, the solution we get is at least as large as the optimal solution (in the case of maximization), but it comes with the advantage that it is efficiently computable. It is desirable to design convex relaxations for problems that yield good approximations. The SoS hierarchy is a family of such convex relaxations.

\subsection{Sum of Squares relaxations}

The SoS hierarchy, sometimes referred to as the Lasserre hierarchy, was first independently studied by \cite{parrilo2000structured, lasserre2001global, shor1987approach} and has been studied in other contexts by \cite{nesterov2000squared, grigoriev2001complexity, Grigoriev01}.
It is a family of convex relaxations for polynomial optimization, parameterized by an integer known as it's degree. As we increase the degree, we get progressively tighter relaxations, but requiring longer times to optimize over.

We now formally describe the Sum of Squares hierarchy, via the so-called pseudoexpectation operator view.

\begin{definition}[Pseudo-expectation values]\label{def: pseudoexpectation}
    Given multivariate polynomial constraints $g_1 = 0$,\ldots,$g_m = 0$ on $n$ variables $x_1, \ldots, x_n$, degree $d$ pseudo-expectation values are a linear map $\pE$ from polynomials of $x_1, \ldots, x_n$ of degree at most $d$ to $\mathbb{R}$ satisfying the following conditions:
    \begin{enumerate}
        \item $\pE[1] = 1$, \label{pe:normalized}
        \item $\pE[f \cdot g_i] = 0$ for every $i \in [m]$ and polynomial $f$ such that $\deg(f \cdot g_i) \leq d$. \label{pe:feasible}
        \item $\pE[f^2] \geq 0$ for every polynomial $f$ such that $\deg(f^2) \le d$. \label{pe:psdness}
    \end{enumerate}
\end{definition}

Any linear map $\pE$ satisfying the above properties is known as a degree $d$ pseudoexpectation operator satisfying the constraints $g_1 = 0, \ldots, g_m = 0$.

\begin{definition}[Degree $d$ SoS]
    The degree $d$ SoS relaxation for the polynomial optimization problem
    \[\text{maximize } p(x)\text{ such that }g_1(x) = 0, \ldots, g_m(x) = 0\]
    is the program that maximizes $\pE[p(x)]$ over all degree $d$ pseudoexpectation operators $\pE$ satisfying the constraints $g_1 = 0, \ldots, g_m = 0$.
\end{definition}

The intuition behind pseudo-expectation values is that the conditions on the pseudo-expectation values are conditions that would be satisfied by any actual expected values over a distribution of solutions, so optimizing over pseudo-expectation values gives a relaxation of the problem.

The main observation is that the SoS relaxation can be efficiently solved! This is because the conditions on pseudo-expectation values can be captured by a semidefinite program. In particular, \cref{pe:psdness} in \cref{def: pseudoexpectation} can be reexpressed in terms of a matrix called the moment matrix.

\begin{definition}[Moment Matrix of $\pE$]
    Given a degree $d$ pseudo-expectation operator $\pE$, define the associated moment matrix $\Lda$ to be a matrix with rows and columns indexed by monomials $p$ and $q$ such that the entry corresponding to row $p$ and column $q$ is
    \[
    \Lda[p, q] \defeq \pE\left[pq\right].
    \]
\end{definition}

It is easy to verify that \cref{pe:psdness} in~\cref{def: pseudoexpectation} equivalent to $\Lda \succeq 0$. Therefore, solving the degree $d$ SoS relaxation can be done via semidefinite programming, see for e.g. \cite{vandenberghe1996semidefinite}.
In general, for degree-$d$ SoS, we can solve it in $n^{O(d)}$ time\footnote{This is not completely accurate due to issues of bit complexity \cite{o2017sos} but this doesn't occur for most problems of interest \cite{RW17:sos}}. Therefore, constant degree SoS can be solved in polynomial time.

\subsubsection{Analyzing degree $2$ SoS for maximum clique}

To illustrate the use of this technique, let's analyze the degree $2$ SoS relaxation for the maximum clique problem on \Erdos-\Renyi random graphs $G_{n, 1/2}$. We use the program from \cref{ex: max_clique}.

Let $A$ be the adjacency matrix of a graph $G$ sampled from $G_{n, 1/2}$ and let $J$ be the matrix with all $1$s. Then, with high probability over the choice of $G$, from random matrix theory, we have $\lda_{max}(A- J/2) = O(\sqrt{n})$ where $\lda_{max}(.)$ denotes the maximum singular value. Now, suppose a set $S$ of vertices form a clique and let $\one_S$ denote the indicator vector of the set $S$, then
\begin{align*}
    \frac{k(k - 1)}{2} &= \ip{\one_S}{(A - J/2)\one_S}\\
    &\le \norm{\one_S}^2\cdot \lda_{max}(A - J / 2)\\
    &\le k \cdot O(\sqrt{n})
\end{align*}
which shows $k \le O(\sqrt{n})$.

The crux of this simple argument is that this is a \textit{low-degree proof}, more specifically degree $2$ proof, that SoS can capture. That is, if we solve the degree $2$ SoS relaxation, we will be able to show that $\pE[\sum x_i] = O(\sqrt{n})$ whp.

To see this formally, we start with the following inequality: $O(\sqrt{n})I - (A - J / 2) \succeq 0$ whp. This implies
\[x^\T(O(\sqrt{n})I - (A - J / 2))x = \sum p_i(x)^2\]
is a sum of squares of polynomials of degree at most $1$.
A simple computation yields
\[x^\T(A - J/2)x = \frac{1}{2}(\sum_{i = 1}^n x_i)^2 - \sum_{i, j} x_ix_j \one_{(i, j) \not\in E(G)}\]

For our program variables $x$, we have $x_i^2 = x_i$ and $x_ix_j \one_{(i, j) \not\in E(G)} = 0$. Therefore,
\[\sum p_i(x)^2 = O(\sqrt{n}) (\sum_{i = 1}^n x_i) - \frac{1}{2}(\sum_{i = 1}^n x_i)^2\]
Apply $\pE$ both sides. We finally use the fact that for a polynomial $p(x)$, we have $\pE[p(x)^2] \ge \pE[p(x)]^2$, which is true because this rearranges to $\pE[(p(x)- \pE[p(x)])^2] \ge 0$, which is true because the left hand side is the the pseudo-expectation of a square polynomial, which is nonnegative by definition. This simple fact is essentially saying that the pseudo-variance is nonnegative. Using the linearity of $\pE$, we finally get

\begin{align*}
    O(\sqrt{n}) \pE[\sum_{i = 1}^n x_i] - \frac{1}{2}(\pE[\sum_{i = 1}^n x_i])^2 &\ge O(\sqrt{n}) \pE[\sum_{i = 1}^n x_i] - \frac{1}{2}\pE[(\sum_{i = 1}^n x_i)^2]\\
    &= \pE[O(\sqrt{n}) (\sum_{i = 1}^n x_i) - \frac{1}{2}(\sum_{i = 1}^n x_i)^2]\\
    &= \pE[\sum p_i(x)^2]\\
    &= \sum \pE[p_i(x)^2]\\
    & \ge 0
\end{align*}

Therefore, $\pE[\sum_{i = 1}^n x_i] = O(\sqrt{n})$ like we wanted to show.

This shows that the degree $2$ SoS relaxation certifies an upper bound of $O(\sqrt{n})$ whp on the size of the maximum clique of an \Erdos-\Renyi random graph. In contrast, the size of the true maximum clique is $(2 + o(1)) \log n$ \cite{matula1976largest}. Despite intense effort, polynomial time algorithms can only detect a planted $k$-clique when $k = \Omega(\sqrt{n})$. Therefore, SoS already achieves the best known guarantees for this problem up to constant factors. It was shown in \cite{BHKKMP16} that higher degree SoS (up to degree $O(\log n)$) doesn't necessarily do much better, which is a SoS lower bound of the type we will study in this work.

\subsubsection{Alternate viewpoints of SoS}

In the polynomial optimization problem of maximizing $p(x)$ subject to the constraints $g_1(x) = 0, \ldots, g_m(x) = 0$, if there does not exist any degree $d$ pseudo-expectation operator $\pE$ satisfying $g_1 = 0, \ldots, g_m = 0$ such that $\pE[p] > c$, then we say that degree $d$ SoS certifies that $\pE[p(x)] \le c$.

A degree $d$ SoS proof that $p(x) \le c$ given $g_1(x) = 0, \ldots, g_m(x) = 0$ is an expression of the form
\[-1 = \sum_{i \le m} g_i(x) q_i(x) + \sum_{i \le a} s_i(x)^2 + (p(x) - c) \sum_{i \le b} t_i(x)^2\]
where $q_1, \ldots, q_m, s_1, \ldots, s_a, t_1, \ldots, t_b$ are polynomials in $x$ such that each term on the right hand side of the above expression has degree at most $d$. Indeed, the existence of such an expression automatically implies that $p(x) \le c$ whenever $g_1(x) = 0, \ldots, g_m(x) = 0$.

When degree $d$ SoS certifies that $\pE[p(x)] \le c$, by duality, this will imply that there exists a degree $d$ SoS proof that $p(x) \leq c$ given $g_1(x), \ldots, g_m(x) = 0$. The Positivstellensatz of Krivine and Stengle \cite{krivine1964anneaux, stengle1974nullstellensatz} says that for any $c$, either there exists $x$ such that $p(x) > c, g_1(x) = 0, \ldots, g_m(x) = 0$, or there is an SoS proof that $p(x) \le c$ given $g_1(x) = 0, \ldots, g_m(x) = 0$.

For a fixed $d$, degree $d$ SoS can indeed be construed as finding the best $c$ so that there is a degree $d$ SoS proof of $\pE[p(x)] \le c$. This also intuitively explains why higher degree SoS gives tighter relaxations. For most programs stemming from combinatorial optimization problems, degree $n$ SoS usually finds the optimal bound, where $n$ is the number of variables. So, for instance, degree $n$ SoS exactly outputs the size of the maximum clique of a graph. For efficient algorithms, we usually want constant degree SoS.
Therefore, for sum of squares lower bounds, the higher the degree, the stronger the lower bound.
In this work, all our lower bounds are for degree $n^{\eps}$ SoS, which corresponds to subexponential time!

The viewpoint we have studied here is the dual view aka the search for simple proofs, which will suit our purposes. There is also the primal viewpoint where SoS can be viewed directly as a semi-definite programming relaxation of the program. This is sometimes useful for algorithm design.

Similar to the maximum clique application shown above, the SoS hierarchy has been shown formally to obtain the state-of-the art approximation guarantees for many fundamental problems both in the worst case and the average case setting. This includes constraint satisfaction problems \cite{Raghavendra08}, maximum cut \cite{GW94}, sparsest cut \cite{AroraRV04}, tensor PCA \cite{hopkins2015tensor}, etc.
Therefore, it's natural to study the limits of SoS by studying SoS lower bounds.

Before we discuss SoS lower bounds, we introduce the framework of hypothesis testing problem in more detail, suited to our purposes.

\section{Hypothesis testing}

Let $\Omega$ be a sample space. Let $\nu, \mu$ be probability distributions on $\Omega^n$. The hypothesis testing problem is the problem of distinguishing $\nu, \mu$ given access to a sample. Formally, input $x \sim \Omega^n$ is sampled from either
\begin{itemize}
    \item $H_0$: $x \sim \mu$
    \item $H_1$: $x \sim \nu$.
\end{itemize}
Our objective is to determine which distribution it came from, with high probability. This is the hypothesis testing problem in general, where traditionally, $H_0$ is known as the null hypothesis and $H_1$ the alternate hypothesis. We abuse notation and use $H_0, H_1$ to also denote the probability distributions $\mu, \nu$ respectively as well.

For example, $H_0$ could be the distribution of \Erdos-\Renyi random graphs and $H_1$ could be the distribution of \Erdos-\Renyi random graphs with a large planted clique. Given the graph, we would like to determine which of the two distributions it came from, or in other words, whether it contains a large clique.

A hypothesis test $f$ is a function $f: \Omega^n \to \{0, 1\}$. Given the input $x$, if $f(x) = 0$, then we report that $x$ came from the null distribution $H_0$ otherwise we report that $x$ came from the alternate distribution $H_1$.

A successful hypothesis test is a test $f$ such that when $b$ is chosen uniformly at random from $\{0, 1\}$ and $x$ is sampled from $H_b$, we have $\EE_b \mathrm{Pr}_{x \sim H_b}[f(x) \neq b] \le o(1)$. That is, test $f$ has success probability $1 - o(1)$. Here, for simplicity, we don't distinguish type $1$ and type $2$ errors.

Indeed, for a test to be useful, it should be computable efficiently. When computational efficiency is disregarded, the famous Neyman-Pearson lemma  precisely characterizes the best hypothesis test. To define this test, we need the following standard definition.

\begin{definition}[Likelihood ratio]
    For a given hypothesis testing problem, define the likelihood ratio of an input $x$ to be $LR(x) = \frac{\mathrm{Pr}_{H_1}(x)}{\mathrm{Pr}_{H_0}(x)}$.
\end{definition}

\begin{lemma}[Neyman-Pearson Lemma]
    For a given hypothesis testing problem, the test $f$ that minimizes $\EE_b\mathrm{Pr}_{x \sim H_b} [f(x) \neq b]$ is the likelihood ratio test
    \[f(x) = \begin{dcases}
        1 & \text{ if $LR(x) > 1$}\\
        0 & \text{ o.w.}
    \end{dcases}\]
\end{lemma}

In this work, our focus will be on efficiently computable tests $f$.

\subsection{Low degree likelihood ratio}\label{subsec: ldlr}

Consider a given hypothesis testing problem. We focus on a special class of efficiently computable hypothesis tests involving low degree multivariate polynomials. These are termed low-degree distinguishers. We give a brief treatment in this section and refer the readers to \cite{hop18, kunisky2021spectral} for a more detailed treatment.

In this section, for polynomials to be well-defined, assume $\Omega \subseteq \RR$. Moreover, assume $H_0$ has finite moments. We will consider distinguishers that arise from multivariate polynomials $f: \RR^n \to \RR$. We say that the distinguisher has degree $D$ if the degree of $f$ is at most $D$. Since the output of a polynomial need not be boolean, we need an alternate definition of the success of this distinguisher. We use the following definition from \cite{hop18}.

\begin{definition}[Degree $D$ distinguisher]
    For a hypothesis testing problem, the multivariate polynomial $f$ is a successful degree $D$ distinguisher if
    \begin{itemize}
        \item (Low degree) $f$ is a multivariate polynomial of degree at most $D$.
        \item (Normalization) $\EE_{x \sim H_0}[f(x)] = 0, \EE_{x \sim H_0} [f(x)^2] = 1$
        \item (Distinguishability) $\lim_{n \to \infty} \EE_{x \sim H_1} [f(x)] \to \infty$.
    \end{itemize}
\end{definition}

The normalization ensures appropriate scaling for the polynomial. Note that the normalization is over the null distribution. Informally, normalized $f$ is a successful distinguisher if it attains unbounded values on the alternate distribution in the limit. Indeed, in applications, a hypothesis test may be obtained by appropriately thresholding on the value of the polynomial.

The limit on the degree imposes the kind of computational restrictions we wish to impose on our distinguishing algorithm.
Trying to understand the power of such low-degree distinguishers for hypothesis testing problems is an active area of research.
For instance, we could ask: If degree $O(\log n)$ distinguishers fail for a hypothesis testing problem with input size $n^{O(1)}$, is the problem hard for all polynomial time algorithms?

The first natural question is to ask what's the best degree $D$ distinguisher for a given hypothesis testing problem. This has been answered in prior works and is simply the projection of the likelihood ratio $LR(x) = \frac{\mathrm{Pr}_{H_1}(x)}{\mathrm{Pr}_{H_0}(x)}$ to degree $D$ polynomials.

To make this precise, for $f, g: \RR^n \to \RR$, define the inner product $\ip{f}{g} = \EE_{x \sim H_0} f(x)g(x)$. Then, we can canonically define the projection $f^{\le D}$ of a function $f$ to degree $D$ polynomials via this inner product. Take an orthonormal basis $\chi_0 = 1, \chi_1, \ldots, \chi_t$ of multivariate polynomials of degree at most $D$ where $\chi_0 = 1$ is the constant function. Then, $f^{\le D}(x) = \sum_{i \le t} \ip{f}{\chi_t}\chi_t(x)$.

The following lemma is implicit in prior works (e.g. \cite{hop17efficient, hopkins2018integrality}). We include a proof for completeness.

\begin{lemma}\label{lem: best_distinguisher}
    For a hypothesis testing problem, the optimal degree $D$ test $f$ that maximizes $\EE_{x \sim H_1} f(x)$ is the normalized low-degree likelihood ratio $\frac{LR^{\le D} - 1}{\norm{LR^{\le D} - 1}}$. Moreover, its value is $\EE_{x \sim H_1}[f(x)] = \norm{LR^{\le D} - 1}$.
\end{lemma}

\begin{proof}
    Let $f$ be a normalized degree $D$ polynomial with $f = \sum_{i = 0}^t c_t\chi_t$. Then, $c_0 = \EE[f] = 0$ and $\sum c_i^2 = \EE[f^2] = 1$. Then, \[\EE_{x \sim H_1} f(x) = \sum_{1 \le i \le t} c_i \EE_{x \sim H_1} \chi_i \le \sqrt{(\sum_{1 \le i \le t} c_i^2)(\sum_{1 \le i \le t} (\EE_{x \sim H_1}\chi_i)^2)} = \sqrt{\sum_{1 \le i \le t} (\EE_{x \sim H_1}\chi_i)^2}\]
    On the other hand, equality is attained by the polynomial $g = \frac{LR^{\le D} - 1}{\norm{LR^{\le D} - 1}}$. Indeed, we have $\EE_{x \sim H_0}[g] = 0$ because $\EE_{x \sim H_0}[LR^{\le D}(x)] = \EE_{x \sim H_0}[LR(x)] = 1$ and trivially, we have $\EE_{x \sim H_0}[g(x)^2] = 1$ since we scaled by the norm. Finally,
    \[\EE_{x \sim H_1} g(x) = \frac{1}{\norm{LR^{\le D} - 1}} \sum_{1 \le i \le t} \ip{LR(x)}{\chi_i}^2\]
    We complete the proof by observing that $\ip{LR(x)}{\chi_i} = \EE_{x \sim H_0} [LR(x)\chi_i(x)] = \EE_{x \sim H_1}[\chi_i(x)]$. Computing the value is straightforward.
\end{proof}

The low-degree likelihood ratio hypothesis \cite{hop17, hop18, kunisky19notes} hypothesizes that if $H_0, H_1$ are \textit{sufficiently nice} distributions, then there is a successful hypothesis test with running itme $n^{O(D)}$ if and only if there exists a successful degree $D$ distinguisher. In particular, based on the above discussion, if $\norm{LR^{\le D} - 1} = O(1)$, then we expect that there is no $n^{O(D)}$ time successful hypothesis test.

A main contribution of this work is to provide strong evidence that this conjecture is true for many fundamental problems, by exhibiting strong SoS lower bounds. To see this connection a bit more formally, we will introduce pseudo-calibration and connect it with low-degree distinguishers.

\section{Pseudo-calibration}\label{subsec: pseudocalibration}

Consider an optimization problem we are trying to show SoS lower bounds for.
To obtain SoS integrality gaps on random instances, we need to construct valid pseudo-expectation values for a random input instance of the problem. Naturally, these pseudo-expectation values will depend on the input.

Psuedo-calibration is a heuristic introduced by \cite{BHKKMP16} to construct such candidate pseudo-expectation values almost mechanically by considering a planted distribution supported on instances of the problem with large objective value and using this planted distribution as a guide to construct the pseudo-expectation values. Note here that, for historic reasons, we use the term random distribution instead of null distribution and the term planted distribution instead of alternative distribution.

Unfortunately, psuedo-calibration doesn't guarantee feasibility of these candidate pseudo-expectation values and the corresponding moment matrix and this has to be verified separately for different problems. This verification of feasibility is relatively easy except for the PSDness condition. This is where the main contribution of this work lies, where we analyze the behavior of the constructed random moment matrix.

Indeed for our applications, psuedocalibration is used to obtain a candidate pseudoexpectation operator $\pE$ and a corresponding moment matrix $\Lda$
from the random vs planted problem. This will be the starting point for all our applications. Pseudo-calibration gives lower bounds for many problems, such as the ones considered in the works \cite{Grigoriev01, Schoenebeck08, KothariMOW17, chlamtavc2018sherali, mohanty2020lifting}, making it an intriguing but poorly understood technique.

Here, we do not attempt to motivate and describe pseudo-calibration in great detail. Instead, we will briefly describe the heuristic, the intuition behind it and show an example of how to use it. A detailed treatment can be found in \cite{BHKKMP16}.

Let $\nu$ denote the random distribution and $\mu$ denote the planted distribution. Let $v$ denote the input and $x$ denote the variables for our SoS relaxation. The main idea is that, for an input $v$ sampled from $\nu$ and any polynomial $f(x)$ of degree at most the SoS degree, pseudo-calibration proposes that for any low-degree test $g(v)$, the correlation of $\pE[f]$ should match in the planted and random distributions. That is,
\[\EE_{v \sim \nu}[\pE[f(x)]g(v)] = \EE_{(x, v) \sim \mu}[f(x)g(v)]\]

Here, the notation $(x, v) \sim \mu$ means that in the planted distribution $\mu$, the input is $v$ and $x$ denotes the planted structure in that instance. For example, in planted clique, $x$ would be the indicator vector of the clique. If there are multiple, pick an arbitrary one.

Let $\calF$ denote the Fourier basis of polynomials for the input $v$. By choosing different basis functions from $\calF$ as choices for $g$ such that the degree is at most some truncation parameter $D$, we get all lower order Fourier coefficients for $\pE[f(x)]$ when considered as a function of $v$. Furthermore, the higher order coefficients are set to be $0$ so that the candidate pseudoexpectation operator can be written as
\[\pE f(x) = \sum_{\substack{g \in \calF\\deg(g) \le n^{\eps}}} \EE_{v \sim \nu}[\pE[f(x)]g(v)] g(v) = \sum_{\substack{g \in \calF\\deg(g) \le n^{\eps}}} \EE_{(x, v) \sim \mu}[[f(x)]g(v)] g(v)\]

The coefficients $\EE_{(x, v) \sim \mu}[[f(x)]g(v)]$ can be explicitly computed in many settings, which therefore gives an explicit pseudoexpectation operator $\pE$.

One intuition for pseudo-calibration is as follows. The planted distribution is usually chosen to be a maximum entropy distribution which still has the planted structure. This conforms to the philosophy that random instances are hard for SoS, such as the uniform Bernoulli distribution for planted clique or the Gaussian distribution for Tensor PCA. By conditioning on the lower order moments matching such a planted distribution, pseudo-calibration can be interpreted as sort of interpolating between the random and planted distributions by only looking at lower order Fourier characters. This intuition has proven to be successful, since pseudo-calibration been successfully exploited to construct SoS lower bounds for a wide variety of dense as well as sparse problems.

An advantage of pseudo-calibration is that this construction automatically satisfies some nice properties that the pseudoexpectation $\pE$ should satisfy. It's linear in $v$ by construction. For all polynomial equalities of the form $f(x) = 0$ that is satisfied in the planted distribution, it's true that $\pE[f(x)] = 0$. For other polynomial equalities of the form $f(x, v) = 0$ that are satisfied in the planted distribution, the equality $\pE[f(x, v)] = 0$ is approximately satisfied. In most cases, $\pE$ can be mildly adjusted to satisfy these exactly.

The condition $\pE[1] = 1$ is not automatically satisfied but in most applications, we usually require that $\pE[1] = 1 \pm \littleoh(1)$. Indeed, this has been the case for all known successful applications of pseudo-calibration. Once we have this, we simply set our final pseudoexpectation operator to be $\pE'$ defined as $\pE'[f(x)] = \pE[f(x)] / \pE[1]$.

We remark that the condition $\pE[1] = 1 \pm \littleoh(1)$ has been quite successful in predicting the right thresholds between approximability and inapproximability\cite{hop17, hop18, kunisky19notes}. This will be crucial when we connect pseudo-calibration to low degree distinguishers.

\paragraph{Example: Planted Clique}
As a warmup, we review the pseudo-calibration calculation for planted clique. Here, the random distribution $\nu$ is $G(n, \frac{1}{2})$.

The planted distribution $\mu$ is as follows. For a given integer $k$, first sample $G'$ from $G(n, \frac{1}{2})$, then choose a random subset $S$ of the vertices where each vertex is picked independently with probability $\frac{k}{n}$. For all pairs $i, j$ of distinct vertices in $S$, add the edge $(i, j)$ to the graph if not already present. Set $G$ to be the resulting graph.

The input is given by $G \in \{-1, 1\}^{\binom{[n]}{2}}$ where $G_{i, j}$ is $1$ if the edge $(i, j)$ is present and $-1$ otherwise. Let $x_1, \ldots, x_n$ be the boolean variables for our SoS program such that $x_i$ indicates if $i$ is in the clique.

Given a set of vertices $V \subseteq [n]$, define $x_V = \prod_{v \in V}{x_v}$. Given a set of possible edges $E \subseteq \binom{[n]}{2}$, define $\chi_E = (-1)^{|E \setminus E(G)|} = \prod_{(i, j) \in E}G_{i, j}$.

Pseudo-calibration says that for all small $V$ and $E$,
\[
\EE_{G \sim \nu}\left[\tilde{E}[x_V]\chi_E\right] = \EE_{\mu}\left[x_V{\chi_E}\right]
\]
Using standard Fourier analysis, this implies that if we take
\[
c_E = \EE_{\mu}\left[x_V{\chi_E}\right] = \left(\frac{k}{n}\right)^{|V \cup V(E)|}
\]
where $V(E)$ is the set of the endpoints of the edges in $E$, then for all small $V$,
\[
\pE[x_V] = \sum_{E:E \text{ is small}}{{c_E}\chi_E} = \sum_{E:E \text{ is small}}{\left(\frac{k}{n}\right)^{|V \cup V(E)|}\chi_E}
\]

Since the values of $\pE[x_V]$ are known, by multi-linearity, this can be naturally extended to obtain values $\pE[f(x)]$ for any polynomial $f$ of degree at most the SoS degree.

Here, we only set the Fourier coefficients for small $E$ and set the other larger Fourier coefficients to $0$. Usually, the choice of the truncation parameter is problem specific but there are some basic requirements \cite{hop17}. We now outline our general strategy to show SoS lower bounds. We employ this in all our results.

\subsection{Strategy to show SoS lower bounds}\label{sec: strategy_for_sos_lower_bounds}

In this work, the general strategy to show SoS lower bounds can be summarized as follows.

\begin{itemize}
	\item Given a random distribution, identify a suitable planted distribution
	\item Pseudocalibrate with respect the two distributions and obtain a candidate pseudoexpectation operator
    \item Show that the moment matrix satisfies the constraints
\end{itemize}

The most technically challenging part of this approach usually is to show that the moment matrix is positive semidefinite. Much of our contributions lies in this step, where we analyze the behavior of the random moment matrix thus obtained. Now, we connect pseudo-calibration to low-degree distinguishers.

\subsection{Connection to Low-degree distinguishers}

We are ready to connect psuedo-calibration to low-degree tests. Recall that in pseudo-calibration, we set the higher order Fourier coefficients to $0$. This is known as truncation. In particular, we truncate so that the resulting pseudoexpectation  has degree at most $D$ in the input. By construction, $\EE[\pE[1]] = 1$ and we would like to understand how much $\pE[1]$ deviates from $1$. The following lemma says that the variance of $\pE[1]$ behaves like the squared value of the optimal degree-$D$ distinguisher.

\begin{lemma}\label{lem: pcal_to_ldlr}
    The pseudo-calibrated pseudo-expectation $\pE$, truncated to degree $D$, satisfies
    \[\mathrm{var}(\pE[1]) = \norm{LR^{\le D} - 1}^2\]
\end{lemma}

\begin{proof}
    Pseudocalibration sets $\EE_{x \sim H_0}[\pE[1] \chi_i] = \EE_{x \sim H_1} [\chi_i]$ for all $i \le t$. Therefore, $\pE[1] = 1 + \sum_{1 \le i \le t} \EE_{x \sim H_1} [\chi_i] \chi_i$ giving $\mathrm{var}(\pE[1]) = \sum_{1 \le i \le t} (\EE_{x \sim H_1}\chi_i)^2 = \norm{LR^{\le D} - 1}^2$.
\end{proof}

One of the essential steps in our SoS lower bound proofs is to verify, after pseudo-calibration, that $\pE[1]$ is well-behaved. In particular, for strong SoS lower bounds, we expect $\pE[1] = 1 + o(1)$. Although this is not formally necessary, it has often been the case in our applications and we expect it to be necessary for obtaining strong SoS lower bounds via this approach.

But when this is indeed the case and we exhibit SoS lower bounds, note that this is already strong evidence towards the low-degree likelihood ratio hypothesis. In more detail, because of \cref{lem: best_distinguisher} and \cref{lem: pcal_to_ldlr}, the best degree $D$ distinguisher does not distinguish the two distributions $\mu, \nu$. Our lower bounds affirm that the powerful SoS hierarchy cannot distinguish the two distributions as well, which is an important step towards the general hypothesis.

It's an important open problem in this field to prove that for sufficiently nice distributions $\mu, \nu$, after pseudo-calibrating, $\pE[1] = 1 + o(1)$ implies the existence of strong SoS lower bounds.

\section{Why Sum of Squares?}

We briefly remark on the successes of SoS in the last decade, especially in robust machine learning, a branch of machine learning where the underlying dataset is noisy, with the noise being either random or adversarial.
Robust machine learning has gotten a lot of attention in recent years because of its wide variety of use cases in machine learning and other downstream applications, including safety-critical ones like autonomous driving. For example, there has been a high volume of practical works in computer vision \cite{szegedy2013intriguing, goodfellow2014explaining, xie2019feature, hendrycks2021natural, sebe2013robust, xie2020adversarial, fischer2017adversarial, kurakin2016adversarial} and speech recognition \cite{hsu2021robust, wang2022wav2vec, rajendran2022analyzing, ravanelli2020multi, li2015robust, alzantot2018did, neekhara2019universal, olivier2022recent}.
In this important field, SoS has recently lead to breakthrough algorithms for
long-standing open problems \cite{bakshi2020robustly, liu2021settling, hopkins2020mean, klivans2018efficient, FKP19, kothari2017outlier, bakshi2020outlier, bakshi2020list, schramm2017fast}. Highlights include
\begin{itemize}
    \item Robustly learning mixtures of high dimensional Gaussians. This is an extremely important problem that has been subjected to intense scrutiny, with a long line of work culminating in \cite{bakshi2020robustly, liu2021settling}.
    \item Efficient algorithms for the fundamental problems of regression \cite{klivans2018efficient}, moment estimation \cite{kothari2017outlier}, clustering \cite{bakshi2020outlier} and subspace recovery \cite{bakshi2020list} in the presence of outliers.
\end{itemize}

Moreover, SoS algorithms are believed to be the optimal robust algorithm for many statistical problems. In a different direction, SoS algorithms have led to the design of fast algorithms for problems such as tensor decomposition \cite{hopkins2016fast, schramm2017fast}.

Broadly speaking, due to its ability to capture a wide variety of algorithmic techniques, SoS has become a fundamental tool in algorithms and optimization. It was and still remains an extremely versatile tool for combinatorial optimization \cite{GW94, AroraRV04, GuruswamiS11, raghavendra2017strongly}) but as we saw above, it is also being extensively used in Statistics and Machine Learning (apart from the references above, see also \cite{BarakBHKSZ12, bks15, HopSS15, pot17}). This sets the stage for the rest of this work where we analyze it for various problems of interest stemming from statistics and statistical physics.

\chapter{Our main results on Sum of Squares lower bounds}\label{chap: main_results}

In this chapter, we state formally the main Sum of Squares lower bounds that we prove in this thesis and put them in the context of prior works. The material in this chapter is adapted from \cite{sklowerbounds, potechin2020machinery}, where the results originally appeared. However, this chapter differs from those works in that we highlight recent progress on these works, mention recently surfaced connections to other problems, and moreover, we present the proof techniques in succession which helps pedagogically since the core principles of the proofs are not entirely dissimilar.

\section{The Sherrington-Kirkpatrick Hamiltonian}

We first define the Gaussian Orthogonal Ensemble, $\GOE(n)$, a random matrix model for $n \times n$ matrices.
\begin{definition}
	The Gaussian Orthogonal Ensemble, denoted $\GOE(n)$, is the distribution of $\frac{1}{\sqrt{2}}(A + A^\T)$
	where $A$ is a random $n\times n$ matrix with i.i.d. standard Gaussian entries.
\end{definition}

Equivalently, we could define $\GOE(n)$ to be a probability distribution over symmetric matrices $W$ such that $W_{ii} \sim \GN(0, 2)$ for $i \le n$ and for $i\neq j$, $W_{ij} = W_{ji} \sim \GN(0, 1)$ independently.

We consider the main optimization task
\begin{equation}\label{eq:general_opt}
	\OPT(W) \defeq \max_{x \in \{\pm 1\}^n} x^\T W x,
\end{equation}
where $W$ is a random symmetric matrix in $\RR^{n\times n}$. This is an important task that arises in computer science and statistical physics.

In computer science, a natural choice of $W$ is to take it to be the Laplacian of a graph~\cite[Section 4]{HooryLW06}. Then, the problem is equivalent to the
Maximum Cut problem, a well-known NP-hard problem in the worst
case~\cite{K72}. The equivalence is immediate by observing that
$x \in \{\pm 1\}^n$ can be thought of as encoding a bipartition of $[n] = \{1,
2, \ldots, n\}$.

In particular, an interesting special case is when we consider sparse random graphs, sampled either from the \Erdos-\Renyi graphs $G(n, \frac{d}{n})$ with average degree $d$ or a uniformly chosen $d$-regular graph, where $d \ge 3$ is a fixed integer.
In this case, it is known that the true size of the maximum cut is asymptotically $n(\frac{d}{4} + f(d) \sqrt{d})$.
Moreover, it was shown in \cite{dembo2017extremal} (originally conjectured in \cite{zdeborova2010conjecture}) that $\lim_{d \to \infty} f(d) = \frac{1}{2}P^* \approx 0.382$, where
\[P^* := \frac{1}{2}\lim_{n\to\infty}\E_{W \sim \GOE(n)}[\frac{1}{n^{3/2}}\OPT(W)] \approx 0.7632\]
is referred to as the Parisi constant. This already strongly motivates the problem of studying \cref{eq:general_opt} when $W \sim \GOE(n)$. Interestingly, this problem is motivated for another fantastic reason.

In statistical physics, when $W \sim \GOE(n)$, our objective, up to scaling, is the Hamiltonian of the famous Sherrington-Kirkpatrick model. Here, $x$ can be thought of as encoding
spin values in a spin-glass model.
$-W_{i,j}$ models the interaction between spin $x_i$ and $x_j$
(with $-W_{i,j} \ge 0$ being ferromagnetic and $-W_{i,j} < 0$ being
anti-ferromagnetic). Then, the optimal value corresponds to the
minimum-energy, or ground state of the system, up to sign.
The works \cite{P79, parisi1980sequence, crisanti2002analysis} predicted, using non-rigorous means, that $P^* \approx 0.7632$. This was eventually formalized
in the works \cite{Tal06, Panchenko2014, guerra2003broken}.

In this work, we will focus on this average case optimization problem when $W \sim \GOE(n)$. The first natural question is whether there exists
a polynomial-time algorithm that
given $W \sim \GOE(n)$ computes an $x$ achieving close to $\OPT(W)$.
In a recent breakthrough work, Montanari~\cite{Montanari19} showed that, for any $\eps > 0$, there exists a polynomial time algorithm that outputs $x$ given $W$ such that with high probability it achieves a value of $(2P^* - \eps)n^{3/2}$ (assuming a widely believed conjecture).

Now we move onto certification: Is there an
efficient algorithm to certify an upper bound on $\OPT(W)$ for any
input $W$?

A simple algorithm will be the spectral algorithm where we just output the largest eigenvalue of $W$, up to scaling, for an upper bound.
Note that $\GOE(n)$ is a particular kind
of Wigner matrix ensemble, thereby satisfying the semicircle law, which
in this case establishes that the largest eigenvalue of $W$ is
$(2+\littleoh_n(1)) \cdot \sqrt{n}$ with probability
$1-\littleoh_n(1)$. Thus, a trivial spectral bound establishes
$\OPT(W) \le (2+\littleoh_n(1)) \cdot n^{3/2}$ with probability
$1-\littleoh_n(1)$.

Now, we can ask if it's possible to beat this spectral algorithm for certification. In
particular, we can ask how well SoS does as a certification
algorithm. The natural upper bound of $(2+\littleoh_n(1)) \cdot
n^{3/2}$ obtained via the spectral norm of $W$ is also the value of
the degree-$2$ SoS relaxation~\cite{MS16}. Two independent recent
works of Mohanty--Raghavendra--Xu~\cite{mohanty2020lifting} and
Kunisky--Bandeira~\cite{KuniskyBandeira19} show that degree-4 SoS does
not perform much better, and a heuristic argument from~\cite{bkw19} suggests that even degree-$(n/\log n)$ SoS cannot certify anything stronger than the trivial spectral bound. Thus we ask,

\begin{center}
	\emph{Can higher-degree SoS certify better upper bounds for the Sherrington--Kirkpatrick problem, \\
		hopefully closer to the true bound $2 \cdot P^* \cdot n^{3/2}$?}
\end{center}

In this work, we answer the question above negatively by showing that even at degree as large as
$n^\delta$, SoS cannot improve upon the basic spectral
algorithm.
\begin{restatable}{theorem}{SKbounds}\label{theo:sk-bounds}
	There exists a constant $\delta > 0$ such that, w.h.p. for $W \sim \GOE(n)$, there is a degree-$n^\delta$ SoS solution
	for the Sherrington--Kirkpatrick problem with value at least $(2-\littleoh_n(1)) \cdot n^{3/2}$.
\end{restatable}

An independent and concurrent work by Kunisky~\cite{kunisky2020} also showed a special case of the above theorem for degree-$6$ SoS, using different techniques.

We will present the proof of this theorem in \cref{chap: sk}.
The above theorem and it's proof originally appeared in \cite{sklowerbounds}, from which the material here is adapted from.
We now present the high level ideas behind the proof of this theorem.


\subsection{Our approach}
In order to prove~\cref{theo:sk-bounds}, we first introduce a new
average-case problem we call Planted Affine Planes (PAP) for which we directly prove a SoS lower bound. We then use
the PAP lower bound to prove a lower bound on the
Sherrington--Kirkpatrick problem. The PAP problem can be informally
described as follows (see~\cref{def:prob:pap} for the formal definition).
\begin{definition}[Informal statement of PAP]
	Given $m$ random vectors $d_1,\ldots,d_m$ in $\mathbb{R}^n$, can we
	prove that there is no vector $v \in \RR^n$ such that for all
	$u \in [m]$, $\langle v, d_u\rangle^2 = 1$? In other words, can we
	prove that $m$ random vectors are not all contained in two parallel
	hyperplanes at equal distance from the origin?
\end{definition}
This problem, when we restrict $v$ to a Boolean vector in $\set{\pm \frac{1}{\sqrt{n}}}^n$,
can be encoded as the feasibility of the polynomial system
\begin{align*}
	\exists v \in \R^n~\text{s.t.} \qquad & \forall i \in [n], \; v_i^2 = \frac{1}{n},\\
	& \forall u \in [m], \; \ip{v}{d_u}^2 = 1.
\end{align*}
Hence it is a ripe candidate for SoS. However, we show that SoS fails
to refute a random instance with high probability over the input. The Boolean restriction on $v$ actually
makes the lower bound result stronger since SoS cannot refute even a
smaller subset of vectors in $\RR^n$. In this work, we will consider two
different random distributions, namely when $d_1, \ldots, d_m$ are
independent samples from the multivariate normal distribution and when
they are independent samples from the uniform distribution on the
boolean hypercube.
\begin{theorem}\label{theo:sos-bounds}
	For both the Gaussian and Boolean settings, there exists a constant $c > 0$ such that for all $\eps > 0$ and $\delta \le c\eps$, for $m \leq n^{3/2 - \eps}$, w.h.p. there is a feasible degree-$n^\delta$ SoS solution for Planted Affine Planes.
\end{theorem}

It turns out that the Planted Affine Plane problem introduced above is
closely related to the following ``Boolean vector in a random
subspace'' problem, which we call the Planted Boolean Vector problem,
introduced by~\cite{mohanty2020lifting} in the context of studying the performance
of SoS on computing the Sherrington--Kirkpatrick Hamiltonian.

The Planted Boolean Vector problem is to certify that a random subspace of $\R^n$ is far from containing a boolean vector.
Specifically, we want to certify an upper bound for
\[
\OPT(V) \defeq  \frac{1}{n}\max_{b \in \{\pm 1\}^n} b^\T \Pi_V b,
\]
where $V$ is a uniformly random $p$-dimensional subspace\footnote{$V$
	can be specified by a basis, which consists of $p$ i.i.d.  samples
	from $\calN(0, I)$.} of $\RR^n$, and $\Pi_V$ is the projector onto
$V$. In brief, the relationship to the Planted Affine Plane problem is
that the PAP vector $v$ represents the coefficients on a linear
combination for the vector $b$ in the span of a basis of
$V$.

An argument of~\cite{mohanty2020lifting} shows that, when $p \ll n$, w.h.p.,
$\OPT(V) \approx \frac{2}{\pi}$, whereas they also show that
w.h.p. assuming $p \geq n^{0.99}$, there is a degree-4 SoS solution
with value $1-\littleoh_n(1)$. They ask whether or not there is a
polynomial time algorithm that can certify a tighter bound; we rule
out SoS-based algorithms for a larger regime both in terms of SoS
degree and the dimension $p$ of the random subspace.

\begin{restatable}{theorem}{booleanSubspace}\label{theo:boolean-subspace}
	There exists a constant $c > 0$ such that, for all $\eps > 0$ and $\delta \le c\eps$, for $p \geq n^{2/3 + \eps}$, w.h.p. over $V$ there is a
	degree-$n^\delta$ SoS solution for Planted Boolean Vector of value $1$.
\end{restatable}

The bulk of our technical contribution lies
in the SoS lower bound for the Planted Affine Planes
problem,~\cref{theo:sos-bounds}. We then show that Planted Affine
Planes in the Gaussian setting is equivalent to the Planted Boolean
Vector problem. The reduction from Sherrington-Kirkpatrick to the
Planted Boolean Vector problem is due to
Mohanty--Raghavendra--Xu~\cite{mohanty2020lifting}.

As a starting point to the PAP lower bound, we employ pseudocalibration to produce a good
candidate SoS solution $\pE$. The operator $\pE$ unfortunately does
not exactly satisfy the PAP constraints ``$\ip{v}{d_u}^2 = 1$'', it
only satisfies them up to a tiny error. In the original work, we use an interesting and
rather generic approach to round $\pE$ to a nearby pseudoexpectation
operator $\pE'$ which does exactly satisfy the constraints, We have omitted this in this thesis for the sake of brevity, but it can be found in the original work \cite{sklowerbounds}.

For degree $D$, the candidate SoS solution can be viewed as a
(pseudo) moment matrix $\calM$ with rows and columns indexed by
subsets $I,J\subset [n]$ with size bounded by $D/2$ and with entries
\[\calM[I,J] \defeq \pE[v^{I} v^{J}].\]
The matrix $\calM$ is a random function of the inputs $d_1, \dots, d_m$, and the most challenging part of the
analysis consists of showing that $\calM$ is positive
semi-definite (PSD) with high probability.

Similarly to~\cite{BHKKMP16}, we decompose
$\calM$ as a linear combination of graph matrices, i.e., $\calM = \sum_{\alpha} \lambda_{\alpha} \cdot M_{\alpha}$, where $M_{\alpha}$
is the graph matrix associated with shape $\alpha$. In brief, each
graph matrix aggregates all terms with shape $\alpha$  in the Fourier expansions of the entries of $\calM$ -- the shape $\alpha$ is informally a graph with labeled edges
with size bounded by $\poly(D)$. A graph
matrix decomposition of $\calM$ is particularly handy in the PSD
analysis since the operator norm of individual graph matrices $M_{\alpha}$ is (with high probability)
determined by simple combinatorial properties of the graph
$\alpha$. One technical difference from~\cite{BHKKMP16} is that our
graph matrices have two types of vertices $\square{}$ and $\circle{}$; these graph matrices fall into the general framework developed by Ahn et al. in~\cite{ahn2016graph}.

To show that the matrix $\calM$ is PSD, we need to study the graph matrices that appear with nonzero coefficients in the decomposition. The matrix $\calM$ can be split into blocks and each diagonal block contains in the decomposition a (scaled) identity matrix. From the graph matrix perspective, this means that certain ``trivial'' shapes appear in the decomposition, with appropriate coefficients. If we could bound the norms of all other graph matrices that appear against these trivial shapes and show that, together, they have negligible norm compared to the sum of these scaled identity blocks, then we would be in good shape.

Unfortunately, this approach will not work. The kernel of the matrix $\calM$ is nontrivial, as a consequence of satisfying the PAP constraints ``$\ip{v}{d_u}^2 = 1$", and hence there is no hope of showing that the contribution of all nontrivial shapes in the decomposition of $\calM$ has small norm. Indeed, certain shapes $\alpha$ appearing in the
decomposition of $\calM$ are such that $\norm{\lambda_{\alpha} \cdot M_{\alpha}}$ is large. As it turns out, all such shapes have a simple graphical substructure, and so we call these shapes \textit{spiders}.

To get around the null space issue, we restrict ourselves to $\nullspace(\calM)^\perp$, which is the complement of the nullspace of $\calM$.
We show that the substructure present in a spider implies that the spider is close to the zero matrix in $\nullspace(\calM)^\perp$. Because of this, we can almost freely
add and subtract $M_\alpha$ for spiders $\alpha$ while preserving the action of $\calM$ on $\nullspace(\calM)^\perp$. Our strategy is to ``kill'' the spiders
by subtracting off $\lambda_\alpha \cdot M_\alpha$ for each spider $\alpha$. However, because $M_{\alpha}$ is only approximately in $\nullspace(\calM)^\perp$, this
strategy could potentially introduce new graph matrix terms, and in particular it could introduce new spiders. To handle this,
we recursively kill them while carefully analyzing how the coefficients of all the graph matrices change. After all spiders
are killed, the resulting moment matrix becomes
$$
\sum_{0 \le k \le D/2} \frac{1}{n^{k}} \cdot I_k + \sum_{\gamma \colon \textup{non-spiders}} \lambda_{\gamma}' \cdot M_{\gamma},
$$
for some new coefficients $\lambda_{\gamma}'$. Here, $I_k$ is the
matrix which has an identity in the $k$th block and the remaining
entries $0$. Using a novel charging argument, we finally show that the
latter term is negligible compared to the former term, thus
establishing $\calM \succeq 0$.

\subsection{Related work}
Degree-$4$ SoS lower bounds on the
Sherrington-Kirkpatrick Hamiltonian problem were proved independently
by Mohanty--Raghavendra--Xu~\cite{mohanty2020lifting} and
Kunisky--Bandeira~\cite{KuniskyBandeira19}. The concurrent and independent work by Kunisky~\cite{kunisky2020} obtained degree $6$ SoS lower bounds. In this work, we prove an
improved degree-$n^{\delta}$ SoS lower bound for some constant $\delta
> 0$.  Our result is obtained by reducing the Sherrington-Kirkpatrick
problem to the ``Boolean Vector in a Random Subspace'' problem which
is equivalent to our new Planted Affine Planes problem on the normal
distribution. The reduction from Sherrington-Kirkpatrick problem to
the ``Boolean Vector in a Random Subspace'' is due to
Mohanty--Raghavendra--Xu~\cite{mohanty2020lifting}. The results of
Mohanty--Raghavendra--Xu~\cite{mohanty2020lifting} and
Kunisky--Bandeira~\cite{KuniskyBandeira19} build on a degree-$2$ SoS
lower bounds of Montanari and Sen~\cite{MS16}.


Degree-$4$ SoS lower bounds on the ``Boolean Vector in a Random
Subspace'' problem for $p~\ge~n^{0.99}$ were proved by
Mohanty--Raghavendra--Xu in~\cite{mohanty2020lifting} where this problem was
introduced. We improve the dependence on $p$ to $p \ge n^{2/3
	+ \epsilon}$ for any $\epsilon > 0$ and obtain a stronger
degree-$n^{c\eps}$ SoS lower bound for some absolute constant $c > 0$.

Interestingly, the recent work \cite{zadik2021latticebased} exhibited a polynomial-time algorithm for the search variant of Planted Affine Planes for $m \ge n + 1$, achieving statistical optimality. In particular, they beat prior known polynomial time algorithms, including SoS based ones, all of which required $m \gg n^2$ \cite{mao2021optimal}. This new algorithm is a lattice-based method that uses the specific algebraic structure present in this problem. Because of this, their algorithm is not robust to small perturbations, that is, they require the points to lie exactly on the planes. On the other hand, the spectral algorithms such as the work of \cite{mao2021optimal} are robust to noise. Because of this necessity of lack of noise, the lattice based algorithm is of a similar flavor to how Gaussian elimination can beat SoS lower bounds in the absense of noise. Specifically, this means that this lattice based algorithm does not refute our certification lower bound, or the low degree likelihood ratio hypothesis described in \cref{subsec: ldlr}.


\section{Sparse PCA}

Principal components analysis (PCA) \cite{joliffe1992principal} is a popular data processing and dimension reduction routine that is widely used. It has numerous applications in Machine Learning, Statistics, Engineering, Biology, etc. Given a dataset, PCA projects the data to a lower dimensional space spanned by the principal components. The intuition is that PCA sheds lower order information such as noise but importantly preserves much of the intrinsic information present in the data that are needed for downstream tasks.

However, despite great optimality properties, PCA has its drawbacks. Firstly, because the principal components are linear combinations of all the original variables, it's notoriously hard to interpret them \cite{mahoney2009cur}. Secondly, it's well known that PCA does not yield good estimators in high dimensional settings \cite{baik2005phase, paul2007asymptotics, johnstone_lu2009}.

To address these issues, a variant of PCA known as Sparse PCA is often used. Sparse PCA searches for principal components of the data with the added constraint of sparsity.
Concretely, consider given data $v_1, v_2, \ldots, v_m \in \RR^d$. In Sparse PCA, we want to find the top principal component of the data under the extra constraint that it has sparsity at most $k$. That is, we want to find a vector $v \in \RR^d$ that maximizes $\sum_{i = 1}^m \ip{v}{v_i}^2$ such that $\norm{v}_0 \le k$.

Sparse PCA has enjoyed applications in a diverse range of fields ranging from medicine, computational biology, economics, image and signal processing, finance and of course, machine learning and statistics (e.g. \cite{wang2012online, naikal2011informative, majumdar2009image, tan2014classification, chun2009expression, allen2011sparse}).
Moreover, sparse PCA comes with the important benefit that the components are easier to interpret. A notable example of this is to recover topics from documents \cite{d2004direct, papailiopoulos2013sparse}. Moreover, interpretability has important benefits for algorithmic fairness in machine learning.

A large volume of research has been devoted to study Sparse PCA and its variants.
Algorithms have been proposed and studied by several works, e.g. \cite{amini_wainwright2008, ma2013sparse, krauthgamer2015, deshpande2016, wang2016statistical, berthet2013complexity, ma_wigderson_15, diakonikolas2017statistical, hop17, brennan2019optimal, ding2019subexponential,  chowdhury2020approximation, d2020sparse}.
For example, simple variants of PCA such as thresholding on top of standard PCA \cite{johnstone_lu2009, chowdhury2020approximation} work well in certain parameter settings. This leads to the natural question whether more sophisticated algorithms can do better either for these settings or other parameter settings.

On the other hand, there have been works from the inapproximability perspective as well (e.g. \cite{berthet2013complexity, hop17, brennan2019optimal, krauthgamer2015, ding2019subexponential, wang2016statistical}, we will give a more detailed overview after stating our main result).
In particular, a lot of these inapproximability results have relied on various other conjectures, due to the difficulty of proving unconditional lower bounds.
Despite these prior works, exactly understanding the limits of efficient algorithms to this problem is still an active research area. This is natural considering the importance of sparse PCA and how fundamental it is to a multitude of applications.

Therefore, we naturally ask (also raised by and posed as an open problem in the works \cite{ma_wigderson_15, hop17, hop18})

\begin{it}
    Can Sum of Squares algorithms beat known algorithms for Sparse PCA?
\end{it}

In this work, we show that SoS algorithms cannot beat known spectral algorithms, even if we allow sub-exponential time! Therefore, this suggests that currently used algorithms such as thresholding or other spectral algorithms are in a sense optimal for this problem.

To prove our results, we will consider random instances of Sparse PCA and show that they are naturally hard for SoS. In particular, we focus on the Wishart random model of Sparse PCA. This model is a more natural modeling assumption compared to other random models that have been studied before, such as the Wigner random model.

Note importantly that our model assumptions only strengthen our results because we are proving impossibility results. In other words, if SoS algorithms do not work for this restricted version of sparse PCA, then it will not work for more general models, e.g. with general covariance or multiple spikes.
We now describe the model.

The Wishart model of Sparse PCA, also known as the Spiked Covariance model, was originally proposed by \cite{johnstone_lu2009}. In this model, we observe $m$ vectors $v_1, \ldots, v_m \in \RR^d$ from the distribution $\GN(0, I_d + \lda uu^T)$ where $u$ is a $k$-sparse unit vector, that is, $\norm{u}_0 \le k$ and we would like to recover the principal component $u$. Here, the sparsity of a vector is the number of nonzero entries and $\lda$ is known as the signal-to-noise ratio.

As the signal to noise ratio $\lda$ gets lower, it becomes harder and maybe even impossible to recover $u$ since the signature left by $u$ in the data becomes fainter. However, it's possible that this may be mitigated if the number of samples $m$ grows. Therefore, there is a tradeoff between $m, n$ and $k$ at play here. Algorithms proposed earlier have been able to recover $u$ at various regimes.
For example, if the number of samples is really large, namely $m \gg \max(\frac{d}{\lda}, \frac{d}{\lda^2})$, then standard PCA will work. If this is not the case, we may still be able to recover $u$ by assuming that the sparsity is not too large compared to the number of samples, namely $m \gg \frac{k^2}{\lda^2}$. To do this, we use a variant of standard PCA known as diagonal thresholding. Similar results have been obtained for various regimes, while some regimes have resisted attack to algorithms.

Our results here complete the picture by showing that in the regimes that have so far resisted attack by efficient algorithms, the powerful Sum of Squares algorithms also cannot recover the principal component. We now state our theorem informally, postponing the formal statement to \cref{cor: spca_main}.

\begin{theorem}\label{thm: spca_main_informal}
    For the Wishart model of Sparse PCA, sub-exponential time SoS algorithms fail to recover the principal component when the number of samples $m \ll \min(\frac{d}{\lda^2}, \frac{k^2}{\lda^2})$ .
\end{theorem}

In particular, this theorem resolves an open problem posed by \cite{ma_wigderson_15} and \cite{hop17, hop18}.

In almost all other regimes, algorithms to recover the principal component $u$ exist. We give a summary of such algorithms shortly, captured succinctly in \cref{fig: spca_thresholds}.
We say almost all other regimes because there is one interesting regime, namely $\frac{d}{\lda^2} \le m \le \frac{\min(d, k)}{\lda}$ marked by light green in \cref{fig: spca_thresholds}, where we can show that information theoretically, we cannot recover $u$ but it's possible to do hypothesis testing of Sparse PCA. That is, in this regime, we can distinguish purely random unspiked samples from the spiked samples. However, we will not be able to recover the principal component even if we use an exponential time brute force algorithm.

\begin{figure}[!h]
    \centering
    \begin{subfigure}{\textwidth}
        \centering
        \includegraphics[scale=.4, trim={0 0 0 0},clip]{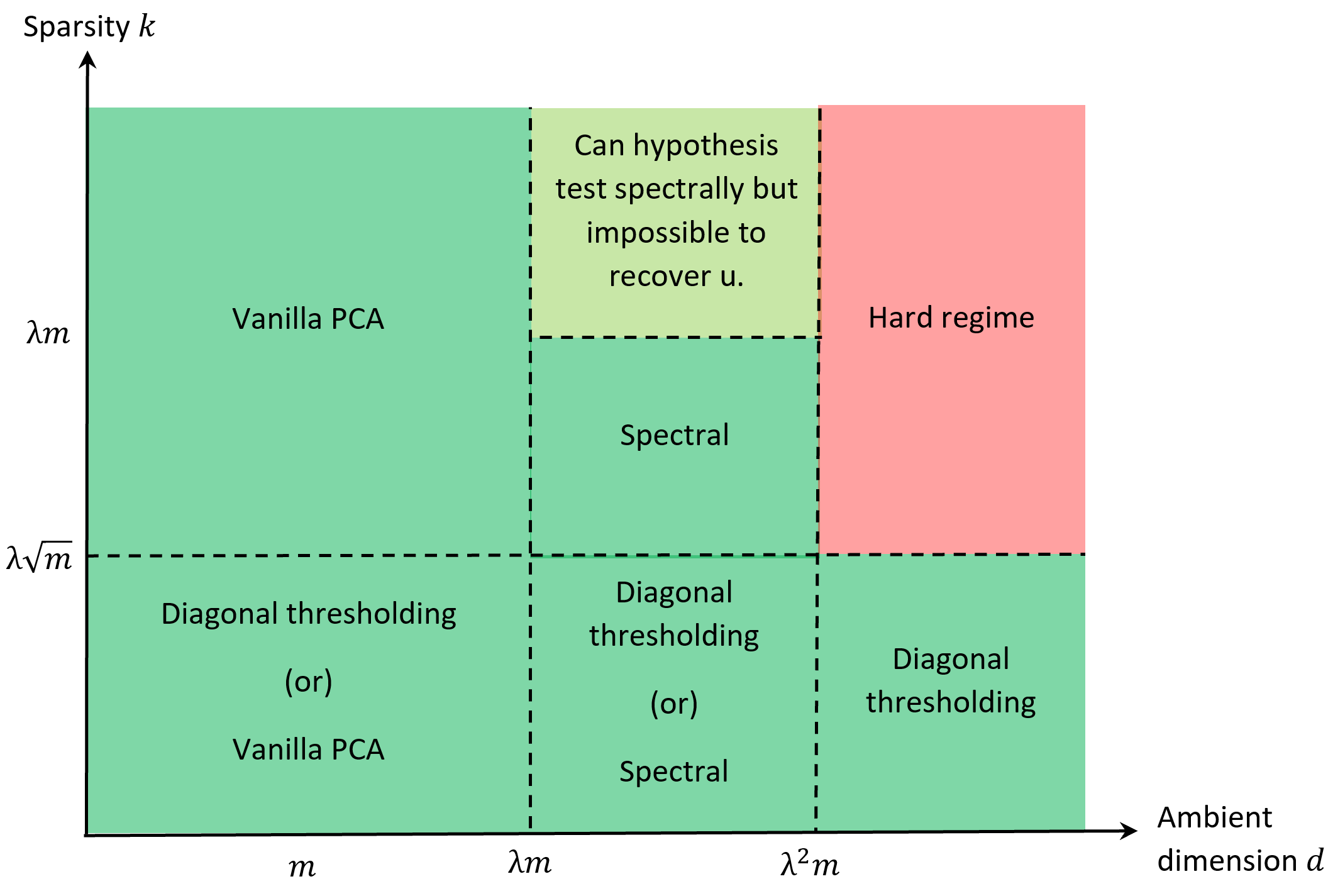}
        \caption{SNR $\lda \ge 1$}
        \label{fig: spca_thresholds1}
    \end{subfigure}%

    \begin{subfigure}{\textwidth}
        \centering
        \includegraphics[scale=.4, trim={0 0 0 0},clip]{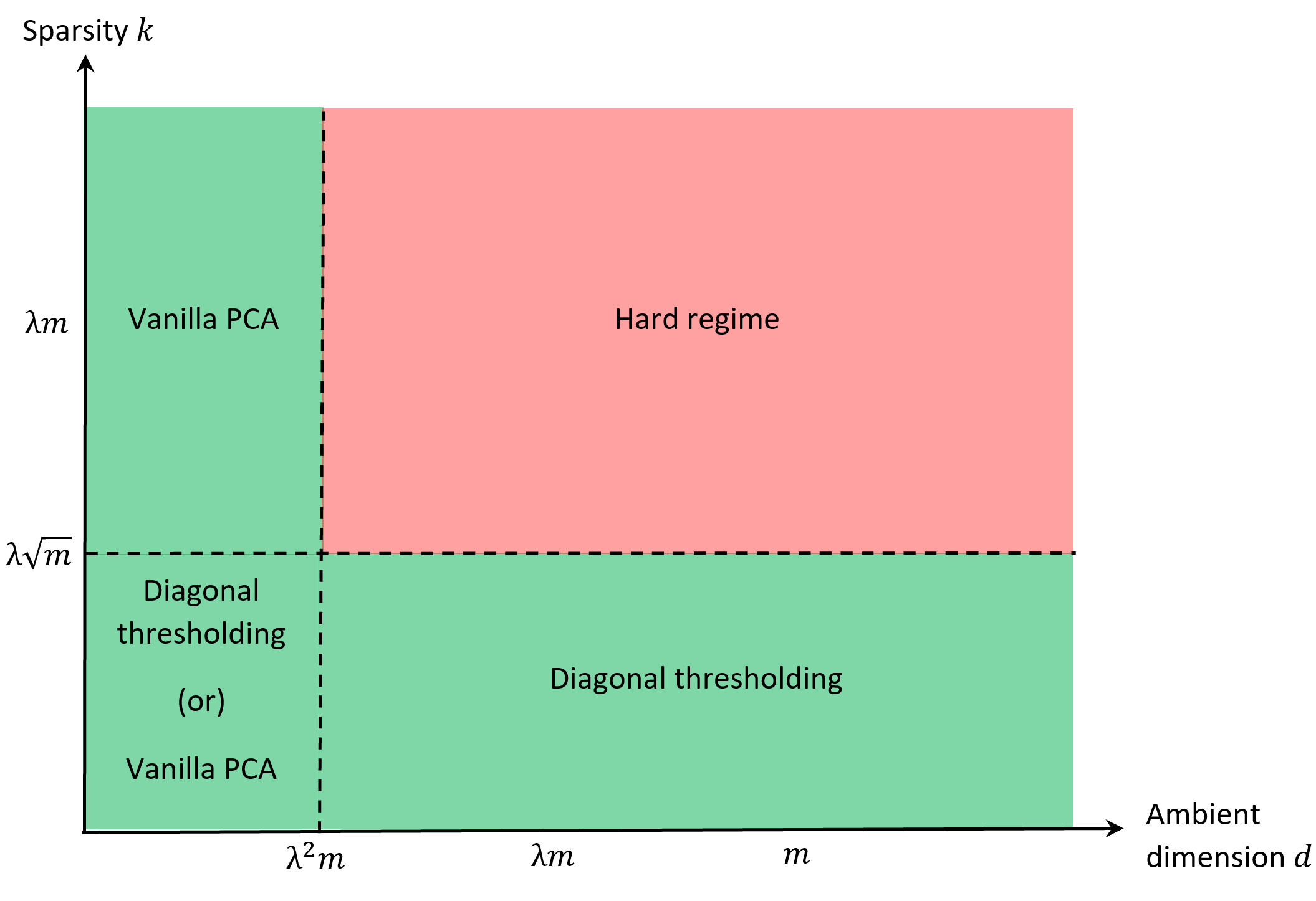}
        \caption{SNR $\lda < 1$}
        \label{fig: spca_thresholds2}
    \end{subfigure}
    \caption{Computational barrier diagram for Sparse PCA}
    \label{fig: spca_thresholds}
\end{figure}

Now, we state our results a bit more formally.
First, we will assume that the entries of $u$ are in $\{-\frac{1}{\sqrt{k}}, 0, \frac{1}{\sqrt{k}}\}$ chosen such that the sparsity is $k$ (and hence, the norm is $1$). Note importantly that this assumption is only strengthening our result: If SoS cannot solve this problem even for this specific $u$, it cannot do any better for the general problem with arbitrary $u$.

Let the vectors from the given dataset be $v_1, \ldots, v_m$. Let them form the rows of a matrix $S \in \RR^{m \times d}$.
Let $\Sigma = \frac{1}{m} \sum_{i = 1}^m v_iv_i^T$ be the sample covariance matrix. Then the standard PCA objective is to maximize $x^T\Sig x$ and recover $x = \sqrt{k}u$. Therefore, the sparse PCA problem can be rephrased as
\[\text{maximize } m\cdot x^T\Sigma x = \sum_{i = 1}^m \ip{x}{v_i}^2\text{ such that }x_i^3 = x_i\text{ for all $i \le d$ and } \sum_{i = 1}^d x_i^2 = k\]
where the program variables are $x_1, \ldots, x_d$.
The constraint $x_i^3 = x_i$ enforces that the entries of $x$ are in $\{-1, 0, 1\}$ and along with these constraints, the last condition $\sum_{i = 1}^d x_i^2 = k$ enforces $k$-sparsity (but we remark that, due to technical reasons, we will only satisfy this condition up to $o(1)$ error in our lower bounds). Then, the vector $u$ can be recovered by setting $u = \frac{1}{\sqrt{k}} x$.

Now, we will consider the series of convex relaxations for Sparse PCA obtained by SoS algorithms. In particular, we will consider SoS degree of $d^{\eps}$ for a small constant $\eps > 0$. Note that this corresponds to SoS algorithms of subexponential running time in the input size $d^{O(1)}$.

Our main result states that for choices of $m$ below a certain threshold, when the vectors $v_1, \ldots, v_m$ are sampled from the unspiked standard Gaussian $\GN(0, I_d)$, then sub-exponential time SoS algorithms will have optimal value close to $m + m\lambda$. This is also the optimal value in the case when the vectors $v_1, \ldots, v_m$ are indeed sampled from the spiked Gaussian $\GN(0, I_d + \lda uu^T)$ and $x = \sqrt{k}u$.
Therefore, SoS is unable to distinguish $\GN(0, I_d)$ from $\GN(0, I_d + \lda uu^T)$ and hence cannot solve sparse PCA. Formally,

\begin{theorem}\label{cor: spca_main}
    For all sufficiently small constants $\eps > 0$, suppose $m \le \frac{d^{1 - \eps}}{\lda^2}, m \le \frac{k^{2 - \eps}}{\lda^2}$, and for some $A > 0$, $d^{A} \le k \le d^{1 - A\eps}, \frac{\sqrt{\lda}}{\sqrt{k}} \le d^{-A\eps}$, then for an absolute constant $C > 0$, with high probability over a random $m \times d$ input matrix $S$ with Gaussian entries, the sub-exponential time SoS algorithm of degree $d^{C\eps}$ for sparse PCA has optimal value at least $m + m\lambda - o(1)$.
\end{theorem}


In other words, sub-exponential time SoS cannot certify that for a random dataset with Gaussian entries, there is no unit vector $u$ with $k$ nonzero entries and $m \cdot u^T\Sigma u \approx m + m\lambda$.

A few remarks are in order.
\begin{enumerate}
    \item Note here that $m + m\lda$ is approximately the value of the SoS program when the input vectors $v_1, \ldots, v_m$ are indeed sampled from the spiked model $\GN(0, I_d + \lda uu^T)$ and $x = \sqrt{k}u$. Therefore, sub-exponential time SoS is unable to distinguish a completely random distribution from the spiked distribution and hence is unable to solve sparse PCA.
    \item The constant $A$ can be thought of as $\approx 0$ and it appears for technical reasons, to ensure that we have sufficient decay in our bounds. In particular, most values of $k, \lda$ fall under the conditions of the theorem.
    \item For technical reasons, the constraint $\sum_{i = 1}^k x_i^2 = k$ is satisfied up to $o(1)$ error in our lower bounds. We leave to future work the problem of satisfying this constraint exactly.
\end{enumerate}

Informally, our main result says that when $m \ll \min\left(\frac{d}{\lda^2}, \frac{k^2}{\lda^2}\right)$, then subexponential time SoS cannot recover the principal component $u$. This is the content of \cref{thm: spca_main_informal}.

To show our results, we use the strategy from \cref{sec: strategy_for_sos_lower_bounds}.
For the Wishart model of Sparse PCA, we use the following distributions.
\begin{restatable}{itemize}{SPCAdistributions}
        \item Random distribution $\nu$: $v_1, \ldots, v_m$ are sampled from $\GN(0, I_d)$ and we take $S$ to be the $m \times d$ matrix with rows $v_1, \ldots, v_m$.
        \item Planted distribution $\mu$: Sample $u$ from $\{-\frac{1}{\sqrt{k}}, 0, \frac{1}{\sqrt{k}}\}^d$ where the values are taken with probabilites $\frac{k}{2d}, 1 - \frac{k}{d}, \frac{k}{2d}$ respectively. Then sample $v_1, \ldots, v_m$ as follows. For each $i \in [m]$, with probability $\Delta$, sample $v_i$ from $\GN(0, I_d + \lda uu^T)$ and with probability $1 - \Delta$, sample $v_i$ from $\GN(0, I_d)$. Finally, take $S$ to be the $m \times d$ matrix with rows $v_1, \ldots, v_m$.
\end{restatable}

In \cref{sec: spca_qual}, we compute the SoS solution obtained by pseudo-calibration. We prove the following theorem.

\begin{restatable}{theorem}{SPCAmain}\label{thm: spca_main}
    There exists a constant $C > 0$ such that for all sufficiently small constants $\eps > 0$, if $m \le \frac{d^{1 - \eps}}{\lda^2}, m \le \frac{k^{2 - \eps}}{\lda^2}$, and there exists a constant $A$ such that $0 < A < \frac{1}{4}$, $d^{4A} \le k \le d^{1 - A\eps}$, and $\frac{\sqrt{\lda}}{\sqrt{k}} \le d^{-A\eps}$, then with high probability, the SoS solution given by pseudo-calibration for degree $d^{C\eps}$ Sum of Squares is feasible.
\end{restatable}

Since we use an average case distribution, this SoS lower bound is a lower bound for certification. An overview of our proof is in \cref{sec: global_approach}.  From this theorem, \cref{cor: spca_main} follows as a corollary.

\paragraph{Prior work on algorithms}

Due to its widespread importance, a tremendous amount of work has been devoted to obtaining algorithms for sparse PCA, both theoretically and practically, \cite{amini_wainwright2008, ma2013sparse, krauthgamer2015, deshpande2016, wang2016statistical, berthet2013complexity, ma_wigderson_15, diakonikolas2017statistical, hop17, brennan2019optimal, ding2019subexponential,  chowdhury2020approximation, d2020sparse} to cite a few.

We now place our result in the context of known algorithms for Sparse PCA and explain why it offers tight tradeoffs between approximability and inapproximability.
Between this work and prior works, we completely understand the parameter regimes where sparse PCA is easy or conjectured to be hard up to polylogarithmic factors. In \cref{fig: spca_thresholds1} and \cref{fig: spca_thresholds2}, we assign the different parameter regimes into the following categories.
\begin{itemize}
    \item Diagonal thresholding: In this regime, Diagonal thresholding \cite{johnstone_lu2009, amini_wainwright2008} recovers the sparse vector. Covariance thresholding \cite{krauthgamer2015, deshpande2016} and SoS algorithms \cite{sparse_pca_focs20} can also be used in this regime. The benefits of these alternate algorithms are that covariance thresholding has better dependence on logarithmic factors and SoS algorithms works in the presence of adversarial errors.
    \item Vanilla PCA: Vanilla PCA (i.e. standard PCA) can recover the vector, i.e. we do not need to use the fact that the vector is sparse (see e.g. \cite{berthet2013, sparse_pca_focs20}).
    \item Spectral: An efficient spectral algorithm recovers the sparse vector (see e.g. \cite{sparse_pca_focs20}).
    \item Can test but not recover: A simple spectral algorithm can solve the hypothesis testing version of Sparse PCA but it is information theoretically impossible to recover the sparse vector \cite[Appendix E]{sparse_pca_focs20}.
    \item Hard: A regime where it is conjectured to be hard for algorithms to recover the sparse principal component. We discuss this in more detail below.
\end{itemize}

In \cref{fig: spca_thresholds1} and \cref{fig: spca_thresholds2}, the regimes corresponding to Diagonal thresholding, Vanilla PCA and Spectral are dark green, while the regimes corresponding to Spectral* and Hard are light green and red respectively.

\paragraph{Prior work on hardness}

Prior works have explored statistical query lower bounds \cite{brennan2020statistical}, basic SDP lower bounds \cite{krauthgamer2015}, reductions from conjectured hard problems \cite{berthet2013, berthet2013complexity, brennan2019optimal, gao2017sparse, wang2016statistical}, lower bounds via the low-degree conjecture \cite{ding2019subexponential, sparse_pca_focs20}, lower bounds via statistical physics \cite{ding2019subexponential, arous2020free}, etc.
We note that similar threshold behaviors as us have been predicted by \cite{sparse_pca_focs20}, but importantly, they assume a conjecture known as the low-degree likelihood conjecture. Similarly, many of these other lower bounds rely on various conjectures. To put this in context, the low-degree likelihood conjecture is a stronger assumption than P $\neq$ NP. In contrast, our results are unconditional and do not assume any conjectures.

Compared to these other lower bounds, there have only been two prior works on lower bounds against SoS algorithms \cite{krauthgamer2015, berthet2013, ma_wigderson_15} which are only for degree $2$ and degree $4$ SoS. In particular, degree $2$ SoS lower bounds have been studied in \cite{krauthgamer2015, berthet2013} although they don't state it this way. Moreover, \cite{ma_wigderson_15} obtained degree $4$ SoS lower bounds but they were very lossy, i.e. they hold for a strict subset of the \textit{Hard} regime $m \ll \frac{k^2}{\lda^2}$ and $m \ll \frac{d}{\lda^2}$. Moreover, the ideas used in these prior works do not generalize for higher degrees.
The lack of other SoS lower bounds can be attributed to the difficulty in proving such lower bounds. In this paper, we vastly strengthen these known results and show almost-tight lower bounds for SoS algorithms of degree $d^{\eps}$ which correspond to sub-exponential running time $d^{d^{O(\eps)}}$.
We note that SoS algorithms get stronger as the degree increases, therefore our results immediately imply these prior results and even in the special case of degree $4$ SoS, we improve the known lossy bounds. In summary, \cref{cor: spca_main} subsumes all these earlier known results and is a vast improvement over prior known SoS lower bounds which provides compelling evidence for the hardness of Sparse PCA in this parameter range.



The work \cite{hop17} also states SoS lower bounds for Sparse PCA but it differs from our work in three important aspects. First, they handle the related but qualitatively different Wigner model of Sparse PCA. Their techniques fail for the Wishart model of Sparse PCA, which is more natural in practice. We overcome this shortcoming and work with the Wishart model. We emphasize that their techniques are insufficient to handle this generality and overcoming this is far from being a mere technicality. On the other hand, our techniques can easily recover their results.
Second, while they sketch a high level proof overview for their lower bound, they don't give a proof. On the other hand, our proofs are fully explicit.
Finally, they assume the input distribution has entries in $\{\pm 1\}$, that is, they work with the $\pm 1$ variant of PCA.
On the other hand, we work with the more realistic setting where the distribution is $\GN(0, 1)$.
Again, our techniques can easily recover their results as well.

\section{Tensor PCA}

We use our techniques to also obtain strong results for the related Tensor Principal components analysis (Tensor PCA) problem.
Tensor PCA, originally introduced by \cite{richard2014statistical}, is a generalization of PCA to higher order tensors. Formally, given an order $k$ tensor of the form $\lda u^{\otimes k} + B$ where $u \in \RR^n$ is a unit vector and $B \in \RR^{[n]^k}$ has independent Gaussian entries, we would like to recover the principal component $u$. Here, $\lda$ is known as the signal-to-noise ratio.

Tensor PCA is a remarkably useful statistical and computational technique to exploit higher order moments of the data.
It was originally envisaged to be applied in latent variable modeling and indeed, it has found multiple applications in this context (e.g. \cite{anandkumar2014tensor, kivva2021learning, anandkumar2014analyzing}). Here, a tensor containing statistics of the input data is computed and then it's decomposed in order to recover the latent variables.
Because of the technique's versatility, it has gathered a lot of attention in machine learning with applications in topic modeling, video processing, collaborative filtering,  community detection, etc. (see e.g. \cite{hsu2012spectral, anandkumar2014guaranteed, richard2014statistical, anandkumar2014tensor, anandkumar2014analyzing, duchenne2011tensor, li2010tensor} and references therein.)

For Tensor PCA, similar to sparse PCA, there has been wide interest in the community to study algorithms (e.g. \cite{arous2020algorithmic, tensorpca16, HopSS15, hopkins2016fast, richard2014statistical, zheng2015interpolating, wein2019kikuchi, kim2017community, anandkumar2017homotopy}) as well as approximability and hardness (e.g. \cite{montanari2015limitation, kunisky19notes, brennan2020reducibility, hop17}, a more detailed overview is presented after stating our main results).
It's worth noting that many of these hardness results are conditional, that is, they rely on various conjectures, sometimes stronger than P $\neq$ NP.
Moreover, there has been widespread interest from the statistics community as well, e.g. \cite{jagannath2020statistical, perry2016statistical, lesieur2017statistical, chen2019phase, chen2018phase}, due to fascinating connections to random matrix theory and statistical physics.

In this work, we study the performance of sub-exponential time Sum of Squares algorithms for Tensor PCA.
Our main result is stated informally below and formally in \cref{cor: tpca_main}.

\begin{theorem}\label{thm: tpca_main_informal}
    For Tensor PCA, sub-exponential time SoS algorithms fail to recover the principal component when the signal to noise ratio $\lda \ll n^{\frac{k}{4}}$.
\end{theorem}

In particular, this resolves an open question posed by the works \cite{HopSS15, tensorpca16, hop17, hop18}.

Let's make this theorem formal. Recall that we are given an order $k$ tensor $A$ of the form $A = \lda u^{\otimes k} + B$ where $u \in \RR^n$ is a unit vector and $B \in \RR^{[n]^k}$ has independent Gaussian entries and we would like to recover the principal component $u$.
Tensor PCA can be rephrased by the program
\[\text{maximize }\ip{A}{x^{\otimes k}} = \ip{A}{\underbrace{x\otimes\ldots\otimes x}_{\text{$k$ times}}}\text{ such that } \sum_{i = 1}^n x_i^2 = 1\]
where the program variables are $x_1, \ldots, x_n$.
The principal component $u$ will then just be the returned solution $x$.
Just like in Sparse PCA, we remark that for technical reasons, we will satisfy the unit vector condition only up to $o(1)$ error in our lower bounds and satisfying the condition exactly is left for future work.
We will again consider sub-exponential time SoS algorithms, in particular degree $n^{\eps}$ SoS, for this problem. This is sub-exponential time because the input size is $n^{O(1)}$.

We then show that if the signal to noise ratio $\lda$ is below a certain threshold, then sub-exponential time SoS for the unspiked input $A \sim \GN(0, I_{[n]^k})$ will have optimal value close to $\lda$, which is also the optimal value in the spiked case when $A = \lda u^{\otimes k} + B, B\sim \GN(0, I_{[n]^k})$ and $x = u$. In other words, SoS cannot distinguish the unspiked and spiked distributions and hence cannot recover the principal component $u$.

\begin{theorem}\label{cor: tpca_main}
    Let $k \ge 2$ be an integer. For all sufficiently small $\eps > 0$, if $\lda \le n^{\frac{k}{4} - \eps}$, for an absolute constant $C > 0$, with high probability over a random tensor $A \sim\GN(0, I_{[n]^k})$, the sub-exponential time SoS algorithm of degree $n^{C\eps}$ for Tensor PCA has optimal value at least $\lda - o(1)$.
\end{theorem}


Therefore, sub-exponential time SoS cannot certify that for a random tensor $A$ sampled from $\GN(0, I_{[n]^k})$, there is no unit vector $u$ such that $\ip{A}{\underbrace{u \otimes\ldots\otimes u}_{\text{$k$ times}}} \approx \lambda$.

We again remark that when the tensor $A$ is actually sampled from the spiked model $A = \lda u^{\otimes k} + B$, the optimal value of the SoS program is approximately $\lda$ when $x = u$. Therefore, this shows that sub-exponential time SoS algorithms cannot solve Tensor PCA.

Informally, the theorem says that when the signal to noise ratio $\lda \ll n^{\frac{k}{4}}$, SoS algorithms cannot solve Tensor PCA, as stated in \cref{thm: tpca_main_informal}.

To show our results for Tensor PCA, we apply the strategy from \cref{sec: strategy_for_sos_lower_bounds} where we use the following distributions. Let $k \ge 2$ be an integer.
\begin{restatable}{itemize}{TPCAdistributions}
        \item Random distribution $\nu$: Sample $A$ from $\GN(0, I_{[n]^k})$.
        \item Planted distribution $\mu$: Let $\lda,\Delta > 0$. Sample $u$ from $\{-\frac{1}{\sqrt{\Delta n}}, 0, \frac{1}{\sqrt{\Delta n}}\}^n$ where the values are taken with probabilites $\frac{\Delta}{2}, 1 - \Delta, \frac{\Delta}{2}$ respectively. Then sample $B$ from $\GN(0, I_{[n]^k})$. Set $A = B + \lda \tens{u}{k}$.
\end{restatable}

In \cref{sec: tpca_qual}, we apply pseudo-calibration and we prove the following theorem.

\begin{restatable}{theorem}{TPCAmain}\label{thm: tpca_main}
    Let $k \ge 2$ be an integer. There exist constants $C,C_{\Del} > 0$ such that for all sufficiently small constants $\eps > 0$, if $\lda \le n^{\frac{k}{4} - \eps}$ and $\Del = n^{-C_{\Del}\eps}$ then with high probability, the SoS solution given by pseudo-calibration for degree $n^{C\eps}$ Sum of Squares is feasible.
\end{restatable}

This theorem can also be naturally interpreted as an SoS lower bound for the certification problem of Tensor PCA. A sketch of our proof follows in \cref{sec: global_approach}. From this theorem, \cref{cor: tpca_main} follows as a corollary.

\paragraph{Prior work}
Algorithms for Tensor PCA have been studied in the works \cite{arous2020algorithmic, tensorpca16, HopSS15, hopkins2016fast, richard2014statistical, zheng2015interpolating, wein2019kikuchi, kim2017community, anandkumar2017homotopy}. It was shown in \cite{tensorpca16} that the degree $q$ SoS algorithm certifies an upper bound of $\frac{2^{O(k)} (n \cdot \text{polylog}(n))^{k/4}}{q^{k/4 - 1/2}}$ for the Tensor PCA problem. When $q = n^{\eps}$ this gives an upper bound of $n^{\frac{k}{4} - O(\eps)}$. Therefore, our result is tight, giving insight into the computational threshold for Tensor PCA.

Lower bounds for Tensor PCA have been studied in various forms including statistical query lower bounds \cite{brennan2020statistical, dudeja2021statistical}, reductions from conjectured hard problems \cite{zhang2018tensor, brennan2020reducibility}, lower bounds from the low-degree conjecture \cite{hop17, hop18, kunisky19notes}, evidence based on the landscape behavior \cite{arous2019landscape, montanari2015limitation}, etc. Compared to a lot of these works which rely on various conjectures, we remark that our lower bounds are unconditional and do not rely on any conjectures.

In \cite{hop17}, similar to Sparse PCA, they state a similar theorem for a different variant of Tensor PCA. However, they do not give a proof whereas we give explicit proofs.
In particular, they state their result without proof for the $\pm{1}$ variant of Tensor PCA whereas we work with the more realistic setting where the distribution is $\GN(0, 1)$. We remark that their techniques do not recover our results but on the other hand, our techniques can recover theirs.


\section{Planted Slightly Denser Subgraph}

In the planted dense subgraph problem, we are given a random graph $G$ where a dense subgraph of size $k$ has been planted and we are asked to find this planted dense subgraph.
This is a natural generalization of the $k$-clique problem \cite{karp1972reducibility} and has been subject to a long line of work over the years (e.g. \cite{feige1997densest, feige2001dense, khot2006ruling, bhaskara2010detecting, bhaskara2012polynomial, braverman2017eth, manurangsi2017almost}).
In this work, we consider the following certification variant of planted dense subgraph.

\begin{quote}
	\em{Given a random graph $G$ sampled from the \Erdos-\Renyi model $G(n, \frac{1}{2})$, certify an upper bound on the edge density of the densest subgraph on $k$ vertices.}
\end{quote}

We show a high degree SoS lower bound for this problem using the strategy from \cref{sec: strategy_for_sos_lower_bounds}. In particular, we use the following distributions.
\begin{restatable}{itemize}{PLDSdistributions}
	\item Random distribution $\nu$: Sample $G$ from $G(n, \frac{1}{2})$
	\item Planted distribution $\mu$: Let $k$ be an integer and let $p > \frac{1}{2}$. Sample a graph $G'$ from $G(n, \frac{1}{2})$. Choose a random subset $S$ of the vertices, where each vertex is picked independently with probability $\frac{k}{n}$. For all pairs $i, j$ of vertices in $S$, rerandomize the edge $(i, j)$ where the probability of $(i, j)$ being in the graph is now $p$. Set $G$ to be the resulting graph.
\end{restatable}
In \cref{sec: plds_qual}, we compute the candidate SoS solution obtained via pseudo-calibration. Our main theorem is as follows, with a proof sketch following in \cref{sec: global_approach}.

\begin{restatable}{theorem}{PLDSmain}\label{thm: plds_main}
	Let $C_p > 0$. There exists a constant $C > 0$ such that for all sufficiently small constants $\eps > 0$, if $k \le n^{\frac{1}{2} - \eps}$ and $p =  \frac{1}{2} + \frac{n^{-C_p\eps}}{2}$, then with high probability, the candidate solution given by pseudo-calibraton for degree $n^{C\eps}$ Sum of Squares is feasible.
\end{restatable}


\paragraph{Related work}
For many different parameter regimes of the random and planted distributions (an example being planting $G_{k, q}$ in $G_{n, p}$ for constants $p < q$), and when $k = o(\sqrt{n})$, the hardness of the easier distinguishing version of planted dense subgraph problem has been posed as formal conjecture (often referred to as the PDS conjecture) before in the literature (see e.g., \cite{hajek2015computational, chen2014statistical, brennan2018reducibility, brennan2019universality}). This has also led to many reductions to other problems \cite{brennan2019optimal}, although it's not clear if these reductions can be made in the SoS framework without loss in the parameter dependence.

In our case, we consider the slightly planted denser subgraph version where for $k \le n^{\frac{1}{2} - \eps}$, we plant a subgraph of density $\frac{1}{2} + \frac{1}{n^{O(\eps)}}$, i.e. $p = \frac{1}{2}, q = \frac{1}{2} + \frac{1}{n^{O(\eps)}}$. This has been widely believed to require sub-exponential time. Our work provides strong evidence towards this by exhibiting unconditional lower bounds against the powerful SoS hierarchy, even if we consider $n^{O(\eps)}$ levels, which corresponds to $n^{n^{O(\eps)}}$ running time! We expect this to lead to this problem being used as a natural starting point for reductions to show sub-exponential time hardness for various problems.

Within the SoS literature, \cite{BHKKMP16} show that for $k \le n^{\frac{1}{2} - \eps}$ for a constant $\eps > 0$, the degree $o(\log n)$ Sum of Squares cannot distinguish between a fully random graph sampled from $G(n, \frac{1}{2})$ from a random graph which has a planted $k$-clique. This implies that degree $o(\log n)$ SoS cannot certify an edge density better than $1$ for the densest $k$-subgraph if $k \le n^{\frac{1}{2} - \eps}$.

In \cref{thm: plds_main}, we show that for $k \le n^{\frac{1}{2} - \eps}$ for a constant $\eps > 0$, degree $n^{\Omega(\eps)}$ SoS cannot certify an edge density better than $\frac{1}{2} + \frac{1}{n^{O(\eps)}}$. The degree of SoS in our setting, $n^{\Omega(\eps)}$ is vastly higher than the earlier known result which uses degree $o(\log n)$. To the best of our knowledge, this is the first result that proves such a high degree lower bound for this problem. 

We remark that when we take $k = n^{\frac{1}{2} - \eps}$,  the true edge density of the densest $k$-subgraph is $\frac{1}{2} + \frac{\sqrt{\log(n/k)}}{\sqrt{k}} + \littleoh(\frac{1}{\sqrt{k}}) \approx \frac{1}{2} + \frac{1}{n^{1/4 - \eps/2}}$ as was shown in \cite[Corollary 2]{gamarnik2019landscape} whereas, by \cref{thm: plds_main}, the SoS optimum is as large as $\frac{1}{2} + \frac{1}{n^{\eps}}$. This highlights a significant difference in the optimum value.

\section{Our approach}\label{sec: global_approach}

In this section, we briefly describe how to prove \cref{thm: spca_main}, \cref{thm: tpca_main} and \cref{thm: plds_main}. We naturally start with pseudocalibration and all constraints except positivity are easily shown to hold by construction. To show positivity and hence the lower bound,  we will essentially apply a general meta-theorem called the machinery. The machinery enables us to show SoS lower bounds for certain kinds of ``noisy'' problems.

In this work, we state and use the machinery, whose proof can be found in the original work where it appeared \cite{potechin2020machinery}. To show PSDness, the machinery constructs certain \emph{coefficient matrices} from the moment matrix $\Lda$ and gives conditions on these coefficient matrices which are sufficient to guarantee that $\Lda$ is PSD with high probability. Some of the ideas involved in the machinery are a generalization of the techniques used to prove the SoS lower bound for planted clique \cite{BHKKMP16}. In this section, we give an informal sketch of the machinery. We also motivate some of the conditions that arise.

\paragraph{Shapes and graph matrices}
We start by describing shapes and graph matrices, which were originally introduced by \cite{BHKKMP16, medarametla2016bounds} and later generalized in \cite{ahn2016graph}. They will be covenient for our analysis.

Shapes $\al$ are graphs that contain extra information about the vertices. Corresponding to each shape $\al$, there is a matrix-valued function $M_{\al}$ (i.e. a matrix whose entries depend on the input) that we call a graph matrix. Graph matrices are analogous to a Fourier basis, but for matrix-valued functions that exhibit a certain kind of symmetry. In our setting, $\Lda$ will be such a matrix-valued function, so we can decompose $\Lda$ as a linear combination of graph matrices $\Lda = \sum_{\text{shapes } \alpha}{\lambda_{\alpha}M_{\alpha}}$.

Shapes and graph matrices have several properties which make them very useful to work with. First, $\norm{M_{\al}}$ can be bounded with high probability in terms of simple combinatorial properties of the shape $\al$. Second, if two shapes $\alpha$ and $\beta$ match up in a certain way, we can combine them to form a larger shape $\alpha \circ \beta$. We call this operation shape composition. Third, each shape $\al$ has a canonical decomposition into three shapes, the left, middle and right parts of $\al$, which we call $\sigma$, $\tau$, and ${\sigma'}^T$. For this canonical decomposition, we have that $\alpha = \sigma \circ \tau \circ {\sigma'}^T$ and $M_{\alpha} \approx M_{\sigma}M_{\tau}M_{{\sigma'}^T}$. This decomposition is crucial for our analysis.

\paragraph{A general framework for SoS lower bounds}
We now sketch the strategy of the machinery.
%
\begin{enumerate}
    \item Decompose the moment matrix $\Lda$ as a linear combination $\Lda = \sum_{\text{shapes } \alpha}{\lambda_{\alpha}M_{\alpha}}$ of graph matrices $M_{\alpha}$.
    \item For each shape $\alpha$, decompose $\alpha$ into a left part $\sigma$, a middle part $\tau$, and a right part ${\sigma'}^T$.
    \item Based on the coefficients $\lambda_{\alpha}$ and the decompositions of the shapes $\alpha$ into left, middle, and right parts, construct coefficient matrices $H_{Id_U}$ and $H_{\tau}$.
    \item Based on the coefficient matrices $H_{Id_U}$ and $H_{\tau}$, obtain an approximate PSD decomposition of $\Lda$.
    \item Show that the error terms (which we call intersection terms) can be bounded by the approximate PSD decomposition of $\Lda$.
\end{enumerate}
This is broadly similar to the work of \cite{BHKKMP16} who showed SoS lower bounds for the planted clique problem.

The machinery shows that this analysis will succeed by distilling it as three conditions on the coefficient matrices.
The rough blueprint to use the machinery to prove SoS lower bounds is as follows.
\begin{enumerate}
    \item Construct a candidate moment matrix $\Lda$.
    \item Decompose the moment matrix $\Lda$ as a linear combination $\Lda = \sum_{\text{shapes } \alpha}{\lambda_{\alpha}M_{\alpha}}$ of graph matrices $M_{\alpha}$ (akin to Fourier decomposition) and find the corresponding coefficient matrices.
    \item Verify the required conditions on the coefficient matrices.
\end{enumerate}

\subsubsection{A sketch of the intuition behind the conditions}\label{ideadescriptionsubsection}

We now motivate and sketch the conditions we present in the machinery.

\paragraph{Giving an approximate PSD factorization}
As discussed above, we decompose the moment matrix $\Lda$ as a linear combination $\Lda = \sum_{\text{shapes } \alpha}{\lambda_{\alpha}M_{\alpha}}$ of graph matrices $M_{\alpha}$. We then decompose each $\alpha$ into left, middle, and right parts $\sigma$, $\tau$, and ${\sigma'}^T$. We now have that
\[
\Lda = \sum_{\alpha = \sigma \circ \tau \circ {\sigma'}^T}{\lambda_{\sigma \circ \tau \circ {\sigma'}^T}M_{\sigma \circ \tau \circ {\sigma'}^T}}
\]

We first consider the terms $\sum_{\sig, \sig'} \lda_{\sig \circ \sig'^T}M_{\sig \circ \sig'^T} \approx \sum_{\sig, \sig'} \lda_{\sig \circ \sig'^T}M_{\sig} M_{\sig'^T}$ where $\tau$ corresponds to an identity matrix and can be ignored.

If there existed real numbers $v_{\sig}$ for all left shapes $\sig$ such that $\lda_{\sig \circ \sig'^T} = v_{\sig}v_{\sig'}$, then we would have
\[
\sum_{\sig, \sig'} \lda_{\sig \circ \sig'^T}M_{\sig} M_{\sig'^T} = \sum_{\sig, \sig'} v_{\sig}v_{\sig'}M_{\sig} M_{\sig'^T} = (\sum_{\sig} v_{\sig}M_{\sig})(\sum_{\sig} v_{\sig}M_{\sig})^T \succeq 0
\]
which shows that the contribution from these terms is positive semidefinite. In fact, this turns out to be the case for the planted clique analysis. However, this may not hold in general. To handle this, we note that the existence of $v_{\sig}$ can be relaxed as follows: Let $H$ be the matrix with rows and columns indexed by left shapes $\sig$ such that $H(\sig, \sig') = \lda_{\sig \circ \sig'^T}$. Up to scaling, $H$ will be one of our coefficient matrices. If $H$ is positive semidefinite then the contribution from these terms will also be positive semidefinite. In fact, this will be
the PSD mass condition of the main theorem.

\paragraph{Handling terms with a non-trivial middle part}

Unfortunately, we also have terms $\lda_{\sig \circ \tau \circ \sig'^T}M_{\sig \circ \tau \circ \sig'^T}$ where $\tau$ is non-trivial. Their strategy is to charge these terms to other terms.
For the sake of simplicity, we will describe how to handle one term. A starting point is the following inequality. For a left shape $\sig$, a middle shape $\tau$, a right shape $\sig'^T$, and real numbers $a, b$,
\[(a M_{\sig} - bM_{\sig'}M_{\tau^T})(a M_{\sig} - bM_{\sig'}M_{\tau^T})^T \succeq 0\]
which rearranges to
\begin{align*}
    ab(M_{\sig}M_{\tau}M_{\sig'^T} + (M_{\sig}M_{\tau}M_{\sig'^T})^T) &\preceq a^2M_{\sig}M_{\sig^T} + b^2M_{\sig'}M_{\tau^T}M_{\tau}M_{\sig'^T}\\
    &\preceq a^2M_{\sig}M_{\sig^T} + b^2\norm{M_{\tau}}^2M_{\sig'}M_{\sig'^T}
\end{align*}

If $\lda_{\sig \circ \tau \circ \sig'^T}^2\norm{M_{\tau}}^2 \le \lda_{\sig \circ \sig^T}\lda_{\sig' \circ \sig'^T}$, then we can choose $a, b$ such that $a^2 \le \lda_{\sig \circ \sig^T}$, $ b^2 \norm{M_{\tau}}^2 \le \lda_{\sig' \circ \sig'^T}$ and $ab = \lda_{\sig \circ \tau \circ \sig'^T}$. This will approximately imply
\[\lda_{\sig \circ \tau \circ \sig'^T}(M_{\sig \circ \tau \circ \sig'^T} + M_{\sig \circ \tau \circ \sig'^T}^T) \preceq \lda_{\sig \circ \sig^T}M_{\sig \circ \sig^T} + \lda_{\sig' \circ \sig'^T}M_{\sig' \circ \sig'^T}\]
which will give us a way to charge terms with a nontrivial middle part against terms with a trivial middle part.

While we could try to apply this inequality term by term, it is not strong enough to give us the main machinery result. Instead, they generalize this inequality to work with the entire set of shapes $\sig, \sig'$ for a fixed $\tau$. This will lead us to
the middle shape bounds condition.

\paragraph{Handing intersection terms}

There's one important technicality in the above calculations. Whenever we decompose $\alpha$ into left, middle, and right parts $\sigma$, $\tau$, and ${\sigma'}^T$, $M_{\sigma}M_{\tau}M_{{\sigma'}^T}$ is only approximately equal to $M_{\alpha} = M_{\sigma \circ \tau \circ {\sigma'}^T}$. All the other error terms have to be carefully handled in the analysis. We call these terms intersection terms.

We exploit the fact that these intersection terms themselves are graph matrices. Therefore, we recursively decompose them into $\sig_2 \circ \tau_2 \circ \sig_2'^T$ and apply the previous ideas. To do this methodically, the machinery employs several ideas such as the notion of intersection patterns and the generalized intersection tradeoff lemma. Properly handling the intersection terms is one of the most technically intensive parts of their work. 
This analysis leads us to the intersection term bounds condition.

\paragraph{Applying the machinery}

To apply the machinery to our problems of interest, we verify the spectral conditions that our coefficients should satisfy and then we can use the main theorem. The Planted slightly denser subgraph application is straightforward and will serve as a good warmup to understand the machinery. In the applications to Tensor PCA and Sparse PCA, the shapes corresponding to the graph matrices with nonzero coefficients have nice structural properties that will be crucial for our analysis. We exploit this structure and use novel charging arguments to verify the conditions of the machinery. We do this in this work.

\section{Related work on Sum of Squares Lower Bounds for Certification Problems}

\cite{KothariMOW17} proved that for random constraint satisfaction problems (CSPs) where the predicate has a balanced pairwise independent distribution of solutions, with high probability, degree $\Omega(n)$ SoS is required to certify that these CSPs do not have a solution. While they don't state it in this manner, the pseudo-expectation values used by \cite{KothariMOW17} can also be derived using pseudo-calibration \cite{rajendran2018combinatorial, brown2020extended}. The analysis for showing that the moment matrix is PSD is very different. It is an interesting question whether or not it is possible to unify these analyses.

\cite{mohanty2020lifting} showed that it's possible to lift degree $2$ SoS solutions to degree $4$ SoS solutions under suitable conditions, and used it to obtain degree $4$ SoS lower bounds for average case $d$-regular Max-Cut and the Sherrington Kirkpatrick problem. Their construction is inspired by pseudo-calibration and their analysis also goes via graph matrices.


\cite{kunisky2020} recently proposed a technique to lift degree $2$ SoS lower bounds to higher levels and applied it to construct degree $6$ lower bounds for the Sherrington-Kirkpatrick problem. Interestingly, their construction does not go via pseudo-calibration.

\section{Organization of the proofs}

We prove the Sherrington-Kirkpatrick lower bound, \cref{theo:sk-bounds}, in \cref{chap: sk}. The proofs for planted slightly denser subgraph, tensor PCA and sparse PCA, namely \cref{thm: plds_main}, \cref{thm: tpca_main} and \cref{thm: spca_main}, are split between \cref{chap: qual} and \cref{chap: quant}. The latter proofs are split into qualitative and quantitative versions. Qualitative theorem statements capture the essence of the inequalities we prove, and serve to illustrate the main forms of the bounds we desire, without getting lost in the details. Quantitative theorems on the other hand build on their qualitative counterparts by stating the precise bounds that are needed. In \cref{chap: qual}, we introduce the machinery and and in \cref{sec: plds_qual}, \cref{sec: tpca_qual} and \cref{sec: spca_qual}, we qualitatively verify the conditions of the machinery for planted slightly denser subgraph, tensor PCA, and sparse PCA respectively. While these sections only verify the qualitative conditions, the results in these sections are precise and will be reused in \cref{chap: quant}, where we fully verify the conditions of the machinery in \cref{sec: plds_quant}, \cref{sec: tpca_quant} and \cref{sec: spca_quant}.

\chapter{The Sherrington-Kirkpatrick Hamiltonian}\label{chap: sk}
In this chapter, we will formally prove Sum of Squares lower bounds for the certification problem of the Sherrington-Kirkpatrick Hamiltonian, in particular \cref{theo:sk-bounds}. The material in this chapter is adapted from \cite{sklowerbounds}, where this work originally appeared. The main difference in this chapter from that work is that we omit the technical section on satisfying constraints exactly.

\section{Technical preliminaries}

In this section we record formal problem statements, then define and discuss one of the main objects in our SoS
lower bound: graph matrices.

For a vector or variable
$v \in \R^n$, and $I \subseteq [n]$, we use the notation
$v^I \defeq \prod_{i \in I}v_i$. When a statement holds with high
probability (w.h.p.), it means it holds with probability $1 - o_n(1)$. In
particular, there is no requirement for small $n$.

\subsection{Problem statements}

We introduce the Planted Affine Planes problem over a distribution $\calD$.
\begin{definition}[Planted Affine Planes (PAP) problem]\label{def:prob:pap}
	Given $d_1, \dots, d_m \sim \calD$ where each $d_u$ is a vector in $\RR^n$,
	determine whether there exists $v \in \set{\pm \frac{1}{\sqrt{n}}}^n$ such that
	\[
	\ip{v}{d_u}^2 = 1,
	\]
	for every $u \in [m]$.
\end{definition}
Our results hold for the Gaussian setting $\mathcal{D} = \calN(0, I)$ and the boolean setting where $\calD$ is uniformly sampled from $\{\pm 1\}^n$, though we conjecture in \cref{sec:open-problems} that similar SoS bounds hold under more general conditions on $\calD$.

Observe that in both settings the solution vector $v$ is restricted to be Boolean (in the sense that the entries are either $\frac{1}{\sqrt{n}}$ or $\frac{-1}{\sqrt{n}}$) and an SoS lower bound for this restricted version of the problem is
stronger than when $v$ can be an arbitrary vector from $\RR^n$.

As we saw in \cref{chap: main_results}, the Sherrington--Kirkpatrick (SK) problem comes from the spin-glass model
in statistical physics~\cite{SK76}.

\begin{definition}[Sherrington-Kirkpatrick problem]\label{def:prob:sk}
	Given $W \sim \GOE(n)$, compute
	\[
	\OPT(W) \defeq \max_{x \in \{\pm 1\}^n} x^\T W x.
	\]
\end{definition}

The Planted Boolean Vector problem was introduced by
Mohanty--Raghavendra--Xu \cite{mohanty2020lifting}, where it was called the
``Boolean Vector in a Random Subspace''.

\begin{definition}[Planted Boolean Vector problem]\label{def:prob:pbv}
	Given as input a uniformly random $p$-dimensional subspace $V$ of $\mathbb{R}^n$ in the form of
	a projector $\Pi_V$ onto $V$, compute
	\[
	\OPT(V) \defeq  \frac{1}{n}\max_{b \in \{\pm 1\}^n} b^\T \Pi_V b.
	\]
\end{definition}

\subsection{Graph matrices}
To study $\calM$, we decompose it using the framework of \textit{graph matrices}. Originally developed in the context of the planted clique problem, graph matrices are random matrices whose entries are symmetric functions of an underlying random object -- in our case, the set of vectors $d_1, \dots, d_m$. We take the general presentation and results from~\cite{ahn2016graph}. For our purposes, the following definitions are sufficient.

The graphs that we study have two types of vertices, circles $\circle{}$ and squares $\square{}$. We let $\calC_m$ be a set of $m$ circles labeled 1 through $m$, which we denote by $\circle{1}, \circle{2}, \dots, \circle{m}$, and let $\calS_n$ be a set of $n$ squares labeled 1 through $n$, which we denote by $\square{1}, \square{2}, \dots, \square{n}$. We will work with bipartite graphs with edges between circles and squares, which have positive integer labels on the edges. When there are no multiedges (the graph is simple), such graphs are in one-to-one correspondence with Fourier characters on the vectors $d_u$. An edge between $\circle{u}$ and $\square{i}$ with label $l$ represents $h_{l}(d_{u,i})$ where $\{h_k\}$ is the Fourier basis (e.g. Hermite polynomials).

\[ \text{simple graph with labeled edges} \qquad \Longleftrightarrow \qquad \displaystyle\prod_{\substack{\circle{u} \in \calC_m,\\ \square{i} \in \calS_n}} h_{l(\circle{u}, \square{i})}(d_{u,i}) \]

An example of a Fourier polynomial as a graph with labeled edges is given in~\cref{fig:fourier_graph}. Unlabeled edges are implicitly labeled 1.
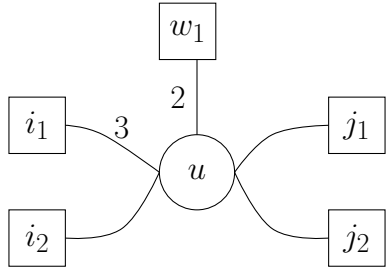
\begin{figure}[h!]
	\centering
	\begin{tikzpicture}[scale=0.5,every node/.style={scale=0.5}]
		\draw  (-2,3) rectangle node {\huge $i_1$}(-0.5,1.5) node (v5) {};
		\draw  (3,1) ellipse (1 and 1) node {\huge $u$};
		\draw  (6.5,3) rectangle node (v10) {\huge $j_1$} (8,1.5);
		\draw  (-2,0) rectangle node {\huge $i_2$} (-0.5,-1.5);
		\draw  (6.5,0) rectangle node {\huge $j_2$} (8,-1.5);
		\node (v1) at (-0.5,3) {};
		\node (v4) at (-0.5,2.25) {};
		\node (v6) at (6.5,2.25) {};
		\node (v8) at (-0.5,-0.75) {};
		\node (v9) at (6.5,-0.75) {};
		\node at (-0.5,0) {};
		\node (v2) at (2,1) {};
		\node at (2,1.5) {};
		\node (v7) at (6.5,1.5) {};
		\node (v3) at (4,1) {};
		\draw  plot[smooth, tension=.7] coordinates {(v3)};
		\draw  plot[smooth, tension=.7] coordinates {(v3)};
		\draw  plot[smooth, tension=.7] coordinates {(v2) (0.5,2) (v4)};
		\node at (1,2.2) {\huge $3$};
		\draw  plot[smooth, tension=.7] coordinates {(v3) (5,2) (v6)};
		\draw  plot[smooth, tension=.7] coordinates {(v3)};
		\draw  plot[smooth, tension=.7] coordinates {(v2) (1,-0.5) (v8)};
		\draw  plot[smooth, tension=.7] coordinates {(v3) (5,-0.5) (v9)};
		\node at (6.5,3) {};
		\draw  plot[smooth, tension=.7] coordinates {(v3)};
		\draw  plot[smooth, tension=.7] coordinates {(v3)};
		\draw  plot[smooth, tension=.7] coordinates {(v10)};
		\node at (6.5,3) {};
		\draw  (2,5.5) rectangle node {\huge $w_1$} (3.5,4);
		\node (v11) at (2.5,4) {};
		\node (v13) at (3,4) {};
		\node (v12) at (3,2) {};
		\draw  plot[smooth, tension=.7] coordinates {(v12) (3,3) (3,4)};
		\node at (2.5,3) {\huge $2$};
	\end{tikzpicture}
	\caption{The Fourier polynomial $h_3(d_{u,i_1})h_1(d_{u,i_2})h_2(d_{u,w_1})h_1(d_{u,j_1})h_1(d_{u,j_2})$ represented as a graph.}
	\label{fig:fourier_graph}
\end{figure}

Define the degree of a vertex $v$,  denoted $\deg(v)$, to be the sum of the labels incident to $v$, and $\abs{E}$ to be the sum of all labels. For
intuition it is mostly enough to work with simple graphs, in which case these quantities make sense as the edge multiplicities in an implicit multigraph.

\begin{definition}[Proper]
	We say an edge-labeled graph is \textit{proper} if it has no multiedges.
\end{definition}
The definitions allow for ``improper'' edge-labeled multigraphs which simplify multiplying graph matrices (\cref{sec:single-spider}).

\begin{definition}[Matrix indices]
	A \textit{matrix index} is a set $A$ of elements from $\calC_m \cup \calS_n$.
\end{definition}
We let $A(\square{i})$ or $A(\circle{u})$ be 0 or 1 to indicate if the vertex is in $A$.

\begin{definition}[Ribbons]\label{def:ribbon}
	A \textit{ribbon} is an undirected, edge-labeled graph $R$ given by $R = (V(R), E(R), A_R, B_R)$, where $V(R) \subseteq \calC_m\cup \calS_n$ and $A_R, B_R$ are two matrix indices (possibly not disjoint) with $A_R, B_R \subseteq V(R)$, representing two distinguished sets of vertices. Furthermore, all edges in $E(R)$ go between squares and circles.
\end{definition}
We think of $A_R$ and $B_R$ as being the ``left'' and ``right'' sides of $R$, respectively. We also define the set of ``middle vertices'' $C_R \defeq V(R) \setminus (A_R \cup B_R)$. If $e \not\in E(R)$, then we define its label $l(e) = 0$. We also abuse notation and write $l(\square{i}, \circle{u})$ instead of $l(\{\square{i}, \circle{u}\})$.

Akin to the picture above, each ribbon corresponds to a Fourier polynomial.
This Fourier polynomial lives inside a single entry of the matrix $M_R$.
In the definition below, the $h_k(x)$ are the Fourier basis corresponding to the respective setting. In the Gaussian case, they are the (unnormalized) Hermite polynomials, and in the boolean case, they are just the parity function, represented by
\[h_0(x) = 1, \qquad h_1(x) = x, \qquad h_k(x) = 0 \;\; (k \geq 2) \]

\begin{definition}[Matrix for a ribbon]\label{def:ribbon-matrix}
	The matrix $M_R$ has rows and columns indexed by subsets of $\calC_m~\cup~\calS_n$, with a single nonzero entry defined by
	\[M_R[I, J] = \left\{\begin{array}{lr}
		\displaystyle\prod_{\substack{e \in E(R), \\ e = \{\square{i}, \circle{u}\}}} h_{l(e)}(d_{u,i}) &  I = A_R, J = B_R\\
		0 & \text{Otherwise}
	\end{array}\right. \]
\end{definition}

Next we describe the shape of a ribbon, which is essentially the ribbon when we have forgotten all the vertex labels and retained only the graph structure and the distinguished sets of vertices.
\begin{definition}[Index shapes]
	An \textit{index shape} is a set $U$ of formal variables. Furthermore, each variable is labeled as either a ``circle'' or a ``square''.
\end{definition}
We let $U(\square{i})$ and $U(\circle{u})$ be either 0 or 1 for whether $\square{i}$ or $\circle{u}$, respectively, is in $U$.

\begin{definition}[Shapes]\label{def:shape}
	A \textit{shape} is an undirected, edge-labeled graph $\al$ given by $\alpha = (V(\alpha), E(\alpha), U_\alpha, V_\alpha)$ where $V(\alpha)$ is a set of formal variables, each of which is labeled as either a ``circle'' or a ``square''. $U_\alpha$ and $V_\alpha$ are index shapes (possibly with variables in common) such that $U_\alpha, V_\alpha \subseteq V(\alpha)$. The edge set $E(\alpha)$ must only contain edges between the circle variables and the square variables.
\end{definition}

We'll also use $W_\alpha \defeq V(\alpha) \setminus (U_\alpha \cup V_\alpha)$ to denote the ``middle vertices'' of the shape.

\begin{remk}
	We will abuse notation and use $\square{i}, \square{j}, \circle{u}, \circle{v}, \ldots$ for both the vertices of ribbons and the vertices of shapes. If they are ribbon vertices, then the vertices are elements of $\calC_m\cup\calS_n$ and if they are shape vertices, then they correspond to formal variables with the appropriate type.
\end{remk}

\begin{definition}[Trivial shape]
	Define a shape $\alpha$ to be trivial if $U_\alpha = V_\alpha$, $W_\alpha = \emptyset$ and $E(\alpha) = \emptyset$.
\end{definition}

\begin{definition}[Transpose of a shape]
	For a shape $\alpha = (V(\alpha), E(\alpha), U_\alpha, V_\alpha)$, its transpose is defined
	to be the shape $\alpha^{\T} = (V(\alpha), E(\alpha), V_\alpha, U_\alpha)$.
\end{definition}

For a shape $\alpha$ and an injective map $\sigma :
V(\alpha) \to \calC_m \cup \calS_n$, we define the
realization $\sigma(\alpha)$ as a ribbon in the natural
way, by labeling all the variables using the map
$\sigma$. We also require $\sigma$ to be
type-preserving i.e. it takes square variables to $\calS_n$ and circle variables to $\calC_m$.
The ribbons that result are referred to as \textit{ribbons of shape $\alpha$}; notice that this partitions the set of all ribbons according to their shape\footnote{Partitions up to equality of shapes, where two shapes are equal if there is a type-preserving bijection between their variables that converts one shape to the other. When we operate on sets of shapes below, we implicitly use each distinct shape only once.}\footnote{Note that in our definition two realizations of a shape may give the same ribbon.}.

Finally, given a shape $\alpha$, the graph matrix $M_\alpha$ consists of all Fourier characters for ribbons of shape $\alpha$.
\begin{definition}[Graph matrices]\label{def:graph-matrix}
	Given a shape $\alpha = (V(\alpha), E(\alpha), U_\alpha, V_\alpha)$, the \textit{graph matrix} $M_\alpha$ is
	\[M_\alpha = \displaystyle\sum_{R \text{ is a ribbon of shape }\alpha} M_R\]
\end{definition}

The moment matrix for PAP will turn out to be defined using graph matrices $M_\alpha$ whose left and right sides only have square vertices, and no circles. However, in the course of the analysis we will factor and multiply graph matrices with circle vertices in the left or right.

\subsection{Norm bounds}
Similar to the norm bounds for graph matrices with only a single type of vertex (see \cref{chap: efron_stein}), the spectral norm of a graph matrix in our setting is determined, up to logarithmic factors, by relatively simple combinatorial properties of the graph. For a subset $S \subseteq \calC_m \cup \calS_n$, we define the weight $w(S)~\defeq~(\#\text{ circles in }S)\cdot \log_n(m)+ (\#\text{ squares in }S)$. Observe that $n^{w(S)} = m^{\# \text{ circles in }S}\cdot n^{\#\text{ squares in }S}$.

\begin{definition}[Minimum vertex separator]
	For a shape $\alpha$, a set $S_{\min}$ is a minimum vertex separator if all paths from $U_\alpha$ to $V_\alpha$ pass through $S_{\min}$ and $w(S_{\min})$ is minimized over all such separating sets.
\end{definition}

Let $W_{iso}$ denote the set of isolated vertices in $W_\alpha$. Then essentially the following norm bound holds for all shapes $\alpha$ with high probability (a formal statement can be found in~\cref{app:norm_bounds}):
\[\norm{M_\alpha} \leq  \widetilde\bigoh\left(n^{\frac{w(V(\alpha)) - w(S_{\min}) + w(W_{iso})}{2}}\right)\]

In fact, the only probabilistic property required of the inputs $d_1, \dots, d_m$ by our proof is that the above norm bounds hold for all shapes that arise in the analysis.
We henceforth assume that the norm bounds in~\cref{lem:gaussian-norm-bounds} (for the Gaussian case) and~\cref{lem:norm-bounds} (for the boolean case) hold.

\section{Proof Strategy}\label{sec:strategy}

Now we explain in more detail the ideas for the Planted Affine Planes
lower bound. Towards the proof of~\cref{theo:sos-bounds}, fix a
constant $\eps > 0$ and a random instance $d_1, \dots, d_m$ with
$n \leq m \leq n^{3/2-\eps}$. We will construct a pseudoexpectation operator
and show that it is PSD up to degree $D = 2\cdot n^\delta$
with high probability.

We start by
pseudocalibrating to obtain a pseudoexpectation operator $\pE$. The
operator $\pE$ will exactly satisfy the ``booleanity" constraints
``$v_i^2 = \frac{1}{n}$" though it may not exactly satisfy the
constraints ``$\ip{v}{d_u}^2 = 1$" due to truncation error in the
pseudocalibration. Taking the truncation parameter $n^{\tau}$ to be larger than the degree $D$ of the SoS solution, i.e., $\delta \ll \tau$, the truncation error is small enough that we can
round $\pE$ to a nearby $\pE'$ that exactly satisfies the
constraints. This is formally accomplished by viewing
$\pE \in \RR^{\binom{[n]}{\leq D}}$ as a vector and expressing the
constraints as a matrix $Q$ such that $\pE$ satisfies the constraints
iff it lies in the null space of $Q$. The choice of $\pE'$ is then the
projection of $\pE$ to $\nullspace(Q)$. The end result is that we
construct a moment matrix $\calM_{fix} = \calM + \calE$ that exactly
satisfies the constraints such that $\norm{\calE}$ is tiny. For the sake of brevity, we omit this technicality in this work, see \cite{sklowerbounds} for the details.

After performing pseudocalibration, in both settings, we will have
essentially the graph matrix decomposition
\[
\calM = \sum_{\text{shapes }\alpha} \lambda_\alpha M_\alpha = \displaystyle\sum_{\substack{\text{shapes }\alpha:\\ \deg(\square{i}) + U(\square{i}) + V(\square{i}) \text{ even},\\ \deg(\circle{u})\text{ even}}} \frac{1}{n^{\frac{\abs{U_\alpha} + \abs{V_\alpha}}{2}}}\cdot \left(\prod_{\circle{u} \in V(\alpha)} h_{\deg(\circle{u})}(1)\right) \cdot \frac{M_\alpha}{n^{\abs{E(\alpha)}/2}}
\]
Here $h_k(1)$ is in both settings the $k$-th Hermite polynomial, evaluated on 1.

In this decomposition of $\calM$, the trivial shapes will be the
dominant terms which we will use to bound the other terms. Recall that
a shape $\alpha = (V(\alpha), E(\alpha), U_\alpha, V_\alpha)$ is
trivial if $U_\alpha = V_\alpha$, $W_\alpha = \emptyset$ and
$E(\alpha) = \emptyset$. These shapes contribute scaled identity
matrices on different blocks of the main diagonal of $\calM$, with
trivial shape $\alpha$ contributing an identity matrix with
coefficient $n^{-\abs{U_\alpha}}$. Two trivial shapes are illustrated
in~\cref{fig:trivial_shapes}.

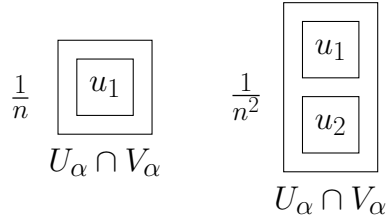
\begin{figure}[h!]
  \centering
  \begin{tikzpicture}[scale=0.5,every node/.style={scale=0.5}]
   \draw  (-1,1) rectangle node {\huge $u_1$} (0.5,-0.5);
   \draw  (-1.5,1.5) rectangle (1,-1);
   \node at (-0.25,-1.75) {\huge $U_{\alpha} \cap V_{\alpha}$};
   \node at (-2.5,0) {\huge $\frac{1}{n}$};
   \draw  (4.5,2.5) rectangle (7,-2);
   \draw  (5,2) rectangle node {\huge $u_1$} (6.5,0.5);
   \draw  (5,0) rectangle node {\huge $u_2$} (6.5,-1.5);
   \node at (3.5,0) {\huge $\frac{1}{n^2}$};
   \node at (5.75,-2.75) {\huge $U_{\alpha} \cap V_{\alpha}$};
  \end{tikzpicture}
  \caption{Two examples of trivial shapes.}
  \label{fig:trivial_shapes}
\end{figure}

Let $\calM_{\text{triv}}$ be this diagonal matrix of trivial shapes in
the above decomposition of $\calM$. To prove that $\calM \psdgeq 0$,
we attempt the simple strategy of showing that the norm of all other
terms can be ``charged'' against this diagonal matrix
$\calM_{\text{triv}}$. For several shapes this strategy is indeed
viable. To illustrate, let's consider one such shape $\alpha$ depicted
in~\cref{fig:non_spider}.

\begin{figure}[h!]
\centering
\begin{tikzpicture}[scale=0.5,every node/.style={scale=0.5}]
  \draw  (-6.5,2) rectangle node {\huge $u_1$} (-5,0.5);
  \draw  (2,6) rectangle node {\huge $w_1$} (3.5,4.5);
  \draw  (10.5,2) rectangle node {\huge $v_1$} (12,0.5);
  \draw  (2,-0.25) rectangle node {\huge $w_3$} (3.5,-1.75);
  \draw  (2,1.25) rectangle node {\huge $w_2$} (3.5,2.75);
  \draw  (-1.5,2) ellipse (1 and 1) node {\huge $u$};
  \draw  (7,2) ellipse (1 and 1) node {\huge $u'$};
  \node (v1) at (2,2) {};
  \node (v5) at (2,5.25) {};
  \node (v3) at (3.5,5.25) {};
  \node (v6) at (2,-1) {};
  \node (v2) at (-0.5,2) {};
  \node (v4) at (-0.5,2) {};
  \draw  plot[smooth, tension=.7] coordinates {(v1)};
  \draw  plot[smooth, tension=.7] coordinates {(v1) (1,2) (v2)};
  \draw  plot[smooth, tension=.7] coordinates {(v4)};
  \draw  plot[smooth, tension=.7] coordinates {(v4) (1,4.5) (v5)};
  \draw  (-7,3) rectangle (-4.5,-0.5);
  \node at (-5.5549,-1.061) {\huge $U_{\alpha}$};
  \draw  (10,3) rectangle (12.5,-0.5);
  \node at (11.3364,-1.1181) {\huge $V_{\alpha}$};
  \node at (2,5.25) {};
  \node at (2,2) {};
  \node (v11) at (3.5,2) {};
  \node (v13) at (3.5,-1) {};
  \node (v8) at (-2.5,2) {};
  \node (v12) at (6,2) {};
  \node (v10) at (8,2) {};
  \node (v7) at (-5,1.25) {};
  \node (v9) at (10.5,1.25) {};
  \draw (v7);
  \draw  plot[smooth, tension=.7] coordinates {(v2)};
  \draw  plot[smooth, tension=.7] coordinates {(-0.5,2) (1,-0.25) (v6)};
  \draw  plot[smooth, tension=.7] coordinates {(v3) (4.5,4.5) (6,2)};
  \draw  plot[smooth, tension=.7] coordinates {(v10) (v9)};
  \draw  plot[smooth, tension=.7] coordinates {(v12) (4.5,2) (v11)};
  \draw  plot[smooth, tension=.7] coordinates {(v12) (4.5,-0.25) (v13)};
  \draw  plot[smooth, tension=.7] coordinates {(v7) (v8)};
  \end{tikzpicture}
  \caption{Picture of basic non-spider shape $\alpha$.}
  \label{fig:non_spider}
\end{figure}
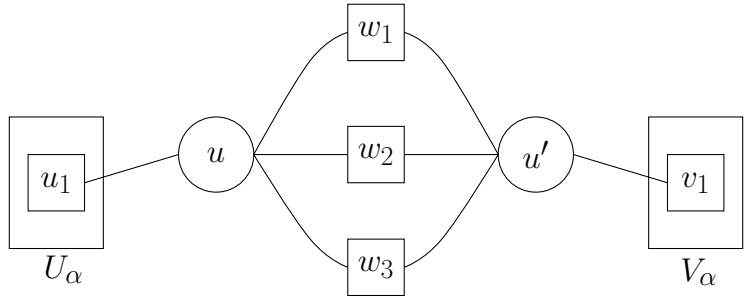

This graph matrix has $\abs{\lambda_\alpha}
= \Theta(\frac{1}{n^5})$. Using the graph matrix norm bounds, with
high probability the norm of this graph matrix is $\tilde{O}({n^2}m)$:
there are four square vertices and two circle vertices which are not
in the minimum vertex separator. Thus, for this shape $\alpha$, with
high probability $\abs{\lambda_{\alpha}}\norm{M_{\alpha}}$ is
$\tilde{O}\left(\frac{m}{n^3}\right)$ and thus
$\lambda_{\alpha}M_{\alpha} \preceq \frac{1}{n}Id$ (which is the
multiple of the identity appearing in the corresponding block of
$\calM_{\text{triv}}$).

Unfortunately, some shapes
$\alpha$ that appear in the decomposition have $\norm{\lambda_\alpha
M_\alpha}$ too large to be charged against
$\calM_{\textup{triv}}$. These are shapes with a certain substructure
(actually the same structure that appears in the matrix $Q$ used to
project the pseudoexpectation operator!) whose norms cannot be handled
by the preceding argument, and which we denote \textit{spiders}.  The
following graph depicts one such \textit{spider} shape (and also
motivates this terminology):
\begin{figure}[h!]
  \centering
  \begin{tikzpicture}[scale=0.5,every node/.style={scale=0.5}]
    \draw  (-4,4) rectangle node {\huge $u_1$} (-2.5,2.5);
    \draw  (-4,-0.5) rectangle node {\huge $u_2$} (-2.5,-2);
    \draw  (5.5,4) rectangle node {\huge $v_1$} (7,2.5);
    \draw  (5.5,-0.5) rectangle node {\huge $v_2$} (7,-2);
   \draw  (1.5,1) ellipse (1 and 1) node {\huge $u$};
   \node (v1) at (-2.5,3.25) {};
   \node (v5) at (5.5,3.25) {};
  \node (v3) at (-2.5,-1.25) {};
  \node (v6) at (5.5,-1.25) {};
  \node (v2) at (0.5,1) {};
  \node (v4) at (2.5,1) {};
  \draw  plot[smooth, tension=.7] coordinates {(v1)};
  \draw  plot[smooth, tension=.7] coordinates {(v1) (-0.5,3) (v2)};
  \draw  plot[smooth, tension=.7] coordinates {(v3) (-0.5,-1) (v2)};
  \draw  plot[smooth, tension=.7] coordinates {(v4)};
  \draw  plot[smooth, tension=.7] coordinates {(v4) (3.5,3) (v5)};
  \draw  plot[smooth, tension=.7] coordinates {(v4) (3.5,-1) (v6)};
  \draw  (-4.5,5) rectangle (-2,-3.3727);
  \node at (-3.0549,-3.9337) {\huge $U_{\alpha}$};
  \draw  (4.8834,5.3694) rectangle (7.5016,-3.1943);
  \node at (6.338,-3.8124) {\huge $V_{\alpha}$};
\end{tikzpicture}
  \caption{Picture of basic spider shape $\alpha$.}
\end{figure}
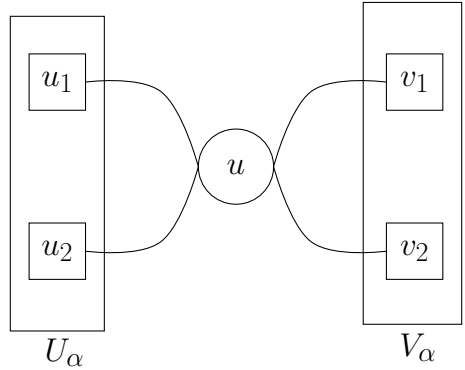

The norm $\norm{\lambda_\alpha M_\alpha}$ of this graph is
$\widetilde\bigomega(\frac{1}{n^{2}})$, as can be easily estimated through the
norm bounds (the coefficient is $\lambda_\alpha = \frac{-2}{n^4}$, the
minimum vertex separator is $\circle{u}$, and there are no isolated
vertices). This is too large to bound against $\frac{1}{n^2}Id$, which is the coefficient of $M_\text{triv}$ on this spider's block.

To skirt this and other spiders, we restrict ourselves to
vectors $x \perp \nullspace(M)$, and observe that this spider $\alpha$ satisfies $x^\T M_{\alpha} \approx 0$. To be more precise, consider the following argument. Consider the two shapes in~\cref{fig:betas}, $\beta_1$ and $\beta_2$ (take note of the label 2 on the edge in $\beta_2$).

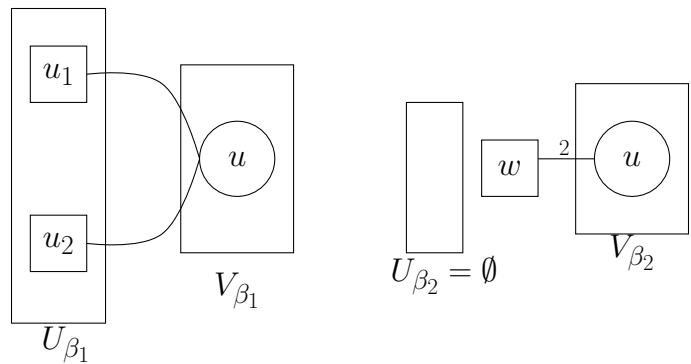
\begin{figure}[h!]
\centering
\begin{tikzpicture}[scale=0.5,every node/.style={scale=0.5}]
\draw  (-4,4) rectangle node {\huge $u_1$} (-2.5,2.5);
\draw  (-4,-0.5) rectangle node {\huge $u_2$} (-2.5,-2);
\draw  (1.5,1) ellipse (1 and 1) node {\huge $u$};
\draw  (12,1) ellipse (1 and 1) node {\huge $u$};
\node (v1) at (-2.5,3.25) {};
\node (v3) at (-2.5,-1.25) {};
\node (v2) at (0.5,1) {};
\node (v4) at (11,1) {};
\draw  plot[smooth, tension=.7] coordinates {(v1)};
\draw  plot[smooth, tension=.7] coordinates {(v1) (-0.5,3) (v2)};
\draw  plot[smooth, tension=.7] coordinates {(v3) (-0.5,-1) (v2)};
\draw  plot[smooth, tension=.7] coordinates {(v4)};
\draw  (-4.5,5) rectangle (-2,-3.3727);
\node at (-3.0549,-3.9337) {\huge $U_{\beta_1}$};
\draw  (0,3.5) rectangle (3,-1.5);
\node at (1.5,-2.5) {\huge $V_{\beta_1}$};
\draw  (10.5,3) rectangle  (13.5,-1);
\node at (12,-1.5) {\huge $V_{\beta_2}$};
\draw  (6,2.5) rectangle (7.5,-1.5);
\draw  (9.5,1.5) rectangle node {\huge $w$} (8,0);
\draw  plot[smooth, tension=.7] coordinates {(v4)};
\draw  plot[smooth, tension=.7] coordinates {(v4)};
\node (v5) at (9.5,1) {};
\draw  plot[smooth, tension=.7] coordinates {(v4) (v5)};
\node at (7,-2.1) {\huge $U_{\beta_2} = \emptyset$};
\node at (10.2,1.3) {\Large $2$};
\end{tikzpicture}
\caption{Picture of shapes $\beta_1$ and $\beta_2$.}
\label{fig:betas}
\end{figure}

We claim that every column of the matrix $2M_{\beta_1} + \frac{1}{n}M_{\beta_2}$
is in the null space of $\calM$. There are $m$ nonzero columns indexed
by assignments to $V$, which can be a single circle
$\circle{1}, \circle{2}, \dots, \circle{m}$. The nonzero rows are
$\emptyset$ in $\beta_2$ and $\{\square{i}, \square{j}\}$ for $i \neq j$ in $\beta_1$. Fixing $I \subseteq [n]$, entry
$(I, \circle{u})$ of the product matrix $\calM(2M_{\beta_1} + \frac{1}{n}M_{\beta_2})$ is
\begin{align*}
2& \displaystyle\sum_{i < j}\pE [v^I v_i v_j] \cdot d_{ui} d_{uj} + \frac{1}{n}\pE[v^I] \cdot \sum_{i}(d_{ui}^2 - 1)\\
&= 2\displaystyle\sum_{i < j}\pE [v^I v_i v_j] \cdot d_{ui} d_{uj} + \pE[v^Iv_i^2] \cdot \sum_{i}d_{ui}^2 - \pE[v^I] && (\pE \text{ satisfies ``}v_i^2 = \frac{1}{n}")\\
&= \sum_{i,j} \pE[v^I v_i v_j] d_{ui}d_{uj} - \pE[v^I] \\
&= \pE[v^I(\ip{v}{d_u}^2 - 1)]\\
&= 0 && (\pE \text{ satisfies ``}\ip{v}{d_u}^2 = 1")
\end{align*}
In words, the constraint ``$\ip{v}{d_u}^2 = 1 $'' creates a shape
$2\beta_1 + \frac{1}{n}\beta_2$ that lies in the null space of the moment
matrix. On the other hand, we can approximately factor the spider
$\alpha$ across its central vertex, and when we do so, the shape
$\beta_1$ appears on the left side.
\begin{figure}[h!]
\centering
\begin{tikzpicture}[scale=0.5,every node/.style={scale=0.5}]
\draw  (4.25,-5.5) rectangle node {\huge $u_1$} (5.75,-7);
\draw  (4.25,-10) rectangle node {\huge $u_2$} (5.75,-11.5);
\draw  (9.75,-8.5) ellipse (1 and 1) node {\huge $u$};
\draw  (13.25,-8.5) ellipse (1 and 1) node {\huge $u$};
\node (v1) at (5.75,-6.25) {};
\node (v3) at (5.75,-10.75) {};
\node (v2) at (8.75,-8.5) {};
\node (v4) at (12.25,-8.5) {};
\draw  plot[smooth, tension=.7] coordinates {(v1)};
\draw  plot[smooth, tension=.7] coordinates {(v1) (7.75,-6.5) (v2)};
\draw  plot[smooth, tension=.7] coordinates {(v3) (7.75,-10.5) (v2)};
\draw  plot[smooth, tension=.7] coordinates {(v4)};
\draw  (3.75,-4.5) rectangle (6.25,-12.8727);
\node at (5.1951,-13.4337) {\huge $U_{\beta_1}$};
\draw  (16.75,-4.5) rectangle (19.25,-13);
\draw  (17.25,-5.5) rectangle node {\huge $u_1$} (18.75,-7);
\draw  (17.25,-10) rectangle node {\huge $u_2$} (18.75,-11.5);
\node (v5) at (14.25,-8.5) {};
\node (v6) at (17.25,-6.25) {};
\node (v7) at (17.25,-10.75) {};
\draw  plot[smooth, tension=.7] coordinates {(v5)};
\draw  plot[smooth, tension=.7] coordinates {(v5) (15.25,-6.5) (v6)};
\draw  plot[smooth, tension=.7] coordinates {(v5) (15.25,-10.5) (v7)};
\node at (18.257,-13.5621) {\huge $U_{\beta_1}$};
\node at (11.5,-8.5) {\huge $\times$};
\node at (20.75,-9) {\Huge $\approx$};
\draw  (23,-5.5) rectangle node {\huge $u_1$} (24.5,-7);
\draw  (23,-10) rectangle node {\huge $u_2$} (24.5,-11.5);
\draw  (32.5,-5.5) rectangle node {\huge $v_1$} (34,-7);
\draw  (32.5,-10) rectangle node {\huge $v_2$} (34,-11.5);
\draw  (28.5,-8.5) ellipse (1 and 1) node {\huge $u$};
\node (v1) at (24.5,-6.25) {};
\node (v5) at (32.5,-6.25) {};
\node (v3) at (24.5,-10.75) {};
\node (v6) at (32.5,-10.75) {};
\node (v2) at (27.5,-8.5) {};
\node (v4) at (29.5,-8.5) {};
\draw  plot[smooth, tension=.7] coordinates {(v1)};
\draw  plot[smooth, tension=.7] coordinates {(v1) (26.5,-6.5) (v2)};
\draw  plot[smooth, tension=.7] coordinates {(v3) (26.5,-10.5) (v2)};
\draw  plot[smooth, tension=.7] coordinates {(v4)};
\draw  plot[smooth, tension=.7] coordinates {(v4) (30.5,-6.5) (v5)};
\draw  plot[smooth, tension=.7] coordinates {(v4) (30.5,-10.5) (v6)};
\draw  (22.5,-4.5) rectangle (25,-13);
\node at (23.9451,-13.561) {\huge $U_{\alpha}$};
\draw  (32,-4.5) rectangle (34.5,-13);
\node at (33.5,-13.5) {\huge $V_{\alpha}$};
\draw  (8.5,-6) rectangle (11,-11);
\node at (10,-11.5) {\huge $V_{\beta_1}$};
\draw  (12,-6) rectangle (14.5,-11);
\node at (13.5,-11.5) {\huge $V_{\beta_1}$};
\end{tikzpicture}
\caption{Approximation $\beta_1 \times \beta_1^\T \approx \alpha$.}
\end{figure}
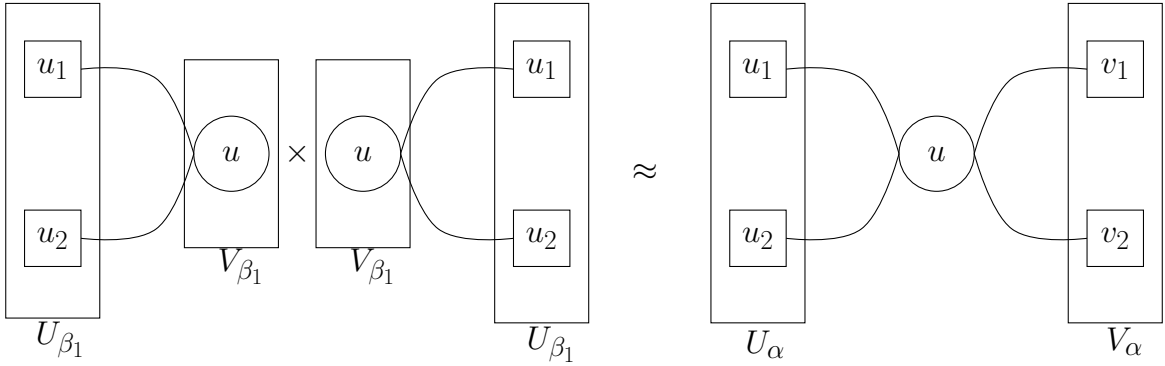
Therefore $M_\alpha \approx M_{\beta_1} M_{\beta_1}^\T \approx
(M_{\beta_1} + \frac{1}{2n}M_{\beta_2}) M_{\beta_1}^\T$. The columns of
the matrix $M_{\beta_1} + \frac{1}{2n}M_{\beta_2}$ are in the null
space of $\calM$, so for $x \perp \nullspace(\calM)$ we have $x^\T
M_\alpha \approx 0$.
More formally, we are able to find coefficients $c_\beta$ so that all
columns of the matrix
\[A = M_\alpha + \displaystyle\sum_{\beta} c_\beta M_\beta \]
are in $\nullspace(\calM)$. We then observe the following fact:
\begin{fact}\label{fact:null-space}
	If $x \perp \nullspace(\calM)$ and $\calM A = 0$, then $x^\T(AB + \calM)x = x^\T(B^\T A^\T + \calM)x= x^\T \calM x$.
\end{fact}
Using the fact, we can freely add multiples of $A$ to $\calM$ without
changing the action of $\calM$ on $\nullspace(\calM)^\perp$. A
judicious choice is to subtract $\lambda_\alpha A$ which will ``kill''
the spider from $\calM$. Doing this for all spiders, we produce a
matrix whose action is equivalent on $\nullspace(\calM)^\perp$, and
which has high minimum eigenvalue by virtue of the fact that it has no
spiders, showing that $\calM$ is PSD.
The catch is two-fold: first, the coefficients $c_\beta$ may
contribute to the coefficients on the non-spiders; second, the further
intersection terms $M_\beta$ may themselves be spiders ( though they
will always have fewer square vertices than $\alpha$). Thus we must
recursively kill these spiders, until there are no spiders remaining
in the decomposition of $\calM$. The resulting matrix has some new
coefficients on the non-spiders
\[ \calM' = \displaystyle\sum_{\text{non-spiders }\beta} \lambda_\beta' M_\beta. \]
We must bound the accumulation on the coefficients
$\lambda_\beta'$. We do this by considering the \textit{web} of
spiders and non-spiders created by each spider and using bounds on the $c_\beta$
and $\lambda_\alpha$ to argue that the contributions do not blow up, via an interesting charging scheme that exploits the structure of these graphs.

\section{Pseudocalibration}\label{sec:pseudo_calib}

As we saw in \cref{chap: sos}, to be able to apply the pseudocalibration technique
to an average-case feasibility problem, in our case the PAP
problem, one needs to design a planted distribution supported on
feasible instances. This is done
in~\cref{subsec:pap:dist}.
In \cref{subsec:pseudo_calib_technique}, we recall the precise details in applying pseudocalibration.
Then we pseudocalibrate in the Gaussian~(\cref{subsec:calib:gauss}) and
boolean~(\cref{subsec:calib:bool}) settings.

\subsection{PAP planted distribution}\label{subsec:pap:dist}

We formally define the random and the planted distributions for the
Planted Affine Planes problem in the Gaussian and boolean
settings. These two (families of) distributions are required by the
pseudocalibration machinery in order to define a candidate
pseudoexpectation operator $\pE$. For the Gaussian setting, we have
the following distributions.

\begin{definition}[Gaussian PAP distributions]\label{def:prob:pap:gauss:dist}
  The Gaussian PAP distributions are as follows.
  \begin{enumerate}
      \item (Random distribution) $m$ i.i.d. vectors $d_u \sim \gauss{0}{I}$.
      \item (Planted distribution) A vector $v$ is sampled uniformly from $\left\{\pm \frac{1}{\sqrt{n}}\right\}^n$, as well as signs $b_u \unif \{\pm 1\}$,
             and $m$ vectors $d_u$ are drawn from $\mathcal{N}(0, I)$ conditioned on $\ip{d_u}{v} = b_u$.
  \end{enumerate}
\end{definition}

For the boolean setting, we have the following distributions.
\begin{definition}[Boolean PAP distributions]\label{def:prob:pap:bool:dist}
  The boolean PAP distributions are as follows
  \begin{enumerate}
      \item (Random distribution) $m$ i.i.d. vectors $d_u \unif \{-1,+1\}^n$.
      \item (Planted distribution) A vector $v$ is sampled uniformly from $\left\{\pm \frac{1}{\sqrt{n}}\right\}^n$, as well as signs $b_u \unif \{\pm 1\}$, and $m$ vectors $d_u$ are drawn from $\left\{\pm 1\right\}^n$ conditioned on $\ip{d_u}{v} = b_u$.
  \end{enumerate}
\end{definition}

\subsection{Pseudocalibration technique}\label{subsec:pseudo_calib_technique}

We will use the shorthand $\E_{\text{ra}}$ and $\E_{\text{pl}}$ for
the expectation under the random and planted distributions.
Pseudocalibration gives a method for constructing a candidate
pseudoexpectation operator $\pE$.
The idea behind pseudocalibration is that
$\E_{\text{ra}} \pE f(v)$ should match with $\E_{\text{pl}} f(v)$ for
every low-degree test of the data $t = t(d) = t(d_1, \dots, d_m)$,
\[\E_{\text{ra}} t(d) \pE f(v) = \E_{\text{pl}} t(d) f(v) .\]
When pseudocalibrating, one can freely choose the ``outer'' basis in
which to express the polynomial $f(v)$, as well as the ``inner'' basis
of low-degree tests which should agree with the planted
distribution. Though we attempted to use alternate bases to simplify
the analysis, ultimately we opted for the standard choice of bases: a
Fourier basis for the inner basis in each setting (Hermite functions
for the Gaussian setting, parity functions for the boolean setting),
and the coordinate basis $v^I$ for the outer basis.

When the inner basis is orthonormal under the random distribution (as a Fourier basis is), the pseudocalibration condition
gives a formula for the coefficients of $\pE f(v)$ in the orthonormal basis (though it only gives the coefficients of the low-degree functions $t(d)$). Concretely, letting the inner basis be indexed by $\alpha \in \calF$, as a function of $d$ the pseudocalibration condition enforces
\[ \pE f(v) = \displaystyle\sum_{\substack{\alpha \in \calF: \\ \abs{\alpha} \leq n^\tau}} \left( \E_{\text{pl}} t_\alpha(d) f(v) \right)t_\alpha (d).\]
Here we use ``$\abs{\alpha} \leq n^\tau$'' to describe the set of low-degree tests. The pseudocalibration condition does not prescribe any coefficients for functions $t_\alpha(d)$ with $\abs{\alpha} > n^\tau$ and an economical choice is to set these coefficients to zero.

When pseudocalibrating, our pseudoexpectation operator is guaranteed to be linear, as the expression above is linear in $f$. It is guaranteed to satisfy all constraints of the form ``$f(v) =0$''. It will approximately satisfy constraints of the form ``$f(v, d) = 0$'', though only up to truncation error.

\begin{fact}[Proof in \cite{sklowerbounds}]\label{lem:pE-constraints}
  If $p(v)$ is a polynomial which is uniformly zero on the planted
  distribution, then $\pE[p]$ is the zero function. If $p(v,d)$ is a polynomial which is uniformly zero on the planted distribution, then the only nonzero Fourier coefficients of $\pE[p]$ are those with size between $n^\tau \pm \deg_d(p)$.
\end{fact}

Truncation
introduces a tiny error in the constraints, which we are able to handle in \cite{sklowerbounds}, omitted in this work for brevity.

For the pseudocalibration we truncate to only Fourier coefficients of
size at most $n^\tau$. The relationship between the parameters is $\delta \le c\tau \le c'\eps$ where $c' < c < 1$ are absolute constants. We will assume that they are sufficiently small for all our proofs to go through.

Pseudocalibration also by default does not enforce the condition $\pE[1] = 1$. However, this is easily fixed by dividing the operator by $\pE[1]$. As will be pointed out in~\cref{rmk:pe-one}, w.h.p. in the unnormalized pseudocalibration, $\pE[1] = 1 + \littleoh_n(1)$ and so the error introduced does not impact the statement of any lemmas.

\subsection{Gaussian setting pseudocalibration}\label{subsec:calib:gauss}

We start by computing the pseudocalibration for the Gaussian setting. Here the natural choice of Fourier basis is the Hermite polynomials. Let $\alpha \in (\N^n)^m$ denote a Hermite polynomial index. Define $\alpha! \defeq \prod_{u,i} \alpha_{u,i}!$ and $\abs{\alpha} \defeq \sum_{u, i} \alpha_{u,i}$ and $\abs{\alpha_u} \defeq \sum_i \alpha_{u,i}$. We let $h_\alpha(d_1, \dots, d_m)$ denote an unnormalized Hermite polynomial, so that $h_{\alpha}/\sqrt{\alpha!}$ forms an orthonormal basis for polynomials in the entries of the vectors $d_1, \dots, d_m$, under the inner product $\ip{p}{q} = \E_{d_1, \dots, d_m \sim \calN(0, I)} [p \cdot q]$.

We can view $\alpha$ as an $m\times n$ matrix of natural numbers, and with this view we also define $\alpha^\T \in (\N^m)^n$.
\begin{lemma}\label{lem:gaussian-pseudocal}
For any $I \subseteq [n]$, the pseudocalibration value is
\[\pE v^I = \displaystyle\sum_{\substack{\alpha: \abs{\alpha} \leq n^\tau,\\ \abs{\alpha_u} \text{ even}, \\ \abs{(\alpha^\T)_i} \equiv I_i \; (\mod 2)}} \left(\prod_{u = 1}^m h_{\abs{\alpha_u}}(1) \right)\cdot\frac{1}{n^{\abs{I}/2 + \abs{\alpha}/2}} \cdot\frac{h_{\alpha}(d_1, \dots, d_m)}{\alpha!}. \]
\end{lemma}
In words, the nonzero Fourier coefficients are those which have even row sums, and whose column sums match the parity of $I$.
\begin{proof}
The truncated pseudocalibrated value is defined to be
\[\pE v^I = \displaystyle\sum_{\alpha : \abs{\alpha} \leq n^\tau} \frac{h_{\alpha}(d_1, \dots, d_m)}{\alpha!} \cdot \E_{\text{pl}}[h_{\alpha}(d_1, \dots, d_m) \cdot v^I] \]
So we set about to compute the planted moments. For this computation, the following lemma is crucial. Here, we give a short proof of this lemma using generating functions. For a different combinatorial proof, see \cite{sklowerbounds}.

\begin{lemma}\label{lem:fixed-moments}
Let $\alpha \in \N^n$. When $v$ is fixed and $b$ is fixed (not necessarily $\pm1$) and $d \sim N(0, I)$ conditioned on $\ip{v}{d} = b\norm{v}$,
\[\E_{d}[h_{\alpha}(d)] = \frac{v^\alpha}{\norm{v}^{\abs{\alpha}}} \cdot h_{\abs{\alpha}}(b).\]
\end{lemma}
\begin{proof}
It suffices to prove the claim when $\norm{v} = 1$ since the left-hand side is independent of $\norm{v}$. Express $d = bv + (I - vv^\T)x$ where $x \sim N(0,I)$ is a standard normal variable. Now we want
\[\E_{x \sim N(0,I)} h_{\alpha}\left(bv + (I - vv^\T)x\right). \]
The Hermite polynomial generating function is
\[\displaystyle\sum_{\alpha \in \N^n} \E_{x \sim N(0,I)} h_{\alpha}\left(bv + (I - vv^\T)x\right)\frac{t^{\alpha}}{\alpha!} = \E_x\exp\left(\ip{bv + (I - vv^\T)x}{t} - \frac{\norm{t}_2^2}{2}\right)\]
\[= \int_{\mathbb{R}^n} \frac{1}{(2\pi)^{\frac{n}{2}}} \cdot \exp\left(\ip{bv + (I - vv^\T)x}{t} - \frac{\norm{t}_2^2}{2} - \frac{\norm{x}_2^2}{2}\right) \; dx. \]
Completing the square,
\begin{align*}= &  \int_{\mathbb{R}^n} \frac{1}{(2\pi)^{\frac{n}{2}}} \cdot \exp\left(\ip{bv}{t} - \frac{\ip{v}{t}^2}{2} - \frac{1}{2} \cdot\norm{x- (t - \ip{v}{t}v)}_2^2\right) \; dx \\
=&  \exp\left(\ip{bv}{t} - \frac{\ip{v}{t}^2}{2}\right) \\
=& \exp\left(b\ip{v}{t} - \frac{1}{2} \cdot \ip{v}{t}^2\right).
\end{align*}
How can we Taylor expand this in terms of $t$? The Taylor expansion of $\exp(by - \frac{y^2}{2})$ is $\sum_{i=0}^\infty h_i(b) \frac{y^i}{i!}$. That is, the $i$-th derivative in $y$ of $\exp(by - \frac{y^2}{2})$, evaluated at 0, is $h_i(b)$. Using the chain rule with $y=\ip{v}{t}$, the $\alpha$-derivative in $t$ of our expression, evaluated at 0, is $v^\alpha \cdot h_{\abs{\alpha}}(b)$. This is the expression we wanted when $\norm{v} = 1$, and along with the aforementioned remark about homogeneity in $\norm{v}$ this completes the proof.
\end{proof}

Now we can finish the calculation.
To compute $\E_{\text{pl}}[h_{\alpha}(d_1, \dots, d_m) \cdot v^I]$, marginalize $v$ and the $b_u$ and factor the conditionally independent $b_u$ and $d_u$.
\begin{align*}
    \E_{\text{pl}}[h_{\alpha}(d_1, \dots, d_m) v^I] &= \E_{v, b_u} v^I \prod_{u=1}^m \E_{d}\left[h_{\alpha_u}(d_u) \mid v, b_u\right]\\
    &= \E_{v, b_u} v^I \cdot \prod_{u=1}^m \frac{v^{\alpha_u}}{\norm{v}^{\abs{\alpha_u}}} \cdot h_{\abs{\alpha_u}}(b_u) && (\text{\cref{lem:fixed-moments})}\\
    &= \left(\E_{v} \frac{v^{I + \sum_{u=1}^m \alpha_u}}{\norm{v}^{\sum_{u=1}^m\abs{\alpha_u}}} \right) \cdot \left( \prod_{u=1}^m \E_{b_u}h_{\abs{\alpha_u}}(b_u)\right)
\end{align*}
The Hermite polynomial expectations will be zero in expectation over $b_u$ if the degree is odd, and otherwise $b_u$ is raised to an even power and can be replaced by 1. This requires that $\abs{\alpha_u}$ is even for all $u$. The norm $\norm{v}$ is constantly $1$ and can be dropped. The numerator will be $\frac{1}{n^{\abs{I}/2 + \abs{\alpha}/2}}$ if the parity of every $\abs{(\alpha^\T)_i}$ matches $I_i$, and 0 otherwise. This completes the pseudocalibration calculation.
\end{proof}

We can now write $\calM$ in terms of graph matrices.

\begin{definition}\label{def:calL_valid_shapes}
	Let $\cal{L}$ be the set of all proper shapes $\alpha$ with the following properties
	\begin{itemize}
		\item $U_{\alpha}$ and $V_{\alpha}$ only contain square vertices and $|U_{\alpha}|, |V_{\alpha}| \le n^{\delta}$
		\item $W_\alpha$ has no degree $0$ vertices
		\item $\deg(\square{i}) + U_\alpha(\square{i}) + V_\alpha(\square{i})$ is even for all $\square{i} \in V(\alpha)$
		\item $\deg(\circle{u})$ is even and $\deg(\circle{u}) \ge 4$ for all $\circle{u} \in V(\alpha)$
		\item $|E(\alpha)| \le n^\tau$
	\end{itemize}
\end{definition}

\begin{remark}
  Note that the shapes in $\cal{L}$ can have isolated vertices in $U_{\alpha} \cap V_{\alpha}$.
\end{remark}

\begin{remark}\label{rmk:circle_deg_bound}
	$\calL$ captures all the shapes that have nonzero coefficient when we write $\calM$ in terms of graph matrices. The constraint $\deg(\circle{u}) \ge 4$ arises because pseudocalibration gives us that $\deg(\circle{u})$ is even, $\circle{u}$ cannot be isolated, and $h_2(1) = 0$.
\end{remark}


For a shape $\alpha$, we define
\[\alpha! \defeq \prod_{e \in E(\alpha)} l(e)!\] Note that this equals the factorial of the corresponding index of the Hermite polynomial for this shape.

\begin{definition}
	For any shape $\alpha$, if $\alpha \in \cal{L}$, define \[\lambda_{\alpha} \defeq
	\left( \prod_{\circle{u}\in V(\alpha)} h_{\deg(\circle{u})}(1)\right)
	\cdot \frac{1}{ n^{(\abs{U_\alpha} + \abs{V_\alpha} + \abs{E(\alpha)})/2}}
	\cdot \frac{1}{\alpha!}  \]
	Otherwise, define $\lambda_{\alpha} \defeq 0$.
\end{definition}

\begin{corollary} Modulo the footnote\footnote{Technically, the graph matrices $M_\alpha$ have rows and columns indexed by all subsets of $\calC_m \cup \calS_n$. The submatrix with rows and columns from $\binom{\calS_n}{\leq D/2}$ equals the moment matrix for $\pE$.}, $\calM = \displaystyle\sum_{\text{shapes }\alpha} \lambda_{\alpha} M_{\alpha}$.
\end{corollary}

\subsection{Boolean setting pseudocalibration}\label{subsec:calib:bool}

We now present the pseudocalibration for the boolean setting. For the
sequel, we need notation for vectors on a slice of the boolean
cube.

\begin{definition}[Slice]
  Let $v \in \set{\pm 1}^n$ and $\theta \in \mathbb{Z}$. The slice $\slice_{v}(\theta)$ is defined as
  $$
  \slice_{v}(\theta) \coloneqq \set{d \in \set{\pm 1}^n ~\vert~ \ip{v}{d} = \theta}.
  $$
  We use $\slice_{v}(\pm \theta)$ to denote $\slice_{v}(\theta) \cup \slice_{v}(-\theta)$ and
  $\slice(\theta)$ to denote $\slice_{v}(\theta)$ when $v$ is the all-ones vector.
\end{definition}

\begin{remark}
  With our notation for the slice, the planted distribution in the boolean setting can be equivalently described as
  \begin{enumerate}
    \item Sample $v \in \set{\frac{\pm 1}{\sqrt{n}}}^n$ uniformly, and then
    \item Sample $d_1,\dots,d_m$ independently and uniformly from $\slice_{\sqrt{n} \cdot v}(\pm\sqrt{n})$.
  \end{enumerate}
\end{remark}
The planted distribution doesn't actually exist for every $n$, but this is immaterial, as we can still define the pseudoexpectation via the same formula.

We will also need the expectation of monomials over the slice
$\slice(\sqrt{n})$ since they will appear in the description of the
pseudocalibrated Fourier coefficients.

\begin{definition}
   $
   e(k) \coloneqq \E_{x \unif \mathcal{S}(\sqrt{n})}\left[x_1\cdots x_k\right].
   $
\end{definition}

We now compute the Fourier coefficients of $\pE v^{\beta}$, where
$\beta \in \mathbb{F}_2^n$. The Fourier basis when $d_1, \dots, d_m \unif \{\pm 1\}^n$ is the set of parity functions. Thus a character can be specified by $\alpha \in (\F_2^n)^m$, where $\alpha$
is composed of $m$ vectors
$\alpha_1,\dots,\alpha_m \in \mathbb{F}_2^n$.  More precisely, the
character $\chi_{\alpha}$ associated to $\alpha$ is defined as
\[ \chi_\alpha(d_1,\dots,d_m) \defeq \prod_{u=1}^m d_u^{\alpha_u}\]
We denote by $\abs{\alpha}$ the number of non-zero entries of $\alpha$
and define $\abs{\alpha_u}$ similarly. Thinking of $\alpha$ as an $m\times n$ matrix with entries in $\F_2$, we also define $\alpha^\T \in (\F_2^n)^m$.

\begin{lemma}\label{lem:boolean-pseudocalibration}
We have
  $$
  \pE v^{\beta} = \frac{1}{n^{\abs{\beta}/2}}\sum_{\substack{\alpha \colon \abs{\alpha} \le n^\tau, \\ \abs{\alpha_u} \text{ even},\\ \abs{\alpha^\T_i} \equiv \beta_i \;(\mod 2) }} \prod_{u=1}^m e(\abs{\alpha_u}) \cdot \chi_{\alpha_u}(d_u).
  $$
\end{lemma}
The set of nonzero coefficients has a similar structure as in the Gaussian case: the rows of $\alpha$ must have an even number of entries, and the $i$-th column must have parity matching $\beta_i$.

\begin{proof}
   Given $\alpha \in (\F_2^n)^m$ with $\abs{\alpha} \le n^\tau$, the pseudocalibration equation enforces by construction that
   $$
   \E_{d_1,\dots,d_m \in \set{\pm 1}^n} (\pE v^{\beta})(d_1,\dots,d_m) \cdot \chi_{\alpha}(d_1,\dots,d_m) = \E_{\text{pl}} v^{\beta} \cdot \chi_{\alpha}(d_1,\dots,d_m).
   $$

  Computing the RHS above
  yields
{\footnotesize
  \begin{align*}
    \E_{v \in \set{\pm 1}^n}  \E_{d_1,\dots,d_m \unif \slice_v(\pm \sqrt{n})}\left[ v^{\beta} \prod_{u=1}^m \chi_{\alpha_u}(d_u) \right] &= \E_{v \in \set{\pm 1}^n}  \E_{d_1,\dots,d_m \unif \slice(\pm \sqrt{n})} \left[ v^{\beta} \prod_{u=1}^m \chi_{\alpha_u}(v) \chi_{\alpha_u}(d_u) \right] \\
      & = \E_{v \in \set{\pm 1}^n} \chi_{\alpha_1+\cdots +\alpha_m + \beta}(v) \E_{d_1,\dots,d_m \in \slice(\pm \sqrt{n})}\left[ \prod_{i=1}^m \chi_{\alpha_i}(d_i) \right] \\
      & = \one_{\left[\alpha_1+\cdots +\alpha_m=\beta\right]} \cdot \prod_{i=1}^m \E_{d_i \in \slice(\pm \sqrt{n})} \left[ \chi_{\alpha_i}(d_i) \right]\\
      & = \one_{\left[\alpha_1+\cdots +\alpha_m=\beta\right]} \cdot \prod_{i=1}^m \one_{\left[\abs{\alpha_i} \equiv 0 \pmod{2} \right]} \cdot \prod_{i=1}^m e(\abs{\alpha_i}).
\end{align*}
}%
  Since we have a general expression for the Fourier coefficient of each character,
  applying Fourier inversion  concludes the proof.
\end{proof}

We can now express the moment matrix in terms of graph matrices.

\begin{definition}
    Let $\calL_{bool}$ be the set of shapes in $\calL$ from~\cref{def:calL_valid_shapes} in which the edge labels are all 1.
\end{definition}

\begin{remark}
	$\calL_{bool}$ captures all the shapes that have nonzero coefficient when we write $\calM$ in terms of graph matrices. Similar to~\cref{rmk:circle_deg_bound}, since $e(2) = 0$ (see~\cref{claim:e2}), we have the same condition $deg(\circle{u}) \ge 4$ for shapes in $\calL_{bool}$.
\end{remark}

\begin{definition}
For all shapes $\alpha$, if $\alpha \in \calL_{bool}$ define
    \[\lambda_\alpha \defeq   \frac{1}{n^{(\abs{U_\alpha} + \abs{V_\alpha})/2}}\prod_{\circle{u} \in V(\alpha)} e(\deg(\circle{u}))\]
Otherwise, let $\lambda_\alpha \defeq 0$.
\end{definition}

\begin{corollary}$\calM = \displaystyle\sum_{\text{shapes }\alpha} \lambda_\alpha M_\alpha$
\end{corollary}

\subsubsection{Unifying the analysis}

It turns out that the analysis of the boolean setting mostly
follows from the analysis in the Gaussian setting.
Initially, the boolean pseudocalibration is essentially equal to
the Gaussian pseudocalibration in which we have removed all shapes
containing at least one edge with a label $k \ge 2$. The coefficients
on the graph matrices will actually be slightly different, but they
both admit an upper bound that is sufficient for our purposes
(see~\cref{prop:coefficient-bound} for the precise statement).

To unify the notation in our analysis, we conveniently set the edge
functions of the graphs in the boolean case to be
\[h_k(x) = \left\{\begin{array}{lr}
    1 & \text{if }k = 0\\
    x & \text{if }k = 1\\
    0 & \text{if }k \geq 2
\end{array}
\right. \]
This choice of $h_k(x)$ preserves the fact that
$\{h_0(x)=1,h_1(x)=x\}$ is an orthogonal polynomial basis in
the boolean setting, while zeroing out graphs with larger labels.

During the course of the analysis, we may multiply two graph matrices
and produce graph matrices with improper parallel edges (so-called
``intersections terms"). For a fixed pair ${u,i}$ of vertices, parallel
edges between $u$ and $i$ with labels $l_1,\dots,l_s$ correspond to
the product of orthogonal polynomials $\prod_{j=1}^s
h_{l_j}(d_{u,i}) \eqqcolon q(d_{u,i})$. We will re-express this product
as a linear combination of polynomials in the orthogonal family, \ie
$q(d_{u,i}) = \sum_{i=0}^{\textup{deg}(q)} \lambda_i \cdot
h_i(d_{u,i})$ for some coefficients $\lambda_i \in \mathbb{R}$. For
the boolean case, the polynomial $q(d_{u,i})$ will be either
$h_0(d_{u,i})=1$ or $h_1(d_{u,i})=d_{u,i}$. However, for the Gaussian
setting there may be up to $\textup{deg}(q)$ non-zero, potentially larger coefficients
$\lambda_i$ for the corresponding Hermite polynomials $h_i$.
For the graphs that arise in this way, we will always bound their
contributions to $\calM$ by applying the triangle inequality and norm
bounds. Since we show bounds using the larger coefficients $\lambda_i$ from the Gaussian case,
the same bounds apply when using the 0/1 coefficients in the boolean case.

We will consider separate cases at any point where the analysis differs between the two settings.

\section{Proving PSD-ness}\label{sec:psd}

Looking at the shapes that make up $\calM$, the trivial shape with $k$ square vertices contributes an identity matrix on the degree-$2k$ submatrix of $\calM$. Our ultimate goal will be to bound all shapes against these identity matrices.

\begin{definition}[Block]
    For $k,l \in \{0,1, \dots, D/2\}$, the $(k,l)$ block of $\calM$ is the submatrix with rows from $\binom{[n]}{k}$ and columns from $\binom{[n]}{l}$. Note that when $\calM$ is expressed as a sum of graph matrices, this exactly restricts $\calM$ to shapes $\alpha$ with $\abs{U_\alpha} = k$ and $\abs{V_\alpha} = l$.
\end{definition}

We define the parameter $\eta \defeq 1/\sqrt{n}$. The trivial shapes
live in the diagonal blocks of $\calM$, and on the $(k,k)$ block
contribute a factor of $\frac{1}{n^k} = \eta^{2k}$ on the diagonal.
In principle, we could make $\eta$ as small as we like\footnote{Though
pseudocalibration truncation errors may become nonnegligible for
extremely tiny $\eta$.} by considering the moments of a rescaling of
$v$ rather than $v$ itself. Counterintuitively, it will turn out that
the scaling helps us prove PSD-ness (see \cite{sklowerbounds}
for more details). It turns out that pseudocalibrating $v$ as a unit
vector (equivalently, using $\eta = 1/\sqrt{n}$) is sufficient for our
analysis.

Towards the goal of bounding $\calM$ by the identity terms, we will bound the norm of matrices on each block of $\calM$, and invoke the following lemma to conclude PSD-ness.

\begin{lemma}\label{lem:block-psd}
    Suppose a symmetric matrix $\calA \in \R^{\binom{[n]}{\leq D} \times \binom{[n]}{\leq D}}$ satisfies, for some parameter $\eta \in (0, 1)$,
    \begin{enumerate}
        \item For each $k \in \{0, 1,\dots, D\}$, the $(k,k)$ block has minimum singular value at least $\eta^{2k}(1-\frac{1}{D+1})$
        \item For each $k,l \in \{0, 1,\dots, D\}$ such that $k \neq l$, the $(k,l)$ block has norm at most $\frac{\eta^{k+l}}{D+1}$.
    \end{enumerate}
    Then $\calA \psdgeq 0$.
\end{lemma}
\begin{proof}
We need to show that for all vectors $x$, ${x^\T}{\calA}x \geq 0$. Given a vector $x$, let $x_0,\dots,x_D$ be its components in blocks $0,\dots,D$. Observe that
\begin{align*}
&{x^\T}{\calA}x \geq \sum_{k \in [0,D]}{\eta^{2k}\left(1 - \frac{1}{D+1}\right)\norm{x_k}^2} - \sum_{k \neq l \in [0,D]}{\frac{\eta^{k+l}}{D+1}\norm{x_k}\norm{x_l}} \\
&= (\norm{x_0},{\eta}\norm{x_1},\dots,{\eta}^D\norm{x_D})
\begin{pmatrix}
1-\frac{1}{D+1} & -\frac{1}{D+1} & \cdots & -\frac{1}{D+1} \\
-\frac{1}{D+1} & 1-\frac{1}{D+1} & \cdots & -\frac{1}{D+1} \\
\vdots  & \vdots  & \ddots & \vdots  \\
-\frac{1}{D+1} & -\frac{1}{D+1} & \cdots & 1-\frac{1}{D+1}
\end{pmatrix}
\begin{pmatrix}
\norm{x_0}\\
{\eta}\norm{x_1}\\
\vdots\\
{\eta}^D\norm{x_D}
\end{pmatrix}
\geq 0.
\end{align*}
\end{proof}

We start by defining spiders, which are special shapes $\alpha$ that we will handle separately in the decomposition of $\calM$. Informally, these contain special substructures which allow their norm bounds not to be negligible with respect to the identity matrix. We then show that shapes which are not spiders have bounded norms.

\begin{definition}[Left Spider]
	A left spider is a proper shape $\alpha = (V(\alpha), E(\alpha), U_{\alpha}, V_{\alpha})$ with the property that there exist two distinct square vertices $\square{i}, \square{j} \in U_{\alpha}$ of degree $1$ and a circle vertex $\circle{u} \in V(\alpha)$ such that $E(\alpha)$ contains the edges $(\square{i}, \circle{u})$ and $(\square{j}, \circle{u})$ (these are necessarily the only edges incident to $\square{i}$ and $\square{j}$).
\end{definition}
The vertices $\square{i}$ and $\square{j}$ are called the \textit{end vertices} of $\alpha$. Because of degree parity, the end vertices must lie in $U_\alpha \setminus (U_\alpha \cap V_\alpha)$.

\begin{definition}[Right spider]
	A shape $\alpha = (V(\alpha), E(\alpha), U_{\alpha}, V_{\alpha})$ is a right spider if $\alpha^\T = (V(\alpha), E(\alpha), V_{\alpha}, U_{\alpha})$ is a left spider. The end vertices of $\alpha^\T$ are also called the end vertices of $\alpha$.
\end{definition}

\begin{definition}[Spider]
	A shape $\alpha$ is a spider if it is either a left spider or a right spider.
\end{definition}
\begin{remark}
	A spider can have many pairs of end vertices. For each possible spider shape, we single out a pair of end vertices, so that in what follows we can discuss ``the'' end vertices of the spider.
\end{remark}

\subsection{Non-spiders are negligible}

For non-spiders, we will now show that their norm is small. We point out that this norm bound on non-spiders critically relies on the assumption $m \leq n^{3/2 - \eps}$.

\begin{lemma}\label{lem:charging}
	If $\alpha \in \calL$ is not a trivial shape and not a spider, then
	\[\frac{1}{n^{|E(\alpha)|/2}} n^{\frac{w(V(\alpha)) - w(S_{\min})}{2}} \le \frac{1}{n^{\Omega(\eps |E(\alpha)|)}}\]
	where $S_{min}$ is the minimum vertex separator of $\alpha$.
\end{lemma}

\begin{proof}
The idea behind the proof is as follows. Each square vertex which is not in the minimum vertex separator contributes $\sqrt{n}$ to the norm bound while each circle vertex which is not in the minimum vertex separator contributes $\sqrt{m}$. To compensate for this, we will try and take the factor of $\frac{1}{\sqrt{n}}$ from each edge and distribute it among its two endpoints so that each square vertex which is not in the minimum vertex separator is assigned a factor of $\frac{1}{\sqrt{n}}$ or smaller and each circle vertex which is not in the minimum vertex separator is assigned a factor of $\frac{1}{\sqrt{m}}$ or smaller.
\begin{remark}
Instead of using the minimum vertex separator, we will actually use a set $S$ of square vertices such that $w(S) \leq w(S_{\min})$. For details, see the actual distribution scheme below.
\end{remark}
To motivate the distribution scheme which we use, we first give two attempts which don't quite work. For simplicity, for these first two attempts we assume that $U_{\alpha} \cap V_{\alpha} = \emptyset$ as vertices in $U_{\alpha} \cap V_{\alpha}$ can essentially be ignored.
\begin{enumerate}
\item[] Attempt 1: Take each edge and assign a factor of $\frac{1}{\sqrt[4]{n}}$ to its square endpoint and a factor of $\frac{1}{\sqrt[8]{m}}$ to its circle endpoint.

With this distribution scheme, since each circle vertex has degree at least $4$, each circle vertex is assigned a factor of $\frac{1}{\sqrt{m}}$ or smaller. Since each square vertex in $W_{\alpha}$ has degree at least $2$, each square vertex in $W_{\alpha}$ is assigned a factor of $\frac{1}{\sqrt{n}}$ or smaller. However, square vertices in $U_{\alpha} \cup V_{\alpha}$ may only have degree $1$ in which case they are assigned a factor of $\frac{1}{\sqrt[4]{n}}$ which is not small enough.

To fix this issue, we can have all of the edges which are incident to a square vertex in $U_{\alpha} \cup V_{\alpha}$ give their entire factor of $\frac{1}{\sqrt{n}}$ to the square vertex.
\begin{remark}\label{rmk:pe-one}
For analyzing $\pE[1]$, this first attempt works as $U_{\alpha} = V_{\alpha} = \emptyset$. Thus, as long as $m \leq n^{2 - \epsilon}$, with high probability $\pE[1] = 1 \pm \littleoh_n(1)$ .
\end{remark}
\item[] Attempt 2: For each edge which is between a square vertex in $U_{\alpha} \cup V_{\alpha}$ and a circle vertex, we assign a factor of $\frac{1}{\sqrt{n}}$ to the square vertex and nothing to the circle vertex. For all other edges, we assign a factor of $\frac{1}{\sqrt[4]{n}}$ to its square endpoint and a factor of $\frac{1}{\sqrt[6]{m}}$ to its circle endpoint (which we can do because $m \leq n^{\frac{3}{2} - \epsilon}$).

With this distribution scheme, each square vertex is assigned a factor of $\frac{1}{\sqrt{n}}$. Since $\alpha$ is not a spider, no circle vertex is adjacent to two vertices in $U_{\alpha}$ or $V_{\alpha}$. Thus, any circle vertex which is not adjacent to both a square vertex in $U_{\alpha}$ and a square vertex in $V_{\alpha}$ must be adjacent to at least $3$ square vertices in $W_{\alpha}$ and is thus assigned a factor of $\frac{1}{\sqrt{m}}$ or smaller. However, we can have circle vertices which are adjacent to both a square vertex in $U_{\alpha}$ and a square vertex in $V_{\alpha}$. These circle vertices may be assigned a factor of $\frac{1}{\sqrt[3]{m}}$, which is not small enough.

To fix this, observe that whenever we have a circle vertex which is adjacent to both a square vertex in $U_{\alpha}$ and a square vertex in $V_{\alpha}$, this gives a path of length $2$ from $U_{\alpha}$ to $V_{\alpha}$. Any vertex separator must contain one of the vertices in this path, so we can put one of these two square vertices in $S$ and not assign it a factor of $\frac{1}{\sqrt{n}}$.
\item[] Actual distribution scheme: Based on these observations, we use the following distribution scheme. Here we are no longer assuming that $U_{\alpha} \cap V_{\alpha}$ is empty.
\begin{enumerate}
\item[1.] Choose a set of square vertices $S \subseteq U_{\alpha} \cup V_{\alpha}$ as follows. Start with $S = U_{\alpha} \cap V_{\alpha}$. Whenever we have a circle vertex which is adjacent to both a square vertex in $U_{\alpha} \setminus V_{\alpha}$ and a square vertex in $V_{\alpha} \setminus U_{\alpha}$, put one of these two square vertices in $S$ (this choice is arbitrary). Observe that $w(S) \leq w(S_{\min})$
\item[2.] For each edge which is incident to a square vertex in $S$, assign a factor of $\frac{1}{\sqrt[3]{m}}$ to its circle endpoint and nothing to this square.
\item[3.] For each edge which is incident to a square vertex in $(U_{\alpha} \cup V_{\alpha}) \setminus S$, assign a factor of $\frac{1}{\sqrt{n}}$ to the square vertex and nothing to the circle vertex.
\item[4.] For all other edges, assign a factor of $\frac{1}{\sqrt[4]{n}}$ to its square endpoint and a factor of $\frac{1}{\sqrt[6]{m}}$ to its circle endpoint.
\end{enumerate}
Now each square vertex which is not in $S$ is assigned a factor of $\frac{1}{\sqrt{n}}$ and since $\alpha$ is not a spider, all circle vertices are assigned a factor of $\frac{1}{\sqrt{m}}$ or smaller.
\end{enumerate}
We now make this argument formal.

	Let $\calC_{\alpha}$ and $\calS_{\alpha}$ be the set of circle vertices and the set of square vertices in $\alpha$ respectively. We have $ n^{\frac{w(V(\alpha)) - w(S_{\min})}{2}} \leq n^{0.5|\calS_{\alpha} \setminus S_{min}| + (0.75 - \frac{\eps}{2})|\calC_{\alpha} \setminus S_{min}|}$. So, it suffices to prove that
	\[|E(\alpha)| - |\calS_{\alpha} \setminus S_{min}| - (1.5 - \eps)|\calC_{\alpha} \setminus S_{min}| \ge \Omega(\epsilon |E(\alpha)|)\]

	Let $Q = U_{\alpha} \cap V_{\alpha}, P = (U_{\alpha} \cup V_{\alpha}) \setminus Q$ and let $P'$ be the set of vertices of $P$ that have degree $1$ and are not in $S_{min}$. Let $E_1$ be the set of edges incident to $P'$ and let $E_2 = E(\alpha) \setminus E_1$.

	For each vertex $\square{i}$ (resp. $\circle{u}$), let the number of edges of $E_2$ incident to it be $\deg'(\square{i})$ (resp. $\deg'(\circle{u})$). Since $\alpha$ is bipartite, we have that $|E_2| = \sum_{\square{i} \in \calS_{\alpha}} \deg'(\square{i}) = \sum_{\circle{u} \in \calC_{\alpha}} \deg'(\circle{u})$. We get that
	\[|E(\alpha)| = |E_1| + |E_2| = |P'| + \frac{1}{2}(\sum_{\square{i} \in \calS_{\alpha}} \deg'(\square{i}) + \sum_{\circle{u} \in \calC_{\alpha}} \deg'(\circle{u}))\]

	We also have $|S_{\alpha} \setminus S_{min}| \le |P'| + |\calS_{\alpha} \cap W_{\alpha}| + |\calS_{\alpha} \cap (P \setminus P')| \le |P'| + \frac{1}{2} \sum_{\square{i} \in \calS_{\alpha}} \deg'(\square{i})$ because each square vertex outside $P' \cup Q$ has degree at least $2$ and is not incident to any edge in $E_1$. So, it suffices to prove
	\[\frac{1}{2}\sum_{\circle{u} \in \calC_{\alpha}} \deg'(\circle{u}) - (1.5 - \eps)|\calC_{\alpha} \setminus S_{min}| \ge \Omega(\epsilon |E(\alpha)|)\]

	Now, observe that each $\circle{u} \in \calC_{\alpha}$ is incident to at most two edges in $E_1$. This is because if it were adjacent to at least $3$ edges in $E_1$, then either $\circle{u}$ is adjacent to at least two vertices of degree $1$ in $U_{\alpha}$ or $\circle{u}$ is adjacent to at least two vertices of degree $1$ in $V_{\alpha}$. However, this cannot happen since $\alpha$ is not a spider. This implies that $\deg'(\circle{u}) \ge \deg(\circle{u}) - 2$.

	Note moreover that if $\circle{u} \in \calC_{\alpha} \setminus S_{min}$, we have that $\deg'(\circle{u}) \ge \deg(\circle{u}) - 1$. This is because, building on the preceding argument, $\deg'(\circle{u}) = \deg(\circle{u}) - 2$ can only happen if there exist $\square{i} \in U_{\alpha}, \square{j} \in V_{\alpha}$ such that $(\square{i}, \circle{u}), (\square{j}, \circle{u}) \in E_1$. But then, note that we have $\square{i}, \square{j} \not\in S_{min}$ by definition of $P'$ and also, $\circle{u} \not\in S_{min}$ by assumption. This means that there is a path from $U_{\alpha}$ to $V_{\alpha}$ which does not pass through $S_{min}$, which is a contradiction.

	Finally, we set $\epsilon$ small enough such that the following inequalities are true, both of which follow from the fact that $\deg(\circle{u}) \ge 4$ for all $\circle{u} \in \calC_{\alpha}$.
	\begin{enumerate}
		\item For any $\circle{u} \in \calC_{\alpha} \cap S_{min}$, we have $\frac{\deg(\circle{u}) - 2}{2} \ge \frac{\eps}{10}\deg(\circle{u})$.
		\item For any $\circle{u} \in \calC_{\alpha} \setminus S_{min}$, we have $\frac{\deg(\circle{u}) - 1}{2} - 1.5 + \eps \ge \frac{\eps}{10}\deg(\circle{u})$.
	\end{enumerate}
	Using this, we get
	\begin{align*}
	\frac{1}{2}\sum_{\circle{u} \in \calC_{\alpha}} \deg'(\circle{u})& - (1.5 - \eps)|\calC_{\alpha} \setminus S_{min}| \\
    &\ge \sum_{\circle{u} \in \calC_{\alpha} \cap S_{min}}
	\frac{\deg(\circle{u}) - 2}{2} + \sum_{\circle{u} \in \calC_{\alpha} \setminus S_{min}}
	\frac{\deg(\circle{u}) - 1}{2} - (1.5 - \eps)|\calC_{\alpha} \setminus S_{min}|\\
	&\ge \sum_{\circle{u} \in \calC_{\alpha} \cap S_{min}}
	\frac{\eps}{10}\deg(\circle{u}) + \sum_{\circle{u} \in \calC_{\alpha} \setminus S_{min}} \left(\frac{\deg(\circle{u}) - 1}{2} - 1.5 + \eps\right)\\
	&\ge \sum_{\circle{u} \in \calC_{\alpha} \cap S_{min}}
	\frac{\eps}{10}\deg(\circle{u}) + \sum_{\circle{u} \in \calC_{\alpha} \setminus S_{min}}
	\frac{\eps}{10}\deg(\circle{u})\\
	&= \sum_{\circle{u} \in \calC_{\alpha}} \frac{\eps}{10}\deg(\circle{u}) = \Omega(\eps|E(\alpha)|)
	\end{align*}
\end{proof}

Since $\calL_{bool} \subseteq \calL$, the above result extends to non-trivial non spider shapes in $\calL_{bool}$ too.

\begin{corollary}
	If $\alpha \in \calL_{bool}$ is not a trivial shape and not a spider, then
	\[\frac{1}{n^{|E(\alpha)|/2}} n^{\frac{w(V(\alpha)) - w(S_{\min})}{2}} \le \frac{1}{n^{\Omega(\eps |E(\alpha)|)}}\]
\end{corollary}

\begin{corollary}\label{cor:non_spider_killing}
	If $\alpha \in \calL$ is not a trivial shape and not a spider, then w.h.p. \[\frac{1}{n^{|E(\alpha)|/2}}\norm{M_{\alpha}} \le \frac{1}{n^{\Omega(\epsilon |E(\alpha)|)}}\]
\end{corollary}

\begin{proof}
	Using the norm bounds in~\cref{lem:gaussian-norm-bounds}, we have
    {\footnotesize\begin{align*}
	\norm{M_\alpha} \leq 2\cdot\left(\abs{V(\alpha)} \cdot (1+\abs{E(\alpha)}) \cdot \log(n)\right)^{C\cdot (\abs{V_{rel}(\alpha)} + \abs{E(\alpha)})} \cdot n^q{\frac{w(V(\alpha)) - w(S_{\min}) + w(W_{iso})}{2}}
    \end{align*}}
	We have $W_{iso} = \emptyset$. Observe that since there are no degree $0$ vertices in $V_{rel}(\alpha)$, we have that $|V_{rel}(\alpha)| \le 2|E(\alpha)|$ and since we also have $|V(\alpha)|\cdot (1+\abs{E(\alpha)})\cdot \log n \le n^{O(\tau)}$, the factor $2\cdot(\abs{V(\alpha)} \cdot (1+\abs{E(\alpha)}) \cdot \log(n))^{C\cdot (\abs{V_{rel}(\alpha)} + \abs{E(\alpha)})}$ can be absorbed into $\frac{1}{n^{\Omega(\eps |E(\alpha)|)}}$. The result follows from~\cref{lem:charging}.
\end{proof}

This says that nontrivial non-spider shapes have $\littleoh_n(1)$ norm (ignoring the extra factor $\eta$ for the moment). We now demonstrate how to use this norm bound to control the total norm of all non-spiders in a block of $\calM$,~\cref{cor:non-spider-sum}. We will first need a couple propositions which will also be of use to us later after we kill the spiders.

\begin{proposition}\label{prop:edge-shape-count}
	The number of proper shapes with at most $L$ vertices and exactly $k$ edges is at most $L^{8(k+1)}$.
\end{proposition}

\begin{proof}
    The following process captures all shapes (though many will be constructed multiple times):
  \begin{itemize}
      \item Choose the number of square and circle variables in each of the four sets $U \cap V, U \setminus (U \cap V), V \setminus (U \cap V), W$. This contributes a factor of $L^{8}$.
      \item Place each edge between two of the vertices. This contributes a factor of $L^{2 k}$.
  \end{itemize}
\end{proof}

\begin{proposition}\label{prop:coefficient-bound}
$\abs{\lambda_\alpha} \leq \eta^{\abs{U_\alpha} + \abs{V_\alpha}} \cdot  \frac{\abs{E(\alpha)}^{3\cdot \abs{E(\alpha)}}}{n^{\abs{E(\alpha)}/2}}$ where we assume by convention that $0^0 = 1$.
\end{proposition}
\begin{proof}
\noindent\textbf{(Gaussian setting)} Recall that the coefficients $\lambda_\alpha$ are either zero or are defined by the formula
\[\lambda_\alpha  = \eta^{\abs{U_\alpha} + \abs{V_\alpha}}\cdot \left( \prod_{\circle{u}\in V(\alpha)} h_{\deg(\circle{u})}(1)\right)
	\cdot \frac{1}{ n^{\abs{E(\alpha)}/2}}
	\cdot \frac{1}{\alpha!}\]

	The sequence $h_k(1)$ satisfies the recurrence $h_0(1) = h_1(1) = 1, h_{k + 1}(1) = h_k(1) - kh_{k - 1}(1)$. We can prove by induction that $\abs{h_k(1)} \le k^k$ and hence,
	\[\prod_{\circle{u}\in V(\alpha)} \abs{h_{\deg(\circle{u})}(1)} \le \prod_{\circle{u}\in V(\alpha)} (\deg(\circle{u}))^{\deg(\circle{u})} \le \abs{E(\alpha)}^{\abs{E(\alpha)}}.\]

\noindent\textbf{(Boolean setting)} In the boolean setting the coefficients $\lambda_\alpha$ are defined by
    \[\lambda_\alpha =  \eta^{\abs{U_\alpha} + \abs{V_\alpha}} \cdot \left(\prod_{\circle{u} \in V(\alpha)} e(\deg(\circle{u})) \right)\]
    Using~\cref{cor:bound_on_coeff_e_k}, we have that $\abs{e(k)} \le k^{3k} \cdot n^{-k/2}$. Thus,
    \[
    \abs{\lambda_\alpha} =  \eta^{\abs{U_\alpha} + \abs{V_\alpha}} \cdot \prod_{\circle{u} \in V(\alpha)} \abs{e(\deg(\circle{u}))} \le  \eta^{\abs{U_\alpha} + \abs{V_\alpha}} \cdot \frac{\abs{E(\alpha)}^{3\abs{E(\alpha)}}}{n^{\abs{E(\alpha)}/2}}.
    \]
\end{proof}

\begin{corollary}\label{cor:non-spider-sum}
For $k, l \in \{0, 1, \dots , D/2\}$, let $\calB_{k,l} \subseteq \calL$ denote the set of nontrivial, non-spiders $\alpha \in \calL$ on the $(k,l)$ block i.e. $\abs{U_\alpha} = k, \abs{V_\alpha} = l$. The total norm of the non-spiders in $\calB_{k, l}$ satisfies
\[\sum_{\alpha \in \calB_{k, l}} \abs{\lambda_\alpha} \norm{M_\alpha} = \eta^{k + l} \cdot \frac{1}{n^{\Omega(\eps)}} \]
\end{corollary}
\begin{proof}
\begin{align*}
    \sum_{\alpha \in \calB_{k, l}} \abs{\lambda_\alpha} \norm{M_\alpha} & \leq \sum_{\alpha \in \calB_{k, l}}\eta^{k+l} \cdot \frac{\abs{E(\alpha)}^{3\abs{E(\alpha)}}}{n^{\abs{E(\alpha)}/2}} \norm{M_\alpha} && \text{(\cref{prop:coefficient-bound})}\\
    & \leq \eta^{k+l} \cdot\sum_{\alpha \in \calB_{k, l}}\left(\frac{\abs{E(\alpha)}^3}{n^{\Omega(\eps)}}\right)^{\abs{E(\alpha)}} && \text{(\cref{cor:non_spider_killing})}\\
    & \leq\eta^{k+l} \cdot \sum_{\alpha \in \calB_{k, l}}\left(\frac{n^{3\tau}}{n^{\Omega(\eps)}}\right)^{\abs{E(\alpha)}} && (\alpha \in \calL)\\
    & \leq \eta^{k+l} \cdot \sum_{\alpha \in \calB_{k, l}}\frac{1}{n^{\Omega(\eps\abs{E(\alpha)})}}\\
    & \leq \eta^{k+l} \cdot\sum_{i=1}^\infty \frac{n^{O(\tau i)}}{n^{\Omega(\eps i)}}\\
    & = \eta^{k+l} \cdot \frac{1}{n^{\Omega(\eps)}}
\end{align*}
where the last inequality used \cref{prop:edge-shape-count} and  the fact $|E(\alpha)| \ge 1\text{ for }\alpha \in \calB_{k, l}$.
\end{proof}

\subsection{Killing a single spider}
\label{sec:single-spider}

We saw in the Proof Strategy section that the shape $2\beta_1 + \frac{1}{n}\beta_2$ lies in the nullspace of a moment matrix which
satisfies the constraints ``$\ip{v}{d_u}^2 = 1$". The shape $\beta_1$ is
exactly the kind of substructure that appears in a spider! Therefore it
is natural to hope that if $\alpha$ is a left spider, then
$\calM_{fix}M_{\alpha} = 0$. This
doesn't quite hold because $\ip{v}{d_u}^2$ is ``missing"
some terms: in realizations of $\alpha,$ the end vertices are required to be
distinct from the other squares in $\alpha$, which prevents terms
for all pairs $i,j$ from appearing in the product
$\calM_{fix}M_\alpha$. There are smaller ``intersection terms"
(which we call
collapses of $\alpha$) that we can add so that the end vertices are permitted to take
on all pairs $i, j$. After adding in these terms, we will produce a matrix $L$ with $\calM_{fix}L =
0$.

We first define what it means to collapse a shape into another shape
by merging two vertices. Here, we only define it for merging two
square vertices, since these are the only kind of merges that will
happen in our analysis of intersection terms.

\begin{definition}[Improper collapse]
    Let $\alpha$ be a shape and let $\square{i}, \square{j}$ be two distinct square vertices in $V(\alpha)$. We define the improper collapse of $\square{i}, \square{j}$ by:
    \begin{itemize}
        \item Remove \square{i}, \square{j} from $V(\alpha)$ and replace them by a single new vertex \square{k}.
        \item Replace each edge $\{\square{i}, \circle{u}\}$ and $\{\square{j}, \circle{u}\}$, if present, by $\{\square{k}, \circle{u}\}$, keeping the same labels (note that there may be multiedges and so the new shape may not be proper).
        \item Set $U(\square{k}) = U(\square{i}) + U(\square{j}) (\mod 2)$ and $V(\square{k}) = V(\square{i}) + V(\square{j}) (\mod 2)$.
    \end{itemize}
\end{definition}

Improper collapses have parallel edges, but we can convert them back to a sum
of proper shapes.
This is done by, for each set of parallel edges, expanding the product of Fourier characters in the Fourier basis. For example, two parallel edges with label 1 should be expanded as
\[h_1(z)^2 = (z^2-1) + 1 = h_2(z) + h_0(z)\]
\begin{definition}[Collapsing a shape]
    Let $\alpha$ be a shape with two distinct square vertices $\square{i}, \square{j}$. We say that $\beta$ is a (proper) collapse of $\square{i}, \square{j}$ if $\beta$ appears in the expansion of the improper collapse of $\square{i},\square{j}$.
\end{definition}

\begin{remark}
    If $l_1, \dots, l_k$ are the labels of a set of parallel edges, then the product $h_{l_1}(z) \cdots h_{l_k}(z)$ is even/odd depending on the parity of $l_1 + \cdots + l_k$. Thus the nonzero Fourier coefficients will be the terms of matching parity. Therefore, in both the boolean and Gaussian cases, the shapes that are proper collapses of a given improper collapse are formed by replacing each set of parallel edges by a single edge $e$ such that $l(e) \le l_1 + \ldots + l_k$ and $l(e)~\equiv~l_1 + \cdots + l_k\pmod 2$.
\end{remark}

\begin{remark}\label{rmk:parity}
	Looking at the definition and in light of the previous remark, we have the following.
	\begin{enumerate}
		\item The number of circle vertices does not change by collapsing a shape but the number of square vertices decreases by $1$.
		\item $\alpha \in \calL$ has the property that the vertices have odd degree if and only if they are in $(U_{\alpha} \cup V_{\alpha}) \setminus (U_{\alpha} \cap V_{\alpha})$. When $\alpha$ collapses, this property is preserved.
	\end{enumerate}
\end{remark}

We now define the desired shapes $L_k$ which lie in the null space of $\calM_{fix}$.

\begin{definition}
For $k \geq 2$ define the shape $\ell_k$ on $\{\square{1}, \dots, \square{k}, \circle{1} \}$ with two edges $\{\{\square{1}, \circle{1}\}$, $\{\square{2}, \circle{1}\}\}$. The left side of $\ell_k$ consists of $U_{\ell_k} = \{\square{1},\dots,\square{k}\}$. The right side consists of $V_{\ell_k} =\{\square{3}, \dots, \square{k}, \circle{1}\}$.
\end{definition}

\begin{definition}\label{def:lk}
Define the ``completed'' version $L_k$ of $\ell_k$ to be the matrix which is the sum of $c_\beta M_{\beta}$ for $\beta$ being the following shapes with coefficients:
\begin{itemize}
    \item ($L_{k,1}$): $\ell_k$, with coefficient 2.
    \item ($L_{k,2}$): If $k \geq 3$, collapse $\square{1}$ and $\square{3}$ in $\ell_k$ with coefficient $\frac{2}{n}$
    \item ($L_{k,3}$): If $k \geq 4$, collapse $\square{1}$ and $\square{3}$, and collapse $\square{2}$ and $\square{4}$ in $\ell_k$ with coefficient $\frac{2}{n^2}$
    \item ($L_{k,4}$): Collapse $\square{1}$ and $\square{2}$, replacing the edges by an edge with label 2, with coefficient $\frac{1}{n}$
    \item ($L_{k,5}$): If $k \geq 3$, collapse $\square{1}, \square{2}$, and $\square{3}$, replacing the edges by an edge with label 2, with coefficient $\frac{1}{n}$.
\end{itemize}
\end{definition}

For a pictorial representation of the ribbons/shapes, see ~\cref{fig:Lk} below.

\begin{lemma}\label{lem:completed-left-side}
    $\calM_{fix} L_k = 0$
\end{lemma}
\begin{proof}
    These shapes are constructed so that if we fix a partial realization
    of the vertices $\circle{1}$ and $\square{3}, \dots, \square{k}$ as $\circle{u} \in \calC_m$ and $S \in \binom{\calS_n}{k-2}$, the squares $\square{1}$ and $\square{2}$ can still be realized as any $j_1,j_2 \in [n]$. That is, exactly the following equality holds,
    \begin{align*}
        (\calM_{fix} L_k)_I &= \displaystyle\sum_{\substack{\circle{u} \in \calC_m,\\ S \in \binom{\calS_n}{k-2}} }\left(\sum_{\substack{j_1, j_2 \in [n]:\\ j_1 \neq j_2}} \pE[v^I v^S v_{j_1}v_{j_2}] d_{uj_1}d_{uj_2} + \sum_{j_1 \in [n]} \pE[v^Iv^Sv_{j_1}^2](d_{uj_1}^2 - 1)\right)\\
        &= \displaystyle\sum_{\substack{\circle{u} \in \calC_m,\\ S \in \binom{\calS_n}{k-2}}} \pE[v^Iv^S(\ip{v}{d_u}^2 - 1)]\\
        &= 0
    \end{align*}

    To demonstrate how the coefficients arise, we analyze the ribbons $R$ which $L_k$ is composed of and see how they contribute to the output.
    For pictures of the ribbons/shapes, see~\cref{fig:Lk} below.
    Let the ribbon be partially realized as $\circle{u}$ and $S = \{\square{j_3},\dots, \square{j_k}\}$. Let $(M_{fix}L_k)_{I(u, S)}$ denote the terms in $(M_{fix}L_k)_I$ with this partial realization. In this notation we want to show
    \[(\calM_{fix}L_k)_{I(u, S)} = \sum_{\substack{j_1, j_2 \in [n]:\\ j_1 \neq j_2}} \pE[v^I v^S v_{j_1}v_{j_2}] d_{uj_1}d_{uj_2} + \sum_{j_1 \in [n]} \pE[v^Iv^Sv_{j_1}^2](d_{uj_1}^2 - 1).\]

    \begin{figure}[!ht]
        \centering
        \includegraphics[height=10cm]{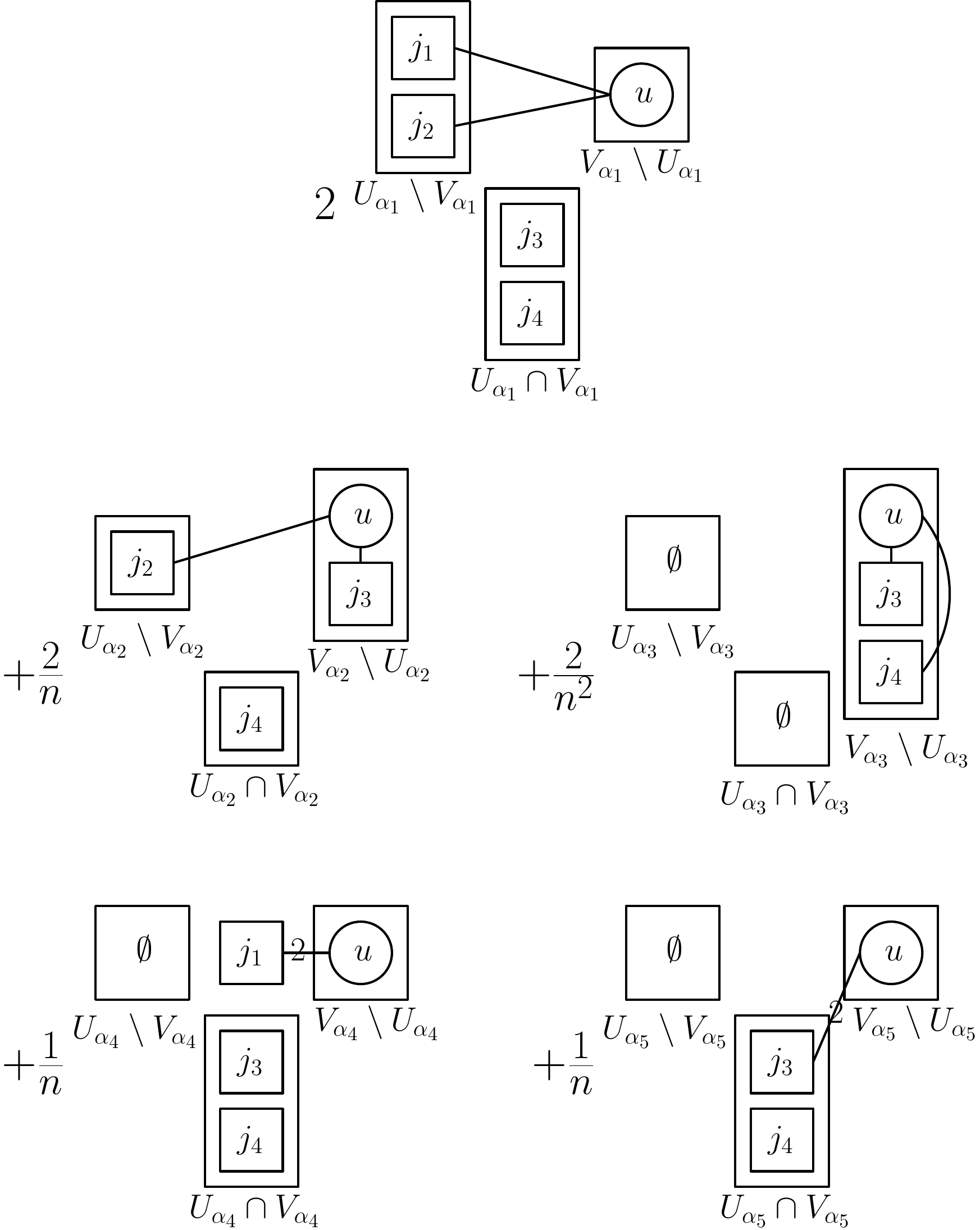}
        \caption{The five shapes that make up $L_4$.}
        \label{fig:Lk}
    \end{figure}

    \begin{enumerate}
        \item If we take a ribbon $R$ with $A_R = \{\square{j_1}, \dots, \square{j_k}\}$, $B_R = \{\square{j_3}, \dots, \square{j_k}\} \cup \{\circle{u}\}$ and $E(R) = \{\{\square{j_1}, \circle{u}\}, \{\square{j_2}, \circle{u}\}\}$ where $j_1 \neq j_2$ and $j_1, j_2 \notin S$ then
        \[ (\calM_{fix}M_R)_{I(u, S)} = \pE[v^Iv^Sv_{j_1}v_{j_2}]d_{uj_1}d_{uj_2}.\]
        This ribbon must ``cover'' both ordered pairs $(j_1, j_2)$ and $(j_2, j_1)$, so we want each such ribbon $R$ to appear with a coefficient of 2 in $L_k$.
        \item If we take a ribbon $R$ with $A_R = \{\square{j_1}, \dots, \square{j_k}\} \setminus \{\square{j_1}, \square{j_3}\}$, $B_R = \{\square{j_3}, \dots, \square{j_k}\} \cup \{\circle{u}\}$ and $E(R) = \{\{\square{j_3}, \circle{u}\}, \{\square{j_2}, \circle{u}\}\}$ where $j_1 = j_3 \in S$ then
        \[ (\calM_{fix}M_R)_{I(u, S)} = \pE[v^Iv^{S\setminus \{j_3\}}v_{j_2}]d_{uj_3}d_{uj_2} = n\pE[v^Iv^Sv_{j_1}v_{j_2}]d_{uj_1}d_{uj_2}.\]
        Taking a coefficient of $\frac{2}{n}$ in $L_k$ covers the two pairs $(j_1, j_2)$ and $(j_2, j_1)$ for this case of overlap with $S$.
        \item If we take a ribbon $R$ with $A_R = \{\square{j_1}, \dots, \square{j_k}\} \setminus \{\square{j_1}, \square{j_2}, \square{j_3}, \square{j_4}\}$, $B_R = \{\square{j_3}, \dots, \square{j_k}\} \cup \{\circle{u}\}$ and $E(R) = \{\{\square{j_3}, \circle{u}\}, \{\square{j_4}, \circle{u}\}\}$ where $j_1 = j_3 \in S$ and $j_2 = j_4 \in S$ then
        \[ (\calM_{fix}M_R)_{I(u, S)} = \pE[v^Iv^{S\setminus \{j_3, j_4\}}]d_{uj_3}d_{uj_4} = n^2\pE[v^Iv^Sv_{j_1}v_{j_2}]d_{uj_1}d_{uj_2}.\]
        Taking a coefficient of $\frac{2}{n^2}$ in $L_k$ covers the two pairs $(j_1, j_2)$ and $(j_2, j_1)$ for this case of overlap with $S$.
        \item If we take a ribbon $R$ with $A_R = \{\square{j_1}, \dots, \square{j_k}\}\setminus \{\square{j_1},\square{j_2}\}$, $B_R = \{\square{j_3}, \dots, \square{j_k}\} \cup \{\circle{u}\}$ and $E(R) = \{\{\square{j_1}, \circle{u}\}_2\}$ where $j_1 = j_2 \notin S$ then
        \[ (\calM_{fix}M_R)_{I(u, S)} = \pE[v^Iv^{S}](d_{uj_1}^2-1) = n\pE[v^Iv^Sv_{j_1}^2](d_{uj_1}^2-1).\]
        Taking a coefficient of $\frac{1}{n}$ in $L_k$ covers these terms.
        \item If we take a ribbon $R$ with $A_R = \{\square{j_1}, \dots, \square{j_k}\} \setminus \{\square{j_1}, \square{j_2}\}$, $B_R = \{\square{j_3}, \dots, \square{j_k}\} \cup \{\circle{u}\}$ and $E(R) = \{\{\square{j_3}, \circle{u}\}_2\}$ where $j_1 = j_2 =j_3\in S$ then
        \[ (\calM_{fix}M_R)_{I(u, S)} = \pE[v^Iv^{S}](d_{uj_3}^2-1) = n\pE[v^Iv^Sv_{j_1}^2](d_{uj_1}^2-1).\]
        Taking a coefficient of $\frac{1}{n}$ in $L_k$ covers these terms.
    \end{enumerate}
\end{proof}

One of the key facts about graph matrices is that multiplication of graph matrices approximately equals a new graph matrix, $M_\alpha \cdot M_\beta \approx M_{\gamma}$, where $\gamma$ is the result of gluing $V_\alpha$ with $U_\beta$ (and if $V_\alpha, U_\beta$ do not have the same number of vertices of each type, the product is zero). The error terms in the approximation are intersection terms (collapses) between the variables in $\alpha$ and $\beta$.
\begin{definition}
    Say that shapes $\alpha$ and $\beta$ are composable if $V_\alpha$ and $U_\beta$ have the same number of square and circle vertices. We say a shape $\gamma$ is a gluing of $\alpha$ and $\beta$, if the graph of $\gamma$ is the disjoint union of the graphs of $\alpha$ and $\beta$, followed by identifying $V_\alpha$ and $U_\beta$ under some type-preserving bijection, and if $U_\gamma = U_\alpha$ and $V_\gamma = V_\beta$.
\end{definition}

\begin{proposition}\label{prop:graph-matrix-multiplication}
    Let $\alpha, \beta$ be composable shapes. Assume that $V(\alpha) \setminus V_\alpha$ has only square vertices. Let $\{\gamma_i\}$ be the distinct gluings of $\alpha$ and $\beta$, and let $\widetilde{\calI}$ be the set of improper collapses of any number of squares (possibly zero) in $V(\alpha) \setminus V_\alpha$ with distinct squares in $V(\beta) \setminus U_\beta$ in any gluing $\gamma_i$. Then there are coefficients $c_\gamma$ for $\gamma \in \widetilde{\calI}$ such that
    \[ M_\alpha\cdot M_\beta = \displaystyle\sum_{\gamma \in \widetilde{\calI}} c_\gamma M_\gamma.\]
    Furthermore, the coefficients satisfy $\abs{c_\gamma} \leq
    2^{\abs{V(\alpha) \setminus V_\alpha}}\abs{V(\gamma)}^{\abs{V(\alpha) \setminus U_\alpha}}$.
\end{proposition}
\begin{proof}
    The product $M_\alpha \cdot M_\beta$ is a matrix which is a symmetric function of the inputs $(d_1, \dots, d_m)$, the space of which is spanned by the $M_\gamma$ over all possible shapes $\gamma$ (not restricted to $\widetilde{\calI}$), so there exist coefficients $c_\gamma$ if we allow all shapes $\gamma$. We need to check that $M_\alpha \cdot M_\beta$ actually lies in the span of shapes in $\widetilde{\calI}$ by showing that all ribbons in $M_\alpha \cdot M_\beta$ have shapes in $\widetilde{\calI}$. Expanding the definition,
{\footnotesize
    \[ M_\alpha \cdot M_\beta = \left(\displaystyle\sum_{R \text{ is a ribbon of shape }\alpha} M_R\right)\left(\sum_{S\text{ is a ribbon of shape }\beta} M_S\right) = \displaystyle\sum_{\substack{R \text{ is a ribbon of shape }\alpha,\\ S \text{ is a ribbon of shape }\beta}} M_R M_S.\]
}%
    In order for $M_RM_S$ to be nonzero, we require $B_R = A_S$ as sets; $R$ may assign the labels arbitrarily inside $B_R$, resulting in different gluings of $\alpha$ and $\beta$. Fix $R$ and $S$, and let $\gamma$ be the corresponding gluing of $\alpha$ and $\beta$ for this $R$ and $S$.

    The matrix $M_RM_S$ has one nonzero entry; we claim that it is a Fourier character for a ribbon $T$ which is a collapse of $\gamma$. The labels of $R$ outside of $B_R$ can possibly overlap with the labels of $S$ outside of $A_S$, and naturally the shape of $T$ is the result of collapsing vertices in $\gamma$ with the same label.

    To bound the coefficients $c_\gamma$ that appear, it suffices to bound the coefficient on a ribbon $M_T$, which is bounded by the number of contributing ribbons $R, S$, where we say ribbons $R$ of shape $\alpha$ and $S$ of shape $\beta$ contribute to $T$ if $M_RM_S = M_T$. From $T$, we can completely recover the sets $A_R$ and $B_S$. The labels of $V(R) \setminus A_R$ must be among the labels of $T$; choose them in at most $\abs{V(\gamma)}^{\abs{V(\alpha) \setminus U_\alpha}}$ ways. This also determines $B_R =A_S$. All that remains is to determine the graph structure of $S$. Since improper collapsing doesn't lose any edges, knowing the labels of $R$ we know exactly which edges of $T$ must come from $R$ and $S$. The vertices $V(T) \setminus V(R)$ must come from $S$, as must $B_R$; pick a subset of $V(R) \setminus B_R$ to include in $2^{\abs{V(\alpha) \setminus V_\alpha}}$ ways.
\end{proof}

Let $\alpha$ be a left spider with end vertices $\square{i}, \square{j}$ which are adjacent to a circle $\circle{u}$. Recall that our goal is to argue that $\calM M_\alpha \approx 0$. To get there, we can try and factor $M_\alpha$ across the vertex separator $S = U_\alpha \cup \{\circle{u}\} \setminus \{\square{i},\square{j}\}$ which separates $\alpha$ into
\[ M_\alpha \approx L_{\abs{U_\alpha}} \cdot M_{\body(\alpha)}\]
where we have defined,
\begin{definition}
    Let $\alpha$ be a left spider with end vertices $\square{i}, \square{j}$.
    Define $\body(\alpha)$ as the shape whose graph is $\alpha$ with $\square{i}$ and $\square{j}$ deleted and with $U_{\body(\alpha)} = U_\alpha \cup \{\circle{u}\} \setminus \{\square{i},\square{j}\}$, $V_{\body(\alpha)} = V_\alpha$. The definition is analogous for right spiders.
\end{definition}
Due to~\cref{lem:completed-left-side}, the right-hand side of the approximation is in the null space of $\calM$. We now formalize this approximate factorization.

\begin{definition}
	Let $\alpha$ be a spider with end vertices $\square{i}, \square{j}$. Define $\widetilde{\calI}_{\alpha}$ to be the set of shapes that can be obtained from $\alpha$ by performing at least one of the following steps:
	\begin{itemize}
	    \item Improperly collapse $\square{i}$ with a square vertex in $\alpha$
	    \item Improperly collapse $\square{j}$ with a square vertex in $\alpha$
	\end{itemize}
	Let $\calI_\alpha$ be the set of proper shapes that can be obtained via the same process but using proper collapses.
\end{definition}
In the above definition, we allow $\square{i}, \square{j}$ to collapse with two distinct squares, or to collapse together, or to both collapse with a common third vertex. For technical reasons we need to work with a refinement of $\calI_\alpha$ into two sets of shapes and use tighter bounds on coefficients of one set.
\begin{definition}
    Let $\calI_{\alpha}^{(1)}$ be the set of shapes that can be obtained from $\alpha$ by performing at least one of the following steps:
	\begin{itemize}
	    \item Collapse $\square{i}$ with a square vertex in $\body(\alpha) \setminus U_\alpha$
	    \item Collapse $\square{j}$ with a square vertex in $\body(\alpha) \setminus U_\alpha$ (distinct from $\square{i}$'s collapse if it happened)
	\end{itemize}
    Let $\calI_\alpha^{(2)}\defeq \calI_\alpha \setminus \calI_\alpha^{(1)}$
    and define the improper versions $\widetilde{\calI}_\alpha^{(1)}, \widetilde{\calI}_\alpha^{(2)}$ analogously.
\end{definition}

\begin{lemma}\label{lem:improper-collapse}
	Let $\alpha$ be a left spider with end vertices \square{i}, \square{j}. There are coefficients $c_\beta$ for $\beta \in \widetilde{\calI}_\alpha$ such that
	\[L_{\abs{U_\alpha}} \cdot M_{\body(\alpha)} = 2M_{\alpha} + \sum_{\beta \in \widetilde{\calI}_\alpha}c_\beta M_\beta,\]
	\[\abs{c_\beta} \leq
	\begin{cases}
	 40\abs{V(\alpha)}^3 & \beta \in \widetilde{\calI}_\alpha^{(1)}\\
	 \frac{40\abs{V(\alpha)}^3}{n} & \beta \in \widetilde{\calI}_\alpha^{(2)}
	\end{cases}.\]
\end{lemma}
\begin{proof}
    First, we can check that the coefficient of $M_\alpha$ is 2. Only the $\ell_k$ term of $L_k$ has the full number of squares, and it has a factor of 2 in $L_k$.

    The shapes in $\widetilde{\calI}_\alpha$ are definitionally the intersection
    terms that appear in this graph matrix product, and furthermore the shapes in
    $\widetilde{\calI}_\alpha$ are definitionally the intersection terms for the $\ell_k$ term.
    Using~\cref{prop:graph-matrix-multiplication}, for each of the five shapes
    in $L_{\abs{U_\alpha}}$ the coefficient it contributes is bounded by
    $4\abs{V(\alpha)}^3$. The coefficient on $\ell_k$ is 2, so the coefficients
    for $\widetilde{\calI}_\alpha^{(1)}$ are at most $8 \abs{V(\alpha)}^3$. The
    maximum coefficient of the other four shapes in $L_{\abs{U_\alpha}}$ is
    $\frac{2}{n}$, so their total contribution to coefficients on
    $\widetilde{\calI}_\alpha^{(2)}$ is at most $\frac{32\abs{V(\alpha)}^3}{n}$.
\end{proof}

We now want to turn our improper shapes into proper ones from $\calI_\alpha$. Unfortunately it is not quite true that to expand an improper shape, one can just expand each edge individually
(though this is true for improper ribbons).
There is an additional difficulty that arises due to ribbon symmetries. To see the difficulty, consider the example given in \cref{fig:ribbon-symmetry} below.

\begin{figure}[!ht]
  \centering
  \begin{tikzpicture}[scale=0.5,every node/.style={scale=0.5}]
    \draw  (-6.5,1.5) rectangle node {\huge $u_1$} (-5,0);
    \draw  (5,1.5) rectangle node {\huge $v_2$} (6.5,0);
    \draw  (0,2.5) ellipse (1 and 1) node {\huge $w_1$};
    \draw  (0,-1) ellipse (1 and 1) node {\huge $w_2$};
    \node (v1) at (-5,0.75) {};
    \node (v3) at (1,2.5) {};
    \node (v2) at (-1,2.5) {};
    \node (v4) at (1,-1) {};
    \draw  (-7,2) rectangle (-4.5,-0.5);
    \node at (-5.5,-1.5) {\huge $U_{\alpha}$};
    \draw  (4.5,2) rectangle (7,-0.5);
    \node at (6,-1.5) {\huge $V_{\alpha}$};
    \node (v6) at (-1,-1) {};
    \node (v5) at (4.95,0.75) {};
    \draw (v3);
    \draw  plot[smooth, tension=.7] coordinates {(v3) (v5)};
    \draw  plot[smooth, tension=.7] coordinates {(v5) (v4)};
    \draw  plot[smooth, tension=.7] coordinates {(v1) (v6)};
    \draw  plot[smooth, tension=.7] coordinates {(v1) (-3,2.5) (v2)};
    \draw  plot[smooth, tension=.7] coordinates {(v1) (-2.5,1) (v2)};
    \node at (-3.5,3) {\Large $1$};
    \node at (-2,0.5) {\Large $1$};
    \node at (-3,-0.5) {\Large $2$};
    \node at (3,2) {\Large $2$};
    \node at (3,-0.5) {\Large $2$};

    \draw  (12.5,-3.5) rectangle node {\huge $u_1$} (14,-5);
    \draw  (24,-3.5) rectangle node {\huge $v_2$} (25.5,-5);
    \draw  (19,-2.5) ellipse (1 and 1) node {\huge $w_1$};
    \draw  (19,-6) ellipse (1 and 1) node {\huge $w_2$};
    \node (v11) at (14,-4.25) {};
    \node (v13) at (20,-2.5) {};
    \node (v12) at (18,-2.5) {};
    \node (v14) at (20,-6) {};
    \draw  (12,-3) rectangle (14.5,-5.5);
    \node at (13.5,-6.5) {\huge $U_{\gamma_2}$};
    \draw  (23.5,-3) rectangle (26,-5.5);
    \node at (25,-6.5) {\huge $V_{\gamma_2}$};
    \node (v16) at (18,-6) {};
    \node (v15) at (23.95,-4.25) {};
    \draw (v13);
    \draw  plot[smooth, tension=.7] coordinates {(v13) (v15)};
    \draw  plot[smooth, tension=.7] coordinates {(v15) (v14)};
    \draw  plot[smooth, tension=.7] coordinates {(v11) (v16)};

    \node at (16,5.5) {\Large $2$};
    \node at (16,-5.5) {\Large $2$};
    \node at (22,-3) {\Large $2$};
    \node at (22,-5.5) {\Large $2$};

    \draw  (12.5,5) rectangle node {\huge $u_1$} (14,3.5);
    \draw  (24,5) rectangle node {\huge $v_2$} (25.5,3.5);
    \draw  (19,6) ellipse (1 and 1) node {\huge $w_1$};
    \draw  (19,2.5) ellipse (1 and 1) node {\huge $w_2$};
    \node (v11) at (14,4.25) {};
    \node (v13) at (20,6) {};
    \node (v12) at (18,6) {};
    \node (v14) at (20,2.5) {};
    \draw  (12,5.5) rectangle (14.5,3);
    \node at (13.5,2) {\huge $U_{\gamma_1}$};
    \draw  (23.5,5.5) rectangle (26,3);
    \node at (25,2) {\huge $V_{\gamma_1}$};
    \node (v16) at (18,2.5) {};
    \node (v15) at (23.95,4.25) {};
    \draw (v13);
    \draw  plot[smooth, tension=.7] coordinates {(v13) (v15)};
    \draw  plot[smooth, tension=.7] coordinates {(v15) (v14)};
    \draw  plot[smooth, tension=.7] coordinates {(v11) (v16)};
    \node at (16,3) {\Large $2$};
    \node at (22,5.5) {\Large $2$};
    \node at (22,3) {\Large $2$};

    \draw  plot[smooth, tension=.7] coordinates {(v11)};
    \draw  plot[smooth, tension=.7] coordinates {(v11) (v12)};
    \node at (8.5,0.5) {\Huge $=$};
    \node at (19,0) {\Huge $+$};
    \node at (10.5,4) {\Huge \bf $2 \times$};
  \end{tikzpicture}
  \caption{A surprising equality of graph matrices.}
  \label{fig:ribbon-symmetry}
\end{figure}
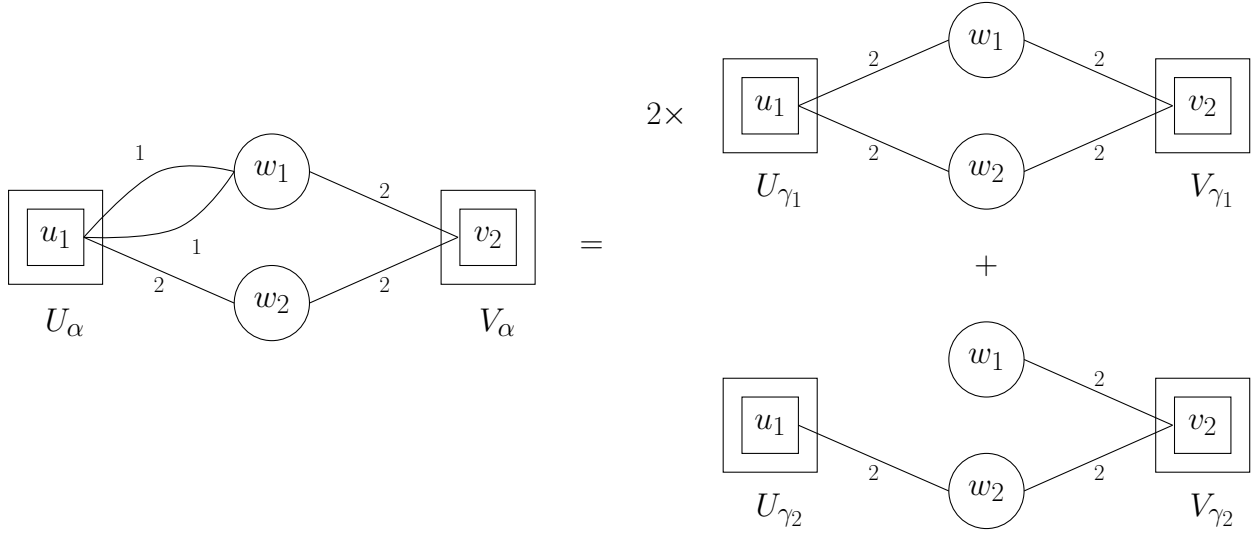

One would expect both coefficients on the right shapes to be 1 since $h_1(z)^2 = h_2(z) + h_0(z)$. However, in the left shape, the two circles are distinguishable, hence summing over all ribbons includes one with $w_1 = i, w_2 =  j$ and a second with $w_1 = j, w_2 = i$. On the top right shape, the circles are indistinguishable, hence the graph/ribbon where the circles are assigned $\{i, j\}$ is counted twice. On the bottom right shape, the circles are distinguishable, so all ribbons are summed once. To bound the new coefficients, we use the concept of shape automorphisms.

\begin{definition}
    An automorphism of a shape $\alpha$ is a function $\phi:V(\alpha) \to V(\alpha)$ that preserves the sets $U_\alpha, V_\alpha$ and is an automorphism of the underlying edge-labeled graph. Let $\aut(\alpha)$ denote the automorphism group of $\alpha$.
\end{definition}

\begin{proposition}\label{prop:expand-improper}
    Let $\alpha$ be an improper shape, and let $\calP$ be the set of proper shapes that can be obtained by expanding $\alpha$. Then there are coefficients $\abs{c_\gamma} \leq C_{Fourier}\cdot C_{Aut}$ such that
    \[M_\alpha = \displaystyle\sum_{\gamma \in \calP} c_\gamma M_\gamma\]
    where $C_{Fourier}$ is a bound on the magnitude of Fourier coefficients in the expansion and $C_{Aut} = \max_{\gamma \in \calP} \frac{\abs{\aut(\gamma)}}{\abs{\aut(\alpha)}}$.
\end{proposition}
\begin{proof}
The number of realizations of a graph matrix giving a particular ribbon is exactly the number of automorphisms, therefore
\begin{align*}
    M_\alpha &= \frac{1}{\abs{\aut(\alpha)}}\displaystyle\sum_{\text{realizations }\sigma} M_{\sigma(\alpha)}
\end{align*}
Expand each improper ribbon $M_{\sigma(\alpha)}$ into proper ribbons with coefficients at most $C_{Fourier}$.
Because the realizations of $\alpha$ and any $\gamma$ are the same, this exactly sums over all $\gamma$ and all realizations of $\gamma$. The
Fourier coefficient on each realization of $\gamma$ is the same; let it be
$c_\gamma'$ with $\abs{c_\gamma'} \leq C_{Fourier}$. Continuing,
\begin{align*}
    &= \displaystyle\frac{1}{\abs{\aut(\alpha)}} \sum_{\gamma \in \calP} c_\gamma'\sum_{\text{realizations }\sigma} M_{\sigma(\gamma)}\\
    &= \sum_{\gamma \in \calP} c_\gamma' \frac{\abs{\aut(\gamma)}}{\abs{\aut(\alpha)}} M_{\gamma}
\end{align*}
\end{proof}

\begin{proposition}\label{prop:hermite-product-coefficients}
Let $l_1 \leq \cdots \leq l_k \in \N$ and let $L = l_1 + \cdots + l_k$. Assume $L \ge 1$. In the Fourier expansion of $h_{l_1}(z)\cdots h_{l_k}(z)$, the maximum coefficient is bounded in magnitude by $(2L)^{L-l_k}$.
\begin{proof}
In the boolean case, the coefficient is 1. In the Gaussian case, the ``linearization coefficient'' of $h_p(z)$ in this product is given by orthogonality to be
\[\frac{\E_{z \sim \calN(0,1)}[h_{l_1}(z) \cdots h_{l_k}(z) \cdot h_p(z)]}{\E_{z \sim \calN(0,1)}[h_p^2(z)]}  = \frac{\E_{z \sim \calN(0,1)}[h_{l_1}(z) \cdots h_{l_k}(z) \cdot h_p(z)]}{p!}\]
A formula from, e.g.,~\cite[Example G (Continued)]{RotaWallstrom97} shows that $\E[h_{l_1} \cdots h_{l_k} \cdot h_p]$ equals the number of ``block perfect matchings'': perfect matchings on $l_1 + \cdots + l_k + p$ elements
divided into blocks of size $l_i$ or $p$ such that no two elements from the same block are matched. Bound the number of block perfect matchings by:
\begin{itemize}
    \item Pick a partial function from blocks $l_1, \dots, l_{k-1}$ to $[L]$ in at most $(L+1)^{L-l_k}$ ways.
    \item If this forms a valid partial matching and there are $p$ unmatched elements remaining, match them with the elements from the block of size $p$ in $p!$ ways.
\end{itemize}
Therefore the coefficient is bounded by $(L+1)^{L - l_k} \leq (2L)^{L-l_k}$.
\end{proof}
\end{proposition}

\begin{proposition}\label{prop:automorphism-ratio}
For a shape $\alpha$, let $\alpha\pm e$ denote the shape with edge $e$ added or deleted.
Then
\[ \frac{\abs{\aut(\alpha \pm e)}}{\abs{\aut(\alpha)}}~\leq~\abs{V(\alpha)}^2.\]
\end{proposition}
\begin{proof}
We show that the two groups have a large subgroup which are equal. Consider $\aut(\alpha\pm e)$ and $\aut(\alpha)$ as group actions on the set $\binom{V(\alpha)}{2}$. Letting $G^e$ denote the stabilizer of edge $e$, observe that $\aut(\alpha\pm e)^e = \aut(\alpha)^e$. By the orbit-stabilizer lemma, the index $\abs{G : G^e}$ is equal to the size of the orbit of $e$, which is at least 1 and at most $\abs{V(\alpha)}^2$. So,
\[\frac{\abs{\aut(\alpha \pm e)}}{\abs{\aut(\alpha)}} = \frac{\abs{\aut(\alpha \pm e) : \aut(\alpha \pm e)^e}}{\abs{\aut(\alpha) : \aut(\alpha)^e}} \leq \abs{V(\alpha)}^2.\qedhere \]
\end{proof}

\begin{lemma}\label{lem:collapse-lemma}
    If $\alpha$ is a left spider, there are coefficients ${c_\beta}$ for each $\beta \in \calI_\alpha$ such that
	\[L_{\abs{U_\alpha}} \cdot M_{\body(\alpha)} = 2M_{\alpha} + \sum_{\beta \in {\calI}_\alpha}c_\beta M_\beta,\]
	\[\abs{c_\beta} \leq
	\begin{cases}
	 160\abs{V(\alpha)}^7\abs{E(\alpha)}^2 & \beta \in {\calI}_\alpha^{(1)}\\
	 \frac{160\abs{V(\alpha)}^7\abs{E(\alpha)}^2}{n} & \beta \in {\calI}_\alpha^{(2)}
	\end{cases}.\]
\end{lemma}
\begin{proof}
    We express each $M_\beta, \beta \in \widetilde{\calI}_\alpha$ in~\cref{lem:improper-collapse} in terms of proper shapes. We apply~\cref{prop:expand-improper} using the following bounds on $C_{Fourier}$ and $C_{Aut}$. The only improperness in $\beta$ comes from
    collapsing (at most) the two end vertices, which have a single incident
    edge each. Therefore the set of labels of any parallel edges is either
    $\{1,k\}$ or $\set{1,1,k},$ for some $k \leq \abs{E(\alpha)}$. By~\cref{prop:hermite-product-coefficients}, we have $C_{Fourier} \leq 4\abs{E(\alpha)}^2$. There are at most two extra parallel edges in $\beta$, so we have $C_{Aut} \leq \abs{V(\alpha)}^4$ using~\cref{prop:automorphism-ratio}. Therefore the coefficients increase by at most $C_{Fourier}\cdot C_{Aut} \leq 4\abs{E(\alpha)}^2\abs{V(\alpha)}^4$.
\end{proof}

\begin{corollary}\label{cor:right-spider-coefs}
    If $\alpha$ is a right spider, there are coefficients $c_\beta$ with the same bounds given in~\cref{lem:collapse-lemma} such that
	\[M_{\body(\alpha)} \cdot L_{\abs{U_\alpha}}^\T = 2M_{\alpha} + \sum_{\beta \in {\calI}_\alpha}c_\beta M_\beta.\]
\end{corollary}

\begin{corollary}\label{cor:spider-killing}
    If $x \perp \nullspace(\calM_{fix})$ and $\alpha$ is a spider, then for some $c_\beta$ with the same bounds given in~\cref{lem:collapse-lemma},
    \[x^\top(M_\alpha - \displaystyle\sum_{\beta \in \calI_\alpha} c_\beta M_\beta) x = 0 \]
\end{corollary}
\begin{proof}
For a left spider, since
\[\calM_{fix} (2M_\alpha + \displaystyle\sum_{\beta \in \calI_\alpha} c_\beta M_\beta) = \calM_{fix} \cdot L_{\abs{U_\alpha}} \cdot M_{\alpha'} = 0\]
we are in position to use~\cref{fact:null-space}. For a right spider, the proof is analogous.
\end{proof}

\subsection{Killing all the spiders}

The strategy is to start with the moment matrix $\calM$ and apply~\cref{cor:spider-killing} repeatedly until we end up with no spiders in our decomposition. For each spider, killing it via~\cref{cor:spider-killing} leaves only intersection terms. Some of those intersection terms may themselves be smaller spiders, in which case we will apply the corollary again and again until only non-spiders remain. The difficulty during this procedure is to bound the total coefficient accumulated on each non-spider. To capture this process, we define the web of a spider $\alpha$, which will be a directed acyclic graph that will capture the spider killing process. For the sake of distinction, we will call the vertices of this graph ``nodes".

\begin{definition}[Web of $\alpha$]
    The web $W(\alpha)$ of a spider $\alpha$ is a rooted directed acyclic graph
    (DAG) whose nodes are shapes and whose root is $\alpha$. Each spider node
    $\gamma$ has edges to nodes $\beta$ for each shape $\beta \in
    \calI_{\gamma}$.
    The non-spider nodes are leaves/sinks of the DAG.
\end{definition}
\begin{remark}
  The DAG structure arises because each shape in $\calI_\gamma$ has strictly fewer square vertices than $\gamma$ for any spider $\gamma$. As a consequence, the height of a web $W(\alpha)$ is at most $\abs{V(\alpha)}$.
\end{remark}
Each node $\gamma$ of $W(\alpha)$ also has an associated value $v_\gamma$, which is defined by the following process:
\begin{itemize}
	\item Initially, set $v_\alpha = 1$ and for all other $\gamma$, set $v_\gamma = 0$.
	\item Starting from the root and in topological order, each spider node $\gamma$ adds $v_{\gamma} c_\beta$ to $v_\beta$ for each child $\beta \in \calI_{\gamma}$, where the $c_\beta$ are the coefficients from~\cref{cor:spider-killing}.
\end{itemize}

\begin{proposition}\label{prop:web-sum}
	If $x \perp \nullspace(\calM_{fix})$, then
	\[\displaystyle x^\T(M_\alpha - \sum_{\text{leaves } \gamma\text{ of }W(\alpha)} v_\gamma M_\gamma)x = 0.\]
\end{proposition}
\begin{proof}
	Start with the equation $x^\T M_{\alpha} x = x^\T
	v_{\alpha}M_{\alpha} x$. In each step, we take the topologically first spider $\gamma$, which in this case means the spider closest to the root of $W(\alpha)$, that is present in the right hand side of our equation and using \cref{cor:spider-killing}, we
	replace $v_{\gamma}M_{\gamma}$ by $\sum_{\beta \in
	\text{children}(\gamma)} v_\gamma c_\beta M_\beta$.
	Precisely by the definition of the $v_{\gamma}$, this
	process ends with the equation
	\[\displaystyle x^\T M_\alpha x = x^\T(\sum_{\text{leaves } \gamma\text{ of }W(\alpha)} v_\gamma M_\gamma)x\]
\end{proof}

\begin{proposition}\label{prop:web-parents}
For any node $\beta$ in $W(\alpha)$, $\abs{\parents(\beta)} \leq 4\abs{V(\alpha)}^3 \cdot \abs{E(\alpha)}^2$ where $parents(\beta)$ is the set of nodes $\gamma$ in $W(\alpha)$ such that $\beta \in \calI_{\gamma}$.
\end{proposition}
\begin{proof}
The following process covers all parent left spiders $\gamma$ which could possibly collapse their end vertices to form $\beta$. Starting from $\gamma = \beta,$
\begin{itemize}
    \item Pick a circle vertex $\circle{u} \in V(\gamma)$ to be the neighbor of the end vertices.
    \item Pick a square vertex $\square{i} \in V(\gamma)$ to be the collapse of the first end vertex. ``Uncollapse'' it by adding a new square to $U_{\gamma}$ with a single edge to $\circle{u}$ with label $1$. Flip the value of $U_{\gamma}(\square{i})$. Modify the label of $\{\square{i}, \circle{u}\}$ to any number up to $\abs{E(\alpha)}$.
    \item Pick a square vertex $\square{j} \in V(\gamma)$ to be the second end vertex. Optionally uncollapse it by adding a new square to $\gamma$ in the same way as above.
\end{itemize}
The process can be carried out in at most $\abs{V(\alpha)}^3\abs{E(\alpha)}(\abs{E(\alpha)}+1) \leq 2\abs{V(\alpha)}^3\abs{E(\alpha)}^2$ ways. We multiply by 2 to accommodate right spiders.
\end{proof}

Let us label each parent-child edge $(\gamma, \beta$) as either a ``type 1'' edge if $\beta \in \calI_{\gamma}^{(1)}$ or a ``type 2'' edge if $\beta \in \calI_{\gamma}^{(2)}$.

\begin{proposition}\label{prop:web-derivation-number}
  Let $p$ be a path in $W(\alpha)$ with $\#_1(p)$ type 1 edges and $\#_2(p)$ type 2 edges. Then $\#_1(p)~\leq~\abs{E(\alpha)}~+~2\#_2(p)$.
\end{proposition}

\begin{proof}
	For a shape $\gamma$, let $S_{\gamma}$ be the set of square vertices in $\gamma$. Then, $S_{\gamma} \cap W_{\gamma}$ will be the set of middle vertices of $\gamma$ which are squares.
  We claim that the quantity $\abs{\calS_\gamma \cap W_\gamma} + \abs{U_\gamma
    \setminus (U_\gamma \cap V_\gamma)} + \abs{V_\gamma \setminus
    (U_\gamma \cap V_\gamma)}$ decreases during a collapse.

Fix a pair of consecutive shapes $(\gamma, \beta)$ which form a type
1 edge. Looking at the definition of $\calI_\gamma^{(1)}$, each end vertex either
  collapses with (1) nothing, or (2) a vertex of $W_\gamma$, or (3) a vertex from
  $V_\gamma \setminus U_\gamma$ (if $\gamma$ is a left spider;
  for a right spider, $U_\gamma \setminus V_\gamma$).
  Furthermore, case (2) or (3) must occur for at least one of the end vertices and also, they do not collapse together.

  If case (2) occurs, then $\abs{\calS_\beta \cap W_\beta} < \abs{\calS_\gamma \cap W_\gamma}$ while $\abs{U_\beta
  \setminus (U_\beta \cap V_\beta)} = \abs{U_\gamma \setminus (U_\gamma \cap
  V_\gamma)}$ and $\abs{V_\beta \setminus
    (U_\beta \cap V_\beta)} = \abs{V_\gamma \setminus
    (U_\gamma \cap V_\gamma)}$.
    On the other hand, if case (3) occurs, then $W_\beta = W_\gamma$ while
  $\abs{U_\beta \setminus(U_\beta \cap V_\beta)}<\abs{U_\gamma
  \setminus(U_\gamma \cap V_\gamma)}$ and $\abs{V_\beta \setminus(U_\beta \cap V_\beta)}<\abs{V_\gamma
    \setminus(U_\gamma \cap V_\gamma)}$.
    In all cases, $\abs{\calS_\beta \cap W_\beta} + \abs{U_\beta
  	\setminus (U_\beta \cap V_\beta)} + \abs{V_\beta \setminus
  	(U_\beta \cap V_\beta)} < \abs{\calS_\gamma \cap W_\gamma} + \abs{U_\gamma
  	\setminus (U_\gamma \cap V_\gamma)} + \abs{V_\gamma \setminus
  	(U_\gamma \cap V_\gamma)}$ as desired.

  Now we bound this expression for $\alpha$. From the definition of $\calL$, \cref{def:calL_valid_shapes},
  for spiders appearing in the pseudocalibration,
  the square vertices in $W_{\alpha}$, $U_\alpha \setminus (U_\alpha \cap
  V_\alpha)$ and $V_\alpha \setminus (U_\alpha \cap
  V_\alpha)$ have degree at least $1$ and can only be connected to circle vertices.
  Therefore their number is bounded by $\abs{E(\alpha)}$. Hence, initially
  $\abs{\calS_\alpha \cap W_\alpha} + \abs{U_\alpha \setminus (U_\alpha \cap V_\alpha)} + \abs{V_\alpha \setminus (U_\alpha \cap V_\alpha)} \leq \abs{E(\alpha)}$.

  Finally, each type 2 edge in $p$ can only increase the quantity $\abs{\calS_\gamma \cap W_\gamma} + \abs{U_\gamma \setminus (U_\gamma \cap V_\gamma)} + \abs{V_\gamma \setminus (U_\gamma \cap V_\gamma)}$
  by at most 2. Therefore, we have the desired inequality $\#_1(p) \leq \abs{E(\alpha)} + 2\#_2(p)$.
\end{proof}
\begin{corollary}\label{cor:web-derivation-number}
    $\#_2(p) \geq \frac{\abs{p}}{3} - \frac{\abs{E(\alpha)}}{3}$.
\end{corollary}
\begin{proof}
  Plug in $\abs{p} = \#_1(p) + \#_2(p)$ and rearrange.
\end{proof}

Finally, we can bound the accumulation on each non-spider by a term which only depends on the parameters of the spider $\alpha$.
\begin{lemma}\label{lem:web-leaves}
There are absolute constants $C_1, C_2$ so that for all leaves $\gamma$ of $W(\alpha)$,
\[ \abs{v_\gamma} \leq (C_1 \cdot \abs{V(\alpha)} \cdot \abs{E(\alpha)})^{C_2 \abs{E(\alpha)}}.\]
\end{lemma}

\begin{proof}
  To bound $\abs{v_\gamma}$ we will sum the contributions of all paths $p = (\beta_0=\alpha,\dots,\beta_r=\gamma)$ in $W(\alpha)$
  starting from $\alpha$ and ending at $\gamma$. This path contributes a product of coefficients $c_\beta$ towards $v_\gamma$.

  \begin{remark}
  Here it is important that type 2 edges have stronger bounds on their coefficients $\abs{c_\beta} \leq  C\cdot (\abs{V(\alpha)}\abs{E(\alpha)})^{O(1)}/n \ll 1$.
  \end{remark}

  Before we proceed with the proof we establish some convenient notation and recall some facts.
  For consecutive shapes $\beta_{i-1},\beta_{i}$ (\ie $\beta_{i}$ is a child of $\beta_{i-1}$),
  we denote by $c_{\beta_i}$ the coefficient from~\cref{cor:spider-killing} applied on $\beta_{i - 1}$.
  By~\cref{prop:web-parents}, the in-degree of $W(\alpha)$ can be bounded as $B_1~\cdot~(\abs{V(\alpha)}\abs{E(\alpha)})^{B_2}$ for some constants $B_1, B_2$. Thus,
  the number of paths of length $r$ ending at $\gamma$ is at most $(B_1\abs{V(\alpha)}\abs{E(\alpha)})^{B_2 r}$. Using \cref{cor:spider-killing}, set $B_1, B_2$ large enough so that $c_{\beta_i}$ is at most $B_1 \cdot (\abs{V(\alpha)}\abs{E(\alpha)})^{B_2}$ for a type $1$ edge (resp. $B_1 \cdot (\abs{V(\alpha)}\abs{E(\alpha)})^{B_2} / n$ for a type $2$ edge).
  {\footnotesize

  \begin{align*}
      \abs{v_\gamma} &\le \sum_{r=0}^\infty \sum_{\substack{p = (\beta_0=\alpha,\dots,\beta_r=\gamma) \\ \textup{path from $\alpha$ to $\gamma$ in } W(\alpha)}} \prod_{i=1}^r \abs{c_{\beta_i}}\\
      &\le \sum_{r=0}^\infty\sum_{\substack{p = (\beta_0=\alpha,\dots,\beta_r=\gamma) \\ \textup{path from $\alpha$ to $\gamma$ in } W(\alpha)}} \left(B_1 \cdot (\abs{V(\alpha)}\abs{E(\alpha)})^{B_2} \right)^{\#_1(p)} \left(B_1 \cdot (\abs{V(\alpha)}\abs{E(\alpha)})^{B_2}/n \right)^{\#_2(p)} \\
      &\le \sum_{r=0}^\infty\sum_{\substack{p = (\beta_0=\alpha,\dots,\beta_r=\gamma) \\ \textup{path from $\alpha$ to $\gamma$ in } W(\alpha)}} \left(B_1 \cdot (\abs{V(\alpha)}\abs{E(\alpha)})^{B_2} \right)^{\abs{E(\alpha)} + 2\#_2(p)} \left(B_1 \cdot (\abs{V(\alpha)}\abs{E(\alpha)})^{B_2}/n \right)^{\#_2(p)} \\
      & = \sum_{r=0}^\infty\sum_{\substack{p = (\beta_0=\alpha,\dots,\beta_r=\gamma) \\ \textup{path from $\alpha$ to $\gamma$ in } W(\alpha)}} \left(B_1 \cdot (\abs{V(\alpha)}\abs{E(\alpha)})^{B_2} \right)^{\abs{E(\alpha)}} \left(B_1' \cdot (\abs{V(\alpha)}\abs{E(\alpha)})^{B_2'}/n \right)^{\#_2(p)}
  \end{align*}
  }
for some constants $B_1', B_2'$
where the first inequality followed by \cref{cor:spider-killing} and the second inequality followed by \cref{prop:web-derivation-number}.
We split the above sum into two sums, $r \le 3|E(\alpha)|$ and $r > 3|E(\alpha)|$. For $r \leq 3\abs{E(\alpha)}$, upper bounding the $\#_2(p)$ term by 1 and upper bounding
  the number of paths by $(B_1\abs{V(\alpha)}\abs{E(\alpha)})^{B_2 r}$ gives a
  bound of $(B_1''\abs{V(\alpha)}\abs{E(\alpha)})^{B_2'' \abs{E(\alpha)}}$ for some constants $B_1'', B_2''$.
  For larger $r$, we lower bound
  $\#_2(p) \geq r/9 = \abs{E(\alpha)}/3$ using~\cref{cor:web-derivation-number}. Applying the same bound on the number of paths,
  the total contribution of the terms corresponding to larger $r$ is bounded by
  1 using the power of $n$ in the denominator (assuming $\delta, \tau$ are
  small enough).
\end{proof}

We define the result of all this spider killing to be a new matrix $\calM^+$.
\begin{definition}
    Define the matrix $\calM^+$ as the result of killing all the spiders,
    \[\calM^+ \defeq \calM - \displaystyle\sum_{\text{spiders }\alpha} \lambda_\alpha \left( M_\alpha - \sum_{\text{leaves }\gamma \text{ of }W(\alpha)} v_\gamma M_\gamma \right)\]
\end{definition}

\subsection{Finishing the proof}\label{sec:finishing-psdness}

The final step of the proof is to argue that, after the spider killing
process is completed, the newly created non-spider terms in $\calM^+$ also
have small norm. Towards this, we would like to prove a statement similar
to~\cref{cor:non_spider_killing}. In that proof, we used special
structural properties of the non-spiders in $\calL$ to
prove that non-spiders in the pseudocalibration were negligible.
But now, the non-spiders in $\calM^+$ need not have the properties of
$\calL$ -- for instance, there could be circle vertices of degree $2$ or
isolated vertices. To handle the potentially larger norms, we will use
that the coefficients of these new non-spider terms $\beta$ come with the
coefficients $\lambda_\alpha$ of the spider terms $\alpha$ in whose web
they lie. Since $\alpha$ has more vertices/edges than $\beta$, the power of $\frac{1}{n}$ in $\lambda_\alpha$ is larger than the ``expected pseudocalibration'' coefficient of $\eta^{\abs{U_\beta} + \abs{V_\beta}} \cdot \frac{1}{n^{\abs{E(\beta)}/2}}$.
We prove that these extra factors of $\frac{1}{n}$ are enough to overpower isolated
vertices or a smaller vertex separator using a careful
charging argument.

\begin{lemma}\label{lem:advanced-charging}
	If $\beta$ is a nontrivial non-spider and $\beta \in W(\alpha)$ for some spider $\alpha \in \calL$, then
	\[\eta^{\abs{U_\alpha} + \abs{V_\alpha}} \cdot \frac{1}{n^{\abs{E(\alpha)}/2}}\cdot n^{\frac{w(V(\beta)) - w(S_{\min}) + w(W_{iso})}{2}} \leq \eta^{\abs{U_\beta} + \abs{V_\beta}}\cdot \frac{1}{n^{\Omega(\eps\abs{E(\alpha)})}}\]
	where $S_{min}$ and $W_{iso}$ are the minimum vertex separator of $\beta$ and the set of isolated vertices of $V(\beta) \setminus (U_\beta \cup V_\beta)$ respectively.
\end{lemma}

\begin{proof}
	We start by giving the idea of the proof. Suppose we try to use the same
	distribution scheme as in the proof of \cref{lem:charging}. It doesn't work
	for two reasons. Firstly, the circle vertices in $\beta$ still have even
	degree, which follows from \cref{rmk:parity}, but now, they could have
	degrees $0$ or $2$. For the previous distribution scheme to go through, we
	needed them to have degree at least $4$ which gave the necessary edge decay
	to handle the norm bounds. Secondly, the square vertices can now have degree
	$0$ hence getting no decay from the edges.

	The first issue is relatively easy to handle. Since $\beta$ was obtained by
	collapsing $\alpha$, the circle vertices of degrees $0$ or $2$ in $\beta$
	must have had degree at least $4$ in $\alpha$ to begin with. Hence, we can
	fix a particular sequence of collapses from $\alpha$ to $\beta$ and then
	assume for the sake of analysis that the removed edges are still present. In
	this case, the same charging argument as in \cref{lem:charging} would go
	through. This is made formal by looking at the sequence of improper
	collapses of this chain of collapses.

	To handle the second issue, let's analyze more carefully how degree $0$
	square vertices appear. Fix a sequence of collapses from $\alpha$ to $\beta$
	and consider a specific step where $\gamma$ collapsed to $\gamma'$ and a
	square vertex of degree $0$ was formed. Let the two square vertices that
	collapsed in $\gamma$ be $\square{i}, \square{j}$ and let the square vertex
	of degree $0$ that formed in $\gamma'$ be $\square{k}$. In light of
	\cref{rmk:parity}, since $\square{k}$ has degree $0$, it must not be in
	$(U_{\gamma'} \cup V_{\gamma'})\setminus (U_{\gamma'} \cap V_{\gamma'})$ and
	hence, $U_{\gamma'}(\square{k}) = V_{\gamma'}(\square{k}) = 0$ or
	$U_{\gamma'}(\square{k}) = V_{\gamma'}(\square{k}) = 1$. But in the latter
	case, this vertex does not contribute to norm bounds since it's in
	$U_{\gamma'} \cap V_{\gamma'}$ so it can be safely disregarded. Note that
	it doesn't have to stay in this set since future collapses might collapse
	this vertex, but this is not a problem as we can charge for this collapse
	if it happens.

	So, assume we have $U_{\gamma'}(\square{k}) = V_{\gamma'}(\square{k}) = 0$.
	But by the definition of collapse, at least one of $\square{i}$ or
	$\square{j}$ must have been in $U_\gamma \setminus (U_\gamma \cap V_\gamma)$
	or $V_\gamma \setminus (U_\gamma \cap V_\gamma)$. Also from the definition
	of collapse, we have $U_{\gamma'}(\square{k}) = U_{\gamma}(\square{i}) +
	U_{\gamma}(\square{j}) (\mod 2)$ and $V_{\gamma'}(\square{k}) =
	V_{\gamma}(\square{i}) + V_{\gamma}(\square{j}) (\mod 2)$. Putting these
	together, we immediately get that the only way this could have happened is
	if either $\square{i}, \square{j} \in U_\gamma \setminus (U_\gamma \cap
	V_\gamma)$ or if $\square{i}, \square{j} \in V_\gamma \setminus (U_\gamma
	\cap V_{\gamma})$.

	When such a collapse happens, observe that $|U_{\gamma}| + |V_{\gamma}| \ge
	|U_{\gamma'}| + |V_{\gamma'}| + 2$. This is precisely where the decay from
	our normalization factor $\eta = \frac{1}{\sqrt{n}}$ kicks in. This
	inequality means that an extra decay factor of $\eta^2 = \frac{1}{n}$ is
	available to us when we compare to the "expected pseudocalibration"
	coefficient of $\beta$. We will use this factor to charge the new square
	vertex of degree $0$.

	We now make these ideas formal.

	Let $Q = U_\beta \cap V_\beta, P = (U_\beta \cup V_\beta) \setminus Q$ and
	let $P'$ be the set of degree $1$ square vertices in $\beta$ that are not in
	$S_{min}$. Let $s_0$ be the number of degree $0$ square vertices in
	$V(\beta)\setminus Q$. All the square vertices outside $P' \cup Q \cup
	S_{min}$ have degree at least $2$, let there be $s_{\ge 2}$ of them.

	Because of parity constraints, \cref{rmk:parity}, and because there are no
	circle vertices in $U_{\beta} \cup V_{\beta}$, all circle vertices have even
	degree in $\beta$. Let $c_0$ be the number of degree $0$ circle vertices in
	$\beta$. Let $c_2, c_{\ge 4}$ be the number of degree $2$ circle vertices and
	the number of circle vertices of degree at least $4$ in $V(\beta) \setminus
	S_{min}$ respectively. Then, we have \[n^{\frac{w(V(\beta)) - w(S_{\min}) +
	w(W_{iso})}{2}} \le n^{\frac{|P'| + s_{\ge 2} + (1.5 - \eps)(c_2 + c_{\ge
	4})}{2}} \cdot n^{s_0 + (1.5 - \epsilon)c_0}\]

	Using $\eta = \frac{1}{\sqrt{n}}$, it suffices to show
	\begin{align*}
        \abs{E(\alpha)} + &(|U_\alpha| + |V_\alpha| - |U_\beta| - |V_\beta|)\\
        &\ge |P'| + s_{\ge 2} + (1.5 - \eps)(c_2 + c_{\ge 4}) + 2s_0 + 2(1.5 - \epsilon)c_0 + \Omega(\eps\abs{E(\alpha)})
    \end{align*}

	There can be many ways to collapse $\alpha$ to $\beta$, fix any one. We first use a charging argument for the degree $0$ square vertices.
	\begin{lemma}\label{lem:phantom_vertex}
		$|U_\alpha| + |V_\alpha| - |U_\beta| - |V_\beta| \ge 2s_0$
	\end{lemma}
	\begin{proof}
		In the collapse process, in each step, a vertex $\square{i} \in U_\gamma \setminus (U_\gamma \cap V_\gamma)$ or $\square{i} \in V_\gamma \setminus (U_\gamma \cap V_\gamma)$ of degree $1$ in an intermediate shape $\gamma$ collapses with another square vertex $\square{k}$. We have that $|U_\gamma| + |V_\gamma|$ decreases precisely when $\square{i}$ collapses with $\square{k} \in U_\gamma$ (resp. $\square{k} \in V_\gamma$). In either case, the quantity decreases by exactly $2$ which we allocate to this new merged vertex. Each degree $0$ square vertex in $V(\beta) \setminus Q$ must have arisen from a collapse, and hence must have had at least an additive quantity of $2$ allocated to it. This proves that $|U_\alpha| + |V_\alpha| - |U_\beta| - |V_\beta| \ge 2s_0$.
	\end{proof}

	We will now prove a structural lemma.
	\begin{lemma}\label{lem:structure}
		Any vertex $\circle{u}$ that has degree at least $2$ in $V(\beta) \setminus S_{min}$ is adjacent to at most $1$ vertex of $P'$.
	\end{lemma}

	\begin{proof}
		Observe that $\circle{u}$ cannot be adjacent to $3$ vertices in $P'$ because otherwise, at least $2$ of them would be in $U_\beta \setminus Q$ or in $V_\beta \setminus Q$ which means $\beta$ would be a spider which is a contradiction. If $\circle{u}$ is adjacent to $2$ vertices in $P'$, then one of them is in $U_\beta \setminus Q$ and the other is in $V_\beta \setminus Q$ respectively. Since both of these vertices are not in $S_{min}$, it follows that $\circle{u}$ is in $S_{min}$ since there is no path from $U_\beta$ to $V_\beta$ that doesn't pass through $S_{min}$. This is a contradiction. Therefore, $\circle{u}$ is adjacent to at most $1$ vertex in $P'$.
	\end{proof}
	This lemma immediately implies $|P'| \le c_2 + c_{\ge 4}$.

	To account for edges of $\alpha$ that are not in $\beta$, we let
	$\widetilde{\beta}$ be the result of improperly collapsing $\alpha$ to
	$\beta$; note that $\abs{E(\alpha)} = \abs{E(\widetilde{\beta})}$.
	We call the edges that disappeared when properly collapsing ``phantom'' edges.
	Let $\deg_{\widetilde{\beta}}(\square{i})$ (resp. $\deg_{\widetilde{\beta}}(\circle{u})$) denote the degree of vertex $\square{i}$ (resp. $\circle{u}$) in $\widetilde{\beta}$. Observe that any circle vertex $\circle{u}$ in $V(\beta)$ has $deg_{\widetilde{\beta}}(\circle{u}) \ge 4$.

	\begin{lemma}\label{lem:phantom_edge}
		$\abs{E(\alpha)} \ge |P'| + s_{\ge 2} + (1.5 - \eps)(c_2 + c_{\ge 4}) + 2(1.5 - \epsilon)c_0 + \Omega(\eps\abs{E(\alpha)})$
	\end{lemma}

	\begin{proof}
		We will use the following charging scheme. Each edge of $\beta$ incident on $P'$ allocates $1$ to the incident square vertex, which is in $P'$. Every other edge of $\beta$ allocates $\frac{1}{2}$ to the incident square vertex and $\frac{1}{2} - \frac{\epsilon}{10}$ to the incident circle vertex.	Each phantom edge allocates $1 - \frac{\eps}{10}$ to the incident circle vertex $\circle{u}$. So, a total of $\frac{\eps}{10}(\abs{E(\alpha)} - |P'|)$ has not been allocated.

		All square vertices in $P'$ have been allocated a value of $1$. And observe that all square vertices of degree at least $2$ in $\beta$ have been allocated at least $1$ from the incident edges of $\beta$, for a total value of $s_{\ge 2}$. So, the square vertices get a total allocation of at least $|P'| + s_{\ge 2}$.

		Consider any degree-$0$ circle vertex $\circle{u}$ in $V(\beta)$. It must be incident to at least $4$ phantom edges and hence, must be allocated at least a value of $4(1 - \frac{\eps}{10}) > 2(1.5 - \eps)$. Hence, the degree-$0$ circle vertices in $V(\beta$) have a total allocation of at least $2(1.5 - \eps)c_0$.

		Suppose the degree of $\circle{u}$ in $V(\beta)$ is $2$. Then, it is incident on at least $2$ phantom edges.
		By \cref{lem:structure}, it is also adjacent to at most one vertex of $P'$ and so, must have been allocated a value of at least $2(1 - \frac{\epsilon}{10}) + (deg_{\widetilde{\beta}}(\circle{u}) - 3)(\frac{1}{2} - \frac{\eps}{10})$. This is at least $1.5 - \epsilon + \frac{\eps}{10}$.

		Suppose the degree of $\circle{u}$ in $V(\beta)$ is at least $4$. By \cref{lem:structure}, it is adjacent to at most one vertex of $P'$. Then it must have been allocated a value of at least $(deg_{\widetilde{\beta}}(\circle{u}) - 1)(\frac{1}{2} - \frac{\eps}{10})$.  Using $deg_{\widetilde{\beta}}(\circle{u}) \ge 4$, this is at least $1.5 - \epsilon + \frac{\eps}{10}$.

		This implies
		\[\abs{E(\alpha)} \ge |P'| + s_{\ge 2} + 2(1.5 - \epsilon)c_0 + (1.5 - \eps + \frac{\eps}{10})(c_2 + c_{\ge 4}) + \frac{\eps}{10}(\abs{E(\alpha)} - |P'|)\]
		Using $|P'| \le c_2 + c_{\ge 4}$ completes the proof.
	\end{proof}

	Adding \cref{lem:phantom_vertex} and \cref{lem:phantom_edge}, we get the result.
\end{proof}

\begin{corollary}\label{cor:non-spider-norm-bound}
	If $\beta$ is a nontrivial non-spider and $\beta \in W(\alpha)$ for some spider $\alpha \in \calL$, then
	\[\eta^{\abs{U_\alpha} + \abs{V_\alpha}} \cdot \frac{1}{n^{\abs{E(\alpha)}/2}}\norm{M_\beta} \leq \eta^{\abs{U_\beta} + \abs{V_\beta}}\cdot \frac{1}{n^{\Omega(\eps\abs{E(\alpha)})}}\]
\end{corollary}

\begin{proof}

	From \cref{lem:gaussian-norm-bounds}, we have
	\[ \norm{M_\beta} \leq 2\cdot\left(\abs{V(\beta)} \cdot (1+\abs{E(\beta)}) \cdot \log(n)\right)^{C\cdot (\abs{V_{rel}(\beta)} + \abs{E(\beta)})} \cdot n^{\frac{w(V(\beta)) - w(S_{\min}) + w(W_{iso})}{2}}\]

	We have $\abs{V(\beta)}\cdot (1+\abs{E(\beta)}) \cdot \log(n) \le n^{O(\tau)}$. Also, $|V_{rel}(\beta)| \le 2(|E(\alpha)| + |E(\beta)|)$ since all the degree $0$ vertices in $V_{rel}(\beta)$ would have had vertices of $V_{rel}(\alpha)$ collapse into it in the chain of collapses and there are no degree $0$ vertices in $V_{rel}(\alpha)$. Finally, since $|E(\alpha)| \ge |E(\beta)|$, the factor
	$2\cdot(\abs{V(\beta)} \cdot (1+\abs{E(\beta)}) \cdot \log(n))^{C\cdot (\abs{V_{rel}(\beta)} + \abs{E(\beta)})}$ can be absorbed into $\frac{1}{n^{\Omega(\eps\abs{E(\alpha)})}}$. The result follows from \cref{lem:advanced-charging}.
\end{proof}

\begin{proposition}\label{prop:m+-diag}
If $\beta$ is a trivial shape, $\lambda_\beta^+ = \lambda_\beta$.
\begin{proof}
    A trivial shape cannot appear in $W(\alpha)$ for any $\alpha$, since every collapse of a spider always keeps its circle vertices around.
\end{proof}
\end{proposition}

\begin{lemma}\label{lem:non-spider-psd}
	For $k,l \in \{0, 1, \dots, D/2\}$, let $\calB_{k,l}$ denote the set of nontrivial non-spiders on block $(k, l)$. Then
	\[\displaystyle\sum_{\beta \in \calB_{k,l}} \abs{\lambda_\beta^+}\norm{M_\beta} \leq \eta^{k+l} \cdot \frac{1}{n^{\Omega(\eps)}} \]
\end{lemma}
\begin{proof}

\begin{align*}
\displaystyle\sum_{\beta \in \calB_{k,l}}\norm{\lambda_\beta^+ M_\beta}
\leq & \sum_{\beta \in \calB_{k,l}} \abs{\lambda_\beta} \norm{M_\beta} + \sum_{\beta \in \calB_{k,l}} \sum_{\substack{\text{spiders }\alpha:\\ \beta \in W(\alpha)}} \abs{v_\beta} \abs{\lambda_\alpha} \norm{M_\beta}
\end{align*}
To bound the first term, we checked previously in~\cref{cor:non-spider-sum} that the total norm of nontrivial non-spiders appearing in the pseudocalibration (i.e. this term) is $\eta^{k+l}o_n(1)$. For the second term, via~\cref{lem:web-leaves} we have a bound on the accumulations $v_\gamma$ of one spider on one non-spider, so it is at most
\[\leq \sum_{\beta \in \calB_{k,l}}\displaystyle\sum_{\substack{\text{spiders }\alpha:\\ \beta \in W(\alpha)}}(C_1\abs{V(\alpha)} \cdot \abs{E(\alpha)})^{C_2 \abs{E(\alpha)}} \cdot \abs{\lambda_\alpha} \norm{M_\beta}.\]
Use the bound on the coefficients $\abs{\lambda_\alpha}$,~\cref{prop:coefficient-bound},
\begin{align*}
\leq  \sum_{\beta \in \calB_{k,l}}\displaystyle\sum_{\substack{\text{spiders }\alpha:\\ \beta \in W(\alpha)}}(C_1\abs{V(\alpha)} \cdot \abs{E(\alpha)})^{C_2 \abs{E(\alpha)}} \cdot\eta^{\abs{U_\alpha} + \abs{V_\alpha}}\cdot  \frac{\abs{E(\alpha)}^{3\abs{E(\alpha)}}}{n^{\abs{E(\alpha)}/2}} \cdot \norm{M_\beta}
\end{align*}
Invoking the norm bound for non-spiders which are collapses, \cref{cor:non-spider-norm-bound},
\begin{align*}
& \leq  \eta^{k+l} \cdot \sum_{\beta \in \calB_{k,l}}\displaystyle\sum_{\substack{\text{spiders }\alpha:\\ \beta \in W(\alpha)}} \left(\frac{C_1\abs{V(\alpha)} \cdot \abs{E(\alpha)}}{n^{\Omega(\eps)}}\right)^{C_2' \abs{E(\alpha)}}\\
& \leq  \eta^{k+l} \cdot \sum_{\beta \in \calB_{k,l}}\displaystyle\sum_{\substack{\text{spiders }\alpha:\\ \beta \in W(\alpha)}} \left(\frac{C_1n^\tau \cdot n^\tau}{n^{\Omega(\eps)}}\right)^{C_2' \abs{E(\alpha)}}.
\end{align*}

Bound the sum over all spiders by the sum over all shapes. By~\cref{prop:edge-shape-count}, the number of shapes with $i$ edges is $n^{O(\tau (i+1))}$. Summing by the number of edges, observe that $\abs{E(\alpha)} \geq \max(\abs{E(\beta)}, 2)$ since spiders always have at least $2$ edges.
\begin{align*}
    & \leq \eta^{k+l}\sum_{\beta \in \calB_{k,l}}\displaystyle\sum_{i=\max(\abs{E(\beta)}, 2)}^\infty
    n^{O(\tau (i+1))} \cdot \left(\frac{C_1n^\tau \cdot n^{\tau}}{n^{\Omega(\eps)}}\right)^{C_2' i} \\
    &\leq \eta^{k+l}\sum_{\beta \in \calB_{k,l}} \frac{1}{n^{\Omega(\eps \max(\abs{E(\beta)}, 2))}}\\
    & \leq \eta^{k+l}\sum_{i=0}^\infty \frac{n^{O(\delta (i+1))}}{n^{\Omega(\eps \max(i, 2))}}\\
    & = \eta^{k+l} \cdot \frac{1}{n^{\Omega(\eps)}} \qedhere
\end{align*}
\end{proof}

\begin{corollary}\label{cor:m+-diag}
For $k \in \{0, \dots, D/2\}$, the $(k,k)$ block of $\calM^+$ has minimum singular value at least $\eta^{2k}(1 - \frac{1}{n^{\Omega(\eps)}})$, and for $k, l \in \{0, \dots, D/2\}, l \neq k$, the $(k,l)$ off-diagonal block has norm at most $\eta^{k+l} \cdot \frac{1}{n^{\Omega(\eps)}}$.
\end{corollary}
\begin{proof}
    By~\cref{prop:m+-diag} the identity matrix appears on the $(k,k)$ blocks with coefficient $\eta^{2k}$. By construction, $\calM^+$ has no spider shapes. By~\cref{lem:non-spider-psd}, the total norm of the non-spider shapes on the $(k,l)$ block is at most $\eta^{k+l}\cdot \frac{1}{n^{\Omega(\eps)}}$.
\end{proof}

\begin{theorem}
    W.h.p. $\calM_{fix} \psdgeq 0$.
\end{theorem}
\begin{proof}
    For any $x \in \nullspace(\calM_{fix})$, we of course have $x^\T \calM_{fix} x = 0$.
    For any $x \perp \nullspace(\calM_{fix})$ with $\norm{x}_2 = 1$,
    \begin{align*}
        x^\T \calM_{fix} x & = x^\T (\calM + \calE) x\\
        &= x^\T \calM^+ x + x^\T \left(\displaystyle\sum_{\text{spiders }\alpha} \lambda_\alpha\left( \calM_\alpha - \sum_{\text{leaves }\gamma\text{ of }W(\alpha)}v_\gamma M_\gamma \right)\right)x + x^\T \calE x\\
        &= x^\T (\calM^+ +  \calE) x
    \end{align*}
    where  the last equality follows from \cref{prop:web-sum}.
    Because the norm bound on $\calE$ is significantly less than $\eta^{D} = n^{-n^{\delta}}$ (see \cite{sklowerbounds}), the bound on the norm of each block of $\calM^+$
    in~\cref{cor:m+-diag} also applies to the blocks of $\calM^+ + \calE$. Therefore,
    we use~\cref{lem:block-psd} to conclude $\calM^+ + \calE \psdgeq 0$ and the above expression is nonnegative.
\end{proof}

\section{Sherrington-Kirkpatrick Lower Bounds}\label{sec:sk}

Here, we prove~\cref{theo:boolean-subspace} and~\cref{theo:sk-bounds}.

Recall that in the Planted Boolean Vector problem, we wish to optimize
\[
\OPT(V) \defeq  \frac{1}{n}\max_{b \in \{\pm 1\}^n} b^\T \Pi_V b,
\]
where $V$ is a uniformly random $p$-dimensional subspace of $\mathbb{R}^n$.

\booleanSubspace*

\begin{proof}
	We wish to produce an SoS solution $\pE$ on boolean variables $b_1, \ldots, b_n$ such that $\pE[b^\T \Pi_V b] = n$.
        Instead of sampling a uniformly random $p$-dimensional subspace $V$ of $\mathbb{R}^n$, we first sample $d_1,\ldots, d_n$ i.i.d. $p$-dimensional
        Gaussian vectors from $\mathcal{N}(0,I)$, then form an $n$-by-$p$ matrix $A$ with rows $d_1,\dots,d_n$, and finally take
        $V$ to be the span of the columns of $A$. Since the columns of $A$ are isotropic i.i.d. random Gaussian vectors, we have
        that $V$ is a uniform $p$-dimensional subspace\footnote{Except for a zero measure event. 
        } of $\mathbb{R}^n$.

        We will consider $V$ as the input for the Planted Boolean Vector problem
        while the vectors $d_1,\dots,d_n$ will be used to construct a pseudoexpectation operator for the Planted Affine Planes
        problem\footnote{Note that the vectors $d_u$ are not ``given" in the Planted Boolean Vector problem, though the construction of $\pE$ is not required to be algorithmic in any sense anyway.}.
        Since $n \le p^{3/2 - \Omega(\epsilon)}$, by~\cref{theo:sos-bounds}, for all $\delta \le c\epsilon$ for a constant $c > 0$, \text{w.h.p.}, there exists a degree-$n^{\delta}$ pseudoexpectation operator $\pE'$ on formal variables $v=(v_1,\dots,v_p)$ such that $\pE'[\ip{v}{d_u}^2] = 1$
        for every $u \in [n]$.

	Define $\pE$ by $\pE[b_u] \defeq \pE'[\ip{v}{d_u}]$ for all $u \in [n]$ and extending it to all polynomials on $\{b_u\}$ by
	multilinearity. This is well defined because $\pE'[\ip{v}{d_u}^2] = 1$. Note that $\pE$ is a valid pseudoexpectation operator
        of the same degree as $\pE'$. Finally, observe that
	\begin{align*}
          \frac{1}{n}\pE [b^\T \Pi_V b] = \frac{1}{n}\pE'[v^\T A^\T \Pi_V Av] = \frac{1}{n}\pE'[v^\T A^\T Av] = 1.
      \end{align*}
\end{proof}

Now we prove lower bounds for the Sherrington-Kirkpatrick problem,
using a reduction and proof due to \cite{mohanty2020lifting}.  We include it here
for completeness. Recall that the SK problem is to
compute
\[
\OPT(W) \defeq \max_{x \in \{\pm 1\}^n} x^\T W x,
\]
where $W$ is sampled from $\GOE(n)$.

\SKbounds*

We will use the following standard results from random matrix theory of $\GOE(n)$.

\begin{fact}\label{fact:goe}
  Let $\lambda_1 \ge \ldots\ge \lambda_n$ be the eigenvalues of $W \sim \GOE(n)$ with corresponding normalized eigenvectors $w_1, \ldots, w_n$.
  Then,
  \begin{enumerate}
    \item  For every $p \in [n]$, the span of $w_1, \ldots, w_p$ is a uniformly random $p$-dimensional subspace of $\RR^n$ (see e.g.~\cite[Section~2]{OVW16}).\label{fact:goe:1}
    \item  W.h.p., $\lambda_{n^{0.67}} \ge (2 - \littleoh(1))\sqrt{n}$ (Corollary of Wigner's semicircle law~\cite{Wig93})
  \end{enumerate}
\end{fact}

\begin{proofof}{\cref{theo:sk-bounds}}
	Let $p = n^{0.67}$ and $W \sim \GOE(n)$. Let $\lambda_1\ge \ldots \ge \lambda_n$ be the eigenvalues of $W$ with corresponding orthonormal set of
        eigenvectors $w_1, \ldots, w_n$. By~\cref{fact:goe}, we have that $\lambda_p \ge (2 - \littleoh(1))\sqrt{n}$ and that $w_1, \ldots, w_p$ span a
        uniformly random $p$-dimensional subspace $V$ of $\RR^n$.

	We consider $V$ as the input of the Boolean Planted Vector problem and by~\cref{theo:boolean-subspace}, for some constant $\delta > 0$,
        \text{w.h.p.} there exists a degree-$n^{\delta}$ pseudoexpectation operator $\pE$ such that $\pE[x_i^2] = 1$ and
        $\pE[\sum_{i = 1}^p\ip{x}{w_i}^2] = \pE[x^\T \Pi_V x] = n$. Now,
	\begin{align*}
	\pE[x^\T Wx] &= \pE[\sum_{i = 1}^n \lambda_i \ip{x}{w_i}^2]\\
    &\ge \lambda_p\pE[x^\T\Pi_V x] - \abs{\lambda_n} \pE[\sum_{i = p + 1}^n\ip{x}{w_i}^2]\\
	&\ge (2 - \littleoh(1))n^{3/2} - \abs{\lambda_n}\pE[\ip{x}{x} - \sum_{i = 1}^p\ip{x}{w_i}^2]\\
    &= (2 - \littleoh(1))n^{3/2}.
	\end{align*}
\end{proofof}

\begin{remk}
  Using the same proof as above, we can obtain~\cref{theo:sk-bounds} even if we were only able to prove SoS lower
  bounds for Planted Affine Planes for some $m = \omega(n)$. So, pushing the value of $m$ up to $n^{3/2 - \epsilon}$, which
  is~\cref{theo:sos-bounds}, offers only a modest improvement.
\end{remk}




\section{Omitted technical details}

\subsection{Norm Bounds}\label{app:norm_bounds}

The precise norm bounds we use come from applying the trace power method
in~\cite{ahn2016graph}, but qualitatively, the bounds from \cref{chap: efron_stein} also work. The paper~\cite{ahn2016graph} uses a slightly different
definition of matrix index. They define a \textit{matrix index piece}
as a tuple of distinct elements from either $\calC_m$ or $\calS_n$
along with a fixed integer denoting multiplicity. A matrix index is
then a set of matrix index pieces. Our graph matrix $M_\alpha$ appears
as a submatrix of those matrices: for a given set of square vertices,
order the squares in increasing order in a tuple, and assign it
multiplicity 1. Hence the same norm bounds apply.

Boolean norm bounds:
\begin{lemma}\label{lem:norm-bounds}
Let $V_{rel}(\alpha) \defeq V(\alpha) \setminus (U_\alpha \cap V_\alpha)$. There is a universal constant $C$ such that the following norm bound holds for all proper shapes $\alpha$ w.h.p.:
\[\norm{M_\alpha} \leq 2\cdot\left(\abs{V(\alpha)} \cdot \log(n)\right)^{C\cdot \abs{V_{rel}(\alpha)}} \cdot n^{\frac{w(V(\alpha)) - w(S_{\min}) + w(W_{iso})}{2}} \]
\end{lemma}
\begin{proof}
From Corollary 8.13 of~\cite{ahn2016graph}, with probability at least $1-\eps$ for a fixed shape $\alpha$,
\[\norm{M_\alpha} \leq 2 \abs{V(\alpha)}^{\abs{V_{rel}(\alpha)}}\cdot \left( 6e \ceil{\frac{\log\left(\frac{n^{w(S_{\min})}}{\eps}\right)}{6\abs{V_{rel}(\alpha)}}}\right)^{\abs{V_{rel}(\alpha)}} \cdot n^{\frac{w(V(\alpha)) - w(S_{\min}) + w(W_{iso})}{2}}\]
Letting $N_k$ be the number of distinct shapes on $k$ vertices (either
circles or squares), we apply the corollary with $\eps = 1/(mn
N_{\abs{V(\alpha)}})$. Union bounding, the failure probability across
all shapes of size $k$ is at most $1/mn$, and since the number of
vertices in a shape is at most $m + n \leq 2m$, we have a bound that
holds with high probability for all shapes. It remains to simplify the
exact bound.

\begin{proposition}\label{prop:boolean-shape-counting}
$N_k \leq 8^k 2^{k^2}$
\begin{proof}
The following process forms all shapes on $k$ vertices: starting from $k$ formal variables, assign each variable to be either a circle or a square, decide whether each variable is in $U_\alpha$ and/or $V_\alpha$, then among the $k^2$ variable pairs put any number of edges.
\end{proof}
\end{proposition}
We also bound $n^{w(S_{\min})} \leq (mn)^{\abs{V(\alpha)}}$.
{\footnotesize
\begin{align*}
    \norm{M_\alpha} & \leq 2 \abs{V(\alpha)}^{\abs{V_{rel}(\alpha)}}\cdot \left( 6e \ceil{\frac{\log\left(n^{w(S_{\min})} \cdot mn N_{\abs{V(\alpha)}}\right)}{6\abs{V_{rel}(\alpha)}}}\right)^{\abs{V_{rel}(\alpha)}} \cdot n^{\frac{w(V(\alpha)) - w(S_{\min}) + w(W_{iso})}{2}} \\
    & \leq 2 \abs{V(\alpha)}^{\abs{V_{rel}(\alpha)}}\cdot \left( 12e \log\left(n^{w(S_{\min})} \cdot mn N_{\abs{V(\alpha)}}\right)\right)^{\abs{V_{rel}(\alpha)}} \cdot n^{\frac{w(V(\alpha)) - w(S_{\min}) + w(W_{iso})}{2}} \\
    & \leq 2 \abs{V(\alpha)}^{\abs{V_{rel}(\alpha)}}\cdot \left( 12e \log\left((mn)^{\abs{V(\alpha)}} \cdot mn\cdot 8^{\abs{V(\alpha)}} 2^{\abs{V(\alpha)}^2}\right)\right)^{\abs{V_{rel}(\alpha)}} \cdot n^{\frac{w(V(\alpha)) - w(S_{\min}) + w(W_{iso})}{2}}\\
    &  \leq 2 \abs{V(\alpha)}^{\abs{V_{rel}(\alpha)}}\cdot \left( 100e \abs{V(\alpha)}^2 \log\left(mn\right)\right)^{\abs{V_{rel}(\alpha)}} \cdot n^{\frac{w(V(\alpha)) - w(S_{\min}) + w(W_{iso})}{2}}\\
    & \leq 2\cdot\left(\abs{V(\alpha)} \cdot \log(mn)\right)^{3\cdot \abs{V_{rel}(\alpha)}} \cdot n^{\frac{w(V(\alpha)) - w(S_{\min}) + w(W_{iso})}{2}}
\end{align*}
}%
Note that we now assume $m \leq n^2$.
\end{proof}

We have the following norm bound for Hermite shapes. For a Hermite shape $\alpha$, define the \textit{total size} to be $\abs{U_\alpha} + \abs{V_\alpha} + \abs{W_\alpha} + \abs{E(\alpha)}$.
\begin{lemma}\label{lem:gaussian-norm-bounds}
Let $V_{rel}(\alpha) \defeq V(\alpha) \setminus (U_\alpha \cap V_\alpha)$ as sets. There is a universal constant $C$ such that the following norm bound holds for all proper shapes $\alpha$ with total size at most $n$ w.h.p.:
\[ \norm{M_\alpha} \leq 2\cdot\left(\abs{V(\alpha)} \cdot (1+\abs{E(\alpha)}) \cdot \log(n)\right)^{C\cdot (\abs{V_{rel}(\alpha)} + \abs{E(\alpha)})} \cdot n^{\frac{w(V(\alpha)) - w(S_{\min}) + w(W_{iso})}{2}}\]
\end{lemma}

The proof performs the same calculation starting from~\cite[Corollary 8.15]{ahn2016graph}. Note that in our notation, $l(\alpha) = \abs{E(\alpha)}$. There is a further difference which is that~\cite{ahn2016graph} uses normalized Hermite polynomials whereas we use unnormalized Hermite polynomials; this contributes the additional term $\prod_{e \in E(\alpha)} l(e)! \leq (1+\abs{E(\alpha)})^{\abs{E(\alpha)}}$. We must replace Proposition~\ref{prop:boolean-shape-counting} with the following:
\begin{proposition}\label{prop:gaussian-shape-counting}
The number of Hermite shapes with total size $k$ is at most $k2^k(k+1)^{2k+k^2}$.
\begin{proof}
Such a shape has at most $k$ distinct variable vertices. Each of these is either a circle or a square. Each variable can be in $U_\alpha$ with multiplicity between 0 and (at most) $k$, and also in $V_\alpha$ with multiplicity between 0 and $k$. The $k^2$ possible pairs of vertices can have edge multiplicity in $E(\alpha)$ between 0 and $k$.
\end{proof}
\end{proposition}

\subsection{Properties of $e(k)$}

In this section, we establish some properties of the $e(k)$ used in the
analysis. Recall that $e(k) = \E_{x \in \mathcal{S}(\sqrt{n})}\left[x_1\dots x_k\right]$ where
$\mathcal{S}(\sqrt{n}) \coloneqq \set{x \in \set{\pm
1}^n \mid \sum_{i=1}^n x_i = \sqrt{n}}$.

\begin{claim}\label{claim:e2}
  $e(2)=0$.
\end{claim}

\begin{proof}
  Fix $y \in \mathcal{S}(\sqrt{n})$. Note that $(\sum_{i=1}^n y_i)^2 = n$ implying
  $\sum_{i < j} y_i y_j = 0$. Using this fact, we get
  $$
  \E_{x \in \mathcal{S}(\sqrt{n})}\left[x_1 x_2\right] = \E_{\sigma \in S_n} y_{\sigma(1)} y_{\sigma(2)} = 0,
  $$
  concluding the proof.
\end{proof}

\begin{definition}
   We say that a tuple $\lambda = (\lambda_1,\dots,\lambda_k)$ of non-negative integers is a partition of $k$
   provided $\sum_{i=1}^k \lambda_i = k$ and $\lambda_1 \ge \cdots \ge \lambda_k$. We use the notation $\lambda \vdash k$
   to denote a partition of $k$. We refer to $\lambda_i$ as a row/part of $\lambda$.
\end{definition}

In the following, we will dealing with polynomials that can be indexed by integer partitions.
For this reason, we now fix a notation for partitions and some associated objects.

\begin{definition}
  The transpose of partition $\lambda = (\lambda_1,\dots,\lambda_k)$ is denoted $\lambda^t$ and defined as
  $\lambda^t_i = \abs{\set{j \in [k] \mid \lambda_j \ge i}}$.
\end{definition}

\begin{remark}
  For a partition $\lambda \vdash k$, $\lambda^t_1$ is the number of rows/parts of $\lambda$.
\end{remark}

\begin{definition}
  The automorphism group of a partition $\aut(\lambda) \leq S_{\lambda^t_1}$ is the group generated by transpositions $(i,j)$
  of rows $\lambda_i = \lambda_j$.
\end{definition}

\begin{remark}
  Let $\lambda \vdash k$ and $p_1(\lambda),\dots,p_k(\lambda)$ be such that $p_i(\lambda) = \abs{\set{j \in [\lambda_1^t] \mid \lambda_j = i}}$.
  Then $\aut(\lambda) \simeq S_{p_1} \times \cdots \times S_{p_k}$.
\end{remark}

\begin{lemma}\label{lem:slice_inv_exact}
  We have
  \[
  \sum_{\lambda \vdash k} \frac{\lambda!}{\lambda_1!\cdots \lambda_k!} \cdot \frac{(n)_{\lambda^t_1}}{\abs{\aut(\lambda)}} \cdot \E_{x \in \mathcal{S}(\sqrt{n})}\left[x_1^{\lambda_1} \dots x_k^{\lambda_k} \right] = n^{k/2}.
  \]
\end{lemma}

\begin{proof}
  For $x \in \mathcal{S}(\sqrt{n})$, we have $(\sum_{i=1}^n x_i)^k =
  n^{k/2}$. Then expanding $(\sum_{i=1}^n x_i)^k$ in the previous equations and
  taking the expectation over $\mathcal{S}(\sqrt{n})$ on both sides
  yields the result of the lemma (after appropriately collecting
  terms).
\end{proof}

\begin{claim}\label{claim:bound_prod_exp_ff}
  Let $\lambda \vdash k$. We have
  $$
  (n)_{\lambda^t_1} \cdot \abs{\E_{x \in \mathcal{S}(\sqrt{n})}\left[x_1^{\lambda_1} \dots x_k^{\lambda_k} \right]} \le 3^{k^3} \cdot n^{k/2}.
  $$
\end{claim}

\begin{proof}
  We induct on $k$. For $k=1$, we have $n \cdot \abs{\E_{x \in \mathcal{S}(\sqrt{n})}\left[x_1\right]} = \sqrt{n} \le 3 \cdot n^{1/2}$.
  Now, suppose $k \ge 2$. We consider three cases:
  \begin{enumerate}
    \item Case $\lambda_1 \ge 3$: Let $\lambda'$ be the partition obtained from $\lambda$ by removing two boxes from $\lambda_1$.
          Note that $\lambda_1^t = (\lambda')^t_1 \le k-2$ and
          $\E_{x \in \mathcal{S}(\sqrt{n})}\left[x_1^{\lambda_1'} \dots x_{k-2}^{\lambda_{k-2}'} \right] = \E_{x \in \mathcal{S}(\sqrt{n})}\left[x_1^{\lambda_1} \dots x_{k-2}^{\lambda_{k-2}} \right]$.
          By the induction hypothesis, we have $(n)_{(\lambda')^t_1} \cdot \abs{\E_{x \in \mathcal{S}(\sqrt{n})}\left[x_1^{\lambda_1'} \dots x_{k-2}^{\lambda_{k-2}'} \right]} \le 3^{(k-2)^2} \cdot n^{(k-2)/2}$.
    \item Case $\lambda_1 = 2$: Let $\lambda'$ be the partition obtained from $\lambda$ by removing $\lambda_1$.
          Note that $\lambda_1^t = (\lambda')^t_1 + 1 \le k-2$. By the induction hypothesis, we have
          \begin{align*}
         (n)_{\lambda^t_1} \cdot \abs{\E_{x \in \mathcal{S}(\sqrt{n})}\left[x_1^{\lambda_1} \dots x_{k-2}^{\lambda_{k-2}} \right]} &\le n \cdot (n)_{(\lambda')^t_1} \cdot \abs{\E_{x \in \mathcal{S}(\sqrt{n})}\left[x_1^{\lambda_1'} \dots x_{k-2}^{\lambda_{k-2}'} \right]}\\
         &\le 3^{(k-2)^3} \cdot n^{k/2}.
          \end{align*}
    \item Case  $\lambda_1 = 1$: To bound $(n)_k \cdot \E_{x \in \mathcal{S}(\sqrt{n})}\left[x_1^{\lambda_1} \dots x_k^{\lambda_k} \right]$, we use~\cref{lem:slice_inv_exact}
          and the two preceding cases. Let $p(k)$ be the partition function, i.e., $p(k) = \abs{\set{\lambda \vdash k}}$. We  deduce that
{\footnotesize
            \begin{align*}
             (n)_k \cdot \abs{\E_{x \in \mathcal{S}(\sqrt{n})}\left[x_1^{\lambda_1} \dots x_k^{\lambda_k}\right]}  & \le n^{k/2} + \sum_{\lambda \vdash k \colon \lambda_1 \ge 2} \frac{\lambda!}{\lambda_1!\cdots \lambda_k!} \cdot \frac{(n)_{\lambda^t_1}}{\abs{\aut(\lambda)}} \cdot \abs{\E_{x \in \mathcal{S}(\sqrt{n})}\left[x_1^{\lambda_1} \dots x_k^{\lambda_k} \right]} \\
             & \le n^{k/2} + k!  \sum_{\lambda \vdash k \colon \lambda_1 \ge 2} (n)_{\lambda^t_1} \cdot \abs{\E_{x \in \mathcal{S}(\sqrt{n})}\left[x_1^{\lambda_1} \dots x_k^{\lambda_k} \right]} \\
             & \le n^{k/2} + k! \sum_{\lambda \vdash k \colon \lambda_1 \ge 3} (n)_{\lambda^t_1} \cdot \abs{\E_{x \in \mathcal{S}(\sqrt{n})}\left[x_1^{\lambda_1} \dots x_k^{\lambda_k} \right]}  + \\
             & \qquad \qquad k !\sum_{\lambda \vdash k \colon \lambda_1 = 2} (n)_{\lambda^t_1} \cdot \abs{\E_{x \in \mathcal{S}(\sqrt{n})}\left[x_1^{\lambda_1} \dots x_k^{\lambda_k} \right]}\\
             & \le 3^{(k-2)^3} \cdot k! \cdot (1 + p(k) + k) \cdot n^{k/2} \le 3^{k^3} \cdot n^{k/2},
          \end{align*}
}%
          as desired.
  \end{enumerate}
\end{proof}

\begin{claim}\label{claim:crude_bound_e}
  Suppose $k < \sqrt{n}/2$. We have
  $$
  \abs{\E_{x \in \mathcal{S}(\sqrt{n})}\left[x_1\dots x_k\right]} \le 2 \cdot 3^{k^3} \cdot n^{-k/2}.
  $$
\end{claim}

\begin{proof}
  Follows from~\cref{claim:bound_prod_exp_ff} and the bound on $k$.
\end{proof}

\begin{remark}
  In~\cref{claim:crude_bound_e}, the factor $3^{k^3}$ is too lossy to
  allow a meaningful bound with $k = n^{\epsilon}$, where $\epsilon >
  0$ is a constant.
\end{remark}

Refining the ideas of~\cref{claim:bound_prod_exp_ff}, we prove a
stronger lemma below which will imply a tighter bound on $e(k)$
sufficient for our application.
\begin{lemma}\label{lem:slice_inv_exp}
   There exists a universal constant $C \ge 1$ such that
   \begin{equation}\label{eq:abs_e_k_sum}
     \sum_{\lambda \vdash k} \frac{\lambda!}{\lambda_1!\cdots \lambda_k!} \cdot \frac{(n)_{\lambda^t_1}}{\abs{\aut(\lambda)}} \cdot \abs{\E_{x \in \mathcal{S}(\sqrt{n})}\left[x_1^{\lambda_1} \dots x_k^{\lambda_k} \right]} \le  k^{C \cdot k} \cdot n^{k/2}.
  \end{equation}
  In particular, for $n \ge 6$,~\cref{eq:abs_e_k_sum} holds with $C=2$.
\end{lemma}

\begin{proof}
  We induct on $k$. For $k = 1$, we have $n \cdot \abs{\E_{x \in \slice(\sqrt{n})} x_1} \le \sqrt{n}$ as desired.
  Using $e(2) = 0$ from~\cref{claim:e2} and the case $k=1$ of~\cref{eq:abs_e_k_sum}, we get that~\cref{lem:slice_inv_exp}
  also holds for $k=2$. Now, consider $k \ge 3$. Let $\Lambda_1 = \set{\lambda \vdash k \mid \lambda_1 = 1}$,
  $\Lambda_2 = \set{\lambda \vdash k \mid \lambda_1 = 2}$ and $\Lambda_{\ge 3} = \set{\lambda \vdash k \mid \lambda_1 = 3}$.
  Note that $\Lambda_1 \sqcup \Lambda_2 \sqcup \Lambda_{\ge 3} = \set{\lambda \vdash k}$ and $\abs{\Lambda_1} = 1$.

  For convenience define $a_{\lambda}$ to be the term associated to $\lambda \vdash k$ on the LHS of~\cref{eq:abs_e_k_sum}, i.e.,
  $$
  a_{\lambda}
  = \frac{\lambda!}{\lambda_1!\cdots \lambda_k!} \cdot \frac{(n)_{\lambda^t_1}}{\abs{\aut(\lambda)}} \cdot \abs{\E_{x \in \mathcal{S}(\sqrt{n})}\left[x_1^{\lambda_1} \dots
  x_k^{\lambda_k} \right]}.
  $$

  First we bound the contribution of the terms associated to
  partitions from $\Lambda_{\ge 3}$ in the LHS
  of~\cref{eq:abs_e_k_sum}. Let $\lambda'$ be the partition obtained
  from $\lambda$ by removing two boxes from $\lambda_1$.  Note that
  $\lambda_1^t = (\lambda')^t_1 \le k-2$ and
  $\E_{x \in \mathcal{S}(\sqrt{n})}\left[x_1^{\lambda_1'} \dots
  x_{k-2}^{\lambda_{k-2}'} \right]
  = \E_{x \in \mathcal{S}(\sqrt{n})}\left[x_1^{\lambda_1} \dots
  x_{k-2}^{\lambda_{k-2}} \right]$. Thus,
  \begin{align*}
  &a_{\lambda} = \frac{\lambda!}{\lambda_1!\cdots \lambda_k!} \cdot \frac{(n)_{\lambda^t_1}}{\abs{\aut(\lambda)}} \cdot \abs{\E_{x \in \mathcal{S}(\sqrt{n})}\left[x_1^{\lambda_1} \dots x_k^{\lambda_k} \right]}  \\
  &  \qquad   = \frac{k(k-1)}{\lambda_1 (\lambda_1-1)} \cdot \frac{\abs{\aut(\lambda')}}{\abs{\aut(\lambda)}} \frac{\lambda'!}{\lambda_1'!\cdots \lambda_k'!} \cdot \frac{(n)_{(\lambda')^t_1}}{\abs{\aut(\lambda')}} \cdot \abs{\E_{x \in \mathcal{S}(\sqrt{n})}\left[x_1^{\lambda_1'} \dots x_{k-2}^{\lambda_{k-2}'} \right]}  \\
  & \qquad = k^2 \cdot \frac{\abs{\aut(\lambda')}}{\abs{\aut(\lambda)}} \cdot a_{\lambda'} \le k^3 \cdot a_{\lambda'},
  \end{align*}
  since $\abs{\aut(\lambda')}/\abs{\aut(\lambda)} \le k-2 \le k$.
  For each $\lambda' \vdash k-2$, we can form a partition $\lambda \vdash k$ in $k-2 \le k$
  ways by adding two blocks to a single row of $\lambda'$. Hence, we have
  \begin{equation}\label{eq:lamb_ge_3_contrib}
    \sum_{\lambda \in \Lambda_{\ge 3}} a_{\lambda} \le k \cdot \sum_{\lambda' \vdash k-2} k^3 \cdot a_{\lambda'} \le k^4 \cdot k^{C \cdot (k-2)} \cdot n^{(k-2)/2},
  \end{equation}
  where the last equality follows from the induction hypothesis.

  Now we bound the contribution of the terms $a_{\lambda}$ associated to partitions $\lambda$
  from $\Lambda_{2}$ in the LHS of~\cref{eq:abs_e_k_sum}. Let $i \ge 1$
  be the number of parts of size two of $\lambda$
  and let $\lambda'$ be the partition obtained
  from $\lambda$ by removing these $i$ parts of size two.  Note that $\lambda_1^t =
  (\lambda')^t_1 + i \le k-1$. We have
  \begin{align*}
  &a_{\lambda} = \frac{\lambda!}{\lambda_1!\cdots \lambda_k!} \cdot \frac{(n)_{\lambda^t_1}}{\abs{\aut(\lambda)}} \cdot \abs{\E_{x \in \mathcal{S}(\sqrt{n})}\left[x_1^{\lambda_1} \dots x_k^{\lambda_k} \right]}  \\
  &  \qquad   \le n^i \cdot \frac{(k)_i}{2^i} \cdot \frac{\abs{\aut(\lambda')}}{\abs{\aut(\lambda)}} \cdot \frac{\lambda'!}{\lambda_1'!\cdots \lambda_k'!} \cdot \frac{(n)_{(\lambda')^t_1}}{\abs{\aut(\lambda')}} \cdot \abs{\E_{x \in \mathcal{S}(\sqrt{n})}\left[x_1 \dots x_{k-2i} \right]} \\
  &  \qquad   = n^i \cdot \frac{(k)_i}{2^i} \cdot \frac{1}{i!} \cdot \frac{\lambda'!}{\lambda_1'!\cdots \lambda_k'!} \cdot \frac{(n)_{(\lambda')^t_1}}{\abs{\aut(\lambda')}} \cdot \abs{\E_{x \in \mathcal{S}(\sqrt{n})}\left[x_1 \dots x_{k-2i} \right]},\\
  \end{align*}
  where in the last equality we used $\abs{\aut(\lambda')}/\abs{\aut(\lambda)} = 1/(i!)$.
  Since $\lambda \in \Lambda_2$ is uniquely specified by its number of parts of size two, applying the induction hypothesis we have
  \begin{align*}
    \sum_{\lambda \in \Lambda_2} a_{\lambda} & \le \sum_{i=1}^{\lfloor k/2 \rfloor} n^i \cdot \frac{(k)_i}{2^i} \cdot \frac{1}{i!} \cdot \left(\frac{\lambda'!}{\lambda_1'!\cdots \lambda_k'!} \cdot \frac{(n)_{(\lambda')^t_1}}{\abs{\aut(\lambda')}} \cdot \abs{\E_{x \in \mathcal{S}(\sqrt{n})}\left[x_1 \dots x_{k-2i} \right]} \right)\\
    &\le \sum_{i=1}^{\lfloor k/2 \rfloor} n^i \cdot \frac{(k)_i}{2^i} \cdot \frac{1}{i!} \cdot k^{C\cdot(k-2i)} \cdot n^{(k-2i)/2}\\
    &\le  k^{C \cdot (k-1) } \cdot n^{k/2} \cdot \sum_{i=0}^{\infty} k^{- C \cdot i} \le  \frac{3}{2} \cdot k^{C \cdot (k-1) } \cdot n^{k/2},
  \end{align*}
  where in the last inequality we used $k \ge 3$ and $C \ge 1$.

  Finally, we consider the case $\lambda_1 = 1$. To bound $a_{\lambda}$, we use~\cref{lem:slice_inv_exact}
  and the two preceding cases. We deduce that

  \begin{align*}
    a_{\lambda} & \le n^{k/2} + \sum_{\mu \in \Lambda_2} a_{\mu} + \sum_{\mu \in \Lambda_{\ge 3}} a_{\mu} \le n^{k/2} + k^4 \cdot k^{C \cdot (k-2)} \cdot n^{(k-2)/2} + \frac{3}{2} \cdot k^{C \cdot (k-1) } \cdot n^{k/2}\\
              & = k^{C \cdot k} \cdot n^{k/2} \left( \frac{1}{k^{C \cdot k}} + \frac{k^4}{n \cdot k^{2 \cdot C}} + \frac{3}{2 \cdot k^{C}} \right).
  \end{align*}
   We can bound the LHS of~\cref{eq:abs_e_k_sum} as
  \begin{align*}
    \sum_{\mu \in \Lambda_1} a_{\mu} +  \sum_{\mu \in \Lambda_2} a_{\mu} + \sum_{\mu \in \Lambda_{\ge 3}} a_{\mu} &\le
               k^{C \cdot k} \cdot n^{k/2} \left( \frac{1}{k^{C \cdot k}} + \frac{2 \cdot k^4}{n \cdot k^{2 \cdot C}} + \frac{3}{k^{C}} \right)\\
               &\le k^{C \cdot k} \cdot n^{k/2},
  \end{align*}
  provided $C > 0$ is a sufficiently large constant. In particular, the constant $C$ can be taken to be $2$ for $n \ge 6$.
\end{proof}

\begin{corollary}\label{cor:bound_on_coeff_e_k}
  We have
  $$
  \abs{\E_{x \in \mathcal{S}(\sqrt{n})}\left[x_1\dots x_k\right]} \le k^{3\cdot k} \cdot n^{-k/2}.
  $$
\end{corollary}

\begin{proof}
   Suppose $k \le \sqrt{n}$.
   Note that~\cref{lem:slice_inv_exp} implies that for $\lambda \vdash k$ with $\lambda_1$
   there exists a constant $C > 0$ such that
   \begin{align*}
     \frac{\lambda!}{\lambda_1!\cdots \lambda_k!} \cdot \frac{(n)_{\lambda^t_1}}{\abs{\aut(\lambda)}} \cdot \abs{\E_{x \in \mathcal{S}(\sqrt{n})}\left[x_1^{\lambda_1} \dots x_k^{\lambda_k} \right]} &= (n)_k \cdot \abs{\E_{x \in \mathcal{S}(\sqrt{n})}\left[x_1\dots x_k\right]}\\
      &\le  k^{C \cdot k} \cdot n^{k/2}.
   \end{align*}
   Simplifying and using the assumption $k \le \sqrt{n}$, we obtain
   $$
   \abs{\E_{x \in \mathcal{S}(\sqrt{n})}\left[x_1\dots x_k\right]} \le \frac{k^{C \cdot k} \cdot n^{-k/2}}{\prod_{i=1}^{k-1} \left(1 - \frac{i}{n}\right)} \le 2 \cdot k^{C \cdot k} \cdot n^{-k/2}.
   $$
   Furthermore, for $n \ge 6$,~\cref{lem:slice_inv_exp} allows us to choose $C=2$.
   Since $\abs{\E_{x \in \mathcal{S}(\sqrt{n})}\left[x_1 \right]} = 1/\sqrt{n}$, the simpler
   bound applies for all values of $k$
   $$
   \abs{\E_{x \in \mathcal{S}(\sqrt{n})}\left[x_1\dots x_k\right]} \le  k^{3 \cdot k} \cdot n^{-k/2},
   $$
   Now the assumption $n \ge 6$ can be removed since, for $k \ge 2$, we have$(k^{3}/\sqrt{n})^{k} \ge 1$,
   where $1$ is the trivial bound. Similarly, our initial assumption of $k \le \sqrt{n}$ can also be removed
   as the bound also becomes trivial in the regime $k > \sqrt{n}$.
\end{proof}

\chapter{The machinery and Qualitative bounds}\label{chap: qual}
In this chapter, we first state the main machinery that we use to prove our results. The machinery is a meta theorem that shows that under several linear algebraic conditions, a large random matrix is positive semidefinite (PSD) with high probability. This is similar in spirit to the PSDness argument in the SoS lower bounds for the Sherrington-Kirkpatrick Hamiltonian from the last chapter, although it's quite a bit more involved.

The machinery originally appeared in \cite{potechin2020machinery, potechin2022sub}, where the complete proof can be found. In \cref{quantitativetheoremstatementsection}, we state the machinery. Compared to that work, we significantly simplify the required definitions needed to state and apply the machinery. Such a decluttering of the definitions is possible since we don't provide the proof and simply apply the theorem. For example, we don't define ribbons, we don't formally define the technical matrix $M^{orth}_{\tau}(H)$ and we work in the simplified setting of Rademacher or Gaussian variables instead of arbitrary distributions with finite moments. Moreover, we interpret the conditions of the machinery as conditions of the problem itself, rather than computational linear algebraic conditions on the moment matrix as in \cite{potechin2020machinery}. Cast in this framework, this makes clear the potential connection to the low degree likelihood ratio hypothesis as described in \cref{chap: sos} and \cref{chap: future_work}.

After stating the machinery in \cref{quantitativetheoremstatementsection}, we exhibit the qualitative bounds for applying the machinery to our problems of interest. The material in these sections is also adapted from \cite{potechin2020machinery}. The main difference is that we improve the exposition by fixing various typos and clarifying various technical arguments.

\section{The machinery}\label{quantitativetheoremstatementsection}

In this section, we describe the machinery we apply to show SoS lower bounds.
As we have already seen in \cref{chap: sk}, the general idea to show SoS lower bounds in this work is to decompose the moment matrix into graph matrices, which are matrix-valued functions of the input entries, and then show PSDness by exhibiting an approximate PSD decomposition. The machinery takes a similar approach and provides an approximate PSD decomposition, using certain decay properties of the Fourier coefficients as well as norm bounds similar to the ones obtained in \cref{chap: efron_stein}.

Consider a hypothesis testing problem $\calP$. We will assume the setup in \cref{subsec: pseudocalibration}. Therefore, $\calP$ is a distinguishing problem between two distributions -- the random distribution (null hypothesis) and the planted distribution (alternative hypothesis). As we saw earlier, we could use the technique of pseudo-calibration to obtain a candidate moment matrix $\Lda$, such that $\Lda$ can potentially serve as an SoS lower bound for the related optimization task on the random distribution. The machinery gives general conditions on $\Lda$ that ensure feasibility with high probability. In particular, the machinery is a set of linear algebraic conditions on $\Lda$ that imply positivity of $\Lda$ w.h.p. and as we saw earlier, the other required feasibility conditions follow easily from pseudo-calibration.

In this work, we slightly diverge from this viewpoint (originally presented in \cite{potechin2020machinery}) and instead view these conditions as properties of the problem $\calP$ directly. Therefore, the machinery can be construed as a claim of feasibility of the pseudo-cailbrated pseudo-expectation operator, under certain conditions on the problem $\calP$. To state the machinery in this language, we need some definitions that follow next.

\subsection{Shapes and graph matrices}

Consider the setting when the input distribution is a Rademacher $G_{n, 1/2}$ graph with the input entries being $\chi_e \in \{-1, 1\}$. For $T \subseteq \binom{[n]}{2}$, let $\chi_T = \prod_{e \in T} \chi_e$ be the standard Fourier basis. In this setting, shapes were already defined in \cref{chap: efron_stein}. Here, for technical reasons, we slightly modify the definitions so that the rows and columns are indexed by sub-tuples of $[n]$ rather than subsets of $[n]$. The techniques developed in \cref{chap: efron_stein} still carry over to bound the norms of such graph matrices.

\begin{definition}[Shapes in the setting of Rademacher $G_{n, 1/2}$ inputs]
    A shape $\alpha = (V(\alpha),E(\alpha),U_{\alpha},V_{\alpha})$ is a graph on vertices $V(\alpha)$ and edges $E(\al)$ with two distinguished tuples of vertices $U_{\al}, V_{\al} \subseteq V(\al)$. Note that $U_{\al}, V_{\al}$ are ordered subsets (tuples).
\end{definition}

As we saw earlier, we can define corresponding matrices for each shape, that are termed graph matrices. Recall that a realization is an injective map from $V(\al)$ to $[n]$. The main difference here, as compared to \cref{chap: efron_stein}, is that in the definition of graph matrices, we sum over realizations $\phi$ that correspond to distinct characters, rather than all realizations $\phi$.

To capture this notion precisely, we use the following definition. Define two realizations (injective maps from $V(\al)$ to $[n]$) $\phi, \phi'$ to be equivalent if $\phi(U_{\al}) = \phi'(U_{\al}), \phi(V_{\al}) = \phi'(V_{\al})$ as tuples and $\phi(E(\al)) = \phi'(E(\al))$ as sets. Let the set of non-equivalent realizations of $\al$ be denoted $\mathrm{Real}(\al)$.

\begin{definition}[Graph matrices in the setting of Rademacher $G_{n, 1/2}$ inputs]
    For a shape $\al$, the graph matrix $M_{\al}$ is a matrix-valued function with rows and columns indexed by sub-tuples of $[n]$ of sizes $|U_{\al}|, |V_{\al}|$ respectively, which is defined as follows: It maps input graph $G \in \{\pm 1\}^{\binom{n}{2}}$ (wich associated fourier characters $\chi_E$) to a matrix with the $A, B$-th entry being \[M_{\alpha}(A,B) = \sum_{\substack{\phi(U_{\al}) = A, \phi(V_{\al}) = B\\\phi \in \mathrm{Real}(\al)}}{\chi_{E(\al)}}\]

\end{definition}

\begin{definition}[Shape transposes]
    For a shape $\alpha = (V(\alpha),E(\alpha),U_{\alpha},V_{\alpha})$, define its transpose  $\al^T$ to be $\alpha^T = (V(\alpha),E(\alpha),V_{\alpha},U_{\alpha})$.
    Note that $M_{\alpha^T} = M_{\alpha}^T$ as matrix transpose.
\end{definition}

\begin{example}
    In \cref{fig: sample_shape}, consider the shape $\al$ as shown. We have $U_{\al} = (u_1, u_2), V_{\al} = (v_1), V(\al) = \{u_1, u_2, v_1, w_1\}$ and $E(\al) = \{\{u_1, w_1\}, \{u_2, w_1\}, \{w_1, v_1\}\}$. $M_{\al}$ is a matrix with rows and columns indexed by tuples of length $|U_{\al}| = 2$ and $|V_{\al}| = 1$ respectively. The nonzero entries will have rows and columns indexed by $(a_1, a_2)$ and $b_1$ respectively for all distinct $a_1, a_2, b_1$, with the corresponding entry being $M_{\al}((a_1, a_2), (b_1)) = \sum_{c_1 \in [n] \setminus \{a_1, a_2, b_1\}} \chi_{a_1, c_1}\chi_{a_2, c_1}\chi_{c_1, b_1}$. Here, the injective map $\phi$ maps the vertices $u_1, u_2, w_1, v_1$ to $a_1, a_2, c_1, b_1$ respectively and we sum over all such maps (as they are all pairwise non-equivalent). Succinctly, \[M_{\al} =
    \begin{blockarray}{rl@{}c@{}r}
        & & \makebox[0pt]{column $(b_1)$} \\[-0.5ex]
        & & \,\downarrow \\[-0.5ex]
        \begin{block}{r(l@{}c@{}r)}
            &  & \vdots & \\[-0.2ex]
            \text{row }(a_1, a_2) \to \mkern-9mu & \raisebox{0.5ex}{\makebox[3.2em][l]{\dotfill}} & \sum_{c_1 \in [n] \setminus \{a_1, a_2, b_1\}} \chi_{a_1, c_1}\chi_{a_2, c_1}\chi_{c_1, b_1} & \raisebox{0.5ex}{\makebox[4.2em][r]{\dotfill}} \\[+.5ex]
            &  & \vdots & \\
        \end{block}
    \end{blockarray}\]
\end{example}

\begin{figure}[!h]
    \centering
    \includegraphics[scale=.3, trim={0 0 0 0},clip]{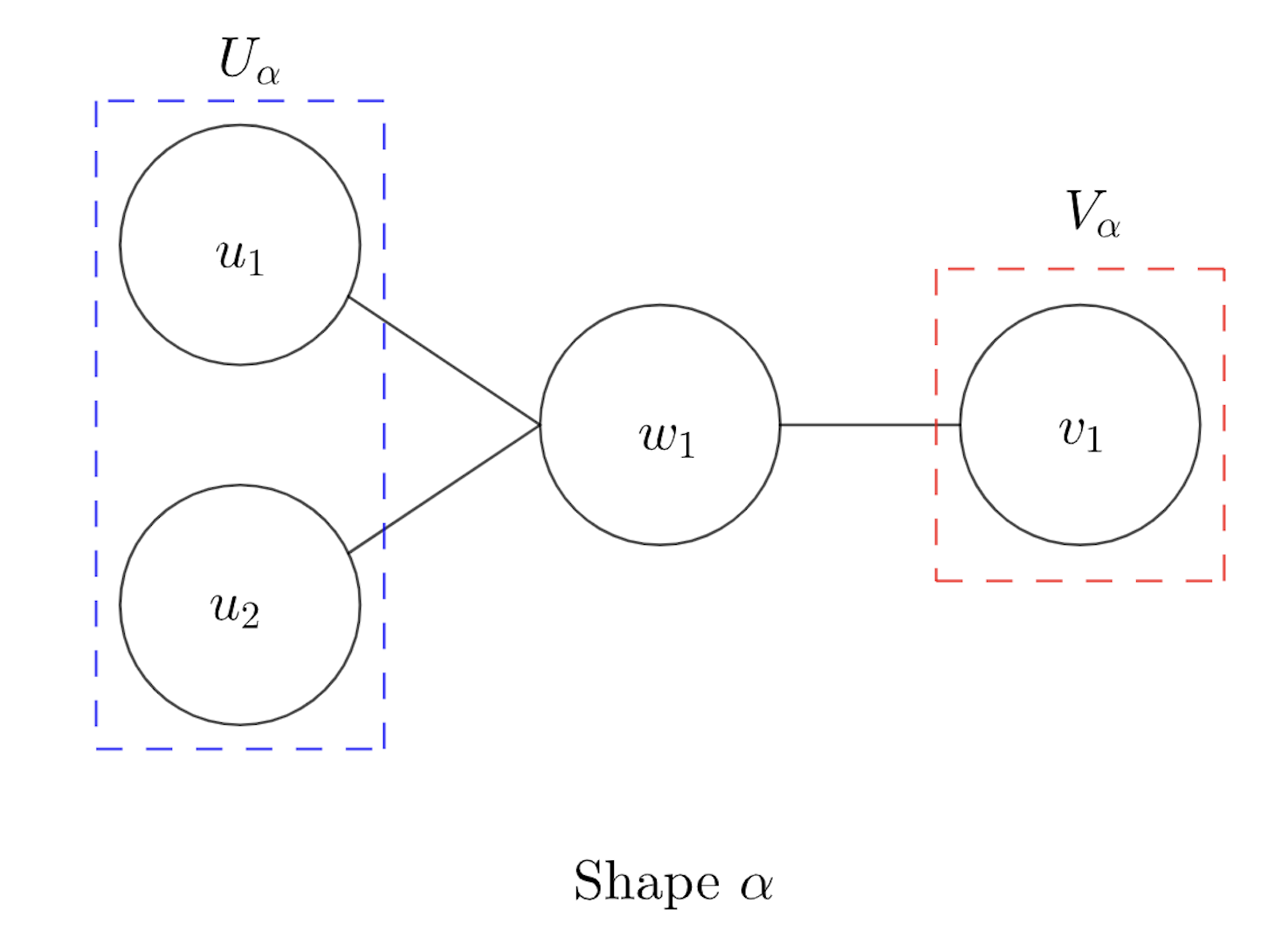}
    \caption{Example of a shape}
    \label{fig: sample_shape}
\end{figure}

Some simple matrices such as the adjacency matrix of a graph and the identity matrix are also graph matrices, as we see below
\begin{itemize}
    \item Take $\al$ to be a shape with two vertices $u, v$ with $U_{\al} = (u), V_{\al} = (v)$ and exactly one edge $\{u, v\}$. Then, $M_{\al}$ has rows and columns indexed by $[n]$ (more specifically tuples of length $1$) with the $i, j$-th entry being $G_{ij}$ if $i \neq j$ and $0$ otherwise. Therefore, $M_{\al}$ is just the $\pm 1$ adjacency matrix of the graph $G$.
    \item Take $\al$ to be the shape with exactly $1$ vertex $u$, no edges and $U_{\al} = V_{\al} = (u)$. Then, $M_{\al}$ is the identity matrix of size $n \times n$.
\end{itemize}

For more examples of graph matrices and why they can be a useful tool to work with, see \cite{ahn2016graph}.
We now define some terms to capture the rows and columns of graph matrices.

\paragraph{Matrix indices and index shapes}
In the above setting of Rademacher $G_{n, 1/2}$ inputs, a matrix index $A$ is a tuple of indices $(a_1, \ldots, a_{|A|})$ where $a_i \in [n]$. When the SoS variables are $y_1, \ldots, y_n$, we associate to this matrix index $A$ the monomial $\prod_{i \le |A|} y_{a_i}$. With this definition, graph matrices have as rows and columns matrix indices.

Define an index shape $U = (u_1, \ldots, u_{|U|})$ to be a tuple of formal variables $u_i$, or in other words, unspecified indices. If $|U| = t$, we say that any matrix index $A$ of length $t$ has shape $U$. We say two index shapes $U, V$ are equivalent, denoted $U \equiv V$ if $|U| = |V|$. Finally, define the weight of $U$ to be $w(U) = |U|$ and the automorphism group $Aut(U) = S_{|U|}$ (the permutations of the elements of $U$). The latter definition is needed for describing coefficients.

\paragraph{Shape definitions}
We say a shape $\al$ is proper if it has no isolated vertices (i.e. no degree $0$ vertices) outside $U_{\al} \cup V_{\al}$. We say a shape $\al$ is trivial if $U_{\al}$ and $V_{\al}$ are equal as sets, and they constitute all the vertices in $\al$. and moreover, there are no edges in $\al$.

A path is a sequence of vertices of $V(\al)$ such that every consecutive pair of vertices form an edge in $V(\al)$. A vertex separator of $\al$ is a set of vertices $S$ such that every path from $U_{\al}$ to $V_{\al}$ passes through $S$. As we saw in \cref{chap: efron_stein}, the norm bounds of the graph matrix $M_{\al}$ rely on the size of the minimum vertex separator of $\al$. Define the weight of a vertex separator $S$ as $|S|$.

The above definitions are sufficient for the application to the Planted Slightly Denser subgraph problem.  But when we work with Tensor PCA and Sparse PCA, we need to generalize the notion of shapes and graph matrices. These generalized shapes and graph matrices were studied in \cite{ahn2016graph}. Now, we describe the required generalizations.

\subsubsection{Definitions for Tensor PCA}

In the Tensor PCA application, the input is a tensor $A \in \RR^{[n]^k}$. To incorporate this, we modify our definitions of shapes and index shapes accordingly. The input entries are now sampled from the distribution $\GN(0, 1)$ instead of $\{-1, 1\}$. So, we will work with the Hermite basis of polynomials.
Let the standard unnormalized Hermite polynomials be denoted as $h_0(x) = 1, h_1(x) = x, h_2(x) = x^2 - 1, \ldots$. Then, we work with the basis $h_a(A) \defeq \prod_{e \in [n]^k} h_e(A_e)$ over $a \in \NN^{[n]^k}$. Accordingly, we will modify the graphs that represent shapes, to have labeled hyperedges of arity $k$. So, an hyperedge $e$ with a label $t$ will correspond to the hermite polynomial $h_t(A_e)$.

\begin{definition}[Hyperedges]
    Instead of standard edges, we will have labeled hyperedges of arity $k$ in the underlying graphs for our ribbons as well as shapes. The label for an hyperedge $e$, denoted $l_e$, is an element of $\NN$ which will correspond to the Hermite polynomial being evaluated on that entry.
\end{definition}

Note that our hyperedges are ordered since the tensor $A$ is not necessarily symmetric.
For variables $x_1, \ldots, x_n$, the rows and columns of our moment matrix will now correspond to monomials of the form $\prod_{i \le n} x_i^{p_i}$ for $p_i \ge 0$. To capture this, we use the notion of index shape pieces and index shapes. Informally, we split the above monomial product into groups based on their powers and each such group will form an index shape piece.

\begin{definition}[Index shape piece]
    An index shape piece $U_i= ((U_{i, 1}, \ldots, U_{i, t}), p_i)$ is a tuple of indices $(U_{i, 1}, \ldots, U_{i, t})$ along with a power $p_i \in \NN$. Let $V(U_i)$ be the set $\{U_{i, 1}, \ldots, U_{i, t}\}$ of vertices of this index shape piece. When clear from context, we use $U_i$ instead of $V(U_i)$.
\end{definition}

If we realize $U_{i, 1}, \ldots, U_{i, t}$ to be indices $a_1, \ldots, a_t \in [n]$, then this realization of this index shape piece corresponds to the monomial $\prod_{j \le t} x_{a_j}^{p_i}$.

\begin{definition}[Index shape]
    An index shape $U$ is a set of index shape pieces $U_i$ that have different powers. Let $V(U)$ be the set of vertices $\cup_i V(U_i)$. When clear from context, we use $U$ instead of $V(U)$.
\end{definition}

Observe that each realization of an index shape corresponds to a row or column of the moment matrix.
Equivalence of index shapes is analogous, namely, for two index shapes $U, V$, we write $U \equiv V$ if for all powers $p$, the index shape pieces of power $p$ in $U$ and $V$ have the same length.
We also define the automorphism group of $U$ as $Aut(U) = \prod_{U_i \in U}{Aut(U_i)}$ where the automorphism group of an index shape piece $U_i$ is $Aut(U_i) = S_{|U_i|}$.
In the definition of shapes, the distinguished set of vertices should now be replaced by index shapes.

\begin{definition}[Shapes]
    Shapes are tuples $\al = (V(\al), E(\al), U_{\al}, V_{\al})$ where $(V(\al), E(\al))$ is a graph with hyperedges of arity $k$ and $U_{\al}, V_{\al}$ are index shapes such that $U_{\al}, V_{\al} \subseteq V(\al)$.
\end{definition}

A shape $\al$ is proper if it has no isolated vertices outside $U_{\al} \cup V_{\al}$, no multi-edges and all the edges have a nonzero label.
To define the notion of vertex separators, we accordingly modify the notion of paths for hyperedges instead of edges. Formally, a path is a sequence of vertices $u_1, \ldots, u_t$ such that $u_i, u_{i + 1}$ are in the same hyperedge, for all $i \le t - 1$.
The notion of vertex separator is identically defined with the above notion of hyperedges and paths.
Finally, the definition of trivial shape $\tau$ is similar, the only change being that we now require $U_{\tau} \equiv V_{\tau}$ instead of saying they're equal as sets.

\subsubsection{Definitions for Sparse PCA}

We are given the $m$ vectors $v_1, \ldots, v_m \in \RR^d$ as input. Similar to Tensor PCA, we will work with the Hermite basis of polynomials since the entries are sampled from the distribution $\GN(0, 1)$.
In particular, if we denote the unnormalized Hermite polynomials by $h_0(x) = 1, h_1(x) = x, h_2(x) = x^2 - 1, \ldots$, then, we work with the basis $h_a(v) \defeq \prod_{i \in [m], j \in [n]} h_{a_{i, j}}(v_{i, j})$ over $a \in \NN^{m \times n}$. To capture this basis, we will modify the graphs that represent shapes to be bipartite graphs with two types of vertices, and have labeled edges that go across vertices of different types. So, an edge $(i, j)$ with label $t$ between a vertex $i$ of type $1$ and a vertex $j$ of type $2$ will correspond to $h_t(v_{i, j})$.

Formally, we will have two types of vertices, the vertices corresponding to the $m$ input vectors that we call type $1$ vertices and the vertices corresponding to ambient dimension of the space that we call type $2$ vertices. For a shape with such vertices, edges will go across vertices of different types, thereby forming a bipartite graph. An edge between a type $1$ vertex $i$ and a type 2 vertex $j$ corresonds to the input entry $v_{i, j}$. Each edge will have a label in $\NN$ corresponding to the Hermite polynomial evaluated on that entry.

We will have variables $x_1, \ldots, x_n$ in our SoS program, so we will work with index shape pieces and index shapes as in Tensor PCA, since the rows and columns of our moment matrix will now correspond to monomials of the form $\prod_{i \le n} x_i^{p_i}$ for $p_i \ge 0$. But since we have $2$ types of vertices, we need to slightly modify the notion of index shape pieces and index shapes.

\begin{definition}[Index shape piece]
    An index shape piece $U_i= ((U_{i, 1}, \ldots, U_{i, t}), t_i, p_i)$ is a tuple of indices $(U_{i, 1}, \ldots, U_{i, t})$ along a type $t_i \in \{1, 2\}$ with a power $p_i \in \NN$. Let $V(U_i)$ be the set $\{U_{i, 1}, \ldots, U_{i, t}\}$ of vertices of this index shape piece. When clear from context, we use $U_i$ instead of $V(U_i)$.
\end{definition}

For an index shape piece $((U_{i, 1}, \ldots, U_{i, t}), t_i, p_i)$ with type $t_i = 2$, if we realize $U_{i_1}, \ldots, U_{i_t}$ to be indices $a_1, \ldots, a_t \in [n]$, then, this index shape pieces correspond this to the monomial $\prod_{j \le n} x_{a_j}^{p_i}$.

\begin{definition}[Index shape]
    An index shape $U$ is a set of index shape pieces $U_i$ that have either have different types or different powers. Let $V(U)$ be the set of vertices $\cup_i V(U_i)$. When clear from context, we use $U$ instead of $V(U)$.
\end{definition}

Each realization of an index shape will correspond to a row or column of the moment matrix. For our moment matrix, the only nonzero rows correspond to index shapes that have only index shape pieces of type $2$, since the only SoS variables are $x_1 \ldots, x_n$, but in order to do our analysis, we need to work with the generalized notion of index shapes that allow index shape pieces of both types.

Analogous to our previous definitions, for two index shapes $U, V$, we write $U \equiv V$ if for all types $t$ and all powers $p$, the index shape pieces of type $t$ and power $p$ in $U$ and $V$ have the same length.
Since we are working with standard graphs, the original notion of path and vertex separator will work , but we will now use the minimum weight vertex separator instead of the minimum vertex separator where we define the weight as follows.

\begin{definition}[Weight of an index shape]
    Suppose we have an index shape $U = \{U_1, U_2\} \in \calI_{mid}$ where $U_1 = ((U_{1, 1}, \ldots, U_{1, |U_1|}), 1, 1)$ is an index shape piece of type $1$ and $U_2 = ((U_{2, 1}, \ldots, U_{2, |U_2|}), 2, 1)$ is an index shape piece of type $2$. Then, define the weight of this index shape to be $w(U) = \sqrt{m}^{|U_1|}\sqrt{n}^{|U_2|}$.
\end{definition}

The definition carries over for a vertex separator as well. We also define the automorphism group of $U$ as $Aut(U) = \prod_{U_i \in U}{Aut(U_i)}$ where the automorphism group of an index shape piece $U_i$ is $Aut(U_i) = S_{|U_i|}$. We now give the modified definition of shapes.

\begin{definition}[Shapes]
    Shapes are tuples $\al = (V(\al), E(\al), U_{\al}, V_{\al})$ where $(V(\al), E(\al))$ is a graph with two types of vertices, has labeled edges only across vertices of different types and $U_{\al}, V_{\al}$ are index shapes such that $U_{\al}, V_{\al} \subseteq V(\al)$.
\end{definition}

The other definitions that follow are analogous. A shape $\al$ is proper if it has no isolated vertices outside $U_{\al} \cup V_{\al}$, no multi-edges and all the edges have a nonzero label. In the definition of trivial shape $\tau$, just as in Tensor PCA, we require $U_{\tau} \equiv V_{\tau}$ instead of saying they're equal as sets.

\subsection{Decomposing shapes}

Compared to the lower bound strategy in the Sherrington-Kirkpatrick lower bound in \cref{chap: sk}, the main strategy in the machinery is to provide an approximate PSD decomposition by decomposing shapes $\al$ into three other shapes $\sig, \tau, \sig'^T$ such that $M_{\al} \approx M_{\sig}M_{\tau} M_{\sig'^T}$. Then, the idea is to argue that the graph matrix coefficients of the moment matrix also decompose similarly, ending with a PSD decomposition showing that the moment matrix is PSD.

We first need to define composition of shapes. We say that shapes $\alpha$ and $\beta$ are composable if $U_{\beta} \equiv V_{\alpha}$. In this case, define their composition to be the shape $\al \circ \beta$ which is obtained by concatenating $\al, \beta$ while gluing together $U_{\beta}, V_{\al}$. Formally, $\al\circ \beta$ is such that $U_{\alpha \circ \beta} = U_{\alpha}$, $V_{\alpha \circ \beta} = V_{\beta}$, and after setting $U_{\beta} = V_{\alpha}$, we take $V(\alpha \circ \beta) = V(\alpha) \cup V(\beta)$, and finally, $E(\alpha \circ \beta) = E(\alpha) \cup E(\beta)$.

Note that by doing this, the concatenated shape could become improper if edges repeat.
We remark that shape composition is not necessarily commutative, but it is associative.

\begin{figure}[!ht]
    \centering
    \includegraphics[scale=0.45, trim={4.5cm 2cm 0 2cm},clip]{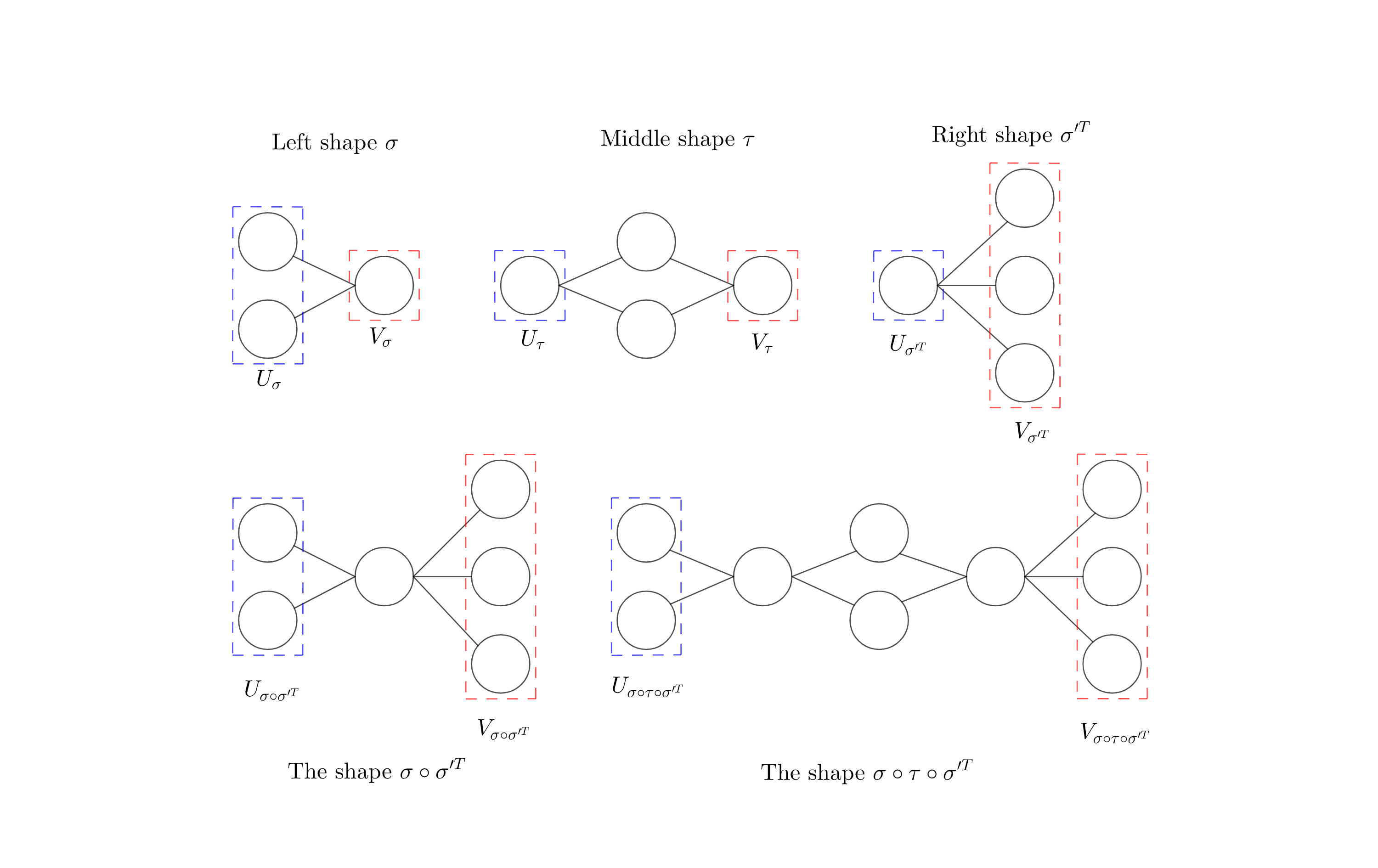}
    \caption{Illustration of shape composition and decomposition.}
    \label{fig: basic_shape_comp}
\end{figure}

\begin{example}
    \cref{fig: basic_shape_comp} illustrates an example of shape composition in the setting where there is only one type of vertex. Observe how the shapes $\sig \circ \sig'^T$ and $\sig \circ \tau \circ \sig'^T$ are obtained from the shapes $\sig, \tau$ and $\sig'^T$.
\end{example}

\begin{example}
    \cref{fig: shape_comp} illustrates an example of shape composition in the setting where there are two types of vertices. We have two types of vertices that we diagrammaticaly represent by squares and circles. Observe how the shapes $\sig \circ \sig'^T$ and $\sig \circ \tau \circ \sig'^T$ are obtained from the shapes $\sig, \tau$ and $\sig'^T$.
\end{example}

\begin{figure}[!ht]
    \centering
    \includegraphics[scale=0.38, trim={8cm 2cm 0 2cm},clip]{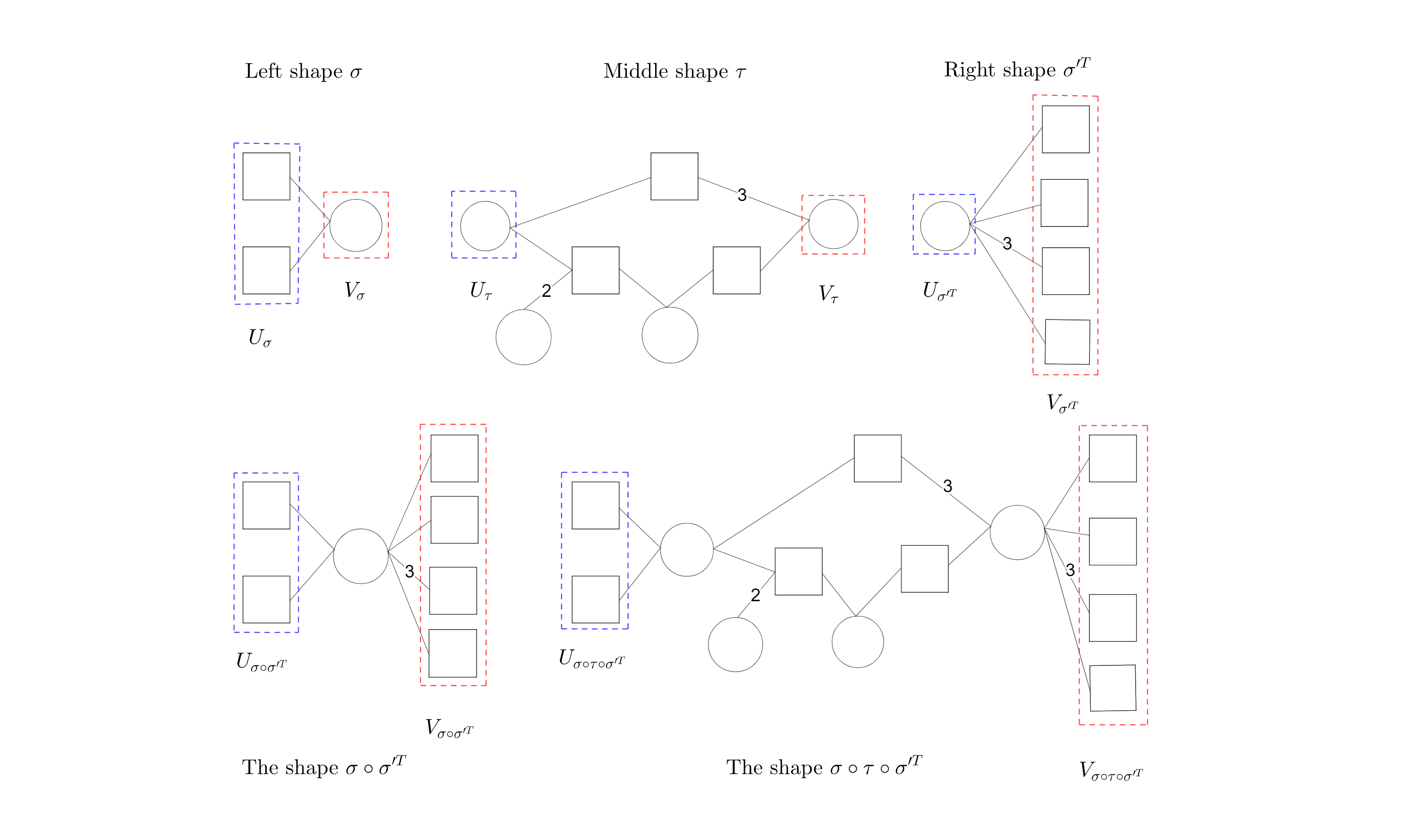}
    \caption{Illustration of shape composition and decomposition.}
    \label{fig: shape_comp}
\end{figure}

Previously, we defined the notion of minimum vertex separators (and analogously, minimum weight vertex separators). In what follows, we collectively term either of them as minimum weight vertex separators. by Define the \textit{leftmost} (resp. \textit{rightmost}) minimum-weight vertex separator $S$ (resp. $T$) to be a minimum-weight vertex separator such that for every other minimum-weight vertex separator $S'$ (resp. $T'$), $S$ separates $U_{\al}$ from $S'$ (resp. $T'$ from $V_{\al}$).
In \cite{BHKKMP16, potechin2020machinery}, it's shown that these are well-defined.

With these definitions in hand, we can now define how to decompose a shape $\al$ into its left, middle and right parts $\sig, \tau, \sig'^T$ respectively.

\begin{definition}[Shape decomposition]
    Let $\alpha$ be a shape. Let $S$ and $T$ be the leftmost and rightmost minimum-weight vertex separators of $\alpha$ together with some orderings $O_S,O_T$ of $S$ and $T$.
    \begin{itemize}
        \item We define the left part $\sigma$ of $\alpha$ to be the shape formed by taking the induced subgraph on all of the vertices of $\alpha$ reachable from $U_{\alpha}$ without passing through $S$ (but including the vertices of $S$) where all edges and hyperedges within $S$ are removed, and we take $U_{\sigma} = U_{\alpha}$ and $V_{\sigma} = (S,O_S)$.
        \item We define the right part ${\sigma'}^T$ of $\alpha$ to be the shape formed by taking the induced subgraph on all of the vertices of $\alpha$ reachable from $V_{\alpha}$ without passing through $T$ (but we include the vertices of $T$) where all edges and hyperedges within $T$ are removed, and we take $V_{{\sigma'}^T} = V_{\alpha}$ and $U_{{\sigma'}^T} = (T,O_T)$.
        \item Finally, we define the middle part $\tau$ of $\alpha$ to be the shape formed by the induced subgraph on all of the vertices of $\alpha$ which are not reachable from $U_{\alpha}$ and $V_{\alpha}$ without touching $S$ and $T$ respectively (but we include the vertices of $S$ and $T$), where we also include the edges or hyperedges entirely within $S$ and the edges or hyperedges entirely within $T$, and we take $U_{\tau} = (S,O_S)$ and $V_{\tau} = (T,O_T)$.
    \end{itemize}
\end{definition}

It's evident from the definition that $\al = \sig \circ \tau \circ \sig'^T$.

\begin{example}
    \cref{fig: basic_shape_comp} illustrates an example decomposition in the setting where there is only one type of vertex.
    \begin{enumerate}
        \item If we start with the shape $\al$ denoted as $\sig \circ \sig'^T$, observe that there is a unique minimum vertex separator, which consists of the middle vertex of degree $5$, i.e. the one that's not in either $U_{\sig \circ \sig'^T}$ or $V_{\sig \circ \sig'^T}$.
        Then, $\al$ is decomposed in to the left part $\sig$, a trivial middle part $\tau$ (not shown in this figure) which has $V(\tau) = \{u\}, U_{\tau} = V_{\tau} = (u), E(\tau) = \emptyset$, and the right part $\sig'^T$.
        \item If we start with the shape $\al$ denoted as $\sig \circ \tau \circ \sig'^T$, then the leftmost minimum vertex separator is the vertex of degree $4$ and the rightmost minimum vertex separator is the vertex of degree $5$. Then, $\al$ is decomposed into the left part $\sig$, the middle part $\tau$ and the right part $\sig'^T$, which are all shown in this figure.
    \end{enumerate}
\end{example}

\begin{example}
    \cref{fig: shape_comp} illustrates an example decomposition in the setting where there are two types of vertices. We have two types of vertices that we diagrammatically represent by squares and circles. In this example, we assume that the set containing a single circle vertex has a lower weight compared to a set of two square vertices.
    \begin{enumerate}
        \item If we start with the shape $\sig \circ \sig'^T$, then it can be decomposed uniquely in to the composition of the left shape $\sig$, the right shape $\sig'^T$. In this case, the middle shape (not shown in this figure) is trivial.
        \item If we start with the shape $\sig \circ \tau \circ \sig'^T$, then it can be decomposed uniquely into the composition of the left shape $\sig$, the middle shape $\tau$ and the right shape $\sig'^T$, which are all shown in this figure.
    \end{enumerate}
\end{example}

At this point, the definitions of left shapes, middle shapes and right shapes are natural.
We say that a shape $\sigma$ is a left shape if $\sigma$ is a proper shape, $V_{\sigma}$ is the left-most and right-most minimum-weight separator of $\sigma$, every vertex in $V(\sigma) \setminus V_{\sigma}$ is reachable from $U_{\sigma}$ without touching $V_{\sigma}$, and $\sigma$ has no hyperedges entirely within $V_{\sigma}$.
Similarly, we say that a shape $\tau$ is a proper middle shape if $\tau$ is a proper shape, $U_{\tau}$ is the left-most minimum-weight separator of $\tau$, and $V_{\tau}$ is the right most minimum-weight separator of $\tau$. We do not define improper middle shapes, which are needed in the machinery proof, but not here.
Finally, we say that a shape ${\sigma}^{T}$ is a right shape if it is the transpose of a left shape.

\subsection{Coefficient matrices}

We have all the necessary definitions in place for shapes and graph matrices. To apply the machinery, we will decompose the given moment matrix $\Lda$ as $\Lda = \sum \lda_{\al}M_{\al}$ where the sum is over all shapes $\al$. The coefficients $\lda_{\al}$ are then assembled into matrices, termed coefficient matrices, that we will define next. The conditions of the machinery will be in terms of these coefficient matrices.

We will begin with some notations for different sets of index shapes. Given a moment matrix $\Lambda$, define $\mathcal{I}(\Lambda)$ to the set of matrix shapes $U$ such that some row or column index of $\Lda$ has shape $U$. Define $w_{max} = \max{\{w(U):U \in \mathcal{I}(\Lambda)\}}$ to be the maximum possible weight of an index shape in $\calI(\Lda)$. Finally, define $\mathcal{I}_{mid}$ to be $\mathcal{I}_{mid} = \{U: w(U) \leq w_{max}, \forall U_i \in U, p_i = 1\}$.
Observe that in the setting of Rademacher $G_{n, 1/2}$ inputs, we have $\mathcal{I}_{mid} = \{U: |U| \leq w_{max}\}$.

In pseudo-calibration, we only keep the shapes that satisfy certain truncation parameters that we choose. Formally, satisfication of truncation parameters is defined as follows.

\begin{definition}[Truncation parameters for the setting of Rademacher $G_{n, 1/2}$ inputs]
    For integers $D_{sos}, D_V \ge 0$, say that a shape $\al$ satisfies the truncation parameters $D_{sos}, D_V$ if
    \begin{itemize}
        \item The degrees of the monomials that $U_{\al}$ and $V_{\al}$ correspond to, are at most $\frac{D_{sos}}{2}$
        \item The left part $\sig$, the middle part $\tau$ and the right part $\sig'$ of $\al$ satisfy the bounds $|V(\sig)|, |V(\tau)|, |V(\sig')| \le D_V$.
    \end{itemize}
\end{definition}

\begin{definition}[Truncation parameters for Tensor PCA and Sparse PCA]
    For integers $D_{sos}, D_V, D_E \ge 0$, say that a shape $\al$ satisfies the truncation parameters $D_{sos}, D_V, D_E$ if
    \begin{itemize}
        \item The degrees of the monomials that $U_{\al}$ and $V_{\al}$ correspond to, are at most $\frac{D_{sos}}{2}$
        \item The left part $\sig$, the middle part $\tau$ and the right part $\sig'^T$ of $\al$ satisfy the bounds $|V(\sig)|, |V(\tau)|, |V(\sig'^T)| \le D_V$
        \item For each $e \in E(\al)$, $l_e \le D_E$.
    \end{itemize}
\end{definition}

We also need to define the sets of shapes which can appear when analyzing $\Lambda$.
Given a moment matrix $\Lambda$, define $\mathcal{L} = \{\sigma: \sigma \text{ is a left shape}, U_{\sigma} \in \mathcal{I}(\Lambda), V_{\sigma} \in \mathcal{I}_{mid}, |V(\sigma)| \leq D_V, \forall e \in E(\sigma), l_e \leq D_E\}$. Moreover, given $V \in \mathcal{I}_{mid}$, define $\mathcal{L}_V = \{\sigma \in \mathcal{L}: V_{\sigma} \equiv V\}$. Finally, given $U \in \mathcal{I}_{mid}$, define $\mathcal{M}_U = \{\tau: \tau \text{ is a non-trivial proper middle shape}, U_{\tau} \equiv V_{\tau} \equiv U, |V(\tau)| \leq D_V, \forall e \in E(\tau), l_e \leq D_E\}$.

We are now ready to define coefficient matrices. Given a moment matrix $\Lambda$, a coefficient matrix is a matrix $H$ whose rows and columns are indexed by left shapes $\sigma,\sigma' \in \mathcal{L}$. $H$ is called SoS-symmetric if  $H(\sigma,\sigma')$ is invariant under the action of the symmetric group, i.e. if we permute the vertices of $U_{\sigma}$ and the vertices of $U_{\sigma'}$ (where we only permute within the same type) under the same permutation, then the entry doesn't change.

For a shape $\tau$, we say that a coefficient matrix $H$ is a $\tau$-coefficient matrix if $H(\sigma,\sigma') = 0$ whenever $V_{\sigma} \not\equiv U_{\tau}$ or $V_{\tau} \not\equiv U_{{\sigma'}^T}$. Given an index shape $U$, we define $Id_{U}$ to be the shape with $U_{Id_{U}} = V_{Id_{U}} = U$, no other vertices, and no edges.

As stated earlier, the coefficients $\lda_{\al}$ of the moment matrix $\Lda$ are assembled to form coefficient matrices, which are used to state the machinery conditions.
Given a shape $\tau$ and a $\tau$-coefficient matrix $H$, we consider the matrix-valued function $M^{fact}_{\tau}(H)$ defined as
\[
M^{fact}_{\tau}(H) = \sum_{\sigma \in \mathcal{L}_{U_{\tau}},\sigma' \in \mathcal{L}_{V_{\tau}}}{H(\sigma,\sigma')M_{\sigma}M_{\tau}M_{\sigma'}^T}
\]

The motivation for this definition is as follows. First observe that $\Lda$ is essentially an expression of the form $\sum_{\tau}\sum_{\sigma, \sigma'}{H(\sigma,\sigma')M_{\sigma \circ \tau \circ {\sigma'}^T}}$. For each $\tau$, the inner expression sort of looks like $M_{\tau}^{fact}(H)$. However, there is a technical difference.
Recall that if we expand out the definition of the graph matrices $M_{\sig}, M_{\tau}$ and $M_{\sig'}^T$, then $M_{\tau}^{fact}(H)$ sums over non-equivalent realizations coming from the sets $\textrm{Real}(\sig), \textrm{Real}(\tau), \textrm{Real}(\sig'^T)$ respectively. Apriori, it's not guaranteed that for each choice of realizations $\phi_1, \phi_2, \phi_3$ from these sets, the corresponding subset of labels $\phi_1(V(\sig)), \phi_2(V(\tau)), \phi_3(V(\sig'^T)) \subseteq [n]$ are disjoint. However, if we enforce that they are disjoint, then we will obtain a matrix closely related to what we desire. The work \cite{potechin2020machinery} terms this matrix obtained by enforcing disjointness of the realizations as $M^{orth}_{\tau}(H)$.

As they remark, it is not true that $M^{orth}_{\tau}(H) = \sum_{\sigma,\sigma'}{H(\sigma,\sigma')M_{\sigma \circ \tau \circ {\sigma'}^T}}$ because of additional terms involving automorphism groups.
Nevertheless, because of this enforced condition that the realizations don't overlap, $\Lda$ can be easily expressed in terms of $M^{orth}$. Indeed, as they show via careful counting, $\Lambda = \sum_{U \in \mathcal{I}_{mid}}{M^{orth}_{Id_U}(H_{Id_U})} + \sum_{U \in \mathcal{I}_{mid}}{\sum_{\tau \in \mathcal{M}_U}{M^{orth}_{\tau}(H_{\tau})}}$, where $H_{Id_U}$ and $H_{\tau}$, formally defined below, are simple coefficient matrices assembled from $\lda_{\al}$

Despite wanting to work with $M^{orth}$, the machinery instead works with $M^{fact}$ instead because showing PSDness is easier with $M^{fact}$ due to the product structure. The proof strategy in the machinery is to then show that the error terms when going from $M^{fact}$ to $M^{orth}$ (and therefore, $\Lda$) are negligible with high probability, concluding the PSDness proof.

Given a matrix-valued function $\Lambda$, we assemble the following coefficient matrices.
\begin{definition}
    Given a matrix-valued function $\Lambda = \sum_{\alpha: \alpha \text{ is proper}}{\lambda_{\alpha}M_{\alpha}}$,
    \begin{enumerate}
        \item For each index shape $U \in \mathcal{I}_{mid}$ and every $\sigma,\sigma' \in \mathcal{L}_{U}$, set $H_{Id_U}(\sigma,\sigma') = \frac{1}{|Aut(U)|}\lambda_{\sigma \circ {\sigma'}^T}$
        \item For each $U \in \mathcal{I}_{mid}$, $\tau \in \mathcal{M}_U$ and $\sigma, \sigma' \in \mathcal{L}_{U}$, set
        $H_{\tau}(\sigma,\sigma') = \frac{1}{|Aut(U_{\tau})|\cdot|Aut(V_{\tau})|}\lambda_{\sigma \circ \tau \circ {\sigma'}^T}$
    \end{enumerate}
\end{definition}

We need a final definition, that of the coefficient matrix $H^{-\gam, \gam}$.
In order to handle error terms in the approximate PSD decomposition, the machinery has to further decompose left shapes $\sigma$ as $\sigma = \sigma_2 \circ \gamma$ where $\sigma_2$ and $\gamma$ are themselves left shapes. In order to capture this operation, the following definitions are needed.

Given a moment matrix $\Lambda$, define $\Gamma = \{\gamma: \gamma \text{ is a non-trivial left shape with } U_{\gamma}, V_{\gamma} \in \mathcal{I}_{mid}, |V(\gamma)| \leq D_V, \forall e \in E(\gamma), l_e \leq D_E\}$. Moreover, given $U,V \in \mathcal{I}_{mid}$ such that $w(U) > w(V)$, define $\Gamma_{U,V} = \{\gamma \in \Gamma: U_{\gamma} \equiv U, V_{\gamma} \equiv V\}$. Finally, given $U \in \mathcal{I}_{mid}$, define $\Gamma_{U,*} = \{\gamma \in \Gamma: U_{\gamma} \equiv U\}$.

We finally define the coefficient matrix $H^{-\gam, \gam}$ given the truncation parameter $D_V$. Given a shape $\tau$ with $U_{\tau} \equiv V_{\tau}$, left shape $\gamma \in {\Gamma}_{*,U_{\tau}}$ (and therefore, $\gamma \in {\Gamma}_{*,V_{\tau}}$), and a $\tau$-coefficient matrix $H$, define $H^{-\gamma,\gamma}$ to be the $(\gamma \circ \tau \circ {\gamma}^T)$-coefficient matrix with entries
    \begin{itemize}
        \item $H^{-\gamma,\gamma}(\sigma,\sigma') = H(\sigma \circ \gamma,\sigma' \circ \gamma)$ if $|V(\sigma \circ \gamma)| \leq D_V$ and $|V(\sigma' \circ \gamma)| \leq D_V$.
        \item $H^{-\gamma,\gamma}(\sigma,\sigma') = 0$ if $|V(\sigma \circ \gamma)| > D_V$ or $|V(\sigma' \circ \gamma)| > D_V$.
    \end{itemize}

\subsection{Main theorems}

For a problem $\calP$, let $\Lda_{\calP}$ be the moment matrix obtained via pseudo-calibration. We then state the conditions that the machinery requires in order to show positivity with high probability.
We will use the following notion of distance between coefficient matrices, which will be useful to bound truncation error.

\begin{definition}
    Given a function $B_{norm}(\alpha)$, define the distance $d_{\tau}(H_{\tau},H'_{\tau})$ between two $\tau$-coefficient matrices $H_{\tau}$ and $H'_{\tau}$ as
    \[
    d_{\tau}(H_{\tau},H'_{\tau}) = \sum_{\sigma \in \mathcal{L}_{U_{\tau}},\sigma' \in \mathcal{L}_{V_{\tau}}}{|H'_{\tau}(\sigma,\sigma') - H_{\tau}(\sigma,\sigma')|B_{norm}(\sigma)B_{norm}(\tau)B_{norm}(\sigma')}
    \]
\end{definition}

We also define $Id_{Sym}$, which is the SoS-symmetric analogue of the identity matrix. For a matrix index $A$, denote by $p_A$ the formal monomial (in terms of the SoS program variables) it corresponds to. Define $Id_{Sym}$ to be the matrix such that the rows and columns of $Id_{Sym}$ are indexed by the matrix indices $A,B$ whose index shape is in $\mathcal{I}(\Lambda)$ and $Id_{Sym}(A,B) = 1$ if $p_A = p_B$ and $Id_{Sym}(A, B) = 0$ if $p_A \neq p_B$.

We introduce a few more notations about shapes in order to state our parameters.
Define $\mathcal{M}'$ to be the set of all shapes $\alpha$ such that $|V(\alpha)| \leq 3D_V$, $\forall e \in E(\alpha), l_e \leq D_E$ and all edges $e \in E(\alpha)$ have multiplicity at most $3D_V$. Note that the latter two conditions are not needed for the setting of Rademacher $G_{n, 1/2}$ inputs but they're needed in the setting of Gaussian inputs.
In the setting of Rademacher $G_{n, 1/2}$ inputs, for a shape $\alpha$, define $S_{\alpha}$ to be the leftmost minimum vertex separator of $\alpha$ and define $I_{\alpha}$ to be the set of vertices in $V(\alpha) \setminus (U_{\alpha} \cup V_{\alpha})$ which are isolated. In the setting of Gaussian $\GN(0, 1)$ inputs, for a shape $\alpha \in \mathcal{M}'$, define $S_{\alpha,min}$ to be the leftmost minimum vertex separator of $\alpha$ if all edges with multiplicity at least $2$ are deleted. Moreover, define $I_{\alpha}$ to be the set of vertices in $V(\alpha) \setminus (U_{\alpha} \cup V_{\alpha})$ such that all edges incident with that vertex have multiplicity at least $2$.

\subsubsection{Choice of parameters in the setting of Rademacher $G_{n, 1/2}$ inputs}

We first state some parameters we will use in this work and then state the main conditions that are needed for the main theorem statement, which is stated after this.

 Let $\eps > 0$ and $D_V, D_E$ be truncation parameters. Define
\begin{itemize}
    \item $q = 3\left\lceil{{D_V}\ln n + \frac{\ln(\frac{1}{\epsilon})}{3} + {D_V}\ln 5 + 3{D^2_V}\ln 2}\right\rceil$
    \item $B_{vertex} = 6{D_V}\sqrt[4]{2eq}$
    \item $B_{norm}(\alpha) = {B_{vertex}^{|V(\alpha) \setminus U_{\alpha}| + |V(\alpha) \setminus V_{\alpha}|}}n^{\frac{w(V(\alpha)) + w(I_{\alpha}) - w(S_{\alpha})}{2}}$
    \item $B(\gamma) = B_{vertex}^{|V(\gamma) \setminus U_{\gamma}| + |V(\gamma) \setminus V_{\gamma}|}n^{\frac{w(V(\gamma) \setminus U_{\gamma})}{2}}$
    \item $N(\gamma) = (3D_V)^{2|V(\gamma) \setminus V_{\gamma}| + |V(\gamma) \setminus U_{\gamma}|}$
    \item $c(\alpha) = 100(3D_V)^{|U_{\alpha} \setminus V_{\alpha}| + |V_{\alpha} \setminus U_{\alpha}| + 2|E(\alpha)|}2^{|V(\alpha) \setminus (U_{\alpha} \cup V_{\alpha})|}$
\end{itemize}

In our application, as stated earlier, we show SoS lower bounds for degree-$n^{\eps}$ SoS, where the input size is $n^{O(1)}$. In this setting, we take $D_V, D_E$ to be of the order of $n^{O(\eps)}$. Therefore, for simplicity, we can interpret the parameters as
\[q = n^{O(\eps)}, B_{vertex} = n^{O(\eps)}, B_{norm}(\alpha) =n^{O(\eps)|V(\al)|}n^{\frac{w(V(\alpha)) + w(I_{\alpha}) - w(S_{\alpha})}{2}}\]
\[B(\gamma) = n^{O(\eps)|V(\gam)|}n^{\frac{w(V(\gamma) \setminus U_{\gamma})}{2}}, N(\gamma) = n^{O(\eps)|V(\gam)|}, c(\alpha) = n^{O(\eps)|V(\al)|}\]

\subsubsection{Choice of parameters in the setting of Gaussian inputs on hypergraphs}

We now state the parameters needed for the more general statement of the machinery where we have Gaussian inputs on hypergraphs.
In this setting, let there be at most $t_{max}$ types of vertices and let $k$ be the maximum arity of an hyperedge. In the setting of Tensor PCA, we take $t_{max} = 1$ and in the setting of Sparse PCA, we take $k = t_{max} = 2$.
For all $\epsilon > 0$ and truncation parameters $D_V, D_E$, define
\begin{enumerate}
    \item $q = \left\lceil{3{D_V}\ln n + \ln(\frac{1}{\epsilon}) + {(3D_V)^k}\ln(D_E + 1) + 3{D_V}\ln 5}\right\rceil$
    \item $B_{vertex} = 6q{D_V}$
    \item $B_{edge}(e) = \left(400{D^2_V}{D^2_E}q\right)^{l_e}$
    \item $B_{norm}(\alpha) =
    2e{B_{vertex}^{|V(\alpha) \setminus U_{\alpha}| + |V(\alpha) \setminus V_{\alpha}|}}\left(\prod_{e \in E(\alpha)}{B_{edge}(e)}\right)n^{\frac{w(V(\alpha)) + w(I_{\alpha}) - w(S_{\alpha,min})}{2}}$
    \item $B(\gamma) = B_{vertex}^{|V(\gamma) \setminus U_{\gamma}| + |V(\gamma) \setminus V_{\gamma}|}\left(\prod_{e \in E(\gamma)}{B_{edge}(e)}\right)n^{\frac{w(V(\gamma) \setminus U_{\gamma})}{2}}$
    \item $N(\gamma) = (3D_V)^{2|V(\gamma) \setminus V_{\gamma}| + |V(\gamma) \setminus U_{\gamma}|}$
    \item $c(\alpha) = 100(3{t_{max}}D_V)^{|U_{\alpha} \setminus V_{\alpha}| + |V_{\alpha} \setminus U_{\alpha}| + k|E(\alpha)|}(2t_{max})^{|V(\alpha) \setminus (U_{\alpha} \cup V_{\alpha})|}$
\end{enumerate}

In our applications, we can interpret the above parameters in a much simpler manner again. More specifically, $k$ is a constant and we work with SoS degree $n^{\eps}$. Then, we can think of each vertex or edge of the shape $\al$ or $\gam$ essentially contributing a factor of $n^{\eps}$. Therefore, we can interpret
\[q = n^{O(\eps)}, B_{vertex} = n^{O(\eps)}, B_{edge} = n^{O(\eps)|E(\al)|}\]
\[B_{norm}(\al) = n^{O(\eps)(|V(\al)| + |E(\al)|)}n^{\frac{w(V(\alpha)) + w(I_{\alpha}) - w(S_{\alpha,min})}{2}}\]
\[B(\gamma) = n^{O(\eps)(|V(\gam)| + |E(\gam)|)}n^{\frac{w(V(\gamma) \setminus U_{\gamma})}{2}}\]
\[N(\gamma) = n^{O(\eps)|V(\gam)|}, c(\alpha) = n^{O(\eps)(|V(\al)| + |E(\al)|)}\]

\subsubsection{Statement of the machinery}

As discussed above, consider the appropriate choice of parameters suited for the problem. Now, we can state our conditions on the problem $\calP$ in terms of its correspondingly constructed pseudo-calibrated moment matrix $\Lda_{\calP}$ and coefficient matrices $H_{Id_U}, H_{\tau}$.

\begin{definition}[PSD mass]
    We say that $\calP$ satisfies \psdmass if for all $U \in \mathcal{I}_{mid}$,  $H_{Id_{U}} \succeq 0$.
\end{definition}

\begin{definition}[Middle shape bounds]
    We say that $\calP$ satisfies \middleshapebounds if for all $U \in \mathcal{I}_{mid}$ and $\tau \in \mathcal{M}_U$,
    \[
    \left[ {\begin{array}{cc}
            \frac{1}{|Aut(U)|c(\tau)}H_{Id_{U}} & B_{norm}(\tau)H_{\tau} \\
            B_{norm}(\tau)H^T_{\tau} & \frac{1}{|Aut(U)|c(\tau)}H_{Id_{U}}
    \end{array}} \right] \succeq 0
    \]
\end{definition}

\begin{definition}[Intersection term bounds]
    For some SoS-symmetric coefficient matrices $\{H'_{\gamma}: \gamma \in \Gamma\}$, $\calP$ satisfies \intersectionbounds with respect to them if for all $U,V \in \mathcal{I}_{mid}$ where $w(U) > w(V)$ and all $\gamma \in \Gamma_{U,V}$,
    \[
    c(\gamma)^2{N(\gamma)}^2{B(\gamma)^2}H^{-\gamma,\gamma}_{Id_{V}} \preceq H'_{\gamma}
    \]
\end{definition}

\begin{definition}[Truncation error bounds]
    For some SoS-symmetric coefficient matrices $\{H'_{\gamma}: \gamma \in \Gamma\}$,
    $\calP$ satisfies \truncationbounds with respect to them if the following condition holds: Whenever $\norm{M_{\alpha}} \leq B_{norm}(\alpha)$ for all $\alpha \in \mathcal{M}'$,
    \[
    \sum_{U \in \mathcal{I}_{mid}}{M^{fact}_{Id_U}{(H_{Id_U})}} \succeq 6\left(\sum_{U \in \mathcal{I}_{mid}}{\sum_{\gamma \in \Gamma_{U,*}}{\frac{d_{Id_{U}}(H'_{\gamma},H_{Id_{U}})}{|Aut(U)|c(\gamma)}}}\right)Id_{sym}
    \]
\end{definition}

Finally, we can state our main theorem.
\begin{theorem}\label{generalmaintheorem}
    For all $\epsilon > 0$, if we take the parameters defined above, and we have SoS-symmetric coefficient matrices $\{H'_{\gamma}: \gamma \in \Gamma\}$ such that $\calP$ satisfies \psdmass, \middleshapebounds, \intersectionbounds and \truncationbounds\hspace{-.8em},
    then with probability at least $1 - \epsilon$, $\Lambda_{\calP} \succeq 0$.
\end{theorem}

In our applications, for problems $\calP$ of interest, we pseudo-calibrate, decompose into graph matrices, exhibit the desired conditions on $\calP$ and invoke the machinery to prove our lower bounds.

\subsubsection{Choice of $H'_{\gamma}$ for our applications}\label{sec: hgamma_qual}
In our applications, we choose $H'_{\gamma}$ as follows.
\begin{enumerate}
    \item $H'_{\gamma}(\sigma,\sigma') = H_{Id_U}(\sigma, \sigma')$ whenever $|V(\sigma \circ \gamma)| \leq D_V$ and $|V(\sigma' \circ \gamma)| \leq D_V$.
    \item $H'_{\gamma}(\sigma,\sigma') = 0$ whenever $|V(\sigma \circ \gamma)| > D_V$ or $|V(\sigma' \circ \gamma)| > D_V$.
\end{enumerate}
Then, the truncation error that we need to bound is
\[
d_{Id_{U_{\gamma}}}(H_{Id_{U_{\gamma}}},H'_{\gamma}) = \sum_{\sigma,\sigma' \in \mathcal{L}_{U_{\gamma}}: V(\sigma) \leq D_V, V(\sigma') \leq D_V,
    \atop |V(\sigma \circ \gamma)| > D_V \text{ or } |V(\sigma' \circ \gamma)| > D_V}{B_{norm}(\sigma)B_{norm}(\sigma')H_{Id_{U_{\gamma}}}(\sigma,\sigma')}
\]

\section{Qualitative bounds for Planted slightly denser subgraph}\label{sec: plds_qual}

\subsection{Pseudo-calibration}

We will pseudo-calibrate with respect the following pair of random and planted distributions which we denote $\nu$ and $\mu$ respectively.

\PLDSdistributions*

We assume that the input is given as $G_{i, j}$ for $i, j \in \binom{[n]}{2}$ where $G_{i, j}$ is $1$ if the edge $(i, j)$ is present in the graph and $-1$ otherwise. We work with the Fourier basis $\chi_E$ defined as $\chi_E(G) \defeq \prod_{(i, j) \in E} G_{i, j}$. For a subset $I \subseteq [n]$, define $x_I \defeq \prod_{i \in I} x_I$.

\begin{lemma}
Let $I \subseteq [n], E \subseteq \binom{[n]}{2}$. Then,
\[\EE_{\mu}[x_I \chi_E(G)] = \left(\frac{k}{n}\right)^{|I \cup V(E)|} (2p - 1)^{|E|}\]
\end{lemma}

\begin{proof}
When we sample $(G, S)$ from $\mu$, we condition on whether $I \cup V(E) \subseteq S$.
\begin{align*}
\EE_{(G, S)\sim \mu}[x_I \chi_E(G)] &= Pr_{(G, S) \sim \mu}[I \cup V(E) \subseteq S]\EE_{(G, S) \sim \mu}[x_I\chi_E(G)|I \cup V(E) \subseteq S]\\
&\qquad + Pr_{(G, S) \sim \mu}[I \cup V(E) \not\subseteq S]\EE_{(G, S) \sim \mu}[x_I\chi_E(G)|I \cup V(E) \not\subseteq S]
\end{align*}
We claim that the second term is $0$. In particular, $\EE_{(G, S) \sim \mu}[x_I\chi_E(G)|I \cup V(E) \not\subseteq S] = 0$ because when $I \cup V(E) \not\subseteq S$, either $S$ doesn't contain a vertex in $I$ or an edge $(i, j) \in E$ is outside $S$. If $S$ doesn't contain a vertex in $I$, then $x_I = 0$ and hence, the quantity is $0$. And if an edge $(i, j) \in E$ is outside $S$, since this edge is sampled with probability $\frac{1}{2}$, by taking expectations, the quantity $\EE_{(G, S) \sim \mu}[x_I\chi_E(G)|I \cup V(E) \not\subseteq S]$ is $0$.

Finally, note that $Pr_{(G, S) \sim \mu}[I \cup V(E) \subseteq S] = \left(\frac{k}{n}\right)^{|I \cup V(E)|}$ and
\[\EE_{(G, S) \sim \mu}[x_I\chi_E(G)|I \cup V(E) \subseteq S] = \EE_{(G, S) \sim \mu}[\chi_E(G)|V(E) \subseteq S] = (2p - 1)^{|E|}\]
The last equality follows because for each edge $e \in E$, since $e$ is present independently with probability $p$, the expected value of $\chi_e$ is $1\cdot p + (-1) \cdot (1 - p) = 2p - 1$.
\end{proof}


Define the degree of SoS to be $D_{sos} = n^{C_{sos}\eps}$ for some constant $C_{sos} > 0$ that we choose later. And define the truncation parameter to be $D_V = n^{C_V\eps}$ for some constant $C_V > 0$.

\begin{remk}[Choice of parameters]\label{rmk: choice_of_params1}
	We first set $\eps$ to be a sufficiently small constant. Based on this choice, we will set $C_V$ to be a sufficiently small constant to satisfy all the inequalities we use in our proof. Based on these choices, we can choose $C_{sos}$ to be sufficiently small to satisfy the inequalities we use.
\end{remk}

We will now describe the decomposition of the moment matrix $\Lda$.

\begin{definition}\label{def: plds_coeffs}
	If a shape $\alpha$ satisfies the following properties:
	\begin{itemize}
		\item $\alpha$ is proper,
		\item $\alpha$ satisfies the truncation parameter $D_{sos}, D_V$.
	\end{itemize}
	then define \[\lambda_{\alpha} = \left(\frac{k}{n}\right)^{|V(\al)|}  (2p - 1)^{|E(\al)|}\]
\end{definition}

\begin{corollary}
	$\Lambda = \sum \lda_{\al}M_{\al}$.
\end{corollary}

\subsection{Qualitative machinery bounds}

In this section, we will prove the PSD mass condition and the qualitative versions of the middle shape and intersection term bounds.

\begin{restatable}[PSD mass]{lemma}{PLDSone}\label{lem: plds_cond1}
	For all $U \in \calI_{mid}$, $H_{Id_U} \succeq 0$
\end{restatable}

While this is easy to prove directly, we would like to introduce appropriate notation so that this lemma as well as the qualitative bounds to follow are immediate.
Therefore, we state the qualitative conditions next and then prove them all together.
Now, we define the following quantities which capture the contribution of the vertices within $\tau, \gam$ to the Fourier coefficients.

\begin{restatable}{definition}{PLDSstau}\label{def: plds_stau}
	For $U \in \calI_{mid}$ and $\tau \in \calM_U$, define
	$S(\tau) = \left(\frac{k}{n}\right)^{|V(\tau)| - |U_{\tau}|}(2p - 1)^{|E(\tau)|}$.
	And for all $U, V \in \calI_{mid}$ where $w(U) > w(V)$ and $\gam \in \Gam_{U, V}$, define
	$S(\gam) = \left(\frac{k}{n}\right)^{|V(\gam)| - \frac{|U_{\gam}| + |V_{\gam}|}{2}}(2p - 1)^{|E(\gam)|}$.
\end{restatable}

We can now state our qualitative bounds, which we prove shortly.

\begin{restatable}[Qualitative middle shape bounds]{lemma}{PLDStwosimplified}\label{lem: plds_cond2_simplified}
	For all $U \in \calI_{mid}$ and $\tau \in \calM_U$,
	\[
	\begin{bmatrix}
		\frac{S(\tau)}{|Aut(U)|}H_{Id_U} & H_{\tau}\\
		H_{\tau}^T & \frac{S(\tau)}{|Aut(U)|}H_{Id_U}
	\end{bmatrix}
	\succeq 0
	\]
\end{restatable}

In the following qualitative intersection term bounds, we use the canonical definition of $H_{\gam}'$ from \cref{sec: hgamma_qual}.

\begin{restatable}[Qualitative intersection term bounds]{lemma}{PLDSthreesimplified}\label{lem: plds_cond3_simplified}
	For all $U, V \in \calI_{mid}$ where $w(U) > w(V)$ and all $\gam \in \Gam_{U, V}$,
	\[\frac{|Aut(V)|}{|Aut(U)|}\cdot\frac{1}{S(\gam)^2}H_{Id_V}^{-\gam, \gam} = H_{\gam}'\]
\end{restatable}

In order to prove these bounds, we define the following quantity to capture the contribution of the vertices within $\sig$ to the Fourier coefficients.

\begin{definition}
	For a shape $\sig\in \calL$, define
	$T(\sig) = \left(\frac{k}{n}\right)^{|V(\sig)| - \frac{|V_{\sig}|}{2}}(2p - 1)^{|E(\sig)|}$.
	For $U \in \calI_{mid}$, define $v_U$ to be the vector indexed by $\sig \in \calL$ such that $v_U(\sig) = T(\sig)$ if $\sig \in \calL_U$ and $0$ otherwise.
\end{definition}

The following propositions are immediate from \cref{def: plds_coeffs}.

\begin{propn}
	For all $U\in \calI_{mid}, \rho \in \calP_U$, $H_{Id_U} = \frac{1}{|Aut(U)|}v_Uv_U^T$.
\end{propn}





\begin{propn}
	For any $U \in \calI_{mid}$ and $\tau \in \calM_U$, $H_{\tau} = \frac{1}{|Aut(U)|^2} S(\tau) v_Uv_U^T$.
\end{propn}


The first proposition implies that for all $U \in \calI_{mid}$, $H_{Id_U} \succeq 0$, which is the PSD mass condition \cref{lem: plds_cond1}.
\cref{lem: plds_cond2_simplified} and \cref{lem: plds_cond3_simplified} also follow easily.

\begin{proof}[Proof of \cref{lem: plds_cond2_simplified}]
\begin{align*}
    \begin{bmatrix}
		\frac{S(\tau)}{|Aut(U)|}H_{Id_U} & H_{\tau}\\
		H_{\tau}^T & \frac{S(\tau)}{|Aut(U)|}H_{Id_U}
	\end{bmatrix} &= \begin{bmatrix}
		\frac{S(\tau)}{|Aut(U)|}v_Uv_U^T & \frac{S(\tau)}{|Aut(U)|^2}v_Uv_U^T\\
		\frac{S(\tau)}{|Aut(U)|^2}v_Uv_U^T & \frac{S(\tau)}{|Aut(U)|}v_Uv_U^T
	\end{bmatrix} \succeq 0
\end{align*}
\end{proof}




\begin{proof}[Proof of \cref{lem: plds_cond3_simplified}]
    Fix $\sig, \sig' \in \calL_{U}$ such that $|V(\sig \circ \gam)|, |V(\sig' \circ \gam)| \le D_V$. Note that $|V(\sig)| - \frac{|V_{\sig}|}{2} + |V(\sig')| - \frac{|V_{\sig'}|}{2} + 2(|V(\gam)| - \frac{|U_{\gam}| + |V_{\gam}|}{2}) = |V(\sig \circ \gam \circ \gam^T \circ \sig'^T)|$. Using \cref{def: plds_coeffs}, we can easily verify that $\lda_{\sig \circ \gam \circ \gam^T \circ \sig'^T} = T(\sigma)T(\sigma') S(\gam)^2$. Therefore, we have $H_{Id_V}^{-\gam, \gam}(\sig, \sig') = \frac{|Aut(U)|}{|Aut(V)|} S(\gam)^2 H_{Id_U}(\sig, \sig')$. Since $H'_{\gam}(\sig, \sig') = H_{Id_U}(\sig, \sig')$ whenever $|V(\sig \circ \gam)|, |V(\sig' \circ \gam)| \le D_V$, this completes the proof.
\end{proof}



\section{Qualitative bounds for Tensor PCA}\label{sec: tpca_qual}

\subsection{Pseudo-calibration}

\begin{definition}[Slack parameter]
	Define the slack parameter to be $\Delta = n^{-C_{\Del}\eps}$ for a constant $C_{\Del} > 0$.
\end{definition}

We will pseudo-calibrate with respect the following pair of random and planted distributions which we denote $\nu$ and $\mu$ respectively.

\TPCAdistributions*

Let the Hermite polynomials be $h_0(x) = 1, h_1(x) = x, h_2(x) = x^2 - 1, \ldots$. For $a \in \NN^{[n]^k}$ and variables $A_e$ for $e \in [n]^k$, define $h_a(A) \defeq \prod_{e \in [n]^k} h_e(A_e)$. We will work with this Hermite basis.

\begin{lemma}
	Let $I \in \NN^n, a \in \NN^{[n]^k}$. For $i \in [n]$, let $d_i = \sum_{i \in e \in [n]^k} a_e$. Let $c$ be the number of $i$ such that $I_i + d_i$ is nonzero. Then, if $I_i + d_i$ are all even, we have
	\[\EE_{\mu}[u^I h_a(A)] = \Delta^c\left(\frac{1}{\sqrt{\Delta n}}\right)^{|I|} \prod_{e \in [n]^k} \left(\frac{\lda}{(\Del n)^{\frac{k}{2}}}\right)^{a_e}\]
	Else, $\EE_{\mu}[u^I h_a(v)] = 0$.
\end{lemma}

\begin{proof}
	When $A \sim \mu$, for all $e \in [n]^k$, we have $A_e = B_e + \lda \prod_{i \le k} u_{e_i}$. where $B_e \sim \GN(0, 1)$.
	Let's analyze when the required expectation is nonzero. We can first condition on $u$ and use the fact that for a fixed $t$, $\EE_{g \sim \GN(0, 1)}[h_k(g + t)] = t^k$ to obtain
	\[\EE_{(u_i, w_e) \sim \mu}[u^I h_a(A)] = \EE_{(u_i) \sim \mu}[u^I\prod_{e \in [n]^k}(\lda \prod _{i \le k}u_{e_i})^{a_e}] = \EE_{(u_i) \sim \mu}[\prod_{i \in [n]} u_i^{I_i + d_i}] \prod_{e \in [n]^k} \lda^{a_e}\]

	Observe that this is nonzero precisely when all $I_i + d_i$ are even, in which case \[\EE_{(u_i) \sim \mu}[\prod_{i \in [n]} u_i^{I_i + d_i}] = \Delta^c\left(\frac{1}{\sqrt{\Del n}}\right)^{\sum_{i \le n} I_i + d_i} =  \Delta^c\left(\frac{1}{\sqrt{\Delta n}}\right)^{|I|} \prod_{e \in [n]^k} \left(\frac{1}{(\Del n)^{\frac{k}{2}}}\right)^{a_e}\]
	where we used the fact that $\sum_{e \in [n]^k} a_e = k \sum_{i \in [n]} d_i$.
	This completes the proof.
\end{proof}


Define the degree of SoS to be $D_{sos} = n^{C_{sos}\eps}$ for some constant $C_{sos} > 0$ that we choose later. And define the truncation parameters to be $D_V = n^{C_V\eps}, D_E = n^{C_E\eps}$ for some constants $C_V, C_E > 0$.

\begin{remk}[Choice of parameters]\label{rmk: choice_of_params2}
	We first set $\eps$ to be a sufficiently small constant. Based on the choice of $\eps$, we will set the constant $C_{\Del} > 0$ sufficiently small so that the planted distribution is well defined. Based on these choices, just as in \cref{rmk: choice_of_params1} we choose $C_V, C_E, C_{sos}$ in that order.
\end{remk}

The underlying graphs for the graph matrices have the following structure; There will be $n$ vertices of a single type and the edges will be ordered hyperedges of arity $k$.
For the analysis of Tensor PCA, we will use the following notation.
\begin{itemize}
	\item For an index shape $U$ and a vertex $i$, define $deg^{U}(i)$ as follows: If $i \in V(U)$, then it is the power of the unique index shape piece $A \in U$ such that $i \in V(A)$. Otherwise, it is $0$.
	\item For an index shape $U$, define $deg(U) = \sum_{i \in V(U)} deg^U(i)$. This is also the degree of the monomial that $U$ corresponds to.
	\item For a shape $\alpha$ and vertex $i$ in $\alpha$, let $deg^{\alpha}(i) = \sum_{i \in e \in E(\alpha)} l_e$.
	\item For any shape $\alpha$, let $deg(\alpha) = deg(U_{\al}) + deg(V_{\al})$.
\end{itemize}

We will now describe the decomposition of the moment matrix $\Lda$.

\begin{definition}\label{def: tpca_coeffs}
	If a shape $\alpha$ satisfies the following properties:
	\begin{itemize}
		\item $deg^{\alpha}(i) + deg^{U_{\alpha}}(i) + deg^{V_{\alpha}}(i)$ is even for all $i \in V(\alpha)$,
		\item $\alpha$ is proper,
		\item $\alpha$ satisfies the truncation parameters $D_{sos}, D_V, D_E$.
	\end{itemize}
	then define \[\lambda_{\alpha} = \Delta^{|V(\al)|} \left(\frac{1}{\sqrt{\Delta n}}\right)^{deg(\alpha)}  \prod_{e \in E(\al)} \left(\frac{\lda}{(\Del n)^{\frac{k}{2}}}\right)^{l_e}\]
	Otherwise, define $\lambda_{\alpha} = 0$.
\end{definition}

\begin{corollary}
	$\Lambda = \sum \lda_{\al}M_{\al}$.
\end{corollary}

\subsection{Qualitative machinery bounds}

Just as in planted slightly denser subgraph, we prove the PSD mass condition and the qualitative middle shape and intersection term bounds, by first stating them and then introducing appropriate notation to prove them all in a unified manner.

\begin{restatable}[PSD mass]{lemma}{TPCAone}\label{lem: tpca_cond1}
	For all $U \in \calI_{mid}$, $H_{Id_U} \succeq 0$
\end{restatable}

We define the following quantities to capture the contribution of the vertices within $\tau, \gam$ to the Fourier coefficients.

\begin{restatable}{definition}{TPCAstau}\label{def: tpca_stau}
	For $U \in \calI_{mid}$ and $\tau \in \calM_U$, if $deg^{\tau}(i)$ is even for all vertices $i \in V(\tau) \setminus U_{\tau} \setminus V_{\tau}$, define
	\[S(\tau) = \Delta^{|V(\tau)| - |U_{\tau}|}\prod_{e \in E(\tau)}\left(\frac{\lda}{(\Del n)^{\frac{k}{2}}}\right)^{l_e}\]
	Otherwise, define $S(\tau) = 0$.
	For all $U, V \in \calI_{mid}$ where $w(U) > w(V)$ and $\gam \in \Gam_{U, V}$, if $deg^{\gam}(i)$ is even for all vertices $i$ in $V(\gam) \setminus U_{\gam} \setminus V_{\gam}$, define
	\[S(\gam) = \Delta^{|V(\gam)| - \frac{|U_{\gam}| + |V_{\gam}|}{2}}\prod_{e \in E(\gam)}\left(\frac{\lda}{(\Del n)^{\frac{k}{2}}}\right)^{l_e}\]
	Otherwise, define $S(\gam) = 0$.
\end{restatable}

We now state the qualitative bounds in terms of these quantities.

\begin{restatable}[Qualitative middle shape bounds]{lemma}{TPCAtwosimplified}\label{lem: tpca_cond2_simplified}
	For all $U \in \calI_{mid}$ and $\tau \in \calM_U$,
	\[
	\begin{bmatrix}
		\frac{S(\tau)}{|Aut(U)|}H_{Id_U} & H_{\tau}\\
		H_{\tau}^T & \frac{S(\tau)}{|Aut(U)|}H_{Id_U}
	\end{bmatrix}
	\succeq 0
	\]
\end{restatable}



We again use the canonical definition of $H_{\gam}'$ from \cref{sec: hgamma_qual}.

\begin{restatable}[Qualitative intersection term bounds]{lemma}{TPCAthreesimplified}\label{lem: tpca_cond3_simplified}
	For all $U, V \in \calI_{mid}$ where $w(U) > w(V)$ and all $\gam \in \Gam_{U, V}$,
	\[\frac{|Aut(V)|}{|Aut(U)|}\cdot\frac{1}{S(\gam)^2}H_{Id_V}^{-\gam, \gam} \preceq H_{\gam}'\]
\end{restatable}

\subsubsection{Proof of PSD mass condition}

We introduce some notation which makes it easy to show the qualitative bounds and which also sheds light on the structure of the coefficient matrices. When we compose shapes $\sig, \sig'$, from \cref{def: tpca_coeffs}, in order for $\lda_{\sig\circ \sig'}$ to be nonzero, observe that all vertices $i$ in $\lda_{\sig \circ \sig'}$ should have $deg^{\sig \circ \sig'}(i) + deg^{U_{\sig \circ \sig'}}(i) + deg^{V_{\sig \circ \sig'}}(i)$ to be even. To partially capture this notion conveniently, we will introduce the notion of parity vectors.

\begin{definition}
	Define a parity vector $\rho$ to be a vector whose entries are in $\{0, 1\}$.
	For $U\in \calI_{mid}$, define $\calP_U$ to be the set of parity vectors $\rho$ whose coordinates are indexed by $U$.
\end{definition}

\begin{definition}
	For a left shape $\sig$, define $\rho_{\sig} \in \calP_{V_{\sig}}$, called the parity vector of $\sig$, to be the parity vector such that for each vertex $i \in V_{\sig}$, the $i$-th entry of $\rho_{\sig}$ is the parity of $deg^{U_{\sig}}(i) + deg^{\sig}(i)$, that is $(\rho_{\sig})_i \equiv deg^{U_{\sig}}(i) + deg^{\sig}(i) \pmod 2$.
	For $U \in \calI_{mid}$ and $\rho \in \calP_U$, let $\calL_{U, \rho}$ be the set of all left shapes $\sig \in \calL_U$ such that $\rho_{\sig} = \rho$, that is, the set of all left shapes with parity vector $\rho$.
\end{definition}

For a shape $\tau$, for a $\tau$ coefficient matrix $H_{\tau}$ and parity vectors $\rho \in \calP_{U_{\tau}}, \rho' \in \calP_{V_{\tau}}$, define the $\tau$-coefficient matrix $H_{\tau, \rho, \rho'}$ as $H_{\tau ,\rho, \rho'}(\sig, \sig') = H_{\tau}(\sig, \sig')$ if $\sig \in \calL_{U_{\tau}, \rho}, \sig' \in \calL_{V_{\tau}, \rho'}$ and $0$ otherwise.
The following proposition is immediate.

\begin{propn}
	For any shape $\tau$ and $\tau$-coefficient matrix $H_{\tau}$, we have the equality $H_{\tau} = \sum_{\rho \in \calP_{U_{\tau}}, \rho' \in \calP_{V_{\tau}}} H_{\tau, \rho, \rho'}$
\end{propn}

\begin{propn}
	For any $U \in \calI_{mid}$, $H_{Id_U} = \sum_{\rho \in \calP_U} H_{Id_U, \rho, \rho}$
\end{propn}

\begin{proof}
	For any $\sig, \sig' \in \calL_U$, using \cref{def: tpca_coeffs}, note that in order for $H_{Id_U}(\sig, \sig')$ to be nonzero, we must have $\rho_{\sig} = \rho_{\sig'}$.
\end{proof}

We define the following quantity to capture the contribution of the vertices within $\sig$ to the Fourier coefficients.

\begin{definition}
	For a shape $\sig\in \calL$, if $deg^{\sig}(i) + deg^{U_{\sig}}(i)$ is even for all vertices $i \in V(\sig) \setminus V_{\sig}$, define
	\[T(\sig) = \Delta^{|V(\sig)| - \frac{|V_{\sig}|}{2}}\left(\frac{1}{\sqrt{\Delta n}}\right)^{deg(U_{\sig})}\prod_{e \in E(\sig)}\left(\frac{\lda}{(\Del n)^{\frac{k}{2}}}\right)^{l_e}\]
	Otherwise, define $T(\sig) = 0$.
	For $U \in \calI_{mid}$ and $\rho \in \calP_U$, define $v_{\rho}$ to be the vector indexed by $\sig \in \calL$ such that $v_{\rho}(\sig)$ is $T(\sig)$ if $\sig \in \calL_{U, \rho}$ and $0$ otherwise.
\end{definition}

With this notation, the PSD mass condition is easily shown.

\begin{proof}[Proof of the PSD mass condition \cref{lem: tpca_cond1}]
    For all $U\in \calI_{mid}, \rho \in \calP_U$, \cref{def: tpca_coeffs} implies $H_{Id_U, \rho, \rho} = \frac{1}{|Aut(U)|}v_{\rho}v_{\rho}^T$.
	Therefore, \[H_{Id_U} = \sum_{\rho \in \calP_U} H_{Id_U, \rho, \rho} = \frac{1}{|Aut(U)|} \sum_{\rho \in \calP_U} v_{\rho}v_{\rho}^T \succeq 0\]
\end{proof}

\subsubsection{Qualitative middle shape bounds}

The next proposition captures the fact that when we compose shapes $\sig, \tau, \sig'^T$, in order for $\lda_{\sig \circ \tau \circ \sig'^T}$ to be nonzero, the parities of the degrees of the merged vertices should add up correspondingly.

\begin{propn}\label{propn: tpca_coeff_2}
	For all $U \in \calI_{mid}$ and $\tau \in \calM_U$, there exist two sets of parity vectors $P_{\tau}, Q_{\tau} \subseteq \calP_{U}$ and a bijection $\pi : P_{\tau} \to Q_{\tau}$ such that $H_{\tau} = \sum_{\rho \in P_{\tau}} H_{\tau, \rho, \pi(\rho)}$.
\end{propn}

\begin{proof}
	Using \cref{def: tpca_coeffs}, in order for $H_{\tau}(\sig, \sig')$ to be nonzero, in $\sig \circ \tau \circ \sig'$, we must have that for all $i \in U_{\tau} \cup V_{\tau}$, $deg^{U_{\sig}}(i) + deg^{U_{\sig'}}(i) + deg^{\sigma \circ \tau \circ \sigma'^T}(i)$ must be even. In other words, for any $\rho \in \calP_U$, there is at most one $\rho' \in \calP_U$ such that if we take $\sig \in \calL_{U, \rho}, \sig' \in \calL_U$ with $H_{\tau}(\sig, \sig')$ nonzero, then the parity of $\sig'$ is $\rho'$. Also, observe that $\rho'$ determines $\rho$. We then take $P_{\tau}$ to be the set of $\rho$ such that $\rho'$ exists, $Q_{\tau}$ to be the set of $\rho'$ and in this case, we define $\pi(\rho) = \rho'$.
\end{proof}



A straightforward verification of the conditions of \cref{def: tpca_coeffs} implies the following proposition.

\begin{propn}
	For any $U \in \calI_{mid}$ and $\tau \in \calM_U$, suppose we take $\rho \in P_{\tau}$.  Let $\pi$ be the bijection from \cref{propn: tpca_coeff_2} so that $\pi(\rho) \in Q_{\tau}$. Then, $H_{\tau, \rho, \pi(\rho)} = \frac{1}{|Aut(U)|^2} S(\tau) v_{\rho}v_{\pi(\rho)}^T$.
\end{propn}


We can now prove the qualitative middle shape bounds.


\begin{proof}[Proof of the qualitative middle shape bounds \cref{lem: tpca_cond2_simplified}]
	Let $P_{\tau}, Q_{\tau}, \pi$ be from \cref{propn: tpca_coeff_2}. For $\rho, \rho' \in \calP_U$, let $W_{\rho, \rho'} = v_{\rho}(v_{\rho'})^T$. Then, $H_{Id_U} = \sum_{\rho \in \calP_U} H_{Id_U, \rho, \rho} = \frac{1}{|Aut(U)|} \sum_{\rho \in \calP_U}W_{\rho, \rho}$ and $H_{\tau} = \sum_{\rho \in P_{\tau}} H_{\tau, \rho, \pi(\rho)} = \frac{1}{|Aut(U)|^2}S(\tau)\sum_{\rho \in P_{\tau}} W_{\rho, \pi(\rho)}$. We have

	\begin{align*}
		\begin{bmatrix}
			\frac{S(\tau)}{|Aut(U)|}H_{Id_U} & H_{\tau}\\
			H_{\tau}^T & \frac{S(\tau)}{|Aut(U)|}H_{Id_U}
		\end{bmatrix}
		&= \frac{S(\tau)}{|Aut(U)|^2}
		\begin{bmatrix}
			\sum_{\rho \in \calP_U} W_{\rho, \rho} & \sum_{\rho \in P_{\tau}} W_{\rho, \pi(\rho)}\\
			\sum_{\rho \in P_{\tau}} W_{\rho, \pi(\rho)}^T & \sum_{\rho \in \calP_U} W_{\rho, \rho}
		\end{bmatrix}
	\end{align*}
	Since $\frac{S(\tau)}{|Aut(U)|^2} \ge 0$, it suffices to prove that $\begin{bmatrix}
		\sum_{\rho \in \calP_U} W_{\rho, \rho} & \sum_{\rho \in P_{\tau}} W_{\rho, \pi(\rho)}\\
		\sum_{\rho \in P_{\tau}} W_{\rho, \pi(\rho)}^T & \sum_{\rho \in \calP_U} W_{\rho, \rho}
	\end{bmatrix}\succeq 0$. Consider
	\begin{align*}
		\begin{bmatrix}
			\sum_{\rho \in \calP_U} W_{\rho, \rho} & \sum_{\rho \in P_{\tau}} W_{\rho, \pi(\rho)}\\
			\sum_{\rho \in P_{\tau}} W_{\rho, \pi(\rho)}^T & \sum_{\rho \in \calP_U} W_{\rho, \rho}
		\end{bmatrix} =& \begin{bmatrix}
			\sum_{\rho \in \calP_U \setminus P_{\tau}} W_{\rho, \rho} & 0\\
			0 & \sum_{\rho \in \calP_U \setminus Q_{\tau}} W_{\rho, \rho}
		\end{bmatrix}\\
		& + \begin{bmatrix}
			\sum_{\rho \in P_{\tau}} W_{\rho, \rho} & \sum_{\rho \in P_{\tau}} W_{\rho, \pi(\rho)}\\
			\sum_{\rho \in P_{\tau}} W_{\rho, \pi(\rho)}^T & \sum_{\rho \in P_{\tau}} W_{\pi(\rho), \pi(\rho)}
		\end{bmatrix}\\
	\end{align*}

	We have $\sum_{\rho \in \calP_U \setminus P_{\tau}} W_{\rho, \rho} = \sum_{\rho \in \calP_U \setminus P_{\tau}} v_{\rho}v_{\rho}^T \succeq 0$. Similarly, $\sum_{\rho \in \calP_U \setminus Q_{\tau}} W_{\rho, \rho} \succeq 0$ and so, the first term in the above expression,
	$\begin{bmatrix}
		\sum_{\rho \in \calP_U \setminus P_{\tau}} W_{\rho, \rho} & 0\\
		0 & \sum_{\rho \in \calP_U \setminus Q_{\tau}} W_{\rho, \rho}
	\end{bmatrix}$ is positive semidefinite. For the second term,
	\begin{align*}
		\begin{bmatrix}
			\sum_{\rho \in P_{\tau}} W_{\rho, \rho} & \sum_{\rho \in P_{\tau}} W_{\rho, \pi(\rho)}\\
			\sum_{\rho \in P_{\tau}} W_{\rho, \pi(\rho)}^T & \sum_{\rho \in P_{\tau}} W_{\pi(\rho), \pi(\rho)}
		\end{bmatrix} &= \sum_{\rho \in P_{\tau}}
		\begin{bmatrix}
			W_{\rho, \rho} & W_{\rho, \pi(\rho)}\\
			W_{\rho, \pi(\rho)}^T & W_{\pi(\rho), \pi(\rho)}
		\end{bmatrix}\\
		&= \sum_{\rho \in P_{\tau}}
		\begin{bmatrix}
			v_{\rho}v_{\rho}^T & v_{\rho}(v_{\pi(\rho)})^T\\
			v_{\pi(\rho)}(v_{\rho})^T & v_{\pi(\rho)}(v_{\pi(\rho)})^T
		\end{bmatrix}\\
		&= \sum_{\rho \in P_{\tau}}
		\begin{bmatrix}
			v_{\rho}\\
			v_{\pi(\rho)}
		\end{bmatrix}
		\begin{bmatrix}
			v_{\rho} &
			v_{\pi(\rho)}
		\end{bmatrix}\\
		& \succeq 0
	\end{align*}
\end{proof}

\subsubsection{Qualitative intersection term bounds}

Similar to \cref{propn: tpca_coeff_2}, the next proposition captures the fact that when we compose shapes $\sig, \gam, \gam^T, \sig'^T$, in order for $\lda_{\sig \circ \gam \circ \gam'^T \circ \sig'^T}$ to be nonzero, the parities of the degrees of the merged vertices should add up correspondingly.

We use the following notation.
For all $U, V \in \calI_{mid}$ where $w(U) > w(V)$, for $\gam \in \Gam_{U, V}$ and parity vectors $\rho, \rho' \in \calP_U$,  define the $\gam \circ \gam^T$-coefficient matrix $H_{Id_V, \rho, \rho'}^{-\gam, \gam}$ as $H_{Id_V, \rho, \rho'}^{-\gam, \gam}(\sig, \sig') = H_{Id_V}^{-\gam, \gam}(\sig, \sig')$ if $\sig \in \calL_{U, \rho},  \sig' \in \calL_{U, \rho'}$ and $0$ otherwise.

\begin{propn}
	For all $U, V \in \calI_{mid}$ where $w(U) > w(V)$, for all $\gam \in \Gam_{U, V}$, there exists a set of parity vectors $P_{\gam} \subseteq \calP_U$ such that
	$H_{Id_V}^{-\gam, \gam} = \sum_{\rho \in P_{\gam}} H_{Id_V, \rho, \rho}^{-\gam, \gam}$.
\end{propn}

\begin{proof}
	Take any $\rho \in \calP_U$. For $\sig \in \calL_{U, \rho}, \sig' \in \calL_U$,  since $H_{Id_V}^{-\gam, \gam}(\sigma, \sigma') = \frac{\lda_{\sig \circ \gam \circ \gam^T \circ \sig'^T}}{|Aut(V)|}$, $H_{Id_V}^{-\gam, \gam}(\sig, \sig')$ is nonzero precisely when $\lda_{\sig \circ \gam \circ \gam^T \circ \sig'^T}$ is nonzero. For this quantity to be nonzero, using \cref{def: tpca_coeffs}, we get that it is necessary, but not sufficient, that the parity vector of $\sig'$ must also be $\rho$. And also observe that there exists a set $P_{\gam}$ of parity vectors $\rho$ for which $H_{Id_V, \rho, \rho}^{-\gam, \gam}$ is nonzero and their sum is precisely $H_{Id_V}^{-\gam, \gam}$.
\end{proof}

For all $U, V \in \calI_{mid}$ where $w(U) > w(V)$, for all $\gam \in \Gam_{U, V}$ and parity vector $\rho \in \calP_U$, define the matrix $H'_{\gam, \rho, \rho}$ as $H'_{\gam, \rho, \rho}(\sig, \sig') = H'_{\gam}(\sig, \sig')$ if $\sig, \sig' \in \calL_{U, \rho}$ and $0$ otherwise. The following proposition is immediate from the definition.

\begin{propn}
	For all $U, V \in \calI_{mid}$ where $w(U) > w(V)$, for $\gam \in \Gam_{U, V}$, $H_{\gam}' = \sum_{\rho \in P_{\gam}} H_{\gam, \rho, \rho}'$.
\end{propn}



\begin{propn}
	For all $U, V \in \calI_{mid}$ where $w(U) > w(V)$, for all $\gam \in \Gam_{U, V}$ and $\rho \in P_{\gam}$,
	\[H_{Id_V, \rho, \rho}^{-\gam, \gam} = \frac{|Aut(U)|}{|Aut(V)|} S(\gam)^2 H'_{\gam, \rho, \rho}\]
\end{propn}

\begin{proof}
	Fix $\sig, \sig' \in \calL_{U, \rho}$ such that $|V(\sig \circ \gam)|, |V(\sig' \circ \gam)| \le D_V$. Note that $|V(\sig)| - \frac{|V_{\sig}|}{2} + |V(\sig')| - \frac{|V_{\sig'}|}{2} + 2(|V(\gam)| - \frac{|U_{\gam}| + |V_{\gam}|}{2}) = |V(\sig \circ \gam \circ \gam^T \circ \sig'^T)|$. Using \cref{def: tpca_coeffs}, we can easily verify that $\lda_{\sig \circ \gam \circ \gam^T \circ \sig'^T} = T(\sigma)T(\sigma') S(\gam)^2$. Therefore, $H_{Id_V, \rho, \rho}^{-\gam, \gam}(\sig, \sig') = \frac{|Aut(U)|}{|Aut(V)|} S(\gam)^2 H_{Id_U, \rho, \rho}(\sig, \sig')$. Since $H'_{\gam, \rho, \rho}(\sig, \sig') = H_{Id_U, \rho, \rho}(\sig, \sig')$ whenever we have $|V(\sig \circ \gam)|, |V(\sig' \circ \gam)| \le D_V$, this completes the proof.
\end{proof}

With this, we can prove the qualitative intersection term bounds.

\begin{proof}[Proof of qualitative intersection term bounds \cref{lem: tpca_cond3_simplified}]
	We have
	\begin{align*}
		\frac{|Aut(V)|}{|Aut(U)|}\cdot\frac{1}{S(\gam)^2}H_{Id_V}^{-\gam, \gam} = \sum_{\rho \in P_{\gam}} \frac{|Aut(V)|}{|Aut(U)|}\cdot\frac{1}{S(\gam)^2} H_{Id_V, \rho, \rho}^{-\gam, \gam}
		= \sum_{\rho \in P_{\gam}} H'_{\gam, \rho, \rho}
		&\preceq \sum_{\rho \in \calP_U} H'_{\gam, \rho, \rho}\\
        &= H'_{\gam}
	\end{align*}
	where we used the fact that for all $\rho \in \calP_U$, we have $H'_{\gam,\rho, \rho} \succeq 0$.
\end{proof}

\section{Qualitative bounds for Sparse PCA}\label{sec: spca_qual}

\subsection{Pseudo-calibration}

\begin{definition}[Slack parameter]
	Define the slack parameter to be $\Delta = d^{-C_{\Del}\eps}$ for a constant $C_{\Del} > 0$.
\end{definition}

We will pseudo-calibrate with respect the following pair of random and planted distributions which we denote $\nu$ and $\mu$ respectively.

\SPCAdistributions*

We will again work with the Hermite basis of polynomials. For $a \in \NN^{m \times d}$ and variables $v_{i, j}$ for $i \in [m], j \in [n]$, define $h_a(v) \defeq \prod_{i \in [m], j \in [n]} h_{a_{i, j}}(v_{i, j})$.
For a nonnegative integer $t$, define $t!!= \frac{(2t)!}{t!2^t} = 1 \times 3 \times \ldots \times t$ if $t$ is odd and $0$ otherwise.

\begin{lemma}
	Let $I \in \NN^d, a \in \NN^{m \times d}$. For $i \in [m]$, let $e_i = \sum_{j \in [d]} a_{ij}$ and for $j \in [d]$, let $f_j = I_j + \sum_{i \in [m]} a_{ij}$. Let $c_1$ (resp. $c_2$) be the number of $i$ (resp. $j$) such that $e_i > 0$ (resp. $f_j > 0$). Then, if $e_i, f_j$ are all even, we have
	\[\EE_{\mu}[u^I h_a(v)] = \left(\frac{1}{\sqrt{k}}\right)^{|I|} \left(\frac{k}{d}\right)^{c_2} \Delta^{c_1}\prod_{i \in [m]} (e_i - 1)!! \prod_{i, j} \frac{\sqrt{\lambda}^{a_{ij}}}{\sqrt{k}^{a_{ij}}}\]
	Else, $\EE_{\mu}[u^I h_a(v)] = 0$.
\end{lemma}

\begin{proof}
	$v_1, \ldots, v_m \sim \mu$ can be written as $v_i = g_i + \sqrt{\lda} b_i l_i u$ where $g_i \sim \GN(0, I_d), l_i \sim \GN(0, 1), b_i \in \{0, 1\}$ where $b_i = 1$ with probability $\Del$.
	Let's analyze when the required expectation is nonzero. We can first condition on $b_i, l_i, u$ and use the fact that for a fixed $t$, $\EE_{g \sim \GN(0, 1)}[h_k(g + t)] = t^k$ to obtain
    \begin{align*}
	\EE_{(u, l_i, b_i, g_i) \sim \mu}[u^I h_a(v)] &= \EE_{(u, l_i, b_i) \sim \mu}[u^I\prod_{i, j}(\sqrt{\lda}b_il_iu_j)^{a_{ij}}]\\
    &= \EE_{(u, l_i, b_i) \sim \mu}[\prod_{i \in [m]} (b_il_i)^{e_i}\prod_{j \in [d]} u_j^{f_j}] \prod_{i, j} \sqrt{\lda}^{a_{ij}}
    \end{align*}
	For this to be nonzero, the set of $c_1$ indices $i$ such that $e_i > 0$, should not have been resampled otherwise $b_i = 0$, each of which happens independently with probability $\Del$. And the set of $c_2$ indices $j$ such that $f_j > 0$ should have been such that $u_j$ is nonzero, each of which happens independently with probability $\frac{k}{d}$. Since $l_i, u_j$ are have zero expectation in $\nu$, we need $e_i, f_j$ to be even. The expectation then becomes
    {\footnotesize
    \begin{align*}
	\Del^{c_1} \left(\frac{k}{d}\right)^{c_2}\EE_{(u, l_i) \sim \mu}[\prod_{i \in [m]} l_i^{e_i}\prod_{j \in [d]} u_j^{f_j}] \prod_{i, j} \sqrt{\lda}^{a_{ij}} = \left(\frac{1}{\sqrt{k}}\right)^{|I|} \left(\frac{k}{d}\right)^{c_2} \Del^{c_1}\prod_{i \in [m]} (e_i - 1)!! \prod_{i, j} \frac{\sqrt{\lambda}^{a_{ij}}}{\sqrt{k}^{a_{ij}}}
    \end{align*}}
	The last equality follows because, for each $j$ such that $u_j$ is nonzero, we have $u_j^t = (\frac{1}{\sqrt{k}})^t$ and $\EE_{g \sim \GN(0, 1)}[g^t] = (t - 1)!!$ if $t$ is even.
\end{proof}

Define the degree of SoS to be $D_{sos} = d^{C_{sos}\eps}$ for some constant $C_{sos} > 0$ that we choose later.
Define the truncation parameters to be $D_V = d^{C_V\eps}, D_E = d^{C_E\eps}$ for some constants $C_V, C_E > 0$. Regarding the choice of parameters, although we are working with a different problem, \cref{rmk: choice_of_params2} directly applies.

The underlying graphs for the graph matrices have the following structure:
There will be two types of vertices - $d$ type $1$ vertices corresponding to the dimensions of the space and $m$ type $2$ vertices corresponding to the different input vectors. The shapes will correspond to bipartite graphs with edges going between across of different types.
For the analysis of Sparse PCA, we will use the following notation.
\begin{itemize}
	\item For a shape $\al$ and type $t \in \{1, 2\}$, let $V_t(\al)$ denote the vertices of $V(\al)$ that are of type $t$. Let $|\al|_t = |V_t(\al)|$.
	\item For an index shape $U$ and a vertex $i$, define $deg^{U}(i)$ as follows: If $i \in V(U)$, then it is the power of the unique index shape piece $A \in U$ such that $i \in V(A)$. Otherwise, it is $0$.
	\item For an index shape $U$, define $deg(U) = \sum_{i \in V(U)} deg^U(i)$. This is also the degree of the monomial $p_U$.
	\item For a shape $\alpha$ and vertex $i$ in $\alpha$, let $deg^{\alpha}(i) = \sum_{i \in e \in E(\alpha)} l_e$.
	\item For any shape $\alpha$, let $deg(\alpha) = deg(U_{\al}) + deg(V_{\al})$.
	\item For an index shape $U \in \calI_{mid}$ and type $t \in \{1, 2\}$, let $U_t \in U$ denote the index shape piece of type $t$ in $U$ if it exists, otherwise define $U_t$ to be $\emptyset$. Note that this is well defined since for each type $t$, there is at most one index shape piece of type $t$ in $U$ since $U \in \calI_{mid}$. Also, denote by $|U|_t$ the length of the tuple $U_t$.
\end{itemize}

We will now describe the decomposition of the moment matrix $\Lda$.

\begin{definition}\label{def: spca_coeffs}
	If a shape $\alpha$ satisfies the following properties:
	\begin{itemize}
		\item Both $U_{\alpha}$ and $V_{\alpha}$ only contain index shape pieces of type $1$,
		\item $deg^{\alpha}(i) + deg^{U_{\alpha}}(i) + deg^{V_{\alpha}}(i)$ is even for all $i \in V(\alpha)$,
		\item $\alpha$ is proper,
		\item $\alpha$ satisfies the truncation parameters $D_{sos}, D_V, D_E$.
	\end{itemize}
	then define \[\lambda_{\alpha} = \left(\frac{1}{\sqrt{k}}\right)^{deg(\alpha)}\left(\frac{k}{d}\right)^{|\alpha|_1}\Del^{|\alpha|_2} \prod_{j \in V_2(\alpha)} (deg^{\alpha}(j) - 1)!!\prod_{e \in E(\alpha)} \frac{\sqrt{\lambda}^{l_e}}{\sqrt{k}^{l_e}}\]
	Otherwise, define $\lambda_{\alpha} = 0$.
\end{definition}

\begin{corollary}
	$\Lambda = \sum \lda_{\al}M_{\al}$.
\end{corollary}

\subsection{Qualitative machinery bounds}

In this section, we will prove the main PSD mass condition and obtain qualitative bounds of the other two conditions, which we will reuse in the full verification.
As in prior sections, we will state the bounds first, introduce notation and then prove them all in a unified manner.

\begin{restatable}[PSD mass]{lemma}{SPCAone}\label{lem: spca_cond1}
	For all $U \in \calI_{mid}$, $H_{Id_U} \succeq 0$
\end{restatable}

We define the following quantities to capture the contribution of the vertices within $\tau, \gam$ to the Fourier coefficients.

\begin{restatable}{definition}{SPCAstau}\label{def: spca_stau}
	For $U \in \calI_{mid}$ and $\tau \in \calM_U$, if $deg^{\tau}(i)$ is even for all vertices $i \in V(\tau) \setminus U_{\tau} \setminus V_{\tau}$, define
	\[S(\tau) =
	\left(\frac{k}{d}\right)^{|\tau|_1 - |U_{\tau}|_1}\Del^{|\tau|_2 - |U_{\tau}|_2} \prod_{j \in V_2(\tau) \setminus U_{\tau} \setminus V_{\tau}} (deg^{\tau}(j) - 1)!!\prod_{e \in E(\tau)} \frac{\sqrt{\lambda}^{l_e}}{\sqrt{k}^{l_e}}\]
	Otherwise, define $S(\tau) = 0$. 	For all $U, V \in \calI_{mid}$ where $w(U) > w(V)$ and $\gam \in \Gam_{U, V}$, if $deg^{\gam}(i)$ is even for all vertices $i$ in $V(\gam) \setminus U_{\gam} \setminus V_{\gam}$, define
	\[S(\gam) =
	\left(\frac{k}{d}\right)^{|\gamma|_1 - \frac{|U_{\gamma}|_1 + |V_{\gamma}|_1}{2}}\Del^{|\gamma|_2 - \frac{|U_{\gamma}|_2 + |V_{\gamma}|_2}{2}} \prod_{j \in V_2(\gamma) \setminus U_{\gamma} \setminus V_{\gamma}} (deg^{\gamma}(j) - 1)!!\prod_{e \in E(\gamma)} \frac{\sqrt{\lambda}^{l_e}}{\sqrt{k}^{l_e}}\]
	Otherwise, define $S(\gam) = 0$.
\end{restatable}

For getting the best bounds, it will be convenient to discretize the Normal distribution. The following fact follows from standard results on Gaussian quadrature, see for e.g. \cite[Lemma 4.3]{diakonikolas2017statistical}.

\begin{fact}[Discretizing the Normal distribution]\label{fact: quadrature}
	There is an absolute constant $C_{disc}$ such that, for any positive integer $D$, there exists a distribution $\calE$ over the real numbers supported on $D$ points $p_1, \ldots, p_D$, such that $|p_i| \le C_{disc} \sqrt{D}$ for all $i \le D$ and
    $\EE_{g \sim \calE}[g^t] = \EE_{g \sim \GN(0, 1)}[g^t]$ for all $t = 0, 1, \ldots, 2D - 1$.
\end{fact}

\begin{definition} For any shape $\tau$, suppose $U' = (U_{\tau})_2, V' = (V_{\tau})_2$ are the type $2$ vertices in $U_{\tau}, V_{\tau}$ respectively. Define
$R(\tau) = (C_{disc}\sqrt{D_E})^{\sum_{j \in U' \cup V'} deg^{\tau}(j)}$.
\end{definition}

We can now state our qualitative bounds.

\begin{restatable}[Qualitative middle shape bounds]{lemma}{SPCAtwosimplified}\label{lem: spca_cond2_simplified}
	For all $U \in\calI_{mid}$ and $\tau \in \calM_U$,
	\[
	\begin{bmatrix}
	\frac{S(\tau)R(\tau)}{|Aut(U)|}H_{Id_U} & H_{\tau}\\
	H_{\tau}^T & \frac{S(\tau)R(\tau)}{|Aut(U)|}H_{Id_U}
	\end{bmatrix}
	\succeq 0\]
\end{restatable}



We again use the canonical definition of $H_{\gam}'$ from \cref{sec: hgamma_qual}.

\begin{restatable}[Qualitative intersection term bounds]{lemma}{SPCAthreesimplified}\label{lem: spca_cond3_simplified}
	For all $U, V \in \calI_{mid}$ where $w(U) > w(V)$ and all $\gam \in \Gam_{U, V}$,
	\[\frac{|Aut(V)|}{|Aut(U)|}\cdot\frac{1}{S(\gam)^2R(\gam)^2}H_{Id_V}^{-\gam, \gam} \preceq H_{\gam}'\]
\end{restatable}

\subsubsection{Proof of the PSD mass condition}

Most of the notation and analysis here are similar to the case of Tensor PCA, we just need to appropriately modify them since there are two types of vertices in the Sparse PCA application.
When we compose shapes $\sig, \sig'$, from \cref{def: spca_coeffs}, in order for $\lda_{\sig\circ \sig'}$ to be nonzero, observe that all vertices $i$ in $\lda_{\sig \circ \sig'}$ should have $deg^{\sig \circ \sig'}(i) + deg^{U_{\sig \circ \sig'}}(i) + deg^{V_{\sig \circ \sig'}}(i)$ to be even. To capture this notion conveniently, we again use the notion of parity vectors.

\begin{definition}
	Define a parity vector $\rho$ to be a vector whose entries are in $\{0, 1\}$.
	For $U\in \calI_{mid}$, define $\calP_U$ to be the set of parity vectors $\rho$ whose coordinates are indexed by $U_1$ followed by $U_2$.
\end{definition}

\begin{definition}
	For a left shape $\sig$, define $\rho_{\sig} \in \calP_{V_{\sig}}$, called the parity vector of $\sig$, to be the parity vector such that for each vertex $i \in V_{\sig}$, the $i$-th entry of $\rho_{\sig}$ is the parity of $deg^{U_{\sig}}(i) + deg^{\sig}(i)$, that is, $(\rho_{\sig})_i \equiv deg^{U_{\sig}}(i) + deg^{\sig}(i) \pmod 2$.
	For $U \in \calI_{mid}$ and $\rho \in \calP_U$, let $\calL_{U, \rho}$ be the set of all left shapes $\sig \in \calL_U$ such that $\rho_{\sig} = \rho$, that is, the set of all left shapes with parity vector $\rho$.
\end{definition}

For a shape $\tau$, for a $\tau$ coefficient matrix $H_{\tau}$ and parity vectors $\rho \in \calP_{U_{\tau}}, \rho' \in \calP_{V_{\tau}}$, define the $\tau$-coefficient matrix $H_{\tau, \rho, \rho'}$ as $H_{\tau ,\rho, \rho'}(\sig, \sig') = H_{\tau}(\sig, \sig')$ if $\sig \in \calL_{U_{\tau}, \rho}, \sig' \in \calL_{V_{\tau}, \rho'}$ and $0$ otherwise. This immediately implies the following proposition.

\begin{propn}
	For any shape $\tau$ and $\tau$-coefficient matrix $H_{\tau}$, we have the equality $H_{\tau} = \sum_{\rho \in \calP_{U_{\tau}}, \rho' \in \calP_{V_{\tau}}} H_{\tau, \rho, \rho'}$
\end{propn}

\begin{propn}
	For any $U \in \calI_{mid}$, $H_{Id_U} = \sum_{\rho \in \calP_U} H_{Id_U, \rho, \rho}$
\end{propn}

\begin{proof}
	For any $\sig, \sig' \in \calL_U$, using \cref{def: spca_coeffs}, note that in order for $H_{Id_U}(\sig, \sig')$ to be nonzero, we must have $\rho_{\sig} = \rho_{\sig'}$.
\end{proof}

We now discretize the normal distribution while matching the first $2D_E - 1$ moments.

\begin{definition}\label{def: discretized_gaussian}
	Let $\calD$ be a distribution over the real numbers obtained by setting $D = D_E$ in \cref{fact: quadrature}. So, in particular, for any $x$ sampled from $\calD$, we have $|x| \le C_{disc}\sqrt{D_E}$ and for $t \le 2D_E - 1$, $\EE_{x \sim \calD}[x^t] = (t - 1)!!$.
\end{definition}

We define the following quantities to capture the contribution of the vertices within $\sig$ to the Fourier coefficients.

\begin{definition}
	For a shape $\sig\in \calL$, if $deg^{\sig}(i) + deg^{U_{\sig}}(i)$ is even for all vertices $i \in V(\sig) \setminus V_{\sig}$, define
	\[T(\sig) = \left(\frac{1}{\sqrt{k}}\right)^{deg(U_{\sig})}\left(\frac{k}{d}\right)^{|\sig|_1 - \frac{|V_{\sig}|_1}{2}}\Del^{|\sig|_2 - \frac{|V_{\sig}|_2}{2}} \prod_{j \in V_2(\sig) \setminus V_{\sig}} (deg^{\sig}(j) - 1)!!\prod_{e \in E(\sig)} \frac{\sqrt{\lambda}^{l_e}}{\sqrt{k}^{l_e}}\]
	Otherwise, define $T(\sig) = 0$.
\end{definition}

\begin{definition}
	Let $U \in \calI_{mid}$. Let $x_i$ for $i \in U_2$ be variables. Denote them collectively as $x_{U_2}$. For $\rho \in \calP_U$, define $v_{\rho, x_{U_2}}$ to be the vector indexed by left shapes $\sig \in \calL$ such that the $\sig$th entry is $T(\sig) \prod_{i \in {U_2}} x_i^{deg^{\sig}(i)}$ if $\sig \in \calL_{U, \rho}$ and $0$ otherwise.
\end{definition}

The following proposition is obvious and immediately implies the PSD mass condition.

\begin{propn}
	For any $U\in \calI_{mid}, \rho \in \calP_U$, suppose $x_i$ for $i \in U_2$ are random variables sampled from $\calD$. Then,
	$H_{Id_U, \rho, \rho} = \frac{1}{|Aut(U)|}\EE_{x}[v_{\rho, x_{U_2}}v_{\rho, x_{U_2}}^T]$.
\end{propn}

\begin{proof}
	Observe that for $\sig, \sig' \in \calL_{U, \rho}$ and $t \in \{1, 2\}$, $(|\sig|_t - \frac{|V_{\sig}|_t}{2}) + (|\sig'|_t - \frac{|V_{\sig'}|_t}{2}) = |\sig \circ \sig'|_t$. The result follows by verifying the conditions of \cref{def: spca_coeffs} and using \cref{def: discretized_gaussian}.
\end{proof}


\begin{proof}[Proof of the PSD mass condition \cref{lem: spca_cond1}]
	We have $H_{Id_U} = \sum_{\rho \in \calP_U} H_{Id_U, \rho, \rho} \succeq 0$ because of the above proposition.
\end{proof}

\subsubsection{Qualitative middle shape bounds}

The next proposition captures the fact that when we compose shapes $\sig, \tau, \sig'^T$, in order for $\lda_{\sig \circ \tau \circ \sig'^T}$ to be nonzero, the parities of the degrees of the merged vertices should add up correspondingly.

\begin{propn}\label{propn: spca_coeff_2}
	For all $U \in \calI_{mid}$ and $\tau \in \calM_U$, there exist two sets of parity vectors $P_{\tau}, Q_{\tau} \subseteq \calP_{U}$ and a bijection $\pi : P_{\tau} \to Q_{\tau}$ such that $H_{\tau} = \sum_{\rho \in P_{\tau}} H_{\tau, \rho, \pi(\rho)}$.
\end{propn}

\begin{proof}
	Using \cref{def: spca_coeffs}, in order for $H_{\tau}(\sig, \sig')$ to be nonzero, we must have that, in $\sig \circ \tau \circ \sig'$, for all $i \in U_{\tau} \cup V_{\tau}$, $deg^{U_{\sig}}(i) + deg^{U_{\sig'}}(i) + deg^{\sigma \circ \tau \circ \sigma'^T}(i)$ must be even. In other words, for any $\rho \in \calP_U$, there is at most one $\rho' \in \calP_U$ such that if we take $\sig \in \calL_{U, \rho}, \sig' \in \calL_U$ with $H_{\tau}(\sig, \sig')$ nonzero, then the parity of $\sig'$ is $\rho'$. Also, observe that $\rho'$ determines $\rho$. We then take $P_{\tau}$ to be the set of $\rho$ such that $\rho'$ exists, $Q_{\tau}$ to be the set of $\rho'$ and in this case, we define $\pi(\rho) = \rho'$.
\end{proof}



\begin{propn}
	For any $U \in \calI_{mid}$ and $\tau \in \calM_U$, suppose we take $\rho \in P_{\tau}$.  Let $\pi$ be the bijection from \cref{propn: spca_coeff_2} so that $\pi(\rho) \in Q_{\tau}$. Let $U' = (U_{\tau})_2, V' = (V_{\tau})_2$ be the type $2$ vertices in $U_{\tau}, V_{\tau}$ respectively. Let $x_i$ for $i \in U' \cup V'$ be random variables independently sampled from $\calD$. Define $x_{U'}$ (resp. $x_{V'}$) to be the subset of variables $x_i$ for $i \in U'$ (resp. $i \in V'$). Then,
	\[H_{\tau, \rho, \pi(\rho)} = \frac{1}{|Aut(U)|^2} S(\tau) \EE_x\left[v_{\rho, x_{U'}}\left(\prod_{i \in U' \cup V'} x_i^{deg^{\tau}(i)}\right)v_{\pi(\rho), x_{V'}}^T\right]\]
\end{propn}

\begin{proof}
	For $\sig \in L_{U, \rho}, \sig' \in \calL_{U, \pi(\rho)}$ and $t \in \{1, 2\}$, we have $(|\tau|_t - |U_{\tau}|_t) + (|\sig|_t - \frac{|V_{\sig}|_t}{2}) + (|\sig'|_t - \frac{|V_{\sig'}|_t}{2}) = |\sig \circ \tau\circ \sig'|_t$.
	The result then follows by a straightforward verification of the conditions of \cref{def: spca_coeffs} using \cref{def: discretized_gaussian}.
\end{proof}


We are ready to show the qualitative middle shape bounds.

\begin{proof}[Proof of the qualitative middle shape bounds \cref{lem: spca_cond2_simplified}]
	Let $P_{\tau}, Q_{\tau}, \pi$ be from \cref{propn: spca_coeff_2}. Let $U' = (U_{\tau})_2, V' = (V_{\tau})_2$ be the type $2$ vertices in $U_{\tau}, V_{\tau}$ respectively. Let $x_i$ for $i \in U' \cup V'$ be random variables independently sampled from $\calD$. Define $x_{U'}$ (resp. $x_{V'}$) to be the subset of variables $x_i$ for $i \in U'$ (resp. $i \in V'$).

	For $\rho \in \calP_U$, define $W_{\rho, \rho} = \EE_{y_{U_2} \sim \calD^{U_2}}[v_{\rho, y_{U_2}}v_{\rho, y_{U_2}}^T]$ so that $H_{Id_U, \rho, \rho} = \frac{1}{|Aut(U)|} W_{\rho, \rho}$. Observe that $W_{\rho, \rho} = \EE[v_{\rho, x_{U'}}v_{\rho, x_{U'}}^T] = \EE[v_{\rho, x_{V'}}v_{\rho, x_{V'}}^T]$ because $x_{U'}$ and $x_{V'}$ are also sets of variables sampled from $\calD$ and, $U'$, $V'$ have the same size as $U_2$ because $U_{\tau} = V_{\tau} = U$.

	For $\rho, \rho' \in \calP_U$, define $Y_{\rho, \rho'} = \EE\left[v_{\rho, x_{U'}}\left(\prod_{i \in U' \cup V'} x_i^{deg^{\tau}(i)}\right)v_{\pi(\rho), x_{V'}}^T\right]$. Then, $H_{\tau} = \sum_{\rho \in P_{\tau}} H_{\tau, \rho, \pi(\rho)} = \frac{1}{|Aut(U)|^2}S(\tau)\sum_{\rho \in P_{\tau}} Y_{\rho, \pi(\rho)}$. We have

	\begin{align*}
	\begin{bmatrix}
	\frac{S(\tau)R(\tau)}{|Aut(U)|}H_{Id_U} & H_{\tau}\\
	H_{\tau}^T & \frac{S(\tau)R(\tau)}{|Aut(U)|}H_{Id_U}
	\end{bmatrix}
	&= \frac{S(\tau)}{|Aut(U)|^2}
	\begin{bmatrix}
	R(\tau)\sum_{\rho \in \calP_U} W_{\rho, \rho} & \sum_{\rho \in P_{\tau}} Y_{\rho, \pi(\rho)}\\
	\sum_{\rho \in P_{\tau}} Y_{\rho, \pi(\rho)}^T & R(\tau)\sum_{\rho \in \calP_U} W_{\rho, \rho}
	\end{bmatrix}
	\end{align*}
	Since $\frac{S(\tau)}{|Aut(U)|^2} \ge 0$, it suffices to prove that $\begin{bmatrix}
	R(\tau)\sum_{\rho \in \calP_U} W_{\rho, \rho} & \sum_{\rho \in P_{\tau}} Y_{\rho, \pi(\rho)}\\
	\sum_{\rho \in P_{\tau}} Y_{\rho, \pi(\rho)}^T & R(\tau)\sum_{\rho \in \calP_U} W_{\rho, \rho}
	\end{bmatrix}\succeq 0$. Consider
	\begin{align*}
		\begin{bmatrix}
			R(\tau)\sum_{\rho \in \calP_U} W_{\rho, \rho} & \sum_{\rho \in P_{\tau}} Y_{\rho, \pi(\rho)}\\
			\sum_{\rho \in P_{\tau}} Y_{\rho, \pi(\rho)}^T & R(\tau)\sum_{\rho \in \calP_U} W_{\rho, \rho}
		\end{bmatrix} =& R(\tau)\begin{bmatrix}
		\sum_{\rho \in \calP_U \setminus P_{\tau}} W_{\rho, \rho} & 0\\
		0 & \sum_{\rho \in \calP_U \setminus Q_{\tau}} W_{\rho, \rho}
		\end{bmatrix}\\
		& + \begin{bmatrix}
		R(\tau)\sum_{\rho \in P_{\tau}} W_{\rho, \rho} & \sum_{\rho \in P_{\tau}} Y_{\rho, \pi(\rho)}\\
		\sum_{\rho \in P_{\tau}} Y_{\rho, \pi(\rho)}^T & R(\tau)\sum_{\rho \in P_{\tau}} W_{\pi(\rho), \pi(\rho)}
		\end{bmatrix}\\
	\end{align*}

	We have $\sum_{\rho \in \calP_U \setminus P_{\tau}} W_{\rho, \rho} = \sum_{\rho \in \calP_U \setminus P_{\tau}} \EE[v_{\rho, x_{U'}}v_{\rho, x_{U'}}^T] \succeq 0$. Similarly, $\sum_{\rho \in \calP_U \setminus Q_{\tau}} W_{\rho, \rho} \succeq 0$. Also, $R(\tau) \ge 0$ and therefore, we have that the first term in the above expression,
	$R(\tau)\begin{bmatrix}
	\sum_{\rho \in \calP_U \setminus P_{\tau}} W_{\rho, \rho} & 0\\
	0 & \sum_{\rho \in \calP_U \setminus Q_{\tau}} W_{\rho, \rho}
	\end{bmatrix}$, is positive semidefinite. For the second term,
{\footnotesize
	\begin{align*}
	&\begin{bmatrix}
	R(\tau)\sum_{\rho \in P_{\tau}} W_{\rho, \rho} & \sum_{\rho \in P_{\tau}} Y_{\rho, \pi(\rho)}\\
	\sum_{\rho \in P_{\tau}} Y_{\rho, \pi(\rho)}^T & R(\tau)\sum_{\rho \in P_{\tau}} W_{\pi(\rho), \pi(\rho)}
	\end{bmatrix}\\
	&\qquad= \sum_{\rho \in P_{\tau}}
	\begin{bmatrix}
	R(\tau)\EE[v_{\rho, x_{U'}}v_{\rho, x_{U'}}^T] & \EE\left[v_{\rho, x_{U'}}\left(\prod_{i \in U' \cup V'} x_i^{deg^{\tau}(i)}\right)v_{\pi(\rho), x_{V'}}^T\right]\\
	\EE\left[v_{\rho, x_{U'}}^T\left(\prod_{i \in U' \cup V'} x_i^{deg^{\tau}(i)}\right)v_{\pi(\rho), x_{V'}}\right] & R(\tau)\EE[v_{\pi(\rho), x_{V'}}v_{\pi(\rho), x_{V'}}^T]
	\end{bmatrix}\\
	&\qquad= \sum_{\rho \in P_{\tau}}\EE
	\begin{bmatrix}
	R(\tau)v_{\rho, x_{U'}}v_{\rho, x_{U'}}^T & v_{\rho, x_{U'}}\left(\prod_{i \in U' \cup V'} x_i^{deg^{\tau}(i)}\right)v_{\pi(\rho), x_{V'}}^T\\
	v_{\rho, x_{U'}}^T\left(\prod_{i \in U' \cup V'} x_i^{deg^{\tau}(i)}\right)v_{\pi(\rho), x_{V'}} & R(\tau)v_{\pi(\rho), x_{V'}}v_{\pi(\rho), x_{V'}}^T
	\end{bmatrix}
	\end{align*}}

We will prove that the term inside the expectation is positive semidefinite for each $\rho \in P_{\tau}$ and each sampling of the $x_i$ from $\calD$, which will complete the proof. Fix $\rho \in P_{\tau}$ and any sampling of the $x_i$ from $\calD$. Let $w_1 = v_{\rho, X_{U'}}, w_2 = v_{\pi(\rho), x_{V'}}$. Let $E = \prod_{i \in U' \cup V'} x_i^{deg^{\tau}(i)}$. We would like to prove that $\begin{bmatrix}
	R(\tau)w_1w_1^T & Ew_1w_2^T\\
	Ew_1^Tw_2 & R(\tau)w_2w_2^T
\end{bmatrix} \succeq 0$. For all $y$ sampled from $\calD$, $|y| \le C_{disc}\sqrt{D_E}$ and so, $|E| \le (C_{disc}\sqrt{D_E})^{\sum_{j \in U' \cup V'} deg^{\tau}(j)} = R(\tau)$.

If $E \ge 0$, then
{\footnotesize
\begin{align*}
	\begin{bmatrix}
		R(\tau)w_1w_1^T & Ew_1w_2^T\\
		Ew_1^Tw_2 & R(\tau)w_2w_2^T
	\end{bmatrix} &= (R(\tau) - E)
	\begin{bmatrix}
		w_1w_1^T & 0\\
		0 & w_2w_2^T
	\end{bmatrix}
	+ E\begin{bmatrix}
		w_1w_1^T & w_1w_2^T\\
		w_1^Tw_2 & w_2w_2^T
	\end{bmatrix}\\
	&= (R(\tau) - E)\left(
	\begin{bmatrix}
	w_1\\
	0
	\end{bmatrix}
	\begin{bmatrix}
	w_1 & 0
	\end{bmatrix} +
	\begin{bmatrix}
		0\\
		w_2
	\end{bmatrix}
	\begin{bmatrix}
		0 & w_2
	\end{bmatrix}\right) +
	E\begin{bmatrix}
	w_1\\
	w_2
	\end{bmatrix}
	\begin{bmatrix}
		w_1 & w_2
	\end{bmatrix}\\
& \succeq 0
\end{align*}}
since $R(\tau) - E \ge 0$ And if $E < 0$,
{\footnotesize
\begin{align*}
	\begin{bmatrix}
		R(\tau)w_1w_1^T & Ew_1w_2^T\\
		Ew_1^Tw_2 & R(\tau)w_2w_2^T
	\end{bmatrix} &= (R(\tau) + E)
	\begin{bmatrix}
		w_1w_1^T & 0\\
		0 & w_2w_2^T
	\end{bmatrix}
	- E\begin{bmatrix}
		w_1w_1^T & -w_1w_2^T\\
		-w_1^Tw_2 & w_2w_2^T
	\end{bmatrix}\\
	&= (R(\tau) + E)\left(
	\begin{bmatrix}
		w_1\\
		0
	\end{bmatrix}
	\begin{bmatrix}
		w_1 & 0
	\end{bmatrix} +
	\begin{bmatrix}
		0\\
		w_2
	\end{bmatrix}
	\begin{bmatrix}
		0 & w_2
	\end{bmatrix}\right)
	- E\begin{bmatrix}
		w_1\\
		-w_2
	\end{bmatrix}
	\begin{bmatrix}
		w_1 & -w_2
	\end{bmatrix}\\
	& \succeq 0
\end{align*}}
since $R(\tau) + E \ge 0$.
\end{proof}

\subsubsection{Qualitative intersection term bounds}

Just as in \cref{propn: spca_coeff_2}, the next proposition captures the fact that when we compose shapes $\sig, \gam, \gam^T, \sig'^T$, in order for $\lda_{\sig \circ \gam \circ \gam^T \circ \sig'^T}$ to be nonzero, the parities of the degrees of the merged vertices should add up correspondingly.
Just as in the tensor PCA application, we similarly define $H_{Id_V, \rho, \rho'}^{-\gam, \gam}$ and $H'_{\gam, \rho, \rho}$. The following propositions are simple and proved the same way.


\begin{propn}
	For all $U, V \in \calI_{mid}$ where $w(U) > w(V)$, for all $\gam \in \Gam_{U, V}$, there exists a set of parity vectors $P_{\gam} \subseteq \calP_U$ such that
	$H_{Id_V}^{-\gam, \gam} = \sum_{\rho \in P_{\gam}} H_{Id_V, \rho, \rho}^{-\gam, \gam}$.
\end{propn}



\begin{propn}
	For all $U, V \in \calI_{mid}$ where $w(U) > w(V)$, for $\gam \in \Gam_{U, V}$, $H_{\gam}' = \sum_{\rho \in P_{\gam}} H_{\gam, \rho, \rho}'$.
\end{propn}

We will now define vectors which are truncations of $v_{\rho, x_{U_2}}$. This definition and the following proposition are mostly a matter of technicality and they are essentially similar to the PSD mass condition analysis.

\begin{definition}
    Let $U, V \in \calI_{mid}$ where $w(U) > w(V)$, and let $\gam \in \Gam_{U, V}$. Let $x_i$ for $i \in U_2$ be variables. Denote them collectively as $x_{U_2}$. For $\rho \in \calP_U$, define $v_{\rho, x_{U_2}}^{-\gam}$ to be the vector indexed by left shapes $\sig \in \calL$ such that the $\sig$th entry is $v_{\rho, x_{U_2}}(\sig)$ if $|V(\sig \circ \gam)| \le D_V$ and $0$ otherwise.
\end{definition}



With this, we can decompose each slice $H_{Id_V, \rho, \rho}^{-\gam, \gam}$.

\begin{propn}\label{lem: spca_decomp}
	For any $U, V \in \calI_{mid}$ where $w(U) > w(V)$, and for any $\gam \in \Gam_{U, V}$, suppose we take $\rho \in P_{\gam}$. When we compose $\gam$ with $\gam^T$ to get $\gam \circ \gam^T$, let $U' = (U_{\gam \circ \gam^T})_2, V' = (V_{\gam \circ \gam^T})_2$ be the type $2$ vertices in $U_{\gam \circ \gam^T}, V_{\gam \circ \gam^T}$ respectively. And let $W'$ be the set of type $2$ vertices in $\gam \circ \gam^T$ that were identified in the composition when we set $V_{\gam} = U_{\gam}^T$. Let $x_i$ for $i \in U' \cup W' \cup V'$ be random variables independently sampled from $\calD$. Define $x_{U'}$ (resp. $x_{V'}, x_{W'}$) to be the subset of variables $x_i$ for $i \in U'$ (resp. $i \in V', i \in W'$). Then,
	\[H_{Id_V, \rho, \rho}^{-\gam, \gam} = \frac{1}{|Aut(V)|}S(\gam)^2 \EE_x\left[(v_{\rho, x_{U'}}^{-\gam})\left(\prod_{i \in U' \cup W' \cup V'} x_i^{deg^{\gam \circ \gam^T}(i)}\right)(v_{\rho, x_{V'}}^{-\gam})^T\right]\]
\end{propn}

\begin{proof}
	Fix $\sig, \sig' \in \calL_{U, \rho}$ such that $|V(\sig \circ \gam)|, |V(\sig' \circ \gam)| \le D_V$. Note that for $t \in \{1, 2\}$, $|\sig|_t - \frac{|V_{\sig}|_t}{2} + |\sig'|_t - \frac{|V_{\sig'}|_t}{2} + 2(|\gam|_t - \frac{|U_{\gam}|_t + |V_{\gam}|_t}{2}) = |\sig \circ \gam \circ \gam^T \circ \sig'^T|_t$. We can easily verify the equality using \cref{def: spca_coeffs} and \cref{def: discretized_gaussian}.
\end{proof}

\begin{propn}
	For any $U, V \in \calI_{mid}$ where $w(U) > w(V)$, and for any $\gam \in \Gam_{U, V}$, suppose we take $\rho \in \calP_U$. Then,
	\[H'_{\gam, \rho, \rho} = \frac{1}{|Aut(U)|}\EE_{y_{U_2}\sim \calD^{U_2}}\left[(v_{\rho, y_{U_2}}^{-\gam})(v_{\rho, y_{U_2}}^{-\gam})^T\right]\]
\end{propn}



We can finally show the qualitative intersection term bounds.

\begin{proof}[Proof of the qualitative intersection term bounds \cref{lem: spca_cond3_simplified}]
    Let $U', V', W'$ be defined as in \cref{lem: spca_decomp}. We have
    {\footnotesize
	\begin{align*}
	\frac{|Aut(V)|}{|Aut(U)|}\cdot\frac{1}{S(\gam)^2R(\gam)^2}H_{Id_V}^{-\gam, \gam} &= \sum_{\rho \in P_{\gam}} \frac{|Aut(V)|}{|Aut(U)|}\cdot\frac{1}{S(\gam)^2R(\gam)^2} H_{Id_V, \rho, \rho}^{-\gam, \gam}\\
	&= \sum_{\rho \in P_{\gam}} \frac{1}{|Aut(U)|}\cdot\frac{1}{R(\gam)^2} \EE_x\left[(v_{\rho, x_{U'}}^{-\gam})\left(\prod_{i \in U' \cup W' \cup V'} x_i^{deg^{\gam \circ \gam^T}(i)}\right)(v_{\rho, x_{V'}}^{-\gam})^T\right]
\end{align*}}

We will now prove that, for all $\rho \in P_{\gam}$,
\begin{align*}
    \frac{1}{|Aut(U)|}\cdot \frac{1}{R(\gam)^2} \EE_x\left[(v_{\rho, x_{U'}}^{-\gam})\left(\prod_{i \in U' \cup W' \cup V'} x_i^{deg^{\gam \circ \gam^T}(i)}\right)(v_{\rho, x_{V'}}^{-\gam})^T\right] \preceq H'_{\gam, \rho, \rho}
\end{align*}
which reduces to proving that
{\footnotesize
\begin{align*}
    \frac{2}{R(\gam)^2} \EE_x\left[(v_{\rho, x_{U'}}^{-\gam})\left(\prod_{i \in U' \cup W' \cup V'} x_i^{deg^{\gam \circ \gam^T}(i)}\right)(v_{\rho, x_{V'}}^{-\gam})^T\right] &\preceq 2\EE_{y_{U_2}\sim \calD^{U_2}}\left[(v_{\rho, y_{U_2}}^{-\gam})(v_{\rho, y_{U_2}}^{-\gam})^T\right]\\
    &= \EE_{x}\left[(v_{\rho, x_{U'}}^{-\gam})(v_{\rho, x_{U'}}^{-\gam})^T + (v_{\rho, x_{V'}}^{-\gam})(v_{\rho, x_{V'}}^{-\gam})^T\right]
\end{align*}}
where the last equality followed from linearity of expectation and the fact that $U' \equiv V' \equiv U_2$.

Since $H_{Id_V, \rho, \rho}^{-\gam, \gam}$ is symmetric, we have
{\footnotesize
\[\EE_x\left[(v_{\rho, x_{U'}}^{-\gam})\left(\prod_{i \in U' \cup W' \cup V'} x_i^{deg^{\gam \circ \gam^T}(i)}\right)(v_{\rho, x_{V'}}^{-\gam})^T\right] = \EE_x\left[(v_{\rho, x_{V'}}^{-\gam})\left(\prod_{i \in U' \cup W' \cup V'} x_i^{deg^{\gam \circ \gam^T}(i)}\right)(v_{\rho, x_{U'}}^{-\gam})^T\right]\]}
So, it suffices to prove
{\footnotesize
\begin{align*}
    \frac{1}{R(\gam)^2}&\EE_x\left[(v_{\rho, x_{U'}}^{-\gam})\left(\prod_{i \in U' \cup W' \cup V'} x_i^{deg^{\gam \circ \gam^T}(i)}\right)(v_{\rho, x_{V'}}^{-\gam})^T + (v_{\rho, x_{V'}}^{-\gam})\left(\prod_{i \in U' \cup W' \cup V'} x_i^{deg^{\gam \circ \gam^T}(i)}\right)(v_{\rho, x_{U'}}^{-\gam})^T\right]\\
    &\preceq \EE_{x}\left[(v_{\rho, x_{U'}}^{-\gam})(v_{\rho, x_{U'}}^{-\gam})^T + (v_{\rho, x_{V'}}^{-\gam})(v_{\rho, x_{V'}}^{-\gam})^T\right]
\end{align*}}

We will prove that for every sampling of the $x_i$ from $\calD$, we have
{\footnotesize
\begin{align*}
    \frac{1}{R(\gam)^2}&\left((v_{\rho, x_{U'}}^{-\gam})\left(\prod_{i \in U' \cup W' \cup V'} x_i^{deg^{\gam \circ \gam^T}(i)}\right)(v_{\rho, x_{V'}}^{-\gam})^T + (v_{\rho, x_{V'}}^{-\gam})\left(\prod_{i \in U' \cup W' \cup V'} x_i^{deg^{\gam \circ \gam^T}(i)}\right)(v_{\rho, x_{U'}}^{-\gam})^T\right) \\
    &\preceq (v_{\rho, x_{U'}}^{-\gam})(v_{\rho, x_{U'}}^{-\gam})^T + (v_{\rho, x_{V'}}^{-\gam})(v_{\rho, x_{V'}}^{-\gam})^T
\end{align*}}
Then, taking expectations will give the result. Indeed, fix a sampling of the $x_i$ from $\calD$. Let $E = \prod_{i \in U' \cup W' \cup V'} x_i^{deg^{\gam \circ \gam^T}(i)}$ and let $w_1 = v_{\rho, x_{U'}}^{-\gam}, w_2 = v_{\rho, x_{V'}}^{-\gam}$. Then, the inequality we need to show is
\[\frac{E}{R(\gam)^2}(w_1w_2^T + w_2w_1^T) \preceq w_1w_1^T + w_2w_2^T\]
Now, since $|x_i| \le C_{disc}\sqrt{D_E}$ for all $i$, we have $|E| \le \prod_{i \in U' \cup W' \cup V'} (C_{disc}\sqrt{D_E})^{deg^{\gam \circ \gam^T}(i)} =  R(\gam)^2$.
If $E \ge 0$, using $\frac{E}{R(\gam)^2}(w_1 - w_2)(w_1 - w_2)^T \succeq 0$ gives
\begin{align*}
    \frac{E}{R(\gam)^2} (w_1w_2^T + w_2w_1^T) &\preceq \frac{E}{R(\gam)^2} (w_1w_1^T + w_2w_2^T)
    \preceq w_1w_1^T + w_2w_2^T
\end{align*}
since $0 \le E \le R(\gam)^2$.
And if $E < 0$, using $\frac{-E}{R(\gam)^2}(w_1 + w_2)(w_1 + w_2)^T \succeq 0$ gives
\begin{align*}
    \frac{E}{R(\gam)^2} (w_1w_2^T + w_2w_1^T) &\preceq \frac{-E}{R(\gam)^2} (w_1w_1^T + w_2w_2^T)
    \preceq w_1w_1^T + w_2w_2^T
\end{align*}
since $0 \le -E \le R(\gam)^2$.
Finally, we use the fact that for all $\rho \in \calP_U$, we have $H'_{\gam,\rho, \rho} \succeq 0$ which can be proved the same way as the proof of \cref{lem: spca_cond1}. Therefore,
\begin{align*}
	\frac{|Aut(V)|}{|Aut(U)|}\cdot\frac{1}{S(\gam)^2R(\gam)^2}H_{Id_V}^{-\gam, \gam} &\preceq \sum_{\rho \in P_{\gam}} H'_{\gam, \rho, \rho}
	\preceq \sum_{\rho \in \calP_U} H'_{\gam, \rho, \rho}
	= H'_{\gam}
	\end{align*}
\end{proof}

\subsection{Intuition for quantitative bounds}

In this section, we will give some intuition on the bounds needed for our main theorem \cref{thm: spca_main}, which is formally proved in \cref{sec: spca_quant}. Informally, the theorem states that when $m \le \frac{d}{\lda^2}$ and $m \le \frac{k^2}{\lda^2}$, then $\Lda \succeq 0$ with high probability.

We will try and understand why the inequality $\lda_{\sig \circ \tau \circ \sig'^T}^2\norm{M_{\tau}}^2 \le \lda_{\sig \circ \sig^T}\lda_{\sig' \circ \sig'^T}$ holds. Assume for simplicity that $d < n$ and consider the shapes in \cref{fig: sparse_pca}. The assumption $d < n$ is used in this example since otherwise, if $d > n$, the decomposition differs from what's shown in the figure.

\begin{figure}[!ht]
    \centering
    \includegraphics[scale=0.9, trim={2cm 5cm 0 4cm},clip]{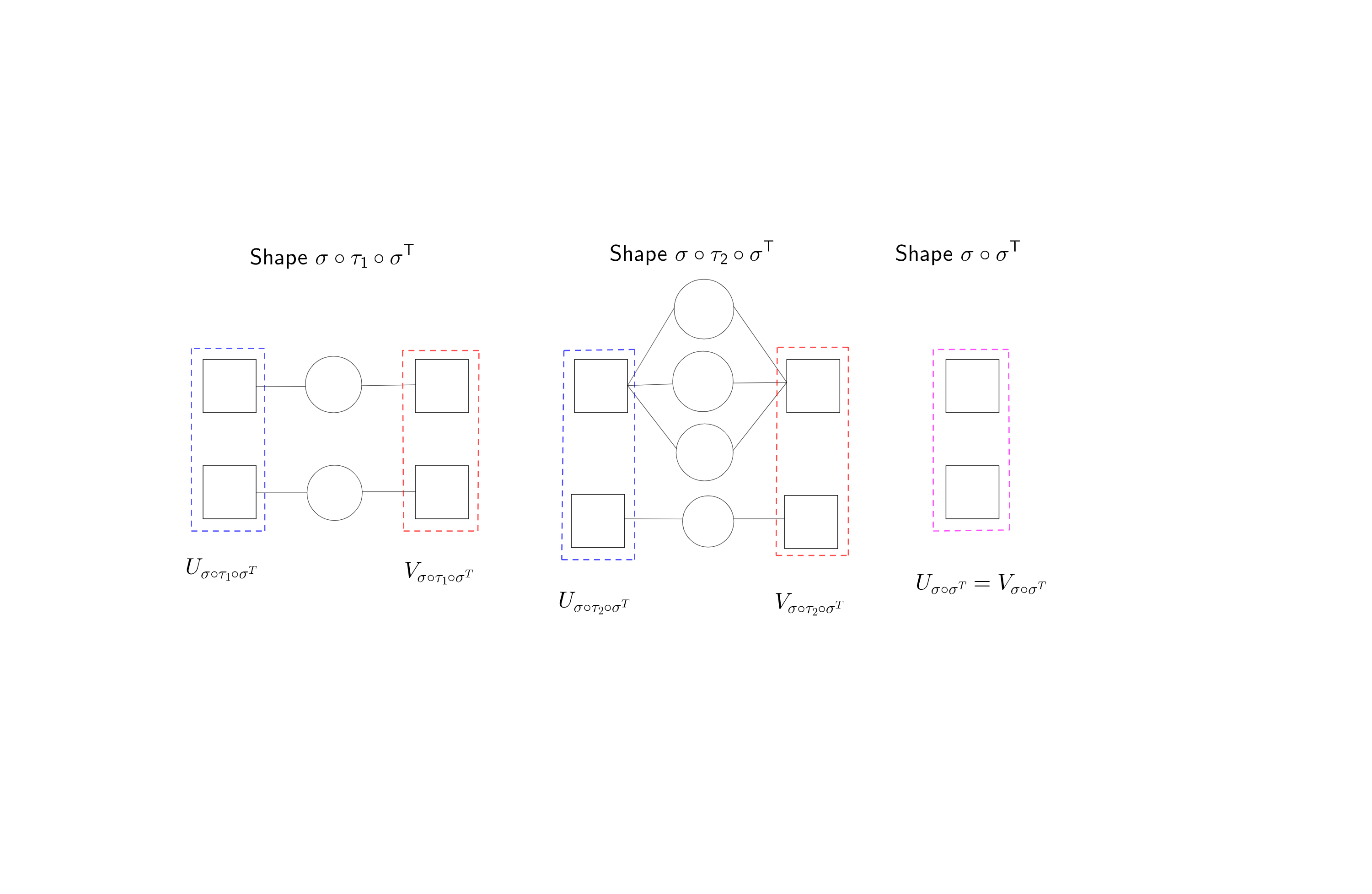}
    \caption{Shapes $\sig \circ \tau_1\circ \sig^T, \sig \circ \tau_2 \circ \sig^T$ and $\sig \circ \sig^T$. All edges have label $1$.}
    \label{fig: sparse_pca}
\end{figure}

Firstly, the shape $\sig \circ \sig^T$ has a coefficient of $\lda_{\sig \circ \sig^T} \approx \left(\frac{1}{\sqrt{k}}\right)^4\left(\frac{k}{d}\right)^2$.
The first shape $\sig \circ \tau_1 \circ \sig^T$ has a coefficient of $\lda_{\sig \circ \tau_1 \circ \sig^T} \approx \left(\frac{1}{\sqrt{k}}\right)^4\left(\frac{k}{d}\right)^4 \left(\frac{\sqrt{\lda}}{\sqrt{k}}\right)^4$ and with high probability, upto lower order terms, $\norm{M_{\tau_1}} \le md$. So, the inequality $\lda_{\sig \circ \tau_1 \circ \sig^T}^2\norm{M_{\tau_1}}^2 \le \lda_{\sig \circ \sig^T}\lda_{\sig \circ \sig^T}$ rearranges to $m \le \frac{d}{\lda^2}$. But this is precisely one of the assumptions on $m$. Moreover, this also confirms that we need this assumption on $m$ in order for our strategy to go through.

The second shape $\sig \circ \tau_2 \circ \sig^T$ has a coefficient of $\lda_{\sig \circ \tau_2 \circ \sig^T} \approx \left(\frac{1}{\sqrt{k}}\right)^4\left(\frac{k}{d}\right)^4 \left(\frac{\sqrt{\lda}}{\sqrt{k}}\right)^8$ and with high probability, upto lower order terms, $\norm{M_{\tau_2}} \le m^2d$. So, the inequality $\lda_{\sig \circ \tau_2 \circ \sig^T}^2\norm{M_{\tau_2}}^2 \le \lda_{\sig \circ \sig^T}\lda_{\sig \circ \sig^T}$ rearranges to $m^2 \le \frac{k^2d}{\lda^4}$. But this is obtained simply by multiplying our assumptions on $m$, namely $m \le \frac{k^2}{\lda^2}$ and $m \le \frac{d}{\lda^2}$.

Moreover, consider a shape of the form $\sig \circ \tau_3 \circ \sig^T$ where $\tau_3$ is similar to $\tau_2$ except it has $t$ (instead of $3$) different circle vertices that are common neighbors to the top 2 square vertices. Analyzing our required inequality, we get for our strategy to go through, $m$ has to satisfy $m \le \frac{k^2}{\lda^2} \cdot \left(\frac{d}{k^2}\right)^{\frac{2}{t + 1}}$. By taking $t$ arbitrarily large, we can see that the condition $m \le \frac{k^2}{\lda^2}$ is needed.

So, we get that for our analysis to go through, the assumptions $m \le \frac{d}{\lda^2}$ and $m \le \frac{k^2}{\lda^2}$ are necessary. We will prove that in fact, these are sufficient. To do this, we use a charging argument that exploits the special structure of the shapes $\al$ that appear in our decomposition of $\Lda$ and their coefficients $\lda_{\al}$, as we obtained in \cref{def: spca_coeffs}. For details, see \cref{sec: spca_quant}.

\chapter{Quantitative bounds}\label{chap: quant}
In this chapter, we will prove the main Sum of Squares lower bounds \cref{thm: plds_main}, \cref{thm: tpca_main} and \cref{thm: spca_main} by building on the qualitative bounds from \cref{chap: qual}. The material in this chapter is adapted from \cite{potechin2020machinery}, however several typos have been fixed and the technical exposition has been improved.

\section{Planted slightly denser subgraph: Full verification}\label{sec: plds_quant}

In this section, we will prove our main theorem on Planted slightly denser subgraph, \cref{thm: plds_main}.

\PLDSmain*

We will apply the machinery. Here, we choose $\eps$ in the theorem, not to be confused with the $\eps$ in \cref{thm: plds_main}, to be an arbitrarily small constant.
We build on the qualitative bounds (and use the same notation) from \cref{sec: plds_qual}.
The result will follow once we verify the main conditions and apply the machinery.

\subsection{\middleshapeboundstwo}

\begin{lemma}\label{lem: plds_charging}
    Suppose $k \le n^{1/2 - \eps}$. For all $U \in \calI_{mid}$ and $\tau \in \calM_U$,
	\[\sqrt{n}^{|V(\tau)| - |U_{\tau}|}S(\tau) \le \frac{1}{n^{C_p\eps|E(\tau)|}}\]
\end{lemma}

\begin{proof}
    This result follows by plugging in the value of $S(\tau)$. Using $k \le n^{1/2 - \eps}$,
	\begin{align*}
	\sqrt{n}^{|V(\tau)| - |U_{\tau}|}S(\tau) &= \sqrt{n}^{|V(\tau)| - |U_{\tau}|} \left(\frac{k}{n}\right)^{|V(\tau)| - |U_{\tau}|}(2(\frac{1}{2} + \frac{1}{2n^{C_p\eps}}) -1 )^{|E(\tau)|}
	\le \frac{1}{n^{C_p\eps|E(\tau)|}}
	\end{align*}
\end{proof}

\begin{corollary}\label{cor: plds_norm_decay}
	For all $U \in \calI_{mid}$ and $\tau \in \calM_U$, we have \[c(\tau)B_{norm}(\tau)S(\tau) \le 1\]
\end{corollary}

\begin{proof}
	Since $\tau$ is a proper middle shape, we have $w(I_{\tau}) = 0$ and $w(S_{\tau}) = w(U_{\tau})$. This implies
	$n^{\frac{w(V(\tau)) + w(I_{\tau}) - w(S_{\tau})}{2}} = \sqrt{n}^{|V(\tau)| - |U_{\tau}|}$.
	Since $\tau$ is proper, every vertex $i \in V(\tau) \setminus U_{\tau}$ or $i \in V(\tau) \setminus V_{\tau}$ has $deg^{\tau}(i) \ge 1$ and hence, $|V(\tau)\setminus U_{\tau}| + |V(\tau)\setminus V_{\tau}| \le 4|E(\tau)|$. Also, $q = n^{O(1) \cdot \eps C_V}$. We can set $C_V$ sufficiently small so that, using \cref{lem: plds_charging},
	{\footnotesize
	\begin{align*}
	c(\tau)&B_{norm}(\tau)S(\tau)\\
	&= 100(3D_V)^{|U_{\tau}\setminus V_{\tau}| + |V_{\tau}\setminus U_{\tau}| + 2|E(\tau)|}2^{|V(\tau)\setminus (U_{\tau}\cup V_{\tau})|}
	\cdot (6D_V\sqrt[4]{2eq})^{|V(\tau)\setminus U_{\tau}| + |V(\tau)\setminus V_{\tau}|}\sqrt{n}^{|V(\tau)| - |U_{\tau}|}S(\tau)\\
	&\le n^{O(1) \cdot \eps C_V \cdot |E(\tau)|} \cdot \sqrt{n}^{|V(\tau)| - |U_{\tau}|}S(\tau)\\
	&\le n^{O(1) \cdot \eps C_V \cdot |E(\tau)|} \cdot \frac{1}{n^{C_p\eps|E(\tau)|}}\\
	&\le 1
	\end{align*}
	}
\end{proof}

We can now obtain middle shape bounds.

\begin{lemma}
    For all $U \in \calI_{mid}$ and $\tau \in \calM_U$,
    \[
\begin{bmatrix}
    \frac{1}{|Aut(U)|c(\tau)}H_{Id_U} & B_{norm}(\tau) H_{\tau}\\
    B_{norm}(\tau) H_{\tau}^T & \frac{1}{|Aut(U)|c(\tau)}H_{Id_U}
\end{bmatrix}
\succeq 0
\]
\end{lemma}

\begin{proof}
	We have
    {\footnotesize
	\begin{align*}
	&\begin{bmatrix}
	\frac{1}{|Aut(U)|c(\tau)}H_{Id_U} & B_{norm}(\tau)H_{\tau}\\
	B_{norm}(\tau)H_{\tau}^T & \frac{1}{|Aut(U)|c(\tau)}H_{Id_U}
	\end{bmatrix}\\
	&\qquad = \begin{bmatrix}
	\left(\frac{1}{|Aut(U)|c(\tau)} - \frac{S(\tau)B_{norm}(\tau)}{|Aut(U)|}\right)H_{Id_U} & 0\\
	0 & \left(\frac{1}{|Aut(U)|c(\tau)} - \frac{S(\tau)B_{norm}(\tau)}{|Aut(U)|}\right)H_{Id_U}
	\end{bmatrix}\\
	&\qquad \qquad + B_{norm}(\tau)\begin{bmatrix}
	\frac{S(\tau)}{|Aut(U)|}H_{Id_U} & H_{\tau}\\
	H_{\tau}^T & \frac{S(\tau)}{|Aut(U)|}H_{Id_U}
	\end{bmatrix}
	\end{align*}}
	By \cref{lem: plds_cond2_simplified}, $\begin{bmatrix}
	\frac{S(\tau)}{|Aut(U)|}H_{Id_U} & H_{\tau}\\
	H_{\tau}^T & \frac{S(\tau)}{|Aut(U)|}H_{Id_U}
	\end{bmatrix}
	\succeq 0$, so the second term above is positive semidefinite. For the first term, by \cref{lem: plds_cond1}, $H_{Id_U} \succeq 0$ and by \cref{cor: plds_norm_decay}, $\frac{1}{|Aut(U)|c(\tau)} - \frac{S(\tau)B_{norm}(\tau)}{|Aut(U)|} \ge 0$, which proves that the first term is also positive semidefinite.
\end{proof}

\subsection{\intersectionboundstwo}

\begin{lemma}\label{lem: plds_charging2}
	Suppose $k \le n^{1/2 - \eps}$. For all $U, V \in \calI_{mid}$ where $w(U) > w(V)$ and for all $\gam \in \Gam_{U, V}$,
	\[n^{w(V(\gam)\setminus U_{\gam})} S(\gam)^2 \le \frac{1}{n^{B\eps (|V(\gam) \setminus (U_{\gam} \cap V_{\gam})| + |E(\gam)|)}}\]
	for some constant $B$ that depends only on $C_p$. In particular, it is independent of $C_V$.
\end{lemma}

\begin{proof}
	Since $\gam$ is a left shape, we have $|U_{\gam}| \ge |V_{\gam}|$ as $V_{\gam}$ is the unique minimum vertex separator of $\gam$ and so, $n^{w(V(\gam) \setminus U_{\gam})} = n^{|V(\gam)| - |U_{\gam}|} \le n^{|V(\gam)| - \frac{|U_{\gam}| + |V_{\gam}|}{2}}$. Also, note that $2|V(\gam)| - |U_{\gam}| - |V_{\gam}| = |U_{\gam} \setminus V_{\gam}| + |V_{\gam} \setminus U_{\gam}| + 2|V(\gam) \setminus U_{\gam} \setminus V_{\gam}| \ge |V(\gam) \setminus (U_{\gam} \cap V_{\gam})|$. Therefore,
	\begin{align*}
	n^{w(V(\gam)\setminus U_{\gam})} S(\gam)^2 &= n^{|V(\gam)\setminus U_{\gam})|} \left(\frac{k}{n}\right)^{2|V(\gam)| - |U_{\gam}| - |V_{\gam}|} (2(\frac{1}{2} + \frac{1}{2n^{C_p\eps}}) - 1)^{2|E(\gam)|}\\
	&\le n^{|V(\gam)| - \frac{|U_{\gam}| + |V_{\gam}|}{2}}\left(\frac{1}{n^{1/2 + \eps}}\right)^{2|V(\gam)| - |U_{\gam}| - |V_{\gam}|}\left(\frac{1}{n^{2C_p\eps}}\right)^{|E(\gam)|}\\
	&\le  \left(\frac{1}{n^{\eps}}\right)^{2|V(\gam)| - |U_{\gam}| - |V_{\gam}|}\left(\frac{1}{n^{2C_p\eps}}\right)^{|E(\gam)|}\\
	&\le \frac{1}{n^{B\eps (|V(\gam) \setminus (U_{\gam} \cap V_{\gam})| + \sum_{e \in E(\gam)} l_e)}}
	\end{align*}
for a constant $B$ that depends only on $C_p$.
\end{proof}

We obtain intersection term bounds.

\begin{lemma}
    For all $U, V \in \calI_{mid}$ where $w(U) > w(V)$ and all $\gam \in \Gam_{U, V}$, \[c(\gam)^2N(\gam)^2B(\gam)^2H_{Id_V}^{-\gam, \gam} \preceq H_{\gam}'\]
\end{lemma}

\begin{proof}
	By \cref{lem: plds_cond3_simplified}, we have
	\begin{align*}
	c(\gam)^2N(\gam)^2B(\gam)^2H_{Id_V}^{-\gam, \gam} &= c(\gam)^2N(\gam)^2B(\gam)^2 S(\gam)^2 \frac{|Aut(U)|}{|Aut(V)|} H'_{\gam}
	\end{align*}
	Using the same proof as in \cref{lem: plds_cond1}, we can see that $H'_{\gam} \succeq 0$. Therefore, it suffices to prove that $c(\gam)^2N(\gam)^2B(\gam)^2 S(\gam)^2 \frac{|Aut(U)|}{|Aut(V)|} \le 1$.
	Since $U, V \in \calI_{mid}$, $|Aut(U)| = |U|!,|Aut(V)| = |V|!$. Therefore, $\frac{|Aut(U)|}{|Aut(V)|} = \frac{|U|!}{|V|!} \le D_V^{|U_{\gam} \setminus V_{\gam}|}$. Also, $q = n^{O(1) \cdot \eps C_V}$. Let $B$ be the constant from \cref{lem: plds_charging2}. We can set $C_V$ sufficiently small so that, using \cref{lem: plds_charging2},

    {\footnotesize
	\begin{align*}
	c(\gam)^2N(\gam)^2B(\gam)^2S(\gam)^2 \frac{|Aut(U)|}{|Aut(V)|} &\le 100^2 (3D_V)^{2|U_{\gam}\setminus V_{\gam}| + 2|V_{\gam}\setminus U_{\gam}| + 4|E(\al)|}4^{|V(\gam) \setminus (U_{\gam} \cup V_{\gam})|}\\
	&\quad\cdot (3D_V)^{4|V(\gam)\setminus V_{\gam}| + 2|V(\gam)\setminus U_{\gam}|} (6D_V\sqrt[4]{2eq})^{2|V(\gam)\setminus U_{\gam}| + 2|V(\gam)\setminus V_{\gam}|}\\
	&\quad\cdot n^{w(V(\gam)\setminus U_{\gam})} S(\gam)^2 \cdot D_V^{|U_\gam \setminus V_{\gam}|} \\
	&\le n^{O(1) \cdot \eps C_V \cdot (|V(\gam) \setminus (U_{\gam} \cap V_{\gam})| + \sum_{e \in E(\gam)} l_e)} \cdot n^{w(V(\gam)\setminus U_{\gam})} S(\gam)^2\\
	&\le n^{O(1) \cdot \eps C_V \cdot (|V(\gam) \setminus (U_{\gam} \cap V_{\gam})| + \sum_{e \in E(\gam)} l_e)} \cdot \frac{1}{n^{B\eps (|V(\gam) \setminus (U_{\gam} \cap V_{\gam})| + \sum_{e \in E(\gam)} l_e)}}\\
	&\le 1
	\end{align*}}
\end{proof}

\subsection{\truncationboundstwo}

In this section, we will prove truncation error bounds.
We use the strategy and notation from \cite[Section 10]{potechin2020machinery}.
First, we will need a bound on $B_{norm}(\sig) B_{norm}(\sig') H_{Id_U}(\sig, \sig')$ that is obtained below.

\begin{lemma}\label{lem: plds_charging3}
	Suppose $k \le n^{1/2 - \eps}$. For all $U \in \calI_{mid}$ and $\sig, \sig' \in \calL_U$,
	\[B_{norm}(\sig) B_{norm}(\sig') H_{Id_U}(\sig, \sig') \le \frac{1}{n^{0.5\eps|V(\al)| + C_p\eps|E(\al)|}} \left(\frac{k}{n}\right)^{|U|}
	\]
\end{lemma}

\begin{proof}
	Let $\al = \sig \circ \sig'$. Observe that $|V(\sig)| + |V(\sig')| = |V(\al)| + |U|$. By choosing $C_V$ sufficiently small,
    {\footnotesize
	\begin{align*}
	B_{norm}(\sig) B_{norm}(\sig') H_{Id_U}(\sig, \sig') &= (6D_V\sqrt[4]{2eq})^{|V(\sig)\setminus U_{\sig}| + |V(\sig)\setminus V_{\sig}|} n^{\frac{w(V(\sig)) - w(U)}{2}}\\
	&\quad\cdot (6D_V\sqrt[4]{2eq})^{|V(\sig')\setminus U_{\sig'}| + |V(\sig')\setminus V_{\sig'}|} n^{\frac{w(V(\sig')) - w(U)}{2}}\\
	&\quad\cdot \frac{1}{|Aut(U)|} \left(\frac{k}{n}\right)^{|V(\al)|} (2(\frac{1}{2} + \frac{1}{2n^{C_p\eps}}) - 1)^{|E(\al)|}\\
	&\le n^{O(1) \cdot \eps C_V \cdot |V(\al)|} \sqrt{n}^{|V(\sig)| - |U|}\sqrt{n}^{|V(\sig')| - |U|} \left(\frac{k}{n}\right)^{|V(\al)|}\frac{1}{n^{C_p\eps|E(\al)|}}\\
	&\le \frac{1}{n^{0.5\eps|V(\al)| + C_p\eps|E(\al)|}} \left(\frac{k}{n}\right)^{|U|}
	\end{align*}}
\end{proof}

Now, we are ready to apply the strategy.

\begin{restatable}{lemma}{PLDSfive}\label{lem: plds_cond5}
	Whenever $\norm{M_{\al}} \le B_{norm}(\al)$ for all $\al \in \calM'$,
	\[
	\sum_{U \in \mathcal{I}_{mid}}{M^{fact}_{Id_U}{(H_{Id_U})}} \succeq \frac{1}{n^{K_1D_{sos}^2}} Id_{sym}
	\]
	for a constant $K_1 > 0$.
\end{restatable}

\begin{proof}
    For $V \in \calI_{mid}$, we have $\lda_V = \left(\frac{k}{n}\right)^{|V|}$. Now, we choose $w_V = \left(\frac{k}{n}\right)^{D_{sos} - |V|}$. Then, for all $\sig \in \calL_{V}$, we have $w_{V} \leq \frac{w_{U_{\sigma}}\lambda_{U_{\sigma}}}{|\mathcal{I}_{mid}|B_{norm}(\sigma)^2{c(\sigma)^2}{H_{Id_V}(\sigma,\sigma)}}$ which is easily verified using \cref{lem: plds_charging3}. The result now follows.
\end{proof}

\begin{restatable}{lemma}{PLDSsix}\label{lem: plds_cond6}
	\[\sum_{U\in \calI_{mid}} \sum_{\gam \in \Gam_{U, *}} \frac{d_{Id_{U}}(H_{Id_{U}}, H'_{\gam})}{|Aut(U)|c(\gam)} \le \frac{n^{K_2 D_{sos}}}{2^{D_V}}\]
	for a constant $K_2 > 0$.
\end{restatable}

\begin{proof}
	We have
	\begin{align*}
	&\sum_{U\in \calI_{mid}} \sum_{\gam \in \Gam_{U, *}} \frac{d_{Id_{U}}(H_{Id_{U}}, H'_{\gam})}{|Aut(U)|c(\gam)} \\
    &= \sum_{U\in \calI_{mid}} \sum_{\gam \in \Gam_{U, *}} \frac{1}{|Aut(U)|c(\gam)}\sum_{\sigma,\sigma' \in \mathcal{L}_{U_{\gamma}}: |V(\sigma)| \leq D_V, |V(\sigma')| \leq D_V,
		\atop |V(\sigma \circ \gamma)| > D_V \text{ or } |V(\sigma' \circ \gamma)| > D_V}{B_{norm}(\sigma)B_{norm}(\sigma')H_{Id_{U_{\gamma}}}(\sigma,\sigma')}
	\end{align*}
	The set of $\sig, \sig'$ that could appear in the above sum must necessarily be non-trivial and hence, $\sig, \sig' \in \calL_U'$. Then,
    {\footnotesize
	\begin{align*}
	&\sum_{U\in \calI_{mid}} \sum_{\gam \in \Gam_{U, *}} \frac{d_{Id_{U}}(H_{Id_{U}}, H'_{\gam})}{|Aut(U)|c(\gam)}\\
	&= \sum_{U\in \calI_{mid}} \sum_{\sigma,\sigma' \in \mathcal{L}'_{U}} {B_{norm}(\sigma)B_{norm}(\sigma')H_{Id_{U}}(\sigma,\sigma')}\sum_{\gam \in \Gam_{U, *}: |V(\sigma \circ \gamma)| > D_V \text{ or } |V(\sigma' \circ \gamma)| > D_V} \frac{1}{|Aut(U)|c(\gam)}
	\end{align*}}
	For $\sig \in \calL'_{U}$, define $m_{\sig} = D_V + 1 - |V(\sig)| \ge 1$. This is precisely set so that for all $\gam \in \Gam_{U, *}$, we have $|V(\sigma \circ \gamma)| > D_V$ if and only if $|V(\gam)| \ge |U| + m_{\sig}$. So, for $\sig, \sig' \in \calL'_U$,
	\begin{align*}
	\sum_{\gam \in \Gam_{U, *}: |V(\sigma \circ \gamma)| > D_V \text{ or } |V(\sigma' \circ \gamma)| > D_V} &\frac{1}{|Aut(U)|c(\gam)} \\
    &=
	\sum_{\gam \in \Gam_{U, *}: |V(\gam)| \ge |U| + \min(m_{\sig}, m_{\sig'})} \frac{1}{|Aut(U)|c(\gam)}\\
	&\le \frac{1}{2^{\min(m_{\sig}, m_{\sig'}) - 1}}
	\end{align*}
	Also, for $\sig, \sig' \in \calL_U'$, we have $|V(\sig \circ \sig')| + min(m_{\sig}, m_{\sig'}) - 1 \ge D_V$.
	Therefore,
	\begin{align*}
		\sum_{U\in \calI_{mid}} \sum_{\gam \in \Gam_{U, *}} &\frac{d_{Id_{U}}(H_{Id_{U}}, H'_{\gam})}{|Aut(U)|c(\gam)} \\
        &\le \sum_{U\in \calI_{mid}} \sum_{\sigma,\sigma' \in \mathcal{L}'_{U}} {B_{norm}(\sigma)B_{norm}(\sigma')H_{Id_{U}}(\sigma,\sigma')\frac{1}{2^{\min(m_{\sig}, m_{\sig'}) - 1}}}\\
		&\le \sum_{U\in \calI_{mid}} \sum_{\sigma,\sigma' \in \mathcal{L}'_{U}}\frac{n^{O(1) D_{sos}}}{n^{0.5\eps|V(\sig \circ \sig')|}2^{\min(m_{\sig}, m_{\sig'}) - 1}}
	\end{align*}
	where we used \cref{lem: plds_charging3}. Using $n^{0.5\eps |V(\sig \circ \sig')|} \ge n^{0.1\eps |V(\sig \circ \sig')|}2^{|V(\sig \circ \sig')|}$,
	\begin{align*}
		\sum_{U\in \calI_{mid}} \sum_{\gam \in \Gam_{U, *}} \frac{d_{Id_{U}}(H_{Id_{U}}, H'_{\gam})}{|Aut(U)|c(\gam)} &\le \sum_{U\in \calI_{mid}} \sum_{\sigma,\sigma' \in \mathcal{L}'_{U}}\frac{n^{O(1) D_{sos}}}{n^{0.1\eps|V(\sig \circ \sig')|} 2^{|V(\sig \circ \sig')|}2^{\min(m_{\sig}, m_{\sig'}) - 1}}\\
		&\le \sum_{U\in \calI_{mid}} \sum_{\sigma,\sigma' \in \mathcal{L}'_{U}}\frac{n^{O(1) D_{sos}}}{n^{0.1\eps|V(\sig \circ \sig')|} 2^{D_V}}\\
		&\le \sum_{U\in \calI_{mid}} \sum_{\sigma,\sigma' \in \mathcal{L}'_{U}}\frac{n^{O(1) D_{sos}}}{D_{sos}^{D_{sos}}n^{0.1\eps|V(\sig \circ \sig')|} 2^{D_V}}
	\end{align*}
	The final step will be to argue that $\sum_{U\in \calI_{mid}} \sum_{\sigma,\sigma' \in \mathcal{L}'_{U}}\frac{1}{D_{sos}^{D_{sos}}n^{0.1 \eps|V(\sig \circ \sig')|}} \le 1$ which will complete the proof. But this will follow if we set $C_V$ small enough.
\end{proof}

We conclude the following.

\begin{lemma}
    Whenever $\norm{M_{\alpha}} \le B_{norm}(\alpha)$ for all $\alpha \in \mathcal{M}'$,
    \[
    \sum_{U \in \mathcal{I}_{mid}}{M^{fact}_{Id_U}{(H_{Id_U})}} \succeq 6\left(\sum_{U \in \mathcal{I}_{mid}}{\sum_{\gamma \in \Gamma_{U,*}}{\frac{d_{Id_{U}}(H'_{\gamma},H_{Id_{U}})}{|Aut(U)|c(\gamma)}}}\right)Id_{sym}
    \]
\end{lemma}

\begin{proof}
	Choose $C_{sos}$ sufficiently small so that $\frac{1}{n^{K_1D_{sos}^2}} \ge 6\frac{n^{K_2D_{sos}}}{2^{D_V}}$ which can be satisfied by setting $C_{sos} < K_3 C_V$ for a sufficiently small constant $K_3 > 0$. Then, since $Id_{Sym} \succeq 0$, using \cref{lem: plds_cond5} and \cref{lem: plds_cond6},
	\begin{align*}
		\sum_{U \in \mathcal{I}_{mid}}{M^{fact}_{Id_U}{(H_{Id_U})}} &\succeq \frac{1}{n^{K_1D_{sos}^2}} Id_{sym}\\
		&\succeq 6\frac{n^{K_2D_{sos}}}{2^{D_V}} Id_{sym}\\
		&\succeq 6\left(\sum_{U \in \mathcal{I}_{mid}}{\sum_{\gamma \in \Gamma_{U,*}}{\frac{d_{Id_{U}}(H'_{\gamma},H_{Id_{U}})}{|Aut(U)|c(\gamma)}}}\right)Id_{sym}
	\end{align*}
\end{proof}

\section{Tensor PCA: Full verification}\label{sec: tpca_quant}

In this section, we will prove all the bounds required to prove \cref{thm: tpca_main}.

\TPCAmain*

We reuse the notation and qualitative bounds from \cref{sec: tpca_qual}.
Once we verify the conditions, this theorem will simply follow from the machinery.

\subsection{\middleshapeboundstwo}

\begin{lemma}\label{lem: tpca_charging}
	Suppose $\lda \le n^{\frac{k}{4} - \eps}$. For all $U \in \calI_{mid}$ and $\tau \in \calM_U$, suppose $deg^{\tau}(i)$ is even for all $i \in V(\tau) \setminus U_{\tau} \setminus V_{\tau}$, then
	\[\sqrt{n}^{|V(\tau)| - |U_{\tau}|}S(\tau) \le \frac{1}{n^{0.5\eps\sum_{e \in E(\tau)} l_e}}\]
\end{lemma}

\begin{proof}
	Firstly, we claim that $\sum_{e \in E(\tau)} kl_e \ge 2(|V(\tau)| - |U_{\tau}|)$. For any vertex $i \in V(\tau) \setminus U_{\tau} \setminus V_{\tau}$, $deg^{\tau}(i)$ is even and is not $0$, hence, $deg^{\tau}(i) \ge 2$. Any vertex $i \in U_{\tau} \setminus V_{\tau}$ cannot have $deg^{\tau}(i) = 0$ otherwise $U_{\tau} \setminus\{i\}$ is a vertex separator of strictly smaller weight than $U_{\tau}$, which is not possible, hence, $deg^{\tau}(i) \ge 1$. Therefore,
	\begin{align*}
	\sum_{e \in E(\tau)}kl_e = \sum_{i \in V(\tau)} deg^{\tau}(i)
	&\ge \sum_{i \in V(\tau) \setminus U_{\tau} \setminus V_{\tau}} deg^{\tau}(i) + \sum_{i \in U_{\tau} \setminus V_{\tau}} deg^{\tau}(i) + \sum_{i \in V_{\tau} \setminus U_{\tau}} deg^{\tau}(i)\\
	&\ge 2|V(\tau) \setminus U_{\tau} \setminus V_{\tau}| + |U_{\tau} \setminus V_{\tau}| + |V_{\tau} \setminus U_{\tau}|\\
	&= 2(|V(\tau)| - |U_{\tau}|)
	\end{align*}
	By choosing $C_{\Del}$ sufficiently small, we have
	\begin{align*}
	\sqrt{n}^{|V(\tau)| - |U_{\tau}|}S(\tau) &= \sqrt{n}^{|V(\tau)| - |U_{\tau}|} \Delta^{|V(\tau)| - |U_{\tau}|}\prod_{e \in E(\tau)}\left(\frac{\lda}{(\Del n)^{\frac{k}{2}}}\right)^{l_e}\\
	&\le \sqrt{n}^{|V(\tau)| - |U_{\tau}|}\Delta^{|V(\tau)| - |U_{\tau}|}\prod_{e \in E(\tau)}n^{(-\frac{k}{4} - 0.5\eps)l_e}\\
	&= \sqrt{n}^{|V(\tau)| - |U_{\tau}| - \frac{\sum_{e \in E(\tau)}kl_e}{2}}\Delta^{|V(\tau)| - |U_{\tau}|}\prod_{e \in E(\tau)}n^{-0.5\eps l_e}\\
	&= \Delta^{|V(\tau)| - |U_{\tau}|}\prod_{e \in E(\tau)}n^{-0.5 \eps l_e}\\
	&\le\frac{1}{n^{0.5\eps\sum_{e \in E(\tau)} l_e}}
	\end{align*}
\end{proof}

\begin{corollary}\label{cor: tpca_norm_decay}
	For all $U \in \calI_{mid}$ and $\tau \in \calM_U$, we have \[c(\tau)B_{norm}(\tau)S(\tau) \le 1\]
\end{corollary}

\begin{proof}
	Since $\tau$ is a proper middle shape, we have $w(I_{\tau}) = 0$ and $w(S_{\tau, min}) = w(U_{\tau})$. This implies
	$n^{\frac{w(V(\tau)) + w(I_{\tau}) - w(S_{\tau, min})}{2}} = \sqrt{n}^{|V(\tau)| - |U_{\tau}|}$.
	If $deg^{\tau}(i)$ is odd for any vertex $i \in V(\tau) \setminus U_{\tau} \setminus V_{\tau}$, then $S(\tau) = 0$ and the inequality is true. So, assume $deg^{\tau}(i)$ is even for all $i \in V(\tau) \setminus U_{\tau} \setminus V_{\tau}$.	As was observed in the proof of \cref{lem: tpca_charging}, every vertex $i \in V(\tau) \setminus U_{\tau}$ or $i \in V(\tau) \setminus V_{\tau}$ has $deg^{\tau}(i) \ge 1$ and hence, $|V(\tau)\setminus U_{\tau}| + |V(\tau)\setminus V_{\tau}| \le 4 \sum_{e \in E(\tau)} l_e$. Also, $|E(\tau)| \le \sum_{e \in E(\tau)} l_e$ and $q = n^{O(1) \cdot \eps (C_V + C_E)}$. We can set $C_V, C_E$ sufficiently small so that, using \cref{lem: tpca_charging},
	\begin{align*}
	c(\tau)B_{norm}(\tau)S(\tau)
	&= 100(3D_V)^{|U_{\tau}\setminus V_{\tau}| + |V_{\tau}\setminus U_{\tau}| + k|E(\tau)|}2^{|V(\tau)\setminus (U_{\tau}\cup V_{\tau})|}\\
	&\quad\cdot 2e(6qD_V)^{|V(\tau)\setminus U_{\tau}| + |V(\tau)\setminus V_{\tau}|}\prod_{e \in E(\tau)} (400D_V^2D_E^2q)^{l_e}\sqrt{n}^{|V(\tau)| - |U_{\tau}|}S(\tau)\\
	&\le n^{O(1) \cdot \eps(C_V + C_E) \cdot \sum_{e \in E(\tau)} l_e} \cdot \sqrt{n}^{|V(\tau)| - |U_{\tau}|}S(\tau)\\
	&\le n^{O(1) \cdot \eps(C_V + C_E) \cdot \sum_{e \in E(\tau)} l_e} \cdot \frac{1}{n^{0.5\eps\sum_{e \in E(\tau)} l_e}}\\
	&\le 1
	\end{align*}
\end{proof}

We can now show middle shape bounds.

\begin{lemma}\label{lem: tpca_cond2}
    For all $U \in \calI_{mid}$ and $\tau \in \calM_U$,
    \[
    \begin{bmatrix}
        \frac{1}{|Aut(U)|c(\tau)}H_{Id_U} & B_{norm}(\tau) H_{\tau}\\
        B_{norm}(\tau) H_{\tau}^T & \frac{1}{|Aut(U)|c(\tau)}H_{Id_U}
    \end{bmatrix}
    \succeq 0
    \]
\end{lemma}

\begin{proof}
	We have
	\begin{align*}
	&\begin{bmatrix}
	\frac{1}{|Aut(U)|c(\tau)}H_{Id_U} & B_{norm}(\tau)H_{\tau}\\
	B_{norm}(\tau)H_{\tau}^T & \frac{1}{|Aut(U)|c(\tau)}H_{Id_U}
	\end{bmatrix}\\
	&\qquad= \begin{bmatrix}
	\left(\frac{1}{|Aut(U)|c(\tau)} - \frac{S(\tau)B_{norm}(\tau)}{|Aut(U)|}\right)H_{Id_U} & 0\\
	0 & \left(\frac{1}{|Aut(U)|c(\tau)} - \frac{S(\tau)B_{norm}(\tau)}{|Aut(U)|}\right)H_{Id_U}
	\end{bmatrix}\\
	&\qquad \qquad+ B_{norm}(\tau)\begin{bmatrix}
	\frac{S(\tau)}{|Aut(U)|}H_{Id_U} & H_{\tau}\\
	H_{\tau}^T & \frac{S(\tau)}{|Aut(U)|}H_{Id_U}
	\end{bmatrix}
	\end{align*}
	By \cref{lem: tpca_cond2_simplified}, $\begin{bmatrix}
	\frac{S(\tau)}{|Aut(U)|}H_{Id_U} & H_{\tau}\\
	H_{\tau}^T & \frac{S(\tau)}{|Aut(U)|}H_{Id_U}
	\end{bmatrix}
	\succeq 0$, so the second term above is positive semidefinite. For the first term, by \cref{lem: tpca_cond1}, $H_{Id_U} \succeq 0$ and by \cref{cor: tpca_norm_decay}, $\frac{1}{|Aut(U)|c(\tau)} - \frac{S(\tau)B_{norm}(\tau)}{|Aut(U)|} \ge 0$, which proves that the first term is also positive semidefinite.
\end{proof}

\subsection{\intersectionboundstwo}

\begin{lemma}\label{lem: tpca_charging2}
	Suppose $\lda \le n^{\frac{k}{4} - \eps}$. For all $U, V \in \calI_{mid}$ where $w(U) > w(V)$ and for all $\gam \in \Gam_{U, V}$,
	\[n^{w(V(\gam)\setminus U_{\gam})} S(\gam)^2 \le \frac{1}{n^{B\eps (|V(\gam) \setminus (U_{\gam} \cap V_{\gam})| + \sum_{e \in E(\gam)} l_e)}}\]
	for some constant $B$ that depends only on $C_{\Del}$. In particular, it is independent of $C_V$ and $C_E$.
\end{lemma}

\begin{proof}
	Suppose there is a vertex $i \in V(\gam) \setminus U_{\gam} \setminus V_{\gam}$ such that $deg^{\gam}(i)$ is odd, then $S(\gam) = 0$ and the inequality is true. So, assume $deg^{\gam}(i)$ is even for all vertices $i \in V(\gam) \setminus U_{\gam} \setminus V_{\gam}$.
	We first claim that $k\sum_{e \in E(\gam)} l_e \ge 2|V(\gam) \setminus U_{\gam}|$. Since $\gam$ is a left shape, all vertices $i$ in $V(\gam) \setminus U_{\gam}$ have $deg^{\gam}(i) \ge 1$. In particular, all vertices $i \in V_{\gam} \setminus U_{\gam}$ have $deg^{\gam}(i) \ge 1$.
	Moreover, if $i \in V(\gam) \setminus U_{\gam} \setminus V_{\gam}$, since $deg^{\gam}(i)$ is even, we must have $deg^{\gam}(i) \ge 2$.

	Let $S'$ be the set of vertices $i \in U_{\gam} \setminus V_{\gam}$ that have $deg^{\gam}(i) \ge 1$. Then, note that $|S'| + |U_{\gam} \cap V_{\gam}| \ge |V_{\gam}| \Longrightarrow |S'| \ge |V_{\gam} \setminus U_{\gam}|$ since otherwise $S' \cup (U_{\gam} \cap V_{\gam})$ will be a vertex separator of $\gam$ of weight strictly less than $V_{\gam}$, which is not possible. Then,
	\begin{align*}
	\sum_{e \in E(\gam)}kl_e &= \sum_{i \in V(\gam)} deg^{\gam}(i)\\
	&\ge \sum_{i \in V(\gam) \setminus U_{\gam} \setminus V_{\gam}} deg^{\gam}(i) + \sum_{i \in U_{\gam} \setminus V_{\gam}} deg^{\gam}(i) + \sum_{i \in V_{\gam} \setminus U_{\gam}} deg^{\gam}(i)\\
	&\ge 2|V(\gam) \setminus U_{\gam} \setminus V_{\gam}| + |S'| + |V_{\gam} \setminus U_{\gam}|\\
	&\ge 2|V(\gam) \setminus U_{\gam} \setminus V_{\gam}| + 2|V_{\gam} \setminus U_{\gam}|\\
	&= 2|V(\gam) \setminus U_{\gam}|
	\end{align*}

	Finally, note that $2|V(\gam)| - |U_{\gam}| - |V_{\gam}| = |U_{\gam} \setminus V_{\gam}| + |V_{\gam} \setminus U_{\gam}| + 2|V(\gam) \setminus U_{\gam} \setminus V_{\gam}| \ge |V(\gam) \setminus (U_{\gam} \cap V_{\gam})|$. By choosing $C_{\Del}$ sufficiently small, we have
	\begin{align*}
	n^{w(V(\gam)\setminus U_{\gam})} S(\gam)^2 &= n^{|V(\gam)\setminus U_{\gam})|} \Delta^{2|V(\gam)| - |U_{\gam}| - |V_{\gam}|} \prod_{e \in E(\gam)} \left(\frac{\lda^2}{(\Del n)^k}\right)^{l_e}\\
	&\le n^{|V(\gam)\setminus U_{\gam})|} \Delta^{2|V(\gam)| - |U_{\gam}| - |V_{\gam}|} \prod_{e \in E(\gam)} n^{-(\frac{k}{2} + \eps)l_e}\\
	&\le \Delta^{2|V(\gam)| - |U_{\gam}| - |V_{\gam}|} \prod_{e \in E(\gam)} n^{-\eps l_e}\\
	&\le \frac{1}{n^{B\eps (|V(\gam) \setminus (U_{\gam} \cap V_{\gam})| + \sum_{e \in E(\gam)} l_e)}}
	\end{align*}
for a constant $B$ that depends only on $C_{\Del}$.
\end{proof}

\begin{remk}
	In the above bounds, note that there is a decay of $n^{B\eps}$ for each vertex in $V(\gam) \setminus (U_{\gam} \cap V_{\gam})$.	One of the main technical reasons for introducing the slack parameter $C_{\Del}$ in the planted distribution was to introduce this decay, which is needed in the current machinery.
\end{remk}

We can now obtain the intersection term bounds.

\begin{lemma}\label{lem: tpca_cond3}
    For all $U, V \in \calI_{mid}$ where $w(U) > w(V)$ and all $\gam \in \Gam_{U, V}$, \[c(\gam)^2N(\gam)^2B(\gam)^2H_{Id_V}^{-\gam, \gam} \preceq H_{\gam}'\]
\end{lemma}

\begin{proof}
	By \cref{lem: tpca_cond3_simplified}, we have
	\begin{align*}
	c(\gam)^2N(\gam)^2B(\gam)^2H_{Id_V}^{-\gam, \gam} &\preceq c(\gam)^2N(\gam)^2B(\gam)^2 S(\gam)^2 \frac{|Aut(U)|}{|Aut(V)|} H'_{\gam}
	\end{align*}
	Using the same proof as in \cref{lem: tpca_cond1}, we can see that $H'_{\gam} \succeq 0$. Therefore, it suffices to prove that $c(\gam)^2N(\gam)^2B(\gam)^2 S(\gam)^2 \frac{|Aut(U)|}{|Aut(V)|} \le 1$.
	Since $U, V \in \calI_{mid}$, $|Aut(U)| = |U|!,|Aut(V)| = |V|!$. Therefore, $\frac{|Aut(U)|}{|Aut(V)|} = \frac{|U|!}{|V|!} \le D_V^{|U_{\gam} \setminus V_{\gam}|}$. Also, $|E(\gam)| \le \sum_{e \in E(\gam)} l_e$ and $q = n^{O(1) \cdot \eps (C_V + C_E)}$. Let $B$ be the constant from \cref{lem: tpca_charging2}. We can set $C_V, C_E$ sufficiently small so that, using \cref{lem: tpca_charging2},
	\begin{align*}
	c(\gam)^2&N(\gam)^2B(\gam)^2S(\gam)^2 \frac{|Aut(U)|}{|Aut(V)|} \\
    &\le 100^2 (3D_V)^{2|U_{\gam}\setminus V_{\gam}| + 2|V_{\gam}\setminus U_{\gam}| + 2k|E(\al)|}4^{|V(\gam) \setminus (U_{\gam} \cup V_{\gam})|}\\
	&\quad\cdot (3D_V)^{4|V(\gam)\setminus V_{\gam}| + 2|V(\gam)\setminus U_{\gam}|} (6qD_V)^{2|V(\gam)\setminus U_{\gam}| + 2|V(\gam)\setminus V_{\gam}|} \prod_{e \in E(\gam)} (400D_V^2D_E^2q)^{2l_e}\\
	&\quad\cdot  n^{w(V(\gam)\setminus U_{\gam})} S(\gam)^2 \cdot D_V^{|U_\gam \setminus V_{\gam}|} \\
	&\le n^{O(1) \cdot \eps(C_V + C_E) \cdot (|V(\gam) \setminus (U_{\gam} \cap V_{\gam})| + \sum_{e \in E(\gam)} l_e)} \cdot n^{w(V(\gam)\setminus U_{\gam})} S(\gam)^2\\
	&\le n^{O(1) \cdot \eps(C_V + C_E) \cdot (|V(\gam) \setminus (U_{\gam} \cap V_{\gam})| + \sum_{e \in E(\gam)} l_e)} \cdot \frac{1}{n^{B\eps (|V(\gam) \setminus (U_{\gam} \cap V_{\gam})| + \sum_{e \in E(\gam)} l_e)}}\\
	&\le 1
	\end{align*}
\end{proof}

\subsection{\truncationboundstwo}

In this section, we will obtain the truncation error bounds using the strategy sketched in section 10 of \cite{potechin2020machinery}. We also reuse the notation. First, we need the following bound on $B_{norm}(\sig) B_{norm}(\sig') H_{Id_U}(\sig, \sig')$.

\begin{lemma}\label{lem: tpca_charging3}
	Suppose $\lda = n^{\frac{k}{4} - \eps}$. For all $U \in \calI_{mid}$ and $\sig, \sig' \in \calL_U$,
	\[B_{norm}(\sig) B_{norm}(\sig') H_{Id_U}(\sig, \sig') \le \frac{1}{n^{0.5\eps C_{\Del}|V(\sig \circ \sig')|}\Del^{D_{sos}}n^{|U|}}\]
\end{lemma}

\begin{proof}
	Suppose there is a vertex $i \in V(\sig) \setminus V_{\sig}$ such that $deg^{\sig}(i) + deg^{U_{\sig}}(i)$ is odd, then $H_{Id_U}(\sig, \sig') = 0$ and the inequality is true. So, assume that $deg^{\sig}(i) + deg^{U_{\sig}}(i)$ is even for all $i \in V(\sig) \setminus V_{\sig}$. Similarly, assume that $deg^{\sig'}(i) + deg^{U_{\sig'}}(i)$ is even for all $i \in V(\sig') \setminus V_{\sig'}$. Also, if $\rho_{\sig} \neq \rho_{\sig'}$, we will have $H_{Id_U}(\sig, \sig') = 0$ and we'd be done. So, assume $\rho_{\sig} = \rho_{\sig'}$.

	Let $\al = \sig \circ \sig'$. We will first prove that $\sum_{e \in E(\al)} kl_e + 2deg(\al) \ge 2|V(\al)| + 2|U|$. Firstly, note that all vertices $i \in V(\al) \setminus (U_{\al} \cup V_{\al})$ have $deg^{\al}(i)$ to be even and nonzero, and hence at least $2$. Moreover, in both the sets $U_{\al} \setminus (U_{\al} \cap V_{\al})$ and $V_{\al} \setminus (U_{\al} \cap V_{\al})$, there are at least $|U| - |U_{\al} \cap V_{\al}|$ vertices of degree at least $1$, because $U$ is a minimum vertex separator. Also, note that $deg(\al) \ge |U_{\al}| + |V_{\al}|$. This implies that
	\begin{align*}
	\sum_{e \in E(\al)} kl_e &+ 2deg(\al)\\
     &\ge 2 |V(\al) \setminus (U_{\al} \cup V_{\al})| + 2(|U| - |U_{\al} \cap V_{\al}|) + 2(|U_{\al}| + |V_{\al}|)\\
	&= 2 (|V(\al)| - |U_{\al} \cup V_{\al}|) + 2(|U| - |U_{\al} \cap V_{\al}|) + 2(|U_{\al} \cup V_{\al}| + |U_{\al} \cap V_{\al}|)\\
	&= 2|V(\al)| + 2|U|
	\end{align*}
	where we used the fact that $U_{\al} \cap V_{\al} \subseteq U$. Finally, by choosing $C_V, C_E$ sufficiently small,
	\begin{align*}
	&B_{norm}(\sig) B_{norm}(\sig') H_{Id_U}(\sig, \sig') \\
    &= 2e(6qD_V)^{|V(\sig)\setminus U_{\sig}| + |V(\sig)\setminus V_{\sig}|}\prod_{e \in E(\sig)} (400D_V^2D_E^2q)^{l_e} n^{\frac{w(V(\sig)) - w(U)}{2}}\\
	&\quad\cdot 2e(6qD_V)^{|V(\sig')\setminus U_{\sig'}| + |V(\sig')\setminus V_{\sig'}|}\prod_{e \in E(\sig')} (400D_V^2D_E^2q)^{l_e} n^{\frac{w(V(\sig')) - w(U)}{2}}\\
	&\quad\cdot \frac{1}{|Aut(U)|} \Delta^{|V(\al)|} \left(\frac{1}{\sqrt{\Delta n}}\right)^{deg(\al)} \prod_{e \in E(\al)} \left(\frac{\lda}{(\Del n)^{\frac{k}{2}}}\right)^{l_e}\\
	&\le n^{O(1) \cdot \eps (C_V + C_E) \cdot (|V(\al)| + \sum_{e \in E(\al)} l_e)} \Delta^{|V(\al)|}\left(\frac{1}{\sqrt{\Del}}\right)^{deg(\al)}\\
	&\quad\cdot \sqrt{n}^{|V(\al)| - |U|} \left(\frac{1}{\sqrt{n}}\right)^{deg(\al)}\prod_{e \in E(\al)}n^{(-\frac{k}{4} - 0.5\eps)l_e}\\
	&\le \frac{n^{O(1) \cdot \eps (C_V + C_E) \cdot (|V(\al)| + \sum_{e \in E(\al)} l_e)}}{n^{\eps C_{\Delta}|V(\al)|}n^{0.5\eps\sum_{e \in E(\al)} l_e}} \cdot \frac{1}{\Del^{D_{sos}}n^{|U|}}\sqrt{n}^{|V(\al)| + |U| - deg(\al) - \frac{1}{2}\sum_{e \in E(\al)} kl_e}\\
	&\le \frac{1}{n^{0.5\eps C_{\Del}|V(\al)|}\Del^{D_{sos}}n^{|U|}}
	\end{align*}
where we used the facts $\Del \le 1, deg(\al) \le 2D_{sos}$.
\end{proof}

We now apply the strategy by showing the following bounds.

\begin{restatable}{lemma}{TPCAfive}\label{lem: tpca_cond5}
	Whenever $\norm{M_{\al}} \le B_{norm}(\al)$ for all $\al \in \calM'$,
	\[
	\sum_{U \in \mathcal{I}_{mid}}{M^{fact}_{Id_U}{(H_{Id_U})}} \succeq \frac{\Del^{2D_{sos}^2}}{n^{D_{sos}}} Id_{sym}
	\]
\end{restatable}

\begin{proof}
    For $V \in \calI_{mid}$, $\lda_V = \frac{1}{n^{|V|}}$. We then choose $w_V = \left(\frac{1}{n}\right)^{D_{sos} - |V|}$. For all left shapes $\sig \in \calL_V$, it's easy to verify $w_{V} \leq \frac{w_{U_{\sigma}}\lambda_{U_{\sigma}}}{|\mathcal{I}_{mid}|B_{norm}(\sigma)^2{c(\sigma)^2}{H_{Id_V}(\sigma,\sigma)}}$ using \cref{lem: tpca_charging3}. This completes the proof.
\end{proof}

\begin{restatable}{lemma}{TPCAsix}\label{lem: tpca_cond6}
	\[\sum_{U\in \calI_{mid}} \sum_{\gam \in \Gam_{U, *}} \frac{d_{Id_{U}}(H_{Id_{U}}, H'_{\gam})}{|Aut(U)|c(\gam)} \le \frac{1}{\Delta^{2D_{sos}}2^{D_V}}\]
\end{restatable}

\begin{proof}
    We use the same argument and notation as in \cref{lem: plds_cond6}. When we plug in the bounds, we get
	\begin{align*}
	\sum_{U\in \calI_{mid}} \sum_{\gam \in \Gam_{U, *}} &\frac{d_{Id_{U}}(H_{Id_{U}}, H'_{\gam})}{|Aut(U)|c(\gam)} \\
    &\le \sum_{U\in \calI_{mid}} \sum_{\sigma,\sigma' \in \mathcal{L}'_{U}} {B_{norm}(\sigma)B_{norm}(\sigma')H_{Id_{U}}(\sigma,\sigma')\frac{1}{2^{\min(m_{\sig}, m_{\sig'}) - 1}}}\\
	&\le \sum_{U\in \calI_{mid}} \sum_{\sigma,\sigma' \in \mathcal{L}'_{U}}\frac{1}{n^{0.5\eps C_{\Del}|V(\sig \circ \sig')|}\Del^{D_{sos}}n^{|U|}2^{\min(m_{\sig}, m_{\sig'}) - 1}}\\
	&\le \sum_{U\in \calI_{mid}} \sum_{\sigma,\sigma' \in \mathcal{L}'_{U}}\frac{1}{n^{0.5\eps C_{\Del}|V(\sig \circ \sig')|}\Del^{D_{sos}}2^{\min(m_{\sig}, m_{\sig'}) - 1}}
	\end{align*}
	where we used \cref{lem: tpca_charging3}. Using $n^{0.5 C_{\Del} |V(\sig \circ \sig')|} \ge n^{0.1\eps C_{\Del} |V(\sig \circ \sig')|}2^{|V(\sig \circ \sig')|}$,
	\begin{align*}
	\sum_{U\in \calI_{mid}} \sum_{\gam \in \Gam_{U, *}} &\frac{d_{Id_{U}}(H_{Id_{U}}, H'_{\gam})}{|Aut(U)|c(\gam)}\\
    &\le \sum_{U\in \calI_{mid}} \sum_{\sigma,\sigma' \in \mathcal{L}'_{U}}\frac{1}{n^{0.1 \eps  C_{\Delta}|V(\sig \circ \sig')|} \Delta^{D_{sos}} 2^{|V(\sig \circ \sig')|}2^{\min(m_{\sig}, m_{\sig'}) - 1}}\\
	&\le \sum_{U\in \calI_{mid}} \sum_{\sigma,\sigma' \in \mathcal{L}'_{U}}\frac{1}{n^{0.1 \eps C_{\Delta}|V(\sig \circ \sig')|}\Delta^{D_{sos}} 2^{D_V}}\\
	&\le \sum_{U\in \calI_{mid}} \sum_{\sigma,\sigma' \in \mathcal{L}'_{U}}\frac{1}{D_{sos}^{D_{sos}}n^{0.1 \eps C_{\Delta}|V(\sig \circ \sig')|}\Delta^{2D_{sos}} 2^{D_V}}
	\end{align*}
	where we set $C_{sos}$ small enough so that $D_{sos} = n^{\eps C_{sos}} \le n^{c\eps C_{\Del}} = \frac{1}{\Del}$. The final step will be to argue that $\sum_{U\in \calI_{mid}} \sum_{\sigma,\sigma' \in \mathcal{L}'_{U}}\frac{1}{D_{sos}^{D_{sos}}n^{0.1 C_{\Delta}\eps|V(\sig \circ \sig')|}} \le 1$ which will complete the proof. But this will follow if we set $C_V, C_E$ small enough.
\end{proof}

We can finally complete the analysis of the truncation error.

\begin{lemma}\label{lem: tpca_cond4}
    Whenever $\norm{M_{\alpha}} \le B_{norm}(\alpha)$ for all $\alpha \in \mathcal{M}'$,
    \[
    \sum_{U \in \mathcal{I}_{mid}}{M^{fact}_{Id_U}{(H_{Id_U})}} \succeq 6\left(\sum_{U \in \mathcal{I}_{mid}}{\sum_{\gamma \in \Gamma_{U,*}}{\frac{d_{Id_{U}}(H'_{\gamma},H_{Id_{U}})}{|Aut(U)|c(\gamma)}}}\right)Id_{sym}
    \]
\end{lemma}

\begin{proof}
	Choose $C_{sos}$ sufficiently small so that $\frac{\Del^{2D_{sos}^2}}{n^{D_{sos}}} \ge \frac{6}{\Delta^{2D_{sos}}2^{D_V}}$ which is satisfied by setting $C_{sos} < 0.5 C_V$. Then, since $Id_{Sym} \succeq 0$, using \cref{lem: tpca_cond5} and \cref{lem: tpca_cond6},
	\begin{align*}
	\sum_{U \in \mathcal{I}_{mid}}{M^{fact}_{Id_U}{(H_{Id_U})}} &\succeq \frac{\Del^{2D_{sos}^2}}{n^{D_{sos}}} Id_{sym}\\
	&\succeq \frac{6}{\Delta^{2D_{sos}}2^{D_V}} Id_{sym}\\
	&\succeq 6\left(\sum_{U \in \mathcal{I}_{mid}}{\sum_{\gamma \in \Gamma_{U,*}}{\frac{d_{Id_{U}}(H'_{\gamma},H_{Id_{U}})}{|Aut(U)|c(\gamma)}}}\right)Id_{sym}
	\end{align*}
\end{proof}

\section{Sparse PCA: Full verification}\label{sec: spca_quant}

In this section, we will full prove \cref{thm: spca_main}.

\SPCAmain*

We already showed the relevant qualitative bounds in \cref{sec: spca_qual}. We use the bounds and also the notation from that section.
We will apply the machinery.

\begin{definition}
    Define $n = \max(d, m)$.
\end{definition}

The above definition conforms with the notation used in the machinery. So, we can use the bounds as stated there.
Once we verify the conditions, the theorem will immediately follow from the machinery.

\subsection{\middleshapeboundstwo}

\begin{lemma}\label{lem: spca_charging}
	Suppose $0 < A < \frac{1}{4}$ is a constant such that $\frac{\sqrt{\lda}}{\sqrt{k}} \le d^{-A\eps}$ and $\frac{1}{\sqrt{k}} \le d^{-2A}$. For all $m$ such that $m \le \frac{d^{1 - \eps}}{\lda^2}, m \le \frac{k^{2 - \eps}}{\lda^2}$, for all $U \in \calI_{mid}$ and $\tau \in \calM_U$, suppose $deg^{\tau}(i)$ is even for all $i \in V(\tau) \setminus U_{\tau} \setminus V_{\tau}$, then
	\[\sqrt{d}^{|\tau|_1 - |U_{\tau}|_1}\sqrt{m}^{|\tau|_2 - |U_{\tau}|_2}S(\tau) \le \prod_{j \in V_2(\tau) \setminus U_{\tau} \setminus V_{\tau}} (deg^{\tau}(j) - 1)!!\cdot \frac{1}{d^{A\eps\sum_{e \in E(\tau)} l_e}}\]
\end{lemma}

\begin{proof}
	Let $r_1 = |\tau|_1 - |U_{\tau}|_1, r_2 = |\tau|_2 - |U_{\tau}|_2$. Since $\Delta \le 1$, it suffices to prove
	\[E := \sqrt{d}^{r_1}\sqrt{m}^{r_2}\left(\frac{k}{d}\right)^{r_1} \left(\frac{\sqrt{\lda}}{\sqrt{k}}\right)^{\sum_{e \in E(\tau)} l_e} \le  \frac{1}{d^{A\eps\sum_{e \in E(\tau)} l_e}}\]

	We will need the following claim.
	\begin{claim}
		$\sum_{e \in E(\tau)} l_e \ge 2 \max(r_1, r_2)$.
	\end{claim}

	\begin{proof}
		We will first prove $\sum_{e \in E(\tau)} l_e \ge 2 r_1$. For any vertex $i \in V_1(\tau) \setminus U_{\tau} \setminus V_{\tau}$, $deg^{\tau}(i)$ is even and is not $0$, hence, $deg^{\tau}(i) \ge 2$. Any vertex $i \in U_{\tau} \setminus V_{\tau}$ cannot have $deg^{\tau}(i) = 0$ otherwise $U_{\tau} \setminus\{i\}$ is a vertex separator of strictly smaller weight than $U_{\tau}$, which is not possible, hence, $deg^{\tau}(i) \ge 1$. Similarly, for $i \in  V_{\tau} \setminus U_{\tau}$, $deg^{\tau}(i) \ge 1$. Also, since $H_{\tau}$ is bipartite, we have $\sum_{i \in V_1(\tau)} deg^{\tau}(i) = \sum_{j \in V_2(\tau)} deg^{\tau}(j)= \sum_{e \in E(\tau)} l_e$. Consider

		\begin{align*}
		\sum_{e \in E(\tau)} l_e &= \sum_{i \in V_1(\tau)} deg^{\tau}(i)\\
		&\ge \sum_{i \in V_1(\tau) \setminus U_{\tau} \setminus V_{\tau}} deg^{\tau}(i) + \sum_{i \in (U_{\tau})_1 \setminus V_{\tau}} deg^{\tau}(i) + \sum_{i \in (V_{\tau})_1 \setminus U_{\tau}} deg^{\tau}(i)\\
		&\ge 2|V_1(\tau) \setminus U_{\tau} \setminus V_{\tau}| + |(U_{\tau})_1 \setminus V_{\tau}| + |(V_{\tau})_1 \setminus U_{\tau}|\\
		&= 2r_1
		\end{align*}
		We can similarly prove $\sum_{e \in E(\tau)} l_e \ge 2 r_2$
	\end{proof}

	To illustrate the main idea, we will start by proving the weaker bound $E \le 1$. Observe that our assumptions imply $m \le \frac{d}{\lda^2}, m \le \frac{k^2}{\lda^2}$ and also, using the fact $\frac{\sqrt{\lda}}{\sqrt{k}} \le d^{-A\eps} \le 1$, we have $E \le \sqrt{d}^{r_1}\sqrt{m}^{r_2}\left(\frac{k}{d}\right)^{r_1} \left(\frac{\sqrt{\lda}}{\sqrt{k}}\right)^{2\max(r_1, r_2)}$.

	\begin{claim}\label{claim: spca_decay}
		For integers $r_1, r_2 \ge 0$, if $m \le \frac{d}{\lda^2}$ and $m \le \frac{k^2}{\lda^2}$, then,
		\[\sqrt{d}^{r_1}\sqrt{m}^{r_2}\left(\frac{k}{d}\right)^{r_1} \left(\frac{\sqrt{\lda}}{\sqrt{k}}\right)^{2\max(r_1, r_2)} \le  1\]
	\end{claim}

	\begin{proof}
	We will consider the cases $r_1 \ge r_2$ and $r_1 < r_2$ separately. If $r_1 \ge r_2$, we have
	\begin{align*}
		\sqrt{d}^{r_1}\sqrt{m}^{r_2}\left(\frac{k}{d}\right)^{r_1} \left(\frac{\sqrt{\lda}}{\sqrt{k}}\right)^{2r_1} &\le \sqrt{d}^{r_1}\left(\frac{\sqrt{d}}{\lda}\right)^{r_2}\left(\frac{k}{d}\right)^{r_1} \left(\frac{\sqrt{\lda}}{\sqrt{k}}\right)^{2r_1}\\
		&= \left(\frac{\lda}{\sqrt{d}}\right)^{r_1 - r_2}\\
		&\le \left(\frac{1}{\sqrt{m}}\right)^{r_1 - r_2}\\
		&\le 1
	\end{align*}
	And if $r_1 < r_2$, we have
	\begin{align*}
	\sqrt{d}^{r_1}\sqrt{m}^{r_2}\left(\frac{k}{d}\right)^{r_1} \left(\frac{\sqrt{\lda}}{\sqrt{k}}\right)^{2r_2} &= \sqrt{d}^{r_1}\sqrt{m}^{r_2 - r_1}\sqrt{m}^{r_1}\left(\frac{k}{d}\right)^{r_1} \left(\frac{\sqrt{\lda}}{\sqrt{k}}\right)^{2r_2}\\
	&\le \sqrt{d}^{r_1}\left(\frac{k}{\lda}\right)^{r_2 - r_1}\left(\frac{\sqrt{d}}{\lda}\right)^{r_1}\left(\frac{k}{d}\right)^{r_1} \left(\frac{\sqrt{\lda}}{\sqrt{k}}\right)^{2r_2}\\
		&= 1
	\end{align*}
	\end{proof}

	For the desired bounds, we mimic this argument while carefully keeping track of factors of $d^{\eps}$.

	\begin{claim}\label{claim: spca_decay2}
		For integers $r_1, r_2 \ge 0$ and an integer $r \ge 2\max(r_1, r_2)$, if $m \le \frac{d^{1 - \eps}}{\lda^2}$ and $m \le \frac{k^{2 - \eps}}{\lda^2}$, then,
		\[\sqrt{d}^{r_1}\sqrt{m}^{r_2}\left(\frac{k}{d}\right)^{r_1} \left(\frac{\sqrt{\lda}}{\sqrt{k}}\right)^r \le  \left(\frac{1}{d^{A\eps}}\right)^r\]
	\end{claim}
	\begin{proof}
	If $r_1 \ge r_2$,
	\begin{align*}
		E &= \sqrt{d}^{r_1}\sqrt{m}^{r_2} \left(\frac{k}{d}\right)^{r_1}\left(\frac{\sqrt{\lda}}{\sqrt{k}}\right)^{2r_1}\left(\frac{\sqrt{\lda}}{\sqrt{k}}\right)^{r - 2r_1}\\
		&\le \sqrt{d}^{r_1}\left(\frac{\sqrt{d}^{1 - \eps}}{\lda}\right)^{r_2} \left(\frac{k}{d}\right)^{r_1}\left(\frac{\sqrt{\lda}}{\sqrt{k}}\right)^{2r_1}\left(\frac{\sqrt{\lda}}{\sqrt{k}}\right)^{r - 2r_1}\\
		& = \left(\frac{\lda}{\sqrt{d}^{1 - \eps}}\right)^{r_1 - r_2} \left(\frac{1}{\sqrt{d}}\right)^{\eps r_1}\left(\frac{\sqrt{\lda}}{\sqrt{k}}\right)^{r - 2r_1}\\
		& \le \left(\frac{1}{\sqrt{m}}\right)^{r_1 - r_2} \left(\frac{1}{\sqrt{d}}\right)^{\eps r_1}\left(\frac{1}{d^{A\eps}}\right)^{r - 2r_1}\\
		&\le \left(\frac{1}{d^{2A}}\right)^{\eps r_1}\left(\frac{1}{d^{A\eps}}\right)^{r - 2r_1}\\
		&= \left(\frac{1}{d^{A\eps}}\right)^r
	\end{align*}
	And if $r_1 < r_2$,
	\begin{align*}
	E &= \sqrt{d}^{r_1}\sqrt{m}^{r_2 - r_1} \sqrt{m}^{r_1} \left(\frac{k}{d}\right)^{r_1}\left(\frac{\sqrt{\lda}}{\sqrt{k}}\right)^{2r_2}\left(\frac{\sqrt{\lda}}{\sqrt{k}}\right)^{r - 2r_2}\\
	&\le \sqrt{d}^{r_1}\left(\frac{\sqrt{k}^{2 - \eps}}{\lda}\right)^{r_2 - r_1}\left(\frac{\sqrt{d}^{1 - \eps}}{\lda}\right)^{r_1} \left(\frac{k}{d}\right)^{r_1}\left(\frac{\sqrt{\lda}}{\sqrt{k}}\right)^{2r_2}\left(\frac{\sqrt{\lda}}{\sqrt{k}}\right)^{r - 2r_2}\\
	&= \left(\frac{\sqrt{k}}{\sqrt{d}}\right)^{\eps r_1}\left(\frac{1}{\sqrt{k}}\right)^{\eps r_2}\left(\frac{\sqrt{\lda}}{\sqrt{k}}\right)^{r - 2r_2}\\
	&\le \left(\frac{1}{\sqrt{k}}\right)^{\eps r_2}\left(\frac{\sqrt{\lda}}{\sqrt{k}}\right)^{r - 2r_2}\\
	&\le \left(\frac{1}{d^{2A}}\right)^{\eps r_2}\left(\frac{1}{d^{A\eps}}\right)^{r - 2r_2}\\
	&\le \left(\frac{1}{d^{A\eps}}\right)^{\sum_{e \in E(\tau)} l_e}
	\end{align*}
\end{proof}
The result follows by setting $r = \sum_{e \in E(\tau)} l_e$ in the above claim.
\end{proof}

\begin{corollary}\label{cor: spca_norm_decay}
	For all $U \in \calI_{mid}$ and $\tau \in \calM_U$, we have
	\[c(\tau) B_{norm}(\tau)S(\tau)R(\tau) \le 1\]
\end{corollary}

\begin{proof}
	First, note that if $deg^{\tau}(i)$ is odd for any vertex $i \in V(\tau) \setminus U_{\tau} \setminus V_{\tau}$, then $S(\tau) = 0$ and the inequality is true. So, assume that $deg^{\tau}(i)$ is even for all $i \in V(\tau) \setminus U_{\tau} \setminus V_{\tau}$.
	Since $\tau$ is a proper middle shape, we have $w(I_{\tau}) = 0$ and $w(S_{\tau, min}) = w(U_{\tau})$. This implies
	$n^{\frac{w(V(\tau)) + w(I_{\tau}) - w(S_{\tau, min})}{2}} = \sqrt{d}^{|\tau|_1 - |U_{\tau}|_1}\sqrt{m}^{|\tau|_2 - |U_{\tau}|_2}$.
	As was observed in the proof of \cref{lem: spca_charging}, every vertex $i \in V(\tau) \setminus U_{\tau}$ or $i \in V(\tau) \setminus V_{\tau}$ has $deg^{\tau}(i) \ge 1$ and hence, $|V(\tau)\setminus U_{\tau}| + |V(\tau)\setminus V_{\tau}| \le 4 \sum_{e \in E(\tau)} l_e$. Also, $q = d^{O(1)\cdot \eps(C_V + C_E)}$. We can set $C_V, C_E$ sufficiently small so that
    {\footnotesize
	\begin{align*}
	c(\tau)B_{norm}(\tau)S(\tau)R(\tau)	&=100(6D_V)^{|U_{\tau}\setminus V_{\tau}| + |V_{\tau}\setminus U_{\tau}| + 2|E(\tau)|}4^{|V(\tau)\setminus (U_{\tau}\cup V_{\tau})|}\\
	&\cdot 2e(6qD_V)^{|V(\tau)\setminus U_{\tau}| + |V(\tau)\setminus V_{\tau}|}\prod_{e \in E(\tau)} (400D_V^2D_E^2q)^{l_e}\\
	&\cdot \sqrt{d}^{|\tau|_1 - |U_{\tau}|_1}\sqrt{m}^{|\tau|_2 - |U_{\tau}|_2} S(\tau) (C_{disc}\sqrt{D_E})^{\sum_{j \in (U_{\tau})_2 \cup (V_{\tau})_2} deg^{\tau}(j)}\\
	&\le d^{O(1) \cdot (C_V + C_E) \cdot \eps\sum_{e \in E(\tau)} l_e}\cdot \prod_{j \in V_2(\tau) \setminus V_2(U_{\tau}) \setminus V_2(V_{\tau})} (deg^{\tau}(j) - 1)!!\cdot \frac{1}{d^{A\eps\sum_{e \in E(\tau)} l_e}}\\
	&\le d^{O(1) \cdot (C_V + C_E) \cdot \eps\sum_{e \in E(\tau)} l_e}\cdot (D_VD_E)^{\sum_{e \in E(\tau)} l_e}\cdot \frac{1}{d^{A\eps\sum_{e \in E(\tau)} l_e}}\\
	&\le 1
	\end{align*}}
\end{proof}

We can now obtain our desired middle shape bounds.

\begin{lemma}\label{lem: spca_cond2}
    For all $U \in \calI_{mid}$ and $\tau \in \calM_U$,
    \[
    \begin{bmatrix}
        \frac{1}{|Aut(U)|c(\tau)}H_{Id_U} & B_{norm}(\tau) H_{\tau}\\
        B_{norm}(\tau) H_{\tau}^T & \frac{1}{|Aut(U)|c(\tau)}H_{Id_U}
    \end{bmatrix}
    \succeq 0
    \]
\end{lemma}

\begin{proof}
	We have
	\begin{align*}
		&\begin{bmatrix}
			\frac{1}{|Aut(U)|c(\tau)}H_{Id_U} & B_{norm}(\tau)H_{\tau}\\
			B_{norm}(\tau)H_{\tau}^T & \frac{1}{|Aut(U)|c(\tau)}H_{Id_U}
		\end{bmatrix}\\
		&\qquad= \begin{bmatrix}
			\left(\frac{1}{|Aut(U)|c(\tau)} - \frac{S(\tau)R(\tau)B_{norm}(\tau)}{|Aut(U)|}\right)H_{Id_U} & 0\\
			0 & \left(\frac{1}{|Aut(U)|c(\tau)} - \frac{S(\tau)R(\tau)B_{norm}(\tau)}{|Aut(U)|}\right)H_{Id_U}
		\end{bmatrix}\\
		&\qquad\quad+ B_{norm}(\tau)\begin{bmatrix}
			\frac{S(\tau)R(\tau)}{|Aut(U)|}H_{Id_U} & H_{\tau}\\
			H_{\tau}^T & \frac{S(\tau)R(\tau)}{|Aut(U)|}H_{Id_U}
		\end{bmatrix}
	\end{align*}
	By \cref{lem: spca_cond2_simplified}, $\begin{bmatrix}
		\frac{S(\tau)R(\tau)}{|Aut(U)|}H_{Id_U} & H_{\tau}\\
		H_{\tau}^T & \frac{S(\tau)R(\tau)}{|Aut(U)|}H_{Id_U}
	\end{bmatrix}
	\succeq 0$, so the second term above is positive semidefinite. For the first term, by \cref{lem: spca_cond1}, $H_{Id_U} \succeq 0$ and by \cref{cor: spca_norm_decay}, $\frac{1}{|Aut(U)|c(\tau)} - \frac{S(\tau)R(\tau)B_{norm}(\tau)}{|Aut(U)|} \ge 0$, which proves that the first term is also positive semidefinite.
\end{proof}

\subsection{\intersectionboundstwo}

\begin{lemma}\label{lem: spca_charging2}
	Suppose $0 < A < \frac{1}{4}$ is a constant such that $\frac{\sqrt{\lda}}{\sqrt{k}} \le d^{-A\eps}, \frac{1}{\sqrt{k}} \le d^{-2A}$ and $\frac{k}{d} \le d^{-A\eps}$. For all $m$ such that $m \le \frac{d^{1 - \eps}}{\lda^2}, m \le \frac{k^{2 - \eps}}{\lda^2}$, for all $U, V \in \calI_{mid}$ where $w(U) > w(V)$ and for all $\gam \in \Gamma_{U, V}$,
	\[n^{w(V(\gam)\setminus U_{\gam})} S(\gam)^2 \le \left(\prod_{j \in V_2(\gam) \setminus U_{\gam} \setminus V_{\gam}}(deg^{\gam}(j)- 1)!!\right)^2\frac{1}{d^{B\eps (|V(\gam) \setminus (U_{\gam} \cap V_{\gam})| + \sum_{e \in E(\gam)} l_e)}}\]
	for some constant $B > 0$ that depends only on $C_{\Del}$. In particular, it is independent of $C_V$ and $C_E$.
\end{lemma}

\begin{proof}
	Suppose there is a vertex $i \in V(\gam) \setminus U_{\gam} \setminus V_{\gam}$ such that $deg^{\gam}(i)$ is odd, then $S(\gam) = 0$ and the inequality is true. So, assume $deg^{\gam}(i)$ is even for all vertices $i \in V(\gam) \setminus U_{\gam} \setminus V_{\gam}$.
	We have $n^{w(V(\gam) \setminus U_{\gam})} = d^{|\gam|_1 - |U_{\gam}|_1}m^{|\gam|_2 - |U_{\gam}|_2}$. Plugging in $S(\gamma)$, we get that we have to prove
    {\footnotesize
	\begin{align*}
		E := d^{|\gam|_1 - |U_{\gam}|_1}m^{|\gam|_2 - |U_{\gam}|_2} \left(\frac{k}{d}\right)^{2|\gamma|_1 - |U_{\gamma}|_1 - |V_{\gamma}|_1}\Del^{2|\gamma|_2 - |U_{\gamma}|_2 - |V_{\gamma}|_2} \prod_{e \in E(\gamma)} \frac{\lambda^{l_e}}{k^{l_e}} \le \frac{1}{d^{B\eps (|V(\gam) \setminus (U_{\gam} \cap V_{\gam})| + \sum_{e \in E(\gam)} l_e)}}
	\end{align*}}

	Let $S'$ be the set of vertices $i \in U_{\gam} \setminus V_{\gam}$ that have $deg^{\gam}(i) \ge 1$. Let $e, f$ be the number of type $1$ vertices and the number of type $2$ vertices in $S'$ respectively. Observe that $S' \cup (U_{\gam} \cap V_{\gam})$ is a vertex separator of $\gam$.
	Let $g = |V_{\gam} \setminus U_{\gam}|_1$ (resp. $h = |V_{\gam} \setminus U_{\gam}|_2$) be the number of type $1$ vertices (resp. type $2$ vertices) in $V_{\gam} \setminus U_{\gam}$.
	We first claim that $d^em^f \ge d^gm^h$. To see this, note that the vertex separator $S' \cup (U_{\gam} \cap V_{\gam})$ has weight $\sqrt{d}^{e + |U_{\gam} \cap V_{\gam}|_1}\sqrt{m}^{f + |U_{\gam} \cap V_{\gam}|_2}$. On the other hand, $V_{\gam}$ has weight $\sqrt{d}^{g + |U_{\gam} \cap V_{\gam}|_1}\sqrt{m}^{h + |U_{\gam} \cap V_{\gam}|_2}$. Since $\gam$ is a left shape, $V_{\gam}$ is the unique minimum vertex separator and hence, we have the inequality $\sqrt{d}^{e + |U_{\gam} \cap V_{\gam}|_1}\sqrt{m}^{f + |U_{\gam} \cap V_{\gam}|_2} \ge \sqrt{d}^{g + |U_{\gam} \cap V_{\gam}|_1}\sqrt{m}^{h + |U_{\gam} \cap V_{\gam}|_2}$ which implies $d^em^f \ge d^gm^h$.
	Let $p = |V(\gam) \setminus (U_{\gam} \cup V_{\gam})|_1$ (resp. $q = |V(\gam) \setminus (U_{\gam} \cup V_{\gam})|_2$) be the number of type $1$ vertices (resp. type $2$ vertices) in $V(\gam) \setminus (U_{\gam} \cup V_{\gam})$.
	To illustrate the main idea, we will first prove the weaker inequality $E \le 1$. Since $\Del \le 1$, it suffices to prove
	\begin{align*}
		d^{|\gam|_1 - |U_{\gam}|_1}m^{|\gam|_2 - |U_{\gam}|_2} \left(\frac{k}{d}\right)^{2|\gamma|_1 - |U_{\gamma}|_1 - |V_{\gamma}|_1} \prod_{e \in E(\gamma)} \frac{\lambda^{l_e}}{k^{l_e}} \le 1
	\end{align*}
	We have
	$d^{|\gam|_1 - |U_{\gam}|_1}m^{|\gam|_2 - |U_{\gam}|_2} = d^{p + g}m^{q + h} \le n^{p + \frac{e + g}{2}}m^{q + \frac{f + h}{2}}$
	since $d^em^f \ge d^gm^h$. Also, $2|\gamma|_1 - |U_{\gamma}|_1 - |V_{\gamma}|_1 = 2p + e + g$. So, it suffices to prove
	\begin{align*}
		n^{p + \frac{e + g}{2}}m^{q + \frac{f + h}{2}}\left(\frac{k}{d}\right)^{2p + e + g} \prod_{e \in E(\gam)} \left(\frac{\lda}{k}\right)^{l_e} \le 1
	\end{align*}

	We will need the following claim.
	\begin{claim}
		$\sum_{e \in E(\gam)} l_e \ge \max(2p + e + g, 2q + f + h)$
	\end{claim}
	\begin{proof}
		Since $H_{\gam}$ is bipartite, we have $\sum_{e \in E(\gam)}l_e = \sum_{i \in V_1(\gam)} deg^{\gam}(i) = \sum_{i \in V_2(\gam)} deg^{\gam}(i)$. Observe that all vertices $i \in V(\gam) \setminus U_{\gam} \setminus V_{\gam}$ have $deg^{\gam}(i)$ nonzero and even, and hence, $deg^{\gam}(i) \ge 2$. Then,
	\begin{align*}
		\sum_{e \in E(\gam)}l_e &= \sum_{i \in V_1(\gam)} deg^{\gam}(i)\\
		&\ge \sum_{i \in V_1(\gam) \setminus U_{\gam} \setminus V_{\gam}} deg^{\gam}(i) + \sum_{i \in (U_{\gam})_1 \setminus V_{\gam}} deg^{\gam}(i) + \sum_{i \in (V_{\gam})_1 \setminus U_{\gam}} deg^{\gam}(i)\\
		&\ge 2p + e + g
	\end{align*}
	Similarly,
	\begin{align*}
	\sum_{e \in E(\gam)}l_e &= \sum_{i \in V_2(\gam)} deg^{\gam}(i)\\
	&\ge \sum_{i \in V_2(\gam) \setminus U_{\gam} \setminus V_{\gam}} deg^{\gam}(i) + \sum_{i \in (U_{\gam})_2 \setminus V_{\gam}} deg^{\gam}(i) + \sum_{i \in (V_{\gam})_2 \setminus U_{\gam}} deg^{\gam}(i)\\
	&\ge 2q + f + h
\end{align*}
Therefore, $\sum_{e \in E(\gam)} l_e \ge \max(2p + e + g, 2q + f + h)$.
\end{proof}

Now, let $r_1 = p + \frac{e + g}{2}, r_2 = q + \frac{f + h}{2}$. Then, $\sum_{e \in E(\gam)} l_e \ge 2\max(r_1, r_2)$ and we wish to prove
	$d^{r_1}m^{r_2} \left(\frac{k}{d}\right)^{2r_1} \left(\frac{\lda}{k}\right)^{2\max(r_1, r_2)} \le 1$
This expression simply follows by squaring \cref{claim: spca_decay}.

Now, to prove that $E \le \frac{1}{d^{B\eps (|V(\gam) \setminus (U_{\gam} \cap V_{\gam})| + \sum_{e \in E(\gam)} l_e)}}$, we mimic this argument while carefully keeping track of factors of $d^{\eps}$. Again, using $d^em^f \ge d^gm^h$, it suffices to prove that
\begin{align*}
	d^{p + \frac{e + g}{2}}m^{q + \frac{f + h}{2}} \left(\frac{k}{d}\right)^{2|\gamma|_1 - |U_{\gamma}|_1 - |V_{\gamma}|_1}&\Del^{2|\gamma|_2 - |U_{\gamma}|_2 - |V_{\gamma}|_2} \prod_{e \in E(\gamma)} \frac{\lambda^{l_e}}{k^{l_e}}\\
    &\le \frac{1}{d^{B\eps (|V(\gam) \setminus (U_{\gam} \cap V_{\gam})| + \sum_{e \in E(\gam)} l_e)}}
\end{align*}

The idea is that the $d^{B\eps}$ decay for the edges are obtained from the stronger assumption on $m$, namely $m \le \frac{d^{1 - \eps}}{\lda^2}, m \le \frac{k^{2 - \eps}}{\lda^2}$. And the $d^{B\eps}$ decay for the type $1$ vertices of $V(\gam) \setminus(U_{\gam} \cap V_{\gam})$ are obtained both from the stronger assumption on $m$ as well as the factors of $\frac{k}{d}$, the latter especially useful for the degree $0$ vertices. Finally, the $d^{B\eps}$ decay for the type $2$ vertices of $V(\gam) \setminus (U_{\gam} \cap V_{\gam})$ are obtained from the factors of $\Delta$.
Indeed, note that for a constant $B$ that depends on $C_{\Del}$, $\Del^{2|\gamma|_2 - |U_{\gamma}|_2 - |V_{\gamma}|_2} \le d^{-B\eps|V(\gam) \setminus (U_{\gam} \cap V_{\gam})|_2}$. So, we would be done if we prove
\begin{align*}
	d^{p + \frac{e + g}{2}}m^{q + \frac{f + h}{2}} \left(\frac{k}{d}\right)^{2|\gamma|_1 - |U_{\gamma}|_1 - |V_{\gamma}|_1}\left(\frac{\lambda}{k}\right)^{\sum_{e \in E(\gamma)} l_e} \le \frac{1}{d^{B\eps (|V(\gam) \setminus (U_{\gam} \cap V_{\gam})|_1 + \sum_{e \in E(\gam)} l_e)}}
\end{align*}

Let $c_0$ be the number of type $1$ vertices $i$ in $V(\gam) \setminus (U_{\gam} \cap V_{\gam})$ such that $deg^{\gam}(i) = 0$. Since they have degree $0$, they must be in $(U_{\gam})_1 \setminus V_{\gam}$. Also, we have $2|\gamma|_1 - |U_{\gamma}|_1 - |V_{\gamma}|_1 = 2p + e + g + c_0$ and hence, $\left(\frac{k}{d}\right)^{2|\gamma|_1 - |U_{\gamma}|_1 - |V_{\gamma}|_1} = \left(\frac{k}{d}\right)^{2p + e + g + c_0}$. For these degree $0$ vertices, we have that the factors of $\frac{k}{d} \le d^{-A\eps}$ offer a decay of $\frac{1}{d^{B\eps}}$. Therefore, it suffices to prove
\begin{align*}
	d^{p + \frac{e + g}{2}}m^{q + \frac{f + h}{2}} \left(\frac{k}{d}\right)^{2p + e + g}\left(\frac{\lambda}{k}\right)^{\sum_{e \in E(\gamma)} l_e} \le \frac{1}{d^{B\eps (p + q + e + f + g + h) + \sum_{e \in E(\gam)} l_e)}}
\end{align*}
for a constant $B > 0$. Observe that $p + q + e + f + g + h \le 2(\sum_{e \in E(\gam)} l_e)$. Therefore, using the notation $r_1 = p + \frac{e + g}{2}, r_2 = q + \frac{f + h}{2}$, it suffices to prove
\begin{align*}
	d^{r_1}m^{r_2} \left(\frac{k}{d}\right)^{2r_1}\left(\frac{\lambda}{k}\right)^{\sum_{e \in E(\gamma)} l_e} \le \frac{1}{d^{B\eps \sum_{e \in E(\gam)} l_e}}
\end{align*}
for a constant $B > 0$. But this follows by squaring \cref{claim: spca_decay2} where we set $r = \sum_{e \in E(\gam)} l_e$.
\end{proof}

\begin{remk}
	In the above bounds, note that there is a decay of $d^{B\eps}$ for each vertex in $V(\gam) \setminus (U_{\gam} \cap V_{\gam})$.	One of the main technical reasons for introducing the slack parameter $C_{\Del}$ in the planted distribution was to introduce this decay, which is needed in the current machinery.
\end{remk}

With this, we obtain intersection term bounds.

\begin{lemma}\label{lem: spca_cond3}
    For all $U, V \in \calI_{mid}$ where $w(U) > w(V)$ and all $\gam \in \Gam_{U, V}$, \[c(\gam)^2N(\gam)^2B(\gam)^2H_{Id_V}^{-\gam, \gam} \preceq H_{\gam}'\]
\end{lemma}

\begin{proof}
	By \cref{lem: spca_cond3_simplified}, we have
	\begin{align*}
		c(\gam)^2N(\gam)^2B(\gam)^2H_{Id_V}^{-\gam, \gam} &\preceq c(\gam)^2N(\gam)^2B(\gam)^2 S(\gam)^2R(\gam)^2 \frac{|Aut(U)|}{|Aut(V)|} H'_{\gam}
	\end{align*}
	Using the same proof as in \cref{lem: spca_cond1}, we can see that $H'_{\gam} \succeq 0$. Therefore, it suffices to prove that $c(\gam)^2N(\gam)^2B(\gam)^2 S(\gam)^2R(\gam)^2 \frac{|Aut(U)|}{|Aut(V)|} \le 1$.
	Since $U, V \in \calI_{mid}$, $Aut(U) = |U|_1!|U|_2!, Aut(V) = |V|_1!|V|_2!$. Therefore, $\frac{|Aut(U)|}{|Aut(V)|} = \frac{|U|_1!|U|_2!}{|V|_1!|V|_2!} \le D_V^{|U_{\gam} \setminus V_{\gam}|}$. Also, $|E(\gam)| \le \sum_{e \in E(\gam)} l_e$ and $q = d^{O(1) \cdot \eps (C_V + C_E)}$. Note $R(\gam)^2 = (C_{disc}\sqrt{D_E})^{2\sum_{j \in (U_{\gam})_2 \cup (V_{\gam})_2} deg^{\gam}(j)} \le d^{O(1)\cdot \eps C_E \cdot \sum_{e \in E(\gam)} l_e}$ and \[\left(\prod_{j \in V_2(\gam) \setminus U_{\gam} \setminus V_{\gam}}(deg^{\gam}(j)- 1)!!\right)^2 \le (D_VD_E)^{2\sum_{e \in E(\tau)} l_e} \le d^{O(1)\cdot \eps (C_V + C_E) \cdot \sum_{e \in E(\gam)} l_e}\]

	Let $B$ be the constant from \cref{lem: spca_charging2}. We can set $C_V, C_E$ sufficiently small so that, using \cref{lem: spca_charging2},
	\begin{align*}
		c(\gam)^2&N(\gam)^2B(\gam)^2S(\gam)^2R(\gam)^2 \frac{|Aut(U)|}{|Aut(V)|} \\
		&\le 100^2 (6D_V)^{2|U_{\gam}\setminus V_{\gam}| + 2|V_{\gam}\setminus U_{\gam}| + |E(\al)|}16^{|V(\gam) \setminus (U_{\gam} \cup V_{\gam})|}\\
		&\quad\cdot (3D_V)^{4|V(\gam)\setminus V_{\gam}| + 2|V(\gam)\setminus U_{\gam}|} (6qD_V)^{2|V(\gam)\setminus U_{\gam}| + 2|V(\gam)\setminus V_{\gam}|} \prod_{e \in E(\gam)} (400D_V^2D_E^2q)^{2l_e}\\
		&\quad\cdot  n^{w(V(\gam)\setminus U_{\gam})} S(\gam)^2 d^{O(1)\cdot \eps C_E \cdot \sum_{e \in E(\gam)} l_e}\cdot D_V^{|U_\gam \setminus V_{\gam}|} \\
		&\le d^{O(1) \cdot \eps(C_V + C_E) \cdot (|V(\gam) \setminus (U_{\gam} \cap V_{\gam})| + \sum_{e \in E(\gam)} l_e)} \cdot n^{w(V(\gam)\setminus U_{\gam})} S(\gam)^2\\
		&\le d^{O(1) \cdot \eps(C_V + C_E) \cdot (|V(\gam) \setminus (U_{\gam} \cap V_{\gam})| + \sum_{e \in E(\gam)} l_e)}\cdot \frac{1}{d^{B\eps (|V(\gam) \setminus (U_{\gam} \cap V_{\gam})| + \sum_{e \in E(\gam)} l_e)}}\\
		&\le 1
	\end{align*}
\end{proof}

\subsection{\truncationboundstwo}

In this section, we will obtain truncation error bounds using the strategy sketched in \cite[Section 10]{potechin2020machinery}. We also reuse the notation. To do this, we need to first obtain a bound on the quantity $B_{norm}(\sig) B_{norm}(\sig') H_{Id_U}(\sig, \sig')$.

\begin{lemma}\label{lem: spca_charging3}
	Suppose $0 < A < \frac{1}{4}$ is a constant such that $\frac{\sqrt{\lda}}{\sqrt{k}} \le d^{-A\eps}$ and $\frac{1}{\sqrt{k}} \le d^{-2A}$. Suppose $m$ is such that $m \le \frac{d^{1 - \eps}}{\lda^2}, m \le \frac{k^{2 - \eps}}{\lda^2}$. For all $U \in \calI_{mid}$ and $\sig, \sig' \in \calL_U$,
	\[B_{norm}(\sig) B_{norm}(\sig') H_{Id_U}(\sig, \sig') \le \frac{1}{d^{0.5A\eps(|V(\sig \circ \sig')| + \sum_{e \in E(\al) l_e}}} \cdot \frac{1}{d^{|U_{\sig}|_1 + |U_{\sig'}|_1}m^{|U_{\sig'}|_2 + |U_{\sig'}|_2}}\]
\end{lemma}

\begin{proof}
	Suppose there is a vertex $i \in V(\sig) \setminus V_{\sig}$ such that $deg^{\sig}(i) + deg^{U_{\sig}}(i)$ is odd, then $H_{Id_U}(\sig, \sig') = 0$ and the inequality is true. So, assume that $deg^{\sig}(i) + deg^{U_{\sig}}(i)$ is even for all $i \in V(\sig) \setminus V_{\sig}$. Similarly, assume that $deg^{\sig'}(i) + deg^{U_{\sig'}}(i)$ is even for all $i \in V(\sig') \setminus V_{\sig'}$. Also, if $\rho_{\sig} \neq \rho_{\sig'}$, we will have $H_{Id_U}(\sig, \sig') = 0$ and we would be done. So, assume $\rho_{\sig} = \rho_{\sig'}$.

	Let there be $e$ (resp. $f$) vertices of type $1$ (resp. type $2$) in $V(\sig) \setminus U_{\sig} \setminus V_{\sig}$. Then, $n^{\frac{w(V(\sig)) - w(U)}{2}} = \sqrt{d}^{|V(\sig)|_1 - |U|_1}\sqrt{m}^{|V(\sig)|_2 - |U|_2} = \sqrt{d}^{|U_{\sig}|_1}\sqrt{m}^{|U_{\sig}|_2} \sqrt{d}^e\sqrt{m}^f$. Let there be $g$ (resp. $h$) vertices of type $1$ (resp. type $2$) in $V(\sig') \setminus U_{\sig'} \setminus V_{\sig'}$. Then, similarly, $n^{\frac{w(V(\sig')) - w(U)}{2}} \le \sqrt{d}^{|U_{\sig'}|_1}\sqrt{m}^{|U_{\sig'}|_2}\sqrt{d}^g\sqrt{m}^h$.

	Let $\al = \sig \circ \sig'$. Since all vertices in $V(\al) \setminus U_{\al} \setminus V_{\al}$ have degree at least $2$, we have $\sum_{e \in E(\al)} l_e \ge \sum_{i \in V_1(\al) \setminus U_{\al} \setminus V_{\al}} deg^{\al}(i) \ge 2(e + g) + |U_{\sig}|_1 + |U_{\sig}|_2$. Similarly, $\sum_{e \in E(\al)} l_e \ge 2(f + h) + |U_{\sig'}|_1 + |U_{\sig'}|_2$. Therefore, by setting $r_1 = e + g, r_2 = f + h$ in \cref{claim: spca_decay2}, we have
	\[\sqrt{d}^{e + g}\sqrt{m}^{f + h} \left(\frac{k}{d}\right)^{e + g}\prod_{e \in E(\alpha)} \frac{\sqrt{\lambda}^{l_e}}{\sqrt{k}^{l_e}} \le \frac{1}{d^{A\eps \sum_{e \in E(\al)} l_e}}\]
	Also, \[\left(\frac{k}{d}\right)^{|\al|_1} \le \left(\frac{k}{d}\right)^{e + g + |U_{\sig}|_1 + |U_{\sig'}|_1}\]
    and \[\prod_{j \in V_2(\alpha)} (deg^{\alpha}(j) - 1)!! \le d^{\eps C_V \sum_{e \in E(\al)} l_e}\]
    Therefore,
	\begin{align*}
		&n^{\frac{w(V(\sig)) - w(U)}{2}}n^{\frac{w(V(\sig')) - w(U)}{2}}H_{Id_U}(\sig, \sig')\\
		&\le d^{O(1)D_{sos}}\sqrt{d}^e\sqrt{m}^f d^{O(1)D_{sos}}\sqrt{d}^g\sqrt{m}^h\\
        &\qquad \cdot\frac{1}{|Aut(U)|}\left(\frac{1}{\sqrt{k}}\right)^{deg(\alpha)}\left(\frac{k}{d}\right)^{|\alpha|_1}\Del^{|\alpha|_2} \prod_{j \in V_2(\alpha)} (deg^{\alpha}(j) - 1)!!\prod_{e \in E(\alpha)} \frac{\sqrt{\lambda}^{l_e}}{\sqrt{k}^{l_e}}\\
		&\le d^{O(1)D_{sos}} d^{\eps C_V \sum_{e \in E(\al)} l_e} \sqrt{d}^{e + g}\sqrt{m}^{f + h} \left(\frac{k}{d}\right)^{e + g}\prod_{e \in E(\alpha)} \frac{\sqrt{\lambda}^{l_e}}{\sqrt{k}^{l_e}} \cdot \frac{1}{d^{|U_{\sig}|_1 + |U_{\sig'}|_1}m^{|U_{\sig'}|_2 + |U_{\sig'}|_2}}\\
		&\le \frac{d^{\eps C_V \sum_{e \in E(\al)} l_e}}{d^{A\eps \sum_{e \in E(\al)} l_e}} \cdot \frac{1}{d^{|U_{\sig}|_1 + |U_{\sig'}|_1}m^{|U_{\sig'}|_2 + |U_{\sig'}|_2}}
	\end{align*}
	By setting $C_V, C_E$ sufficiently small and plugging in the expressions for $B_{norm}(\sig), B_{norm}(\sig')$, we obtain the result.
\end{proof}

We can apply the the strategy now.

\begin{restatable}{lemma}{SPCAfive}\label{lem: spca_cond5}
	Whenever $\norm{M_{\al}} \le B_{norm}(\al)$ for all $\al \in \calM'$,
	\[
	\sum_{U \in \mathcal{I}_{mid}}{M^{fact}_{Id_U}{(H_{Id_U})}} \succeq \frac{1}{d^{K_1D_{sos}^2}} Id_{sym}
	\]
	for a constant $K_1 > 0$ that can depend on $C_{\Del}$.
\end{restatable}

\begin{proof}
    For $V \in \calI_{mid}$, $\lda_V = \frac{\Del^{|V|_2}}{d^{|V|_1}k^{|V|_2}}$. Let the minimum value of this quantity over all $V$ be $N$. We then choose $w_V = N / \lda_V$ so that for all left shapes $\sig \in \calL_V$, \cref{lem: spca_charging3} implies $w_{V} \leq \frac{w_{U_{\sigma}}\lambda_{U_{\sigma}}}{|\mathcal{I}_{mid}|B_{norm}(\sigma)^2{c(\sigma)^2}{H_{Id_V}(\sigma,\sigma)}}$, completing the proof.
\end{proof}

\begin{restatable}{lemma}{SPCAsix}\label{lem: spca_cond6}
	\[\sum_{U\in \calI_{mid}} \sum_{\gam \in \Gam_{U, *}} \frac{d_{Id_{U}}(H_{Id_{U}}, H'_{\gam})}{|Aut(U)|c(\gam)} \le \frac{d^{K_2 D_{sos}}}{2^{D_V}}\]
	for a constant $K_2 > 0$ that can depend on $C_{\Del}$.
\end{restatable}

\begin{proof}
    We do the same calculations as in the proof of \cref{lem: plds_cond6}, until
	\begin{align*}
		\sum_{U\in \calI_{mid}} \sum_{\gam \in \Gam_{U, *}} &\frac{d_{Id_{U}}(H_{Id_{U}}, H'_{\gam})}{|Aut(U)|c(\gam)}\\
        &\le \sum_{U\in \calI_{mid}} \sum_{\sigma,\sigma' \in \mathcal{L}'_{U}} {B_{norm}(\sigma)B_{norm}(\sigma')H_{Id_{U}}(\sigma,\sigma')\frac{1}{2^{\min(m_{\sig}, m_{\sig'}) - 1}}}\\
		&\le \sum_{U\in \calI_{mid}} \sum_{\sigma,\sigma' \in \mathcal{L}'_{U}}\frac{d^{O(1) D_{sos}}}{d^{0.5A\eps|V(\sig \circ \sig')|}2^{\min(m_{\sig}, m_{\sig'}) - 1}}
	\end{align*}
	where we used \cref{lem: spca_charging3}. Using $d^{0.5A\eps |V(\sig \circ \sig')|} \ge d^{0.1A\eps |V(\sig \circ \sig')|}2^{|V(\sig \circ \sig')|}$,
	\begin{align*}
		\sum_{U\in \calI_{mid}} \sum_{\gam \in \Gam_{U, *}} \frac{d_{Id_{U}}(H_{Id_{U}}, H'_{\gam})}{|Aut(U)|c(\gam)} &\le \sum_{U\in \calI_{mid}} \sum_{\sigma,\sigma' \in \mathcal{L}'_{U}}\frac{d^{O(1) D_{sos}}}{d^{0.1A\eps|V(\sig \circ \sig')|}  2^{|V(\sig \circ \sig')|}2^{\min(m_{\sig}, m_{\sig'}) - 1}}\\
		&\le \sum_{U\in \calI_{mid}} \sum_{\sigma,\sigma' \in \mathcal{L}'_{U}}\frac{d^{O(1) D_{sos}}}{d^{0.1A\eps|V(\sig \circ \sig')|} 2^{D_V}}\\
		&\le \sum_{U\in \calI_{mid}} \sum_{\sigma,\sigma' \in \mathcal{L}'_{U}}\frac{d^{O(1) D_{sos}}}{D_{sos}^{D_{sos}}d^{0.1A\eps|V(\sig \circ \sig')|} 2^{D_V}}
	\end{align*}
	The final step will be to argue that $\sum_{U\in \calI_{mid}} \sum_{\sigma,\sigma' \in \mathcal{L}'_{U}}\frac{1}{D_{sos}^{D_{sos}}d^{0.1 A\eps|V(\sig \circ \sig')|}} \le 1$ which will complete the proof. But this will follow if we set $C_V, C_E$ small enough.
\end{proof}

We can finally show that truncation errors can be handled.

\begin{restatable}{lemma}{SPCAfour}\label{lem: spca_cond4}
    Whenever $\norm{M_{\alpha}} \le B_{norm}(\alpha)$ for all $\alpha \in \mathcal{M}'$,
    \[
    \sum_{U \in \mathcal{I}_{mid}}{M^{fact}_{Id_U}{(H_{Id_U})}} \succeq 6\left(\sum_{U \in \mathcal{I}_{mid}}{\sum_{\gamma \in \Gamma_{U,*}}{\frac{d_{Id_{U}}(H'_{\gamma},H_{Id_{U}})}{|Aut(U)|c(\gamma)}}}\right)Id_{sym}
    \]
\end{restatable}

\begin{proof}
	Choose $C_{sos}$ sufficiently small so that $\frac{1}{d^{K_1D_{sos}^2}} \ge 6\frac{d^{K_2D_{sos}}}{2^{D_V}}$ which can be satisfied by setting $C_{sos} < K_3 C_V$ for a sufficiently small constant $K_3 > 0$. Then, since $Id_{Sym} \succeq 0$, using \cref{lem: spca_cond5} and \cref{lem: spca_cond6},
	\begin{align*}
		\sum_{U \in \mathcal{I}_{mid}}{M^{fact}_{Id_U}{(H_{Id_U})}} &\succeq \frac{1}{d^{K_1D_{sos}^2}} Id_{sym}\\
		&\succeq 6\frac{d^{K_2D_{sos}}}{2^{D_V}} Id_{sym}\\
		&\succeq 6\left(\sum_{U \in \mathcal{I}_{mid}}{\sum_{\gamma \in \Gamma_{U,*}}{\frac{d_{Id_{U}}(H'_{\gamma},H_{Id_{U}})}{|Aut(U)|c(\gamma)}}}\right)Id_{sym}
	\end{align*}
\end{proof}

\chapter{Followup and Future work}\label{chap: future_work}

In this chapter, we go over some follow-up works that are not covered in this dissertation and also suggest directions for future work. We then conclude this dissertation with a note on the broader implications of our work for computer science.

\section{Nonlinear concentration for non-product distributions}

Our techniques in \cref{chap: efron_stein} apply to a collection of random variables that are sampled independently of each other. A natural question is to ask if we can generalize to the case when they are not independent. For example, this is useful when instead of analyzing \Erdos-\Renyi random graphs, we wish to analyze uniform $d$-regular graphs. Such a generalization seems extremely likely because our proof techniques essentially requires a Markov Chain that mixes rapidly to the given distribution, and then we can recursively apply the \Poincare inequality. We leave this for future work.

\section{Sum of Squares lower bounds}

In this dissertation, we saw several SoS lower bounds and while they build on fundamental conceptual building blocks such as the nonlinear concentration results we show and simple heuristics like pseudocalibration, an important technical barrier in the current proofs is that the proofs are highly technical and have many moving parts. It's an important research question to understand if the proofs can be simplified. Apart from enabling a better understanding of the SoS hierarchy, this will also help us understand the computational barriers of several fundamental problems in computer science. Examples of such problems follow.

\subsection{Sparse Independent Set}

In a follow-up work \cite{jones2022sum}, we prove SoS lower bounds for the important problem of maximum independent set on sparse \Erdos-\Renyi random graphs.

In this dissertation, the SoS lower bounds studied were in the setting when the input was sampled from product distributions where each distribution was either Rademacher or Gaussian. This is also the case in many prior works on SoS lower bounds. Recall that this was termed the \textit{dense setting} in \cref{chap: efron_stein}. It's equally important to study problems in the fascinating average-case \textit{sparse setting} where the input distribution could have high Orlicz norm, for example when the input is an \Erdos-\Renyi random graph sampled from $G_{n, p}$ instead of $G_{n, \frac{1}{2}}$ for some $p = o(1)$. The techniques developed in this work and prior works for high degree SoS lower bounds do not easily generalize to this setting. The work \cite{jones2022sum} initiates this research direction for the fundamental problem of maximum independent set on random sparse graphs.

Consider the independent set problem on a graph $G \sim G_{n, \frac{d}{n}}$ where $d$ is the average degree. If $d = \frac{n}{2}$, then this is the same as the maximum clique problem and SoS lower bounds were obtained in \cite{BHKKMP16}. We now focus on the setting $d \ll n$. We first state the size of the true optimum.
\begin{fact}[\cite{COE15, DM11, DSS16}]
    W.h.p. the max independent set in $G$ has size $(1+o_d(1)) \cdot \frac{2\ln d}{d} \cdot n$.
\end{fact}

The famous \Lovasz $\vartheta$ function efficiently computes an upper bound on this value and its value is well-known on such random graphs.
\begin{fact}[\cite{CO05}]
    W.h.p. $\vartheta(G) =\Theta(\frac{n}{\sqrt{d}} )$.
\end{fact}

The value of the $\vartheta$ function is also the output of the degree $2$ SoS relaxation for this problem. So, there is an integrality gap of approximately $\sqrt{d}$. We therefore naturally ask whether higher degree SoS can perform better or this gap persists. In our work, we show that this $\sqrt{d}$ integrality gap persists for higher degrees of SoS as well

We prove two main results, one in the setting $(\log n)^2 \le d \le \sqrt{n}$ and the other in the setting $n^{\Omega(1)}  \le d \le \frac{n}{2}$.
Note that we have not covered the case when the average degree $d$ is constant. This is an interesting direction for future work.

In the first setting $(\log n)^2 \le d \le \sqrt{n}$, we show a tradeoff between the degree $D_{sos}$  of the SoS relaxation and the integrality gap.
\begin{theorem}
    There is an absolute constant $c_0 \in \N$ such that for sufficiently large $n \in \N$ and ${d \in [(\log n)^2, n^{0.5}]}$, and parameters $k, \dsos$ satisfying
    $
    k ~\leq~ \frac{n}{\dsos^{c_0}\cdot \log n \cdot d^{1/2}},$
    w.h.p. over $G~\sim~ G_{n,~d/n}$, there exists a degree-$\dsos$ pseudoexpectation for the maximum independent set problem with objective value $(1-o(1)) k$.
\end{theorem}
In particular, when $d \in [n^{\Omega(1)}, \sqrt{n}]$, this exhibits an SoS lower bound against polynomial degree $n^{\Omega(1)}$ SoS.
In the second setting $n^{\Omega(1)}  \le d \le \frac{n}{2}$, we show an SoS lower bound for logarithmic degree SoS.
\begin{theorem}
    \label{thm:informal-logd}
    For any $\eps_1, \eps_2 >0$ there is $\delta > 0$, such that for $d \in [n^{\eps_1}, n/2]$ and $k \leq \frac{n}{d^{1/2+\eps_2}}$, w.h.p. over $G~\sim~ G_{n,~d/n}$, there exists a degree-$(\delta \log d)$ pseudoexpectation with objective value $(1-o(1))k$.
\end{theorem}

We remark that these theorems rule out polynomial-time certification (i.e. constant
degree SoS) for any $d \geq \polylog(n)$.

Broadly speaking, we utilize similar techniques to show these results, namely pseudo-calibration, graph matrices and approximate PSD decomposition. However, the approach does not readily work and we overcome the difficulties with several new ideas and techniques. We summarize some of them below.
\begin{itemize}
    \item The first conceptual difficulty we overcome is that we are unable to apply pseudo-calibration due to the lack of a good candidate planted distribution. For most natural choices of the planted distribution, simple statistics distinguish the random distribution from the planted distribution. While a suitable planted distribution that enables the use of pseudo-calibration may very well exist, we are yet to find one. Instead, in this work, we simply use the na\"ive planted distribution but instead modify the heuristic of pseudo-calibration (that we term \textit{pseudo-calibration with connected trunction}) to construct our candidate moment matrix.
    \item The second conceptual difficulty was the lack of good norm bounds for graph matrices built from sparse graphs. In that work, we utilized the trace method with a careful analysis to obtain better norm bounds. Moreover, as we saw in \cref{chap: efron_stein}, we are able to obtain similar norm bounds without the trace method, using our general recursion theorem.
\end{itemize}
Apart from the above developments, we develop several technical tools such as conditioning, a generalization of the intersection tradeoff lemma, etc.
For more details, see \cite{jones2022sum}.

\subsection{Planted Affine Planes and Maximum Cut}\label{sec:open-problems}

For the Planted Affine Planes problem from \cref{chap: sk} where we sampled $m$ vectors $d_1, \ldots, d_m$ independently from $\GN(0, I_n)$, we showed an SoS lower bound for $m \le n^{3/2 - \eps}$. However, from the analysis of $\pE[1]$ in \cref{rmk:pe-one}, we expect a lower bound to hold for $m \ll n^{2 - \eps}$. This is because, as we saw in \cref{chap: sos} and which we will revisit in the next section, analyzing $\pE[1]$ is an established way to hypothesize about the power of SoS. Therefore, we conjecture
\begin{conjecture}
    \cref{theo:sos-bounds} holds with the bound on the number of sampled vectors $m$ loosened to $m \leq n^{2-\eps}$.
\end{conjecture}

Dual to this (in fact, we exploit the duality in our proof in \cref{chap: sk}), we conjecture an SoS lower bound for the Planted Boolean Vector problem holds whenever $p \geq n^{1/2+\eps}$.
\begin{conjecture}
    \cref{theo:boolean-subspace} holds with the bound on the dimension $p$ of a random subspace
    loosened to $p \geq n^{1/2+\eps}$.
\end{conjecture}

We remark that recent work \cite{zadik2021latticebased} has exhibited a polynomial time for the search variant of Planted Affine Planes for $m \ge n + 1$, as opposed to prior known algorithms that required $m \gg n^2$ \cite{mao2021optimal}. The algorithm in \cite{mao2021optimal} is spectral and robust to noise, moreover it is likely captured by SoS. On the other hand, the algorithm in \cite{zadik2021latticebased} is lattice-based and is not robust to noise, (i.e. it assumes that all vectors must exactly lie in the two planes), and is not captured by SoS.

In our SoS lower bounds for the Planted Boolean Vector problem and the Planted Affine Planes problem, we assumed that the input entries were chosen i.i.d Gaussian or Boolean. In fact, it's plausible that our proof techniques go through when the distribution is ``random enough'', such as the uniform distribution from the sphere. One potential extension of this intuition is as follows: In the Planted Boolean Vector problem, if the subspace is the eigenspace of the bottom
eigenvectors of a random adjacency matrix, the instance should still be
difficult. This last setting arises in Maximum Cut, for which we conjecture the following.

\begin{conjecture}\label{conj: sos_for_max_cut}
    Let $d \geq 3$, and let $G$ be a random $d$-regular graph on $n$ vertices. For some $\delta > 0$, w.h.p. there is a degree-$n^\delta$ pseudoexpectation operator $\pE$ on boolean variables $x_i$ with maximum cut value at least
    \[ \frac{1}{2} + \frac{\sqrt{d-1}}{d}(1 - \littleoh_{d,n}(1)) \]
\end{conjecture}

The above expression is w.h.p. the value of the spectral relaxation for Maximum Cut, therefore qualitatively this conjecture expresses that degree $n^\delta$ SoS cannot significantly tighten the basic spectral relaxation.

We should remark that, with respect to the goal of showing SoS cannot significantly outperform the Goemans-Williamson relaxation, random instances are not integrality gap instances. The main difficulty in comparing (even degree 4) SoS to the Goemans-Williamson algorithm seems to be the lack of a candidate hard input distribution.

Evidence for this conjecture comes from the fact that the only property
required of the random inputs $d_1, \dots, d_m$ was that norm bounds hold for
the graph matrix with Hermite polynomial entries. When the variables
$\{d_{u,i}\}$ are i.i.d from some other distribution, if we use graph matrices
for the orthonormal polynomials under the distribution and assuming suitable
bounds on the moments of the distribution, the same norm bounds
hold~\cite{ahn2016graph}.
When $d_u$ is sampled uniformly from the sphere or another distribution for which the coordinates are not
i.i.d, it seems likely that
similar norm bounds hold. Moreover, as explained in the previous section, the techniques from \cref{chap: efron_stein} will likely be useful to obtain such norm bounds.

\subsection{Unique Games}

The famous Unique Games conjecture (UGC) \cite{Khot02:unique} postulates that a graph theory problem known as the Unique Games problem is NP-hard. This conjecture gained tremendous traction in the community because of it's numerous consequences (e.g. \cite{Khot02:unique, KhotKMO04, Raghavendra08})) and connections to various other fields such as metric geometry \cite{KhotV05} and discrete Fourier analysis \cite{KR03}. An exciting array of recent works \cite{dinur2018towards, barak2018small, subhash2018pseudorandom} has shown that a problem closely related to unique games, known as $2$-to-$2$ games, is NP-hard. This is an important step towards proving the UGC and offers evidence that the UGC is true.

On the algorithmic side, there have been various attempts (see for e.g. \cite{T05:unique, CharikarMM06, arora2015subexponential}) to disprove the UGC. In particular, Barak et al. \cite{barak2012hypercontractivity} showed that degree $8$ SoS can efficiently solve integrality gap instances of the Unique Games problem that were proposed for linear programs and SDPs considered earlier. This work caused significant interest in the community, since it suggests that SoS might be a way to refute the UGC.

Therefore, it's tremendously important to understand the performance of SoS on the unique games problem. A good first step would be to understand the performance of SoS for the problem of maximum cut, which is a special case of the Unique Games problem. In fact, we can be even more concrete and ask for the performance of SoS for the problem of maximum cut on random graphs, more precisely \cref{conj: sos_for_max_cut}. Lower bounds were shown for degree $2$ and degree $4$ in \cite{MS16, mohanty2020lifting} and generalizing their analyses for higher degree SoS is a nontrivial but important open problem.

\section{Low degree likelihood ratio hypothesis}

As explained in \cref{chap: sos}, the low-degree likelihood ratio hypothesis analytically predicts the computational barriers for hypothesis testing in bounded time, for \textit{sufficiently nice} distributions. See \cite{hop18, kunisky19notes, holmgren2020counterexamples} and references therein for more details. A full proof of this hypothesis is beyond current techniques, since it's likely harder than proving say $P \neq NP$. Despite this, confirming the hypothesis in restricted proof systems is a fascinating and important field for future research. In particular, building on the notation from \cref{chap: sos}, we would like to prove that for sufficiently nice distributions $\nu, \mu$, after pseudo-calibrating, if $\pE[1] = 1 + o(1)$, then there exists an SoS lower bound. Indeed, in this work, we confirm this for several fundamental problems. Proving this in general will go a long way towards understanding the power of bounded-time algorithms.

\section{Technical improvements}

Having covered some general directions for future research, we now specify a few directions for improving some technical aspects of our results.

\subsubsection{Improving parameter dependences}

In many of our lower bounds, we require polynomial decay in the Fourier coefficients. For example, we require a decay of $n^{\eps}$ for each new Fourier character, where $n$ is the input size. This is done to handle various other factors that appear in norm bounds when doing the charging arguments. In the proofs, we term these as vertex or edge decay, corresponding to how they are encoded in the graph matrix arguments we use.
By doing this, we obtain a slightly weaker lower bound. For example, instead of getting a $n^{1/4}$ lower bound (up to polylogarithmic factors) for Tensor PCA, we obtain a $n^{1/4 - \epsilon}$ lower bound for any $\eps > 0$. In general, while they facilitate the proof, it's not clear that this sort of decay is necessary and it's open to find a tighter analysis so as to close the gap from known upper bounds up to a polylogarithmic factor.

Related to the above discussion, another open problem is to push the degree of SoS higher in our lower bounds. For example, in the Sherrington-Kirkpatrick lower bound, it's open to push the SoS degree from $n^{\eps}$ to $\widetilde{\Omega}(n)$. Our current techniques do not handle this but we expect the lower bound to nevertheless hold.

\subsubsection{Satisfying constraints exactly}

In some of our lower bounds, our planted distributions only approximately satisfy constraints such as having a subgraph of size $k$, having a unit vector $u$, and having $u$ be $k$-sparse. While we would like to use planted distributions which satisfy such constraints exactly, the moment matrix becomes much harder to analyze.

We do resolve it for the Sherrington-Kirkpatrick lower bound by using a rounding technique \cite{ghosh2020sum}.
This same issue also appeared in the SoS lower bounds for planted clique \cite{BHKKMP16}, which was fixed in a recent paper by Pang \cite{Pang21}.
We leave it to future work to resolve this in general.

\section{Implications for Computer Science}

As we saw in the introduction, the current state of affairs in Theoretical Computer Science research seems to be to understand the limits of computation for various problems. Even though there maybe potential ultimate goals such as settling the P vs NP problem or even relatively modest goals such as settling the Unique Games Conjecture, there's much to be learnt and uncovered from this process. For example, for various problems, there seems to be a discernible gap between what's information theoretically possible and what's computationally feasible. Our work adds insight into this intriguing phenomenon, known as the information-computation tradeoff.
However, there are also other questions that need answering. For example, what makes certification seemingly harder than estimation or recovery? Can we characterize the precise property of problems that potentially make them hard or easy for various classes of algorithms such as Sum of Squares?
While much rich structure is slowly being uncovered in this general pursuit, a proper understanding still eludes us.
However, applications of what we've discovered so far, both technically and philosophically, are already numerous in various branches of mathematics and science, therefore research in this field is more than for the sake of mere curiosity.
We hope our work serves as a meaningful progress towards this grand goal.

\makebibliography


\end{document}